%% file: main.tex
\PassOptionsToPackage{capitalize,nameinlink}{cleveref}
\PassOptionsToPackage{usenames,dvipsnames}{color}
\PassOptionsToPackage{usenames,table}{xcolor}
\PassOptionsToPackage{final}{microtype}
\PassOptionsToPackage{scaled}{inconsolata}

\documentclass[sigconf, nonacm]{acmart}

\newcommand{\name}{{Taurus}\xspace}
\newcommand{\algoSize}{\small}
\newcommand{\codename}{\name}

\newcommand{\paperTitle}{\codename: Lightweight Parallel Logging for\texorpdfstring{\\}{ }In-Memory Database Management Systems (Extended Version)}
\newcommand{\paperShortTitle}{\codename: Lightweight Parallel Logging for In-Memory Database Management Systems}

\usepackage{circledsteps}
\usepackage{tikz}
\newcommand*\circled[1]{\tikz[baseline=(char.base)]{
            \node[shape=circle,draw,inner sep=0.8pt] (char) {\tiny #1};}}

\usepackage{graphicx}
\usepackage{balance}  %

\usepackage{xspace}
\usepackage{booktabs}   %
\usepackage{etex}
\usepackage[linesnumbered, ruled, vlined]{algorithm2e}
\usepackage{xcolor}

\usepackage{tabularx}
\usepackage{float}
\usepackage{subfig}
\usepackage{framed}
\usepackage{caption}
\usepackage{soul}
\usepackage{paralist}
\usepackage{listings}
\usepackage[export]{adjustbox}
\usepackage{boxedminipage}
\usepackage[noend]{algorithmic}
\usepackage{color}
\usepackage{float}
\usepackage{cleveref}
\usepackage{placeins}
\renewcommand{\subparagraph}{}
\usepackage{titlesec}
\usepackage{tikz}
\usetikzlibrary{tikzmark,calc}
\newcommand{\globalGraphScale}{0.296}
\newcommand\DrawBox[3][]{%
  \begin{tikzpicture}[remember picture,overlay]
    \draw[fill=gray!10,dashed,#1] 
    ([xshift=-0.5em,yshift=2.1ex]{pic cs:#2}) 
    rectangle 
    ([xshift=2pt,yshift=-1.1ex]pic cs:#3);
  \end{tikzpicture}%
}

\usepackage[T1]{fontenc}
\usepackage{algorithmic}
\usepackage{array}
\usepackage{eqparbox}

\let\oldnl\nl
\newcommand{\nonl}{\renewcommand{\nl}{\let\nl\oldnl}}

\def\HiLi{\leavevmode\rlap{\hbox to 
\hsize{\color[rgb]{.7,.7,.7}\leaders\hrule height .8\baselineskip 
depth .5ex\hfill}}}

\Crefname{section}{Sec.}{Secs.}
\Crefname{subsection}{Sec.}{Secs.}
\Crefname{algorithm}{Alg.}{Algs.}

\definecolor{todo-color}{rgb}{1,0,0}
\definecolor{comment-color}{rgb}{0.5,0.1,0.1}
\newcommand{\codeComment}[1]{\textnormal{\color{comment-color}{\textit{\# 
#1}}}\unskip}

\newcommand{\dbLV}{LV\xspace}
\newcommand{\myitem}[1]{\textit{\sf #1}}

\newcommand{\RLV}{{\textit{RLV}}\xspace}
\newcommand{\ELV}{{\textit{ELV}}\xspace}
\newcommand{\PLV}{{\textit{PLV}}\xspace}
\newcommand{\PLVs}{{\textit{PLV}s}\xspace}
\newcommand{\LPLV}{{\textit{LPLV}}\xspace}

\newcommand{\writeLV}{{\textit{writeLV}}\xspace}
\newcommand{\readLV}{{\textit{readLV}}\xspace}
\newcommand{\LV}{{\textit{LV}}\xspace}
\newcommand{\LVs}{{\textit{LV}s}\xspace}
\newcommand{\writeset}{{\textit{writeSet}}\xspace}
\newcommand{\readset}{{\textit{readSet}}\xspace}
\newcommand{\elementWiseMax}{{\textit{ElemWiseMax}}\xspace}

\newcommand{\loglsn}{\textit{logLSN}\xspace}
\newcommand{\hide}[1]{}
\newcommand{\txn}[1]{\textit{T#1}\xspace} %
\newcommand{\logger}[1]{\textit{Log #1}}

\usepackage[backgroundcolor=gray!30,linecolor=black,colorinlistoftodos]{todonotes}

\soulregister\cite7
\soulregister\ref7
\soulregister\cref7
\soulregister\name7
\soulregister\wts7
\soulregister\rts7
\soulregister\cts7
\soulregister\label7
\soulregister\textbf7
\soulregister\texttf7
\soulregister\myfun7
\soulregister\url7
\soulregister\nonl7
\soulregister\codeComment7
\soulregister\textit7
\soulregister\txn7
\soulregister\LV7
\soulregister\readLV7
\soulregister\writeLV7
\soulregister\elementWiseMax7
\soulregister\If7
\soulregister\;7
\soulregister\in7

\newcommand{\myhighlightpar}[1]{#1}%
\newcommand{\myhighlight}[1]{#1}%

\newcommand{\tpccNewOrder}{New-Order\xspace}
\newcommand{\tpccPayment}{Payment\xspace}
\newcommand{\tpccStockLevel}{Stock-Level\xspace}
\newcommand{\tpccDelivery}{Delivery\xspace}
\newcommand{\tpccOrderStatus}{Order-Status\xspace}
\definecolor{inlineBG}{HTML}{F1F1F1}
\definecolor{light-yellow}{HTML}{efef00}

\newcommand{\conflictRAW}{RAW\xspace}
\newcommand{\conflictWAR}{WAR\xspace}
\newcommand{\conflictWAW}{WAW\xspace}

\newtheorem{property}{Property}%
\newtheorem{theorem}{Theorem}%

\begin{document}

\title[\paperShortTitle]{\paperTitle}

\newcommand{\mail}[1]{\href{mailto:#1}{#1}}

\newcommand{\superscript}[1]{\ensuremath{^{\textrm{#1}}}}
\def\xCMU{\superscript{$\bigstar$}}
\def\xWISC{\superscript{$\spadesuit$}}
\def\xMIT{\superscript{$\clubsuit$}}

\author{Yu Xia}
\affiliation{
  \institution{Massachusetts Institute of Technology}
}
\email{yuxia@mit.edu}

\author{Xiangyao Yu}
\affiliation{
  \institution{University of Wisconsin--Madison}
}
\email{yxy@cs.wisc.edu}

\author{Andrew Pavlo}
\affiliation{
  \institution{Carnegie Mellon University}
}
\email{pavlo@cs.cmu.edu}

\author{Srinivas Devadas}
\affiliation{
  \institution{Massachusetts Institute of Technology}
}
\email{devadas@mit.edu}

\newcommand{\yxy}[1]{\textcolor{blue}{YXY: #1}}
\newcommand{\yx}[1]{\textcolor{purple}{YX: #1}}
\newcommand{\red}[1]{\textcolor{red}{#1}}

\newcommand{\squishitemize}{
 \begin{list}{$\bullet$}
  { \setlength{\itemsep}{0pt}
     \setlength{\parsep}{2pt}
     \setlength{\topsep}{2pt}
     \setlength{\partopsep}{0pt}
     \setlength{\leftmargin}{1.95em}
     \setlength{\labelwidth}{1.5em}
     \setlength{\labelsep}{0.5em} } }

\newcounter{Lcount}
\newcommand{\squishlist}{
    \begin{list}{\arabic{Lcount}. }
   { \usecounter{Lcount}
        \setlength{\itemsep}{0pt}
        \setlength{\parsep}{3pt}
        \setlength{\topsep}{3pt}
        \setlength{\partopsep}{0pt}
        \setlength{\leftmargin}{2em}
        \setlength{\labelwidth}{1.5em}
        \setlength{\labelsep}{0.5em} } }

\newcommand{\squishend}{\end{list}}

\input{list-data}
\input{abstract}

\maketitle

\input{intro}

\input{background}
\input{protocol}

\input{discussion}

\input{evaluation}

\input{related}

\input{conclusion}
\FloatBarrier

\bibliographystyle{abbrv}
{\small \bibliography{refs,arch,misc,db}}
\appendix
\begin{appendix}
\input{appendix}
\end{appendix}

\end{document}

%% file: list-data.tex
\input{revision/figs/Command-Tm0-TPCC-Lr0-Number-of-Worker-Threads-vs-Throughput.tex}
\input{revision/figs/Command-Tm0-TPCC-Lr1-Number-of-Worker-Threads-vs-MaxThr.tex}
\input{revision/figs/Command-Tm1-TPCC-Lr0-Number-of-Worker-Threads-vs-Throughput.tex}
\input{revision/figs/Command-Tm1-TPCC-Lr1-Number-of-Worker-Threads-vs-MaxThr.tex}
\input{revision/figs/Command-YCSB-Lr0-Number-of-Worker-Threads-vs-Throughput.tex}
\input{revision/figs/Command-YCSB-Lr1-Number-of-Worker-Threads-vs-MaxThr.tex}
\input{revision/figs/contention-hddYCSB-Lr0-Zipfian-Theta-vs-MaxThr.tex}
\input{revision/figs/Data-Tm0-TPCC-Lr0-Number-of-Worker-Threads-vs-Throughput.tex}
\input{revision/figs/Data-Tm0-TPCC-Lr1-Number-of-Worker-Threads-vs-MaxThr.tex}
\input{revision/figs/Data-Tm1-TPCC-Lr0-Number-of-Worker-Threads-vs-Throughput.tex}
\input{revision/figs/Data-Tm1-TPCC-Lr1-Number-of-Worker-Threads-vs-MaxThr.tex}
\input{revision/figs/Data-YCSB-Lr0-Number-of-Worker-Threads-vs-Throughput.tex}
\input{revision/figs/Data-YCSB-Lr1-Number-of-Worker-Threads-vs-MaxThr.tex}
\input{revision/figs/Ln16-silo-Tm0-TPCC-Lr0-Number-of-Worker-Threads-vs-MaxThr.tex}
\input{revision/figs/Ln16-silo-Tm0-TPCC-Lr1-Number-of-Worker-Threads-vs-MaxThr.tex}
\input{revision/figs/Ln16-silo-Tm1-TPCC-Lr0-Number-of-Worker-Threads-vs-MaxThr.tex}
\input{revision/figs/Ln16-silo-Tm1-TPCC-Lr1-Number-of-Worker-Threads-vs-MaxThr.tex}
\input{revision/figs/Ln16-silo-YCSB-Lr0-Number-of-Worker-Threads-vs-MaxThr.tex}
\input{revision/figs/Ln16-silo-YCSB-Lr1-Number-of-Worker-Threads-vs-MaxThr.tex}
\input{revision/figs/Ln16-Tm0-TPCC-Lr0-Number-of-Worker-Threads-vs-MaxThr.tex}
\input{revision/figs/Ln16-Tm0-TPCC-Lr1-Number-of-Worker-Threads-vs-MaxThr.tex}
\input{revision/figs/Ln16-Tm1-TPCC-Lr0-Number-of-Worker-Threads-vs-MaxThr.tex}
\input{revision/figs/Ln16-Tm1-TPCC-Lr1-Number-of-Worker-Threads-vs-MaxThr.tex}
\input{revision/figs/Ln16-YCSB-Lr0-Number-of-Worker-Threads-vs-MaxThr.tex}
\input{revision/figs/Ln16-YCSB-Lr1-Number-of-Worker-Threads-vs-MaxThr.tex}
\input{revision/figs/RAM-YCSB-Lr0-Number-of-Worker-Threads-vs-Throughput.tex}
\input{revision/figs/RAM-YCSB-Lr1-Number-of-Worker-Threads-vs-Throughput.tex}
\input{revision/figs/Tm0-TPCC-Lr0-Number-of-Worker-Threads-vs-Throughput.tex}
\input{revision/figs/Tm1-TPCC-Lr0-Number-of-Worker-Threads-vs-Throughput.tex}
\input{revision/figs/TPCF-Lr0-Number-of-Worker-Threads-vs-Throughput.tex}
\input{revision/figs/TPCF-Lr1-Number-of-Worker-Threads-vs-Throughput.tex}
\input{revision/figs/YCSB-Lr0-Number-of-Worker-Threads-vs-Throughput.tex}
\input{revision/figs/YCSB-Lr1-Number-of-Worker-Threads-vs-MaxThr.tex}

%% file: revision/figs/Command-Tm0-TPCC-Lr0-Number-of-Worker-Threads-vs-Throughput.tex
\newcommand{\mainCommandNewOrderTPCCLogNumberofWorkerThreadsvsThroughputScalabilityNoLoggingDataNOWAIT}{0.0\%\xspace}
\newcommand{\mainCommandNewOrderTPCCLogNumberofWorkerThreadsvsThroughputScalabilityTaurusCommandNOWAIT}{$27.9\times$\xspace}
\newcommand{\mainCommandNewOrderTPCCLogNumberofWorkerThreadsvsThroughputScalabilityTaurusDataNOWAIT}{0.0\%\xspace}
\newcommand{\mainCommandNewOrderTPCCLogNumberofWorkerThreadsvsThroughputCompareTaurusCommandNOWAITOverNoLoggingDataNOWAIT}{$2.8\times$\xspace}
\newcommand{\mainCommandNewOrderTPCCLogNumberofWorkerThreadsvsThroughputCompareTaurusCommandNOWAITOverNoLoggingDataNOWAITMax}{$2.8\times$\xspace}
\newcommand{\mainCommandNewOrderTPCCLogNumberofWorkerThreadsvsThroughputCompareTaurusCommandNOWAITOverSerialCommandNOWAIT}{$6.2\times$\xspace}
\newcommand{\mainCommandNewOrderTPCCLogNumberofWorkerThreadsvsThroughputCompareTaurusCommandNOWAITOverSerialCommandNOWAITMax}{$6.1\times$\xspace}
\newcommand{\mainCommandNewOrderTPCCLogNumberofWorkerThreadsvsThroughputCompareTaurusCommandNOWAITOverSerialDataNOWAIT}{$2.8\times$\xspace}
\newcommand{\mainCommandNewOrderTPCCLogNumberofWorkerThreadsvsThroughputCompareTaurusCommandNOWAITOverSerialDataNOWAITMax}{$2.8\times$\xspace}
\newcommand{\mainCommandNewOrderTPCCLogNumberofWorkerThreadsvsThroughputCompareTaurusCommandNOWAITOverPloverDataNOWAIT}{$2.8\times$\xspace}
\newcommand{\mainCommandNewOrderTPCCLogNumberofWorkerThreadsvsThroughputCompareTaurusCommandNOWAITOverPloverDataNOWAITMax}{$2.8\times$\xspace}
\newcommand{\mainCommandNewOrderTPCCLogNumberofWorkerThreadsvsThroughputCompareTaurusCommandNOWAITOverSerialRAIDCommandNOWAIT}{5.1\%\xspace}
\newcommand{\mainCommandNewOrderTPCCLogNumberofWorkerThreadsvsThroughputCompareTaurusCommandNOWAITOverSerialRAIDCommandNOWAITMax}{5.1\%\xspace}
\newcommand{\mainCommandNewOrderTPCCLogNumberofWorkerThreadsvsThroughputCompareTaurusDataNOWAITOverSerialDataNOWAIT}{0.0\%\xspace}
\newcommand{\mainCommandNewOrderTPCCLogNumberofWorkerThreadsvsThroughputCompareTaurusDataNOWAITOverSerialDataNOWAITMax}{0.0\%\xspace}
\newcommand{\mainCommandNewOrderTPCCLogNumberofWorkerThreadsvsThroughputCompareTaurusDataNOWAITOverPloverDataNOWAIT}{0.0\%\xspace}
\newcommand{\mainCommandNewOrderTPCCLogNumberofWorkerThreadsvsThroughputCompareTaurusDataNOWAITOverPloverDataNOWAITMax}{0.0\%\xspace}
\newcommand{\mainCommandNewOrderTPCCLogNumberofWorkerThreadsvsThroughputCompareTaurusDataNOWAITOverSerialRAIDDataNOWAIT}{0.0\%\xspace}
\newcommand{\mainCommandNewOrderTPCCLogNumberofWorkerThreadsvsThroughputCompareTaurusDataNOWAITOverSerialRAIDDataNOWAITMax}{0.0\%\xspace}
\newcommand{\mainCommandNewOrderTPCCLogNumberofWorkerThreadsvsThroughputCompareTaurusCommandNOWAITOverTaurusDataNOWAIT}{$2.8\times$\xspace}
\newcommand{\mainCommandNewOrderTPCCLogNumberofWorkerThreadsvsThroughputCompareTaurusCommandNOWAITOverTaurusDataNOWAITMax}{$2.8\times$\xspace}
\newcommand{\mainCommandNewOrderTPCCLogNumberofWorkerThreadsvsThroughputCompareTaurusCommandNOWAITOverSiloRDataSilo}{$2.8\times$\xspace}
\newcommand{\mainCommandNewOrderTPCCLogNumberofWorkerThreadsvsThroughputCompareTaurusCommandNOWAITOverSiloRDataSiloMax}{$2.8\times$\xspace}
\newcommand{\mainCommandNewOrderTPCCLogNumberofWorkerThreadsvsThroughputCompareTaurusDataNOWAITOverSiloRDataSilo}{0.0\%\xspace}
\newcommand{\mainCommandNewOrderTPCCLogNumberofWorkerThreadsvsThroughputCompareTaurusDataNOWAITOverSiloRDataSiloMax}{0.0\%\xspace}

%% file: revision/figs/Command-Tm0-TPCC-Lr1-Number-of-Worker-Threads-vs-MaxThr.tex
\newcommand{\mainCommandNewOrderTPCCRecNumberofWorkerThreadsvsMaxThrScalabilityNoLoggingDataNOWAIT}{0.0\%\xspace}
\newcommand{\mainCommandNewOrderTPCCRecNumberofWorkerThreadsvsMaxThrScalabilityTaurusCommandNOWAIT}{$10.4\times$\xspace}
\newcommand{\mainCommandNewOrderTPCCRecNumberofWorkerThreadsvsMaxThrScalabilityTaurusDataNOWAIT}{0.0\%\xspace}
\newcommand{\mainCommandNewOrderTPCCRecNumberofWorkerThreadsvsMaxThrCompareTaurusCommandNOWAITOverNoLoggingDataNOWAIT}{$2.2\times$\xspace}
\newcommand{\mainCommandNewOrderTPCCRecNumberofWorkerThreadsvsMaxThrCompareTaurusCommandNOWAITOverNoLoggingDataNOWAITMax}{$2.2\times$\xspace}
\newcommand{\mainCommandNewOrderTPCCRecNumberofWorkerThreadsvsMaxThrCompareTaurusCommandNOWAITOverSerialCommandNOWAIT}{$17.5\times$\xspace}
\newcommand{\mainCommandNewOrderTPCCRecNumberofWorkerThreadsvsMaxThrCompareTaurusCommandNOWAITOverSerialCommandNOWAITMax}{$17.5\times$\xspace}
\newcommand{\mainCommandNewOrderTPCCRecNumberofWorkerThreadsvsMaxThrCompareTaurusCommandNOWAITOverSerialDataNOWAIT}{$2.2\times$\xspace}
\newcommand{\mainCommandNewOrderTPCCRecNumberofWorkerThreadsvsMaxThrCompareTaurusCommandNOWAITOverSerialDataNOWAITMax}{$2.2\times$\xspace}
\newcommand{\mainCommandNewOrderTPCCRecNumberofWorkerThreadsvsMaxThrCompareTaurusCommandNOWAITOverPloverDataNOWAIT}{$2.2\times$\xspace}
\newcommand{\mainCommandNewOrderTPCCRecNumberofWorkerThreadsvsMaxThrCompareTaurusCommandNOWAITOverPloverDataNOWAITMax}{$2.2\times$\xspace}
\newcommand{\mainCommandNewOrderTPCCRecNumberofWorkerThreadsvsMaxThrCompareTaurusCommandNOWAITOverSerialRAIDCommandNOWAIT}{$17.6\times$\xspace}
\newcommand{\mainCommandNewOrderTPCCRecNumberofWorkerThreadsvsMaxThrCompareTaurusCommandNOWAITOverSerialRAIDCommandNOWAITMax}{$17.6\times$\xspace}
\newcommand{\mainCommandNewOrderTPCCRecNumberofWorkerThreadsvsMaxThrCompareTaurusDataNOWAITOverSerialDataNOWAIT}{0.0\%\xspace}
\newcommand{\mainCommandNewOrderTPCCRecNumberofWorkerThreadsvsMaxThrCompareTaurusDataNOWAITOverSerialDataNOWAITMax}{0.0\%\xspace}
\newcommand{\mainCommandNewOrderTPCCRecNumberofWorkerThreadsvsMaxThrCompareTaurusDataNOWAITOverPloverDataNOWAIT}{0.0\%\xspace}
\newcommand{\mainCommandNewOrderTPCCRecNumberofWorkerThreadsvsMaxThrCompareTaurusDataNOWAITOverPloverDataNOWAITMax}{0.0\%\xspace}
\newcommand{\mainCommandNewOrderTPCCRecNumberofWorkerThreadsvsMaxThrCompareTaurusDataNOWAITOverSerialRAIDDataNOWAIT}{0.0\%\xspace}
\newcommand{\mainCommandNewOrderTPCCRecNumberofWorkerThreadsvsMaxThrCompareTaurusDataNOWAITOverSerialRAIDDataNOWAITMax}{0.0\%\xspace}
\newcommand{\mainCommandNewOrderTPCCRecNumberofWorkerThreadsvsMaxThrCompareTaurusCommandNOWAITOverTaurusDataNOWAIT}{$2.2\times$\xspace}
\newcommand{\mainCommandNewOrderTPCCRecNumberofWorkerThreadsvsMaxThrCompareTaurusCommandNOWAITOverTaurusDataNOWAITMax}{$2.2\times$\xspace}
\newcommand{\mainCommandNewOrderTPCCRecNumberofWorkerThreadsvsMaxThrCompareTaurusCommandNOWAITOverSiloRDataSilo}{$2.2\times$\xspace}
\newcommand{\mainCommandNewOrderTPCCRecNumberofWorkerThreadsvsMaxThrCompareTaurusCommandNOWAITOverSiloRDataSiloMax}{$2.2\times$\xspace}
\newcommand{\mainCommandNewOrderTPCCRecNumberofWorkerThreadsvsMaxThrCompareTaurusDataNOWAITOverSiloRDataSilo}{0.0\%\xspace}
\newcommand{\mainCommandNewOrderTPCCRecNumberofWorkerThreadsvsMaxThrCompareTaurusDataNOWAITOverSiloRDataSiloMax}{0.0\%\xspace}

%% file: revision/figs/Command-Tm1-TPCC-Lr0-Number-of-Worker-Threads-vs-Throughput.tex
\newcommand{\mainCommandPayTPCCLogNumberofWorkerThreadsvsThroughputScalabilityNoLoggingDataNOWAIT}{0.0\%\xspace}
\newcommand{\mainCommandPayTPCCLogNumberofWorkerThreadsvsThroughputScalabilityTaurusCommandNOWAIT}{$12.0\times$\xspace}
\newcommand{\mainCommandPayTPCCLogNumberofWorkerThreadsvsThroughputScalabilityTaurusDataNOWAIT}{0.0\%\xspace}
\newcommand{\mainCommandPayTPCCLogNumberofWorkerThreadsvsThroughputCompareTaurusCommandNOWAITOverNoLoggingDataNOWAIT}{$6.8\times$\xspace}
\newcommand{\mainCommandPayTPCCLogNumberofWorkerThreadsvsThroughputCompareTaurusCommandNOWAITOverNoLoggingDataNOWAITMax}{$7.0\times$\xspace}
\newcommand{\mainCommandPayTPCCLogNumberofWorkerThreadsvsThroughputCompareTaurusCommandNOWAITOverSerialCommandNOWAIT}{$6.8\times$\xspace}
\newcommand{\mainCommandPayTPCCLogNumberofWorkerThreadsvsThroughputCompareTaurusCommandNOWAITOverSerialCommandNOWAITMax}{$6.8\times$\xspace}
\newcommand{\mainCommandPayTPCCLogNumberofWorkerThreadsvsThroughputCompareTaurusCommandNOWAITOverSerialDataNOWAIT}{$6.8\times$\xspace}
\newcommand{\mainCommandPayTPCCLogNumberofWorkerThreadsvsThroughputCompareTaurusCommandNOWAITOverSerialDataNOWAITMax}{$7.0\times$\xspace}
\newcommand{\mainCommandPayTPCCLogNumberofWorkerThreadsvsThroughputCompareTaurusCommandNOWAITOverPloverDataNOWAIT}{$6.8\times$\xspace}
\newcommand{\mainCommandPayTPCCLogNumberofWorkerThreadsvsThroughputCompareTaurusCommandNOWAITOverPloverDataNOWAITMax}{$7.0\times$\xspace}
\newcommand{\mainCommandPayTPCCLogNumberofWorkerThreadsvsThroughputCompareTaurusCommandNOWAITOverSerialRAIDCommandNOWAIT}{$1.2\times$\xspace}
\newcommand{\mainCommandPayTPCCLogNumberofWorkerThreadsvsThroughputCompareTaurusCommandNOWAITOverSerialRAIDCommandNOWAITMax}{$1.1\times$\xspace}
\newcommand{\mainCommandPayTPCCLogNumberofWorkerThreadsvsThroughputCompareTaurusDataNOWAITOverSerialDataNOWAIT}{0.0\%\xspace}
\newcommand{\mainCommandPayTPCCLogNumberofWorkerThreadsvsThroughputCompareTaurusDataNOWAITOverSerialDataNOWAITMax}{0.0\%\xspace}
\newcommand{\mainCommandPayTPCCLogNumberofWorkerThreadsvsThroughputCompareTaurusDataNOWAITOverPloverDataNOWAIT}{0.0\%\xspace}
\newcommand{\mainCommandPayTPCCLogNumberofWorkerThreadsvsThroughputCompareTaurusDataNOWAITOverPloverDataNOWAITMax}{0.0\%\xspace}
\newcommand{\mainCommandPayTPCCLogNumberofWorkerThreadsvsThroughputCompareTaurusDataNOWAITOverSerialRAIDDataNOWAIT}{0.0\%\xspace}
\newcommand{\mainCommandPayTPCCLogNumberofWorkerThreadsvsThroughputCompareTaurusDataNOWAITOverSerialRAIDDataNOWAITMax}{0.0\%\xspace}
\newcommand{\mainCommandPayTPCCLogNumberofWorkerThreadsvsThroughputCompareTaurusCommandNOWAITOverTaurusDataNOWAIT}{$6.8\times$\xspace}
\newcommand{\mainCommandPayTPCCLogNumberofWorkerThreadsvsThroughputCompareTaurusCommandNOWAITOverTaurusDataNOWAITMax}{$7.0\times$\xspace}
\newcommand{\mainCommandPayTPCCLogNumberofWorkerThreadsvsThroughputCompareTaurusCommandNOWAITOverSiloRDataSilo}{$6.8\times$\xspace}
\newcommand{\mainCommandPayTPCCLogNumberofWorkerThreadsvsThroughputCompareTaurusCommandNOWAITOverSiloRDataSiloMax}{$7.0\times$\xspace}
\newcommand{\mainCommandPayTPCCLogNumberofWorkerThreadsvsThroughputCompareTaurusDataNOWAITOverSiloRDataSilo}{0.0\%\xspace}
\newcommand{\mainCommandPayTPCCLogNumberofWorkerThreadsvsThroughputCompareTaurusDataNOWAITOverSiloRDataSiloMax}{0.0\%\xspace}

%% file: revision/figs/Command-Tm1-TPCC-Lr1-Number-of-Worker-Threads-vs-MaxThr.tex
\newcommand{\mainCommandPayTPCCRecNumberofWorkerThreadsvsMaxThrScalabilityNoLoggingDataNOWAIT}{0.0\%\xspace}
\newcommand{\mainCommandPayTPCCRecNumberofWorkerThreadsvsMaxThrScalabilityTaurusCommandNOWAIT}{$7.4\times$\xspace}
\newcommand{\mainCommandPayTPCCRecNumberofWorkerThreadsvsMaxThrScalabilityTaurusDataNOWAIT}{0.0\%\xspace}
\newcommand{\mainCommandPayTPCCRecNumberofWorkerThreadsvsMaxThrCompareTaurusCommandNOWAITOverNoLoggingDataNOWAIT}{$4.7\times$\xspace}
\newcommand{\mainCommandPayTPCCRecNumberofWorkerThreadsvsMaxThrCompareTaurusCommandNOWAITOverNoLoggingDataNOWAITMax}{$5.3\times$\xspace}
\newcommand{\mainCommandPayTPCCRecNumberofWorkerThreadsvsMaxThrCompareTaurusCommandNOWAITOverSerialCommandNOWAIT}{$7.1\times$\xspace}
\newcommand{\mainCommandPayTPCCRecNumberofWorkerThreadsvsMaxThrCompareTaurusCommandNOWAITOverSerialCommandNOWAITMax}{$8.0\times$\xspace}
\newcommand{\mainCommandPayTPCCRecNumberofWorkerThreadsvsMaxThrCompareTaurusCommandNOWAITOverSerialDataNOWAIT}{$4.7\times$\xspace}
\newcommand{\mainCommandPayTPCCRecNumberofWorkerThreadsvsMaxThrCompareTaurusCommandNOWAITOverSerialDataNOWAITMax}{$5.3\times$\xspace}
\newcommand{\mainCommandPayTPCCRecNumberofWorkerThreadsvsMaxThrCompareTaurusCommandNOWAITOverPloverDataNOWAIT}{$4.7\times$\xspace}
\newcommand{\mainCommandPayTPCCRecNumberofWorkerThreadsvsMaxThrCompareTaurusCommandNOWAITOverPloverDataNOWAITMax}{$5.3\times$\xspace}
\newcommand{\mainCommandPayTPCCRecNumberofWorkerThreadsvsMaxThrCompareTaurusCommandNOWAITOverSerialRAIDCommandNOWAIT}{$7.3\times$\xspace}
\newcommand{\mainCommandPayTPCCRecNumberofWorkerThreadsvsMaxThrCompareTaurusCommandNOWAITOverSerialRAIDCommandNOWAITMax}{$8.1\times$\xspace}
\newcommand{\mainCommandPayTPCCRecNumberofWorkerThreadsvsMaxThrCompareTaurusDataNOWAITOverSerialDataNOWAIT}{0.0\%\xspace}
\newcommand{\mainCommandPayTPCCRecNumberofWorkerThreadsvsMaxThrCompareTaurusDataNOWAITOverSerialDataNOWAITMax}{0.0\%\xspace}
\newcommand{\mainCommandPayTPCCRecNumberofWorkerThreadsvsMaxThrCompareTaurusDataNOWAITOverPloverDataNOWAIT}{0.0\%\xspace}
\newcommand{\mainCommandPayTPCCRecNumberofWorkerThreadsvsMaxThrCompareTaurusDataNOWAITOverPloverDataNOWAITMax}{0.0\%\xspace}
\newcommand{\mainCommandPayTPCCRecNumberofWorkerThreadsvsMaxThrCompareTaurusDataNOWAITOverSerialRAIDDataNOWAIT}{0.0\%\xspace}
\newcommand{\mainCommandPayTPCCRecNumberofWorkerThreadsvsMaxThrCompareTaurusDataNOWAITOverSerialRAIDDataNOWAITMax}{0.0\%\xspace}
\newcommand{\mainCommandPayTPCCRecNumberofWorkerThreadsvsMaxThrCompareTaurusCommandNOWAITOverTaurusDataNOWAIT}{$4.7\times$\xspace}
\newcommand{\mainCommandPayTPCCRecNumberofWorkerThreadsvsMaxThrCompareTaurusCommandNOWAITOverTaurusDataNOWAITMax}{$5.3\times$\xspace}
\newcommand{\mainCommandPayTPCCRecNumberofWorkerThreadsvsMaxThrCompareTaurusCommandNOWAITOverSiloRDataSilo}{$4.7\times$\xspace}
\newcommand{\mainCommandPayTPCCRecNumberofWorkerThreadsvsMaxThrCompareTaurusCommandNOWAITOverSiloRDataSiloMax}{$5.3\times$\xspace}
\newcommand{\mainCommandPayTPCCRecNumberofWorkerThreadsvsMaxThrCompareTaurusDataNOWAITOverSiloRDataSilo}{0.0\%\xspace}
\newcommand{\mainCommandPayTPCCRecNumberofWorkerThreadsvsMaxThrCompareTaurusDataNOWAITOverSiloRDataSiloMax}{0.0\%\xspace}

%% file: revision/figs/Command-YCSB-Lr0-Number-of-Worker-Threads-vs-Throughput.tex
\newcommand{\mainCommandYCSBLogNumberofWorkerThreadsvsThroughputScalabilityNoLoggingDataNOWAIT}{0.0\%\xspace}
\newcommand{\mainCommandYCSBLogNumberofWorkerThreadsvsThroughputScalabilityTaurusCommandNOWAIT}{$38.0\times$\xspace}
\newcommand{\mainCommandYCSBLogNumberofWorkerThreadsvsThroughputScalabilityTaurusDataNOWAIT}{0.0\%\xspace}
\newcommand{\mainCommandYCSBLogNumberofWorkerThreadsvsThroughputCompareTaurusCommandNOWAITOverNoLoggingDataNOWAIT}{$12.0\times$\xspace}
\newcommand{\mainCommandYCSBLogNumberofWorkerThreadsvsThroughputCompareTaurusCommandNOWAITOverNoLoggingDataNOWAITMax}{$12.0\times$\xspace}
\newcommand{\mainCommandYCSBLogNumberofWorkerThreadsvsThroughputCompareTaurusCommandNOWAITOverSerialCommandNOWAIT}{$3.8\times$\xspace}
\newcommand{\mainCommandYCSBLogNumberofWorkerThreadsvsThroughputCompareTaurusCommandNOWAITOverSerialCommandNOWAITMax}{$3.0\times$\xspace}
\newcommand{\mainCommandYCSBLogNumberofWorkerThreadsvsThroughputCompareTaurusCommandNOWAITOverSerialDataNOWAIT}{$12.0\times$\xspace}
\newcommand{\mainCommandYCSBLogNumberofWorkerThreadsvsThroughputCompareTaurusCommandNOWAITOverSerialDataNOWAITMax}{$12.0\times$\xspace}
\newcommand{\mainCommandYCSBLogNumberofWorkerThreadsvsThroughputCompareTaurusCommandNOWAITOverPloverDataNOWAIT}{$12.0\times$\xspace}
\newcommand{\mainCommandYCSBLogNumberofWorkerThreadsvsThroughputCompareTaurusCommandNOWAITOverPloverDataNOWAITMax}{$12.0\times$\xspace}
\newcommand{\mainCommandYCSBLogNumberofWorkerThreadsvsThroughputCompareTaurusCommandNOWAITOverSerialRAIDCommandNOWAIT}{$1.6\times$\xspace}
\newcommand{\mainCommandYCSBLogNumberofWorkerThreadsvsThroughputCompareTaurusCommandNOWAITOverSerialRAIDCommandNOWAITMax}{$1.6\times$\xspace}
\newcommand{\mainCommandYCSBLogNumberofWorkerThreadsvsThroughputCompareTaurusDataNOWAITOverSerialDataNOWAIT}{0.0\%\xspace}
\newcommand{\mainCommandYCSBLogNumberofWorkerThreadsvsThroughputCompareTaurusDataNOWAITOverSerialDataNOWAITMax}{0.0\%\xspace}
\newcommand{\mainCommandYCSBLogNumberofWorkerThreadsvsThroughputCompareTaurusDataNOWAITOverPloverDataNOWAIT}{0.0\%\xspace}
\newcommand{\mainCommandYCSBLogNumberofWorkerThreadsvsThroughputCompareTaurusDataNOWAITOverPloverDataNOWAITMax}{0.0\%\xspace}
\newcommand{\mainCommandYCSBLogNumberofWorkerThreadsvsThroughputCompareTaurusDataNOWAITOverSerialRAIDDataNOWAIT}{0.0\%\xspace}
\newcommand{\mainCommandYCSBLogNumberofWorkerThreadsvsThroughputCompareTaurusDataNOWAITOverSerialRAIDDataNOWAITMax}{0.0\%\xspace}
\newcommand{\mainCommandYCSBLogNumberofWorkerThreadsvsThroughputCompareTaurusCommandNOWAITOverTaurusDataNOWAIT}{$12.0\times$\xspace}
\newcommand{\mainCommandYCSBLogNumberofWorkerThreadsvsThroughputCompareTaurusCommandNOWAITOverTaurusDataNOWAITMax}{$12.0\times$\xspace}
\newcommand{\mainCommandYCSBLogNumberofWorkerThreadsvsThroughputCompareTaurusCommandNOWAITOverSiloRDataSilo}{$12.0\times$\xspace}
\newcommand{\mainCommandYCSBLogNumberofWorkerThreadsvsThroughputCompareTaurusCommandNOWAITOverSiloRDataSiloMax}{$12.0\times$\xspace}
\newcommand{\mainCommandYCSBLogNumberofWorkerThreadsvsThroughputCompareTaurusDataNOWAITOverSiloRDataSilo}{0.0\%\xspace}
\newcommand{\mainCommandYCSBLogNumberofWorkerThreadsvsThroughputCompareTaurusDataNOWAITOverSiloRDataSiloMax}{0.0\%\xspace}

%% file: revision/figs/Command-YCSB-Lr1-Number-of-Worker-Threads-vs-MaxThr.tex
\newcommand{\mainCommandYCSBRecNumberofWorkerThreadsvsMaxThrScalabilityNoLoggingDataNOWAIT}{0.0\%\xspace}
\newcommand{\mainCommandYCSBRecNumberofWorkerThreadsvsMaxThrScalabilityTaurusCommandNOWAIT}{$15.5\times$\xspace}
\newcommand{\mainCommandYCSBRecNumberofWorkerThreadsvsMaxThrScalabilityTaurusDataNOWAIT}{0.0\%\xspace}
\newcommand{\mainCommandYCSBRecNumberofWorkerThreadsvsMaxThrCompareTaurusCommandNOWAITOverNoLoggingDataNOWAIT}{$5.3\times$\xspace}
\newcommand{\mainCommandYCSBRecNumberofWorkerThreadsvsMaxThrCompareTaurusCommandNOWAITOverNoLoggingDataNOWAITMax}{$5.3\times$\xspace}
\newcommand{\mainCommandYCSBRecNumberofWorkerThreadsvsMaxThrCompareTaurusCommandNOWAITOverSerialCommandNOWAIT}{$11.3\times$\xspace}
\newcommand{\mainCommandYCSBRecNumberofWorkerThreadsvsMaxThrCompareTaurusCommandNOWAITOverSerialCommandNOWAITMax}{$11.3\times$\xspace}
\newcommand{\mainCommandYCSBRecNumberofWorkerThreadsvsMaxThrCompareTaurusCommandNOWAITOverSerialDataNOWAIT}{$5.3\times$\xspace}
\newcommand{\mainCommandYCSBRecNumberofWorkerThreadsvsMaxThrCompareTaurusCommandNOWAITOverSerialDataNOWAITMax}{$5.3\times$\xspace}
\newcommand{\mainCommandYCSBRecNumberofWorkerThreadsvsMaxThrCompareTaurusCommandNOWAITOverPloverDataNOWAIT}{$5.3\times$\xspace}
\newcommand{\mainCommandYCSBRecNumberofWorkerThreadsvsMaxThrCompareTaurusCommandNOWAITOverPloverDataNOWAITMax}{$5.3\times$\xspace}
\newcommand{\mainCommandYCSBRecNumberofWorkerThreadsvsMaxThrCompareTaurusCommandNOWAITOverSerialRAIDCommandNOWAIT}{$11.0\times$\xspace}
\newcommand{\mainCommandYCSBRecNumberofWorkerThreadsvsMaxThrCompareTaurusCommandNOWAITOverSerialRAIDCommandNOWAITMax}{$10.9\times$\xspace}
\newcommand{\mainCommandYCSBRecNumberofWorkerThreadsvsMaxThrCompareTaurusDataNOWAITOverSerialDataNOWAIT}{0.0\%\xspace}
\newcommand{\mainCommandYCSBRecNumberofWorkerThreadsvsMaxThrCompareTaurusDataNOWAITOverSerialDataNOWAITMax}{0.0\%\xspace}
\newcommand{\mainCommandYCSBRecNumberofWorkerThreadsvsMaxThrCompareTaurusDataNOWAITOverPloverDataNOWAIT}{0.0\%\xspace}
\newcommand{\mainCommandYCSBRecNumberofWorkerThreadsvsMaxThrCompareTaurusDataNOWAITOverPloverDataNOWAITMax}{0.0\%\xspace}
\newcommand{\mainCommandYCSBRecNumberofWorkerThreadsvsMaxThrCompareTaurusDataNOWAITOverSerialRAIDDataNOWAIT}{0.0\%\xspace}
\newcommand{\mainCommandYCSBRecNumberofWorkerThreadsvsMaxThrCompareTaurusDataNOWAITOverSerialRAIDDataNOWAITMax}{0.0\%\xspace}
\newcommand{\mainCommandYCSBRecNumberofWorkerThreadsvsMaxThrCompareTaurusCommandNOWAITOverTaurusDataNOWAIT}{$5.3\times$\xspace}
\newcommand{\mainCommandYCSBRecNumberofWorkerThreadsvsMaxThrCompareTaurusCommandNOWAITOverTaurusDataNOWAITMax}{$5.3\times$\xspace}
\newcommand{\mainCommandYCSBRecNumberofWorkerThreadsvsMaxThrCompareTaurusCommandNOWAITOverSiloRDataSilo}{$5.3\times$\xspace}
\newcommand{\mainCommandYCSBRecNumberofWorkerThreadsvsMaxThrCompareTaurusCommandNOWAITOverSiloRDataSiloMax}{$5.3\times$\xspace}
\newcommand{\mainCommandYCSBRecNumberofWorkerThreadsvsMaxThrCompareTaurusDataNOWAITOverSiloRDataSilo}{0.0\%\xspace}
\newcommand{\mainCommandYCSBRecNumberofWorkerThreadsvsMaxThrCompareTaurusDataNOWAITOverSiloRDataSiloMax}{0.0\%\xspace}

%% file: revision/figs/contention-hddYCSB-Lr0-Zipfian-Theta-vs-MaxThr.tex
\newcommand{\contentioncontentionhddYCSBLogZipfianThetavsMaxThrCompareTaurusCommandNOWAITOverNoLoggingDataNOWAIT}{31.4\%\xspace}
\newcommand{\contentioncontentionhddYCSBLogZipfianThetavsMaxThrCompareTaurusCommandNOWAITOverNoLoggingDataNOWAITMax}{36.8\%\xspace}

%% file: revision/figs/Data-Tm0-TPCC-Lr0-Number-of-Worker-Threads-vs-Throughput.tex
\newcommand{\mainDataNewOrderTPCCLogNumberofWorkerThreadsvsThroughputScalabilityNoLoggingDataNOWAIT}{0.0\%\xspace}
\newcommand{\mainDataNewOrderTPCCLogNumberofWorkerThreadsvsThroughputScalabilityTaurusCommandNOWAIT}{0.0\%\xspace}
\newcommand{\mainDataNewOrderTPCCLogNumberofWorkerThreadsvsThroughputScalabilityTaurusDataNOWAIT}{$8.1\times$\xspace}
\newcommand{\mainDataNewOrderTPCCLogNumberofWorkerThreadsvsThroughputCompareTaurusCommandNOWAITOverNoLoggingDataNOWAIT}{0.0\%\xspace}
\newcommand{\mainDataNewOrderTPCCLogNumberofWorkerThreadsvsThroughputCompareTaurusCommandNOWAITOverNoLoggingDataNOWAITMax}{0.0\%\xspace}
\newcommand{\mainDataNewOrderTPCCLogNumberofWorkerThreadsvsThroughputCompareTaurusCommandNOWAITOverSerialCommandNOWAIT}{0.0\%\xspace}
\newcommand{\mainDataNewOrderTPCCLogNumberofWorkerThreadsvsThroughputCompareTaurusCommandNOWAITOverSerialCommandNOWAITMax}{0.0\%\xspace}
\newcommand{\mainDataNewOrderTPCCLogNumberofWorkerThreadsvsThroughputCompareTaurusCommandNOWAITOverSerialDataNOWAIT}{$22.8\times$\xspace}
\newcommand{\mainDataNewOrderTPCCLogNumberofWorkerThreadsvsThroughputCompareTaurusCommandNOWAITOverSerialDataNOWAITMax}{$20.1\times$\xspace}
\newcommand{\mainDataNewOrderTPCCLogNumberofWorkerThreadsvsThroughputCompareTaurusCommandNOWAITOverPloverDataNOWAIT}{$4.0\times$\xspace}
\newcommand{\mainDataNewOrderTPCCLogNumberofWorkerThreadsvsThroughputCompareTaurusCommandNOWAITOverPloverDataNOWAITMax}{$3.6\times$\xspace}
\newcommand{\mainDataNewOrderTPCCLogNumberofWorkerThreadsvsThroughputCompareTaurusCommandNOWAITOverSerialRAIDCommandNOWAIT}{0.0\%\xspace}
\newcommand{\mainDataNewOrderTPCCLogNumberofWorkerThreadsvsThroughputCompareTaurusCommandNOWAITOverSerialRAIDCommandNOWAITMax}{0.0\%\xspace}
\newcommand{\mainDataNewOrderTPCCLogNumberofWorkerThreadsvsThroughputCompareTaurusDataNOWAITOverSerialDataNOWAIT}{$8.3\times$\xspace}
\newcommand{\mainDataNewOrderTPCCLogNumberofWorkerThreadsvsThroughputCompareTaurusDataNOWAITOverSerialDataNOWAITMax}{$7.5\times$\xspace}
\newcommand{\mainDataNewOrderTPCCLogNumberofWorkerThreadsvsThroughputCompareTaurusDataNOWAITOverPloverDataNOWAIT}{$1.4\times$\xspace}
\newcommand{\mainDataNewOrderTPCCLogNumberofWorkerThreadsvsThroughputCompareTaurusDataNOWAITOverPloverDataNOWAITMax}{$1.3\times$\xspace}
\newcommand{\mainDataNewOrderTPCCLogNumberofWorkerThreadsvsThroughputCompareTaurusDataNOWAITOverSerialRAIDDataNOWAIT}{$1.3\times$\xspace}
\newcommand{\mainDataNewOrderTPCCLogNumberofWorkerThreadsvsThroughputCompareTaurusDataNOWAITOverSerialRAIDDataNOWAITMax}{$1.3\times$\xspace}
\newcommand{\mainDataNewOrderTPCCLogNumberofWorkerThreadsvsThroughputCompareTaurusCommandNOWAITOverTaurusDataNOWAIT}{$2.7\times$\xspace}
\newcommand{\mainDataNewOrderTPCCLogNumberofWorkerThreadsvsThroughputCompareTaurusCommandNOWAITOverTaurusDataNOWAITMax}{$2.7\times$\xspace}
\newcommand{\mainDataNewOrderTPCCLogNumberofWorkerThreadsvsThroughputCompareTaurusCommandNOWAITOverSiloRDataSilo}{$2.7\times$\xspace}
\newcommand{\mainDataNewOrderTPCCLogNumberofWorkerThreadsvsThroughputCompareTaurusCommandNOWAITOverSiloRDataSiloMax}{$2.6\times$\xspace}
\newcommand{\mainDataNewOrderTPCCLogNumberofWorkerThreadsvsThroughputCompareTaurusDataNOWAITOverSiloRDataSilo}{0.0\%\xspace}
\newcommand{\mainDataNewOrderTPCCLogNumberofWorkerThreadsvsThroughputCompareTaurusDataNOWAITOverSiloRDataSiloMax}{1.5\%\xspace}

%% file: revision/figs/Data-Tm0-TPCC-Lr1-Number-of-Worker-Threads-vs-MaxThr.tex
\newcommand{\mainDataNewOrderTPCCRecNumberofWorkerThreadsvsMaxThrScalabilityNoLoggingDataNOWAIT}{0.0\%\xspace}
\newcommand{\mainDataNewOrderTPCCRecNumberofWorkerThreadsvsMaxThrScalabilityTaurusCommandNOWAIT}{0.0\%\xspace}
\newcommand{\mainDataNewOrderTPCCRecNumberofWorkerThreadsvsMaxThrScalabilityTaurusDataNOWAIT}{$7.0\times$\xspace}
\newcommand{\mainDataNewOrderTPCCRecNumberofWorkerThreadsvsMaxThrCompareTaurusCommandNOWAITOverNoLoggingDataNOWAIT}{0.0\%\xspace}
\newcommand{\mainDataNewOrderTPCCRecNumberofWorkerThreadsvsMaxThrCompareTaurusCommandNOWAITOverNoLoggingDataNOWAITMax}{0.0\%\xspace}
\newcommand{\mainDataNewOrderTPCCRecNumberofWorkerThreadsvsMaxThrCompareTaurusCommandNOWAITOverSerialCommandNOWAIT}{0.0\%\xspace}
\newcommand{\mainDataNewOrderTPCCRecNumberofWorkerThreadsvsMaxThrCompareTaurusCommandNOWAITOverSerialCommandNOWAITMax}{0.0\%\xspace}
\newcommand{\mainDataNewOrderTPCCRecNumberofWorkerThreadsvsMaxThrCompareTaurusCommandNOWAITOverSerialDataNOWAIT}{$20.4\times$\xspace}
\newcommand{\mainDataNewOrderTPCCRecNumberofWorkerThreadsvsMaxThrCompareTaurusCommandNOWAITOverSerialDataNOWAITMax}{$20.4\times$\xspace}
\newcommand{\mainDataNewOrderTPCCRecNumberofWorkerThreadsvsMaxThrCompareTaurusCommandNOWAITOverPloverDataNOWAIT}{$3.2\times$\xspace}
\newcommand{\mainDataNewOrderTPCCRecNumberofWorkerThreadsvsMaxThrCompareTaurusCommandNOWAITOverPloverDataNOWAITMax}{$3.2\times$\xspace}
\newcommand{\mainDataNewOrderTPCCRecNumberofWorkerThreadsvsMaxThrCompareTaurusCommandNOWAITOverSerialRAIDCommandNOWAIT}{0.0\%\xspace}
\newcommand{\mainDataNewOrderTPCCRecNumberofWorkerThreadsvsMaxThrCompareTaurusCommandNOWAITOverSerialRAIDCommandNOWAITMax}{0.0\%\xspace}
\newcommand{\mainDataNewOrderTPCCRecNumberofWorkerThreadsvsMaxThrCompareTaurusDataNOWAITOverSerialDataNOWAIT}{$6.7\times$\xspace}
\newcommand{\mainDataNewOrderTPCCRecNumberofWorkerThreadsvsMaxThrCompareTaurusDataNOWAITOverSerialDataNOWAITMax}{$6.7\times$\xspace}
\newcommand{\mainDataNewOrderTPCCRecNumberofWorkerThreadsvsMaxThrCompareTaurusDataNOWAITOverPloverDataNOWAIT}{$1.0\times$\xspace}
\newcommand{\mainDataNewOrderTPCCRecNumberofWorkerThreadsvsMaxThrCompareTaurusDataNOWAITOverPloverDataNOWAITMax}{$1.0\times$\xspace}
\newcommand{\mainDataNewOrderTPCCRecNumberofWorkerThreadsvsMaxThrCompareTaurusDataNOWAITOverSerialRAIDDataNOWAIT}{$2.1\times$\xspace}
\newcommand{\mainDataNewOrderTPCCRecNumberofWorkerThreadsvsMaxThrCompareTaurusDataNOWAITOverSerialRAIDDataNOWAITMax}{$2.1\times$\xspace}
\newcommand{\mainDataNewOrderTPCCRecNumberofWorkerThreadsvsMaxThrCompareTaurusCommandNOWAITOverTaurusDataNOWAIT}{$3.0\times$\xspace}
\newcommand{\mainDataNewOrderTPCCRecNumberofWorkerThreadsvsMaxThrCompareTaurusCommandNOWAITOverTaurusDataNOWAITMax}{$3.0\times$\xspace}
\newcommand{\mainDataNewOrderTPCCRecNumberofWorkerThreadsvsMaxThrCompareTaurusCommandNOWAITOverSiloRDataSilo}{$3.0\times$\xspace}
\newcommand{\mainDataNewOrderTPCCRecNumberofWorkerThreadsvsMaxThrCompareTaurusCommandNOWAITOverSiloRDataSiloMax}{$3.0\times$\xspace}
\newcommand{\mainDataNewOrderTPCCRecNumberofWorkerThreadsvsMaxThrCompareTaurusDataNOWAITOverSiloRDataSilo}{0.5\%\xspace}
\newcommand{\mainDataNewOrderTPCCRecNumberofWorkerThreadsvsMaxThrCompareTaurusDataNOWAITOverSiloRDataSiloMax}{0.5\%\xspace}

%% file: revision/figs/Data-Tm1-TPCC-Lr0-Number-of-Worker-Threads-vs-Throughput.tex
\newcommand{\mainDataPayTPCCLogNumberofWorkerThreadsvsThroughputScalabilityNoLoggingDataNOWAIT}{0.0\%\xspace}
\newcommand{\mainDataPayTPCCLogNumberofWorkerThreadsvsThroughputScalabilityTaurusCommandNOWAIT}{0.0\%\xspace}
\newcommand{\mainDataPayTPCCLogNumberofWorkerThreadsvsThroughputScalabilityTaurusDataNOWAIT}{$10.2\times$\xspace}
\newcommand{\mainDataPayTPCCLogNumberofWorkerThreadsvsThroughputCompareTaurusCommandNOWAITOverNoLoggingDataNOWAIT}{0.0\%\xspace}
\newcommand{\mainDataPayTPCCLogNumberofWorkerThreadsvsThroughputCompareTaurusCommandNOWAITOverNoLoggingDataNOWAITMax}{0.0\%\xspace}
\newcommand{\mainDataPayTPCCLogNumberofWorkerThreadsvsThroughputCompareTaurusCommandNOWAITOverSerialCommandNOWAIT}{0.0\%\xspace}
\newcommand{\mainDataPayTPCCLogNumberofWorkerThreadsvsThroughputCompareTaurusCommandNOWAITOverSerialCommandNOWAITMax}{0.0\%\xspace}
\newcommand{\mainDataPayTPCCLogNumberofWorkerThreadsvsThroughputCompareTaurusCommandNOWAITOverSerialDataNOWAIT}{$7.5\times$\xspace}
\newcommand{\mainDataPayTPCCLogNumberofWorkerThreadsvsThroughputCompareTaurusCommandNOWAITOverSerialDataNOWAITMax}{$6.8\times$\xspace}
\newcommand{\mainDataPayTPCCLogNumberofWorkerThreadsvsThroughputCompareTaurusCommandNOWAITOverPloverDataNOWAIT}{$2.0\times$\xspace}
\newcommand{\mainDataPayTPCCLogNumberofWorkerThreadsvsThroughputCompareTaurusCommandNOWAITOverPloverDataNOWAITMax}{$2.0\times$\xspace}
\newcommand{\mainDataPayTPCCLogNumberofWorkerThreadsvsThroughputCompareTaurusCommandNOWAITOverSerialRAIDCommandNOWAIT}{0.0\%\xspace}
\newcommand{\mainDataPayTPCCLogNumberofWorkerThreadsvsThroughputCompareTaurusCommandNOWAITOverSerialRAIDCommandNOWAITMax}{0.0\%\xspace}
\newcommand{\mainDataPayTPCCLogNumberofWorkerThreadsvsThroughputCompareTaurusDataNOWAITOverSerialDataNOWAIT}{$9.5\times$\xspace}
\newcommand{\mainDataPayTPCCLogNumberofWorkerThreadsvsThroughputCompareTaurusDataNOWAITOverSerialDataNOWAITMax}{$8.7\times$\xspace}
\newcommand{\mainDataPayTPCCLogNumberofWorkerThreadsvsThroughputCompareTaurusDataNOWAITOverPloverDataNOWAIT}{$2.6\times$\xspace}
\newcommand{\mainDataPayTPCCLogNumberofWorkerThreadsvsThroughputCompareTaurusDataNOWAITOverPloverDataNOWAITMax}{$2.6\times$\xspace}
\newcommand{\mainDataPayTPCCLogNumberofWorkerThreadsvsThroughputCompareTaurusDataNOWAITOverSerialRAIDDataNOWAIT}{$1.3\times$\xspace}
\newcommand{\mainDataPayTPCCLogNumberofWorkerThreadsvsThroughputCompareTaurusDataNOWAITOverSerialRAIDDataNOWAITMax}{$1.3\times$\xspace}
\newcommand{\mainDataPayTPCCLogNumberofWorkerThreadsvsThroughputCompareTaurusCommandNOWAITOverTaurusDataNOWAIT}{21.5\%\xspace}
\newcommand{\mainDataPayTPCCLogNumberofWorkerThreadsvsThroughputCompareTaurusCommandNOWAITOverTaurusDataNOWAITMax}{21.8\%\xspace}
\newcommand{\mainDataPayTPCCLogNumberofWorkerThreadsvsThroughputCompareTaurusCommandNOWAITOverSiloRDataSilo}{25.6\%\xspace}
\newcommand{\mainDataPayTPCCLogNumberofWorkerThreadsvsThroughputCompareTaurusCommandNOWAITOverSiloRDataSiloMax}{25.7\%\xspace}
\newcommand{\mainDataPayTPCCLogNumberofWorkerThreadsvsThroughputCompareTaurusDataNOWAITOverSiloRDataSilo}{5.2\%\xspace}
\newcommand{\mainDataPayTPCCLogNumberofWorkerThreadsvsThroughputCompareTaurusDataNOWAITOverSiloRDataSiloMax}{5.1\%\xspace}

%% file: revision/figs/Data-Tm1-TPCC-Lr1-Number-of-Worker-Threads-vs-MaxThr.tex
\newcommand{\mainDataPayTPCCRecNumberofWorkerThreadsvsMaxThrScalabilityNoLoggingDataNOWAIT}{0.0\%\xspace}
\newcommand{\mainDataPayTPCCRecNumberofWorkerThreadsvsMaxThrScalabilityTaurusCommandNOWAIT}{0.0\%\xspace}
\newcommand{\mainDataPayTPCCRecNumberofWorkerThreadsvsMaxThrScalabilityTaurusDataNOWAIT}{$8.8\times$\xspace}
\newcommand{\mainDataPayTPCCRecNumberofWorkerThreadsvsMaxThrCompareTaurusCommandNOWAITOverNoLoggingDataNOWAIT}{0.0\%\xspace}
\newcommand{\mainDataPayTPCCRecNumberofWorkerThreadsvsMaxThrCompareTaurusCommandNOWAITOverNoLoggingDataNOWAITMax}{0.0\%\xspace}
\newcommand{\mainDataPayTPCCRecNumberofWorkerThreadsvsMaxThrCompareTaurusCommandNOWAITOverSerialCommandNOWAIT}{0.0\%\xspace}
\newcommand{\mainDataPayTPCCRecNumberofWorkerThreadsvsMaxThrCompareTaurusCommandNOWAITOverSerialCommandNOWAITMax}{0.0\%\xspace}
\newcommand{\mainDataPayTPCCRecNumberofWorkerThreadsvsMaxThrCompareTaurusCommandNOWAITOverSerialDataNOWAIT}{$7.0\times$\xspace}
\newcommand{\mainDataPayTPCCRecNumberofWorkerThreadsvsMaxThrCompareTaurusCommandNOWAITOverSerialDataNOWAITMax}{$7.0\times$\xspace}
\newcommand{\mainDataPayTPCCRecNumberofWorkerThreadsvsMaxThrCompareTaurusCommandNOWAITOverPloverDataNOWAIT}{$1.6\times$\xspace}
\newcommand{\mainDataPayTPCCRecNumberofWorkerThreadsvsMaxThrCompareTaurusCommandNOWAITOverPloverDataNOWAITMax}{$1.6\times$\xspace}
\newcommand{\mainDataPayTPCCRecNumberofWorkerThreadsvsMaxThrCompareTaurusCommandNOWAITOverSerialRAIDCommandNOWAIT}{0.0\%\xspace}
\newcommand{\mainDataPayTPCCRecNumberofWorkerThreadsvsMaxThrCompareTaurusCommandNOWAITOverSerialRAIDCommandNOWAITMax}{0.0\%\xspace}
\newcommand{\mainDataPayTPCCRecNumberofWorkerThreadsvsMaxThrCompareTaurusDataNOWAITOverSerialDataNOWAIT}{$8.2\times$\xspace}
\newcommand{\mainDataPayTPCCRecNumberofWorkerThreadsvsMaxThrCompareTaurusDataNOWAITOverSerialDataNOWAITMax}{$8.1\times$\xspace}
\newcommand{\mainDataPayTPCCRecNumberofWorkerThreadsvsMaxThrCompareTaurusDataNOWAITOverPloverDataNOWAIT}{$1.9\times$\xspace}
\newcommand{\mainDataPayTPCCRecNumberofWorkerThreadsvsMaxThrCompareTaurusDataNOWAITOverPloverDataNOWAITMax}{$1.9\times$\xspace}
\newcommand{\mainDataPayTPCCRecNumberofWorkerThreadsvsMaxThrCompareTaurusDataNOWAITOverSerialRAIDDataNOWAIT}{$2.2\times$\xspace}
\newcommand{\mainDataPayTPCCRecNumberofWorkerThreadsvsMaxThrCompareTaurusDataNOWAITOverSerialRAIDDataNOWAITMax}{$2.2\times$\xspace}
\newcommand{\mainDataPayTPCCRecNumberofWorkerThreadsvsMaxThrCompareTaurusCommandNOWAITOverTaurusDataNOWAIT}{14.0\%\xspace}
\newcommand{\mainDataPayTPCCRecNumberofWorkerThreadsvsMaxThrCompareTaurusCommandNOWAITOverTaurusDataNOWAITMax}{14.1\%\xspace}
\newcommand{\mainDataPayTPCCRecNumberofWorkerThreadsvsMaxThrCompareTaurusCommandNOWAITOverSiloRDataSilo}{15.9\%\xspace}
\newcommand{\mainDataPayTPCCRecNumberofWorkerThreadsvsMaxThrCompareTaurusCommandNOWAITOverSiloRDataSiloMax}{16.4\%\xspace}
\newcommand{\mainDataPayTPCCRecNumberofWorkerThreadsvsMaxThrCompareTaurusDataNOWAITOverSiloRDataSilo}{2.2\%\xspace}
\newcommand{\mainDataPayTPCCRecNumberofWorkerThreadsvsMaxThrCompareTaurusDataNOWAITOverSiloRDataSiloMax}{2.7\%\xspace}

%% file: revision/figs/Data-YCSB-Lr0-Number-of-Worker-Threads-vs-Throughput.tex
\newcommand{\mainDataYCSBLogNumberofWorkerThreadsvsThroughputScalabilityNoLoggingDataNOWAIT}{0.0\%\xspace}
\newcommand{\mainDataYCSBLogNumberofWorkerThreadsvsThroughputScalabilityTaurusCommandNOWAIT}{0.0\%\xspace}
\newcommand{\mainDataYCSBLogNumberofWorkerThreadsvsThroughputScalabilityTaurusDataNOWAIT}{$7.6\times$\xspace}
\newcommand{\mainDataYCSBLogNumberofWorkerThreadsvsThroughputCompareTaurusCommandNOWAITOverNoLoggingDataNOWAIT}{0.0\%\xspace}
\newcommand{\mainDataYCSBLogNumberofWorkerThreadsvsThroughputCompareTaurusCommandNOWAITOverNoLoggingDataNOWAITMax}{0.0\%\xspace}
\newcommand{\mainDataYCSBLogNumberofWorkerThreadsvsThroughputCompareTaurusCommandNOWAITOverSerialCommandNOWAIT}{0.0\%\xspace}
\newcommand{\mainDataYCSBLogNumberofWorkerThreadsvsThroughputCompareTaurusCommandNOWAITOverSerialCommandNOWAITMax}{0.0\%\xspace}
\newcommand{\mainDataYCSBLogNumberofWorkerThreadsvsThroughputCompareTaurusCommandNOWAITOverSerialDataNOWAIT}{$5.5\times$\xspace}
\newcommand{\mainDataYCSBLogNumberofWorkerThreadsvsThroughputCompareTaurusCommandNOWAITOverSerialDataNOWAITMax}{$5.5\times$\xspace}
\newcommand{\mainDataYCSBLogNumberofWorkerThreadsvsThroughputCompareTaurusCommandNOWAITOverPloverDataNOWAIT}{21.6\%\xspace}
\newcommand{\mainDataYCSBLogNumberofWorkerThreadsvsThroughputCompareTaurusCommandNOWAITOverPloverDataNOWAITMax}{21.6\%\xspace}
\newcommand{\mainDataYCSBLogNumberofWorkerThreadsvsThroughputCompareTaurusCommandNOWAITOverSerialRAIDCommandNOWAIT}{0.0\%\xspace}
\newcommand{\mainDataYCSBLogNumberofWorkerThreadsvsThroughputCompareTaurusCommandNOWAITOverSerialRAIDCommandNOWAITMax}{0.0\%\xspace}
\newcommand{\mainDataYCSBLogNumberofWorkerThreadsvsThroughputCompareTaurusDataNOWAITOverSerialDataNOWAIT}{$7.1\times$\xspace}
\newcommand{\mainDataYCSBLogNumberofWorkerThreadsvsThroughputCompareTaurusDataNOWAITOverSerialDataNOWAITMax}{$7.1\times$\xspace}
\newcommand{\mainDataYCSBLogNumberofWorkerThreadsvsThroughputCompareTaurusDataNOWAITOverPloverDataNOWAIT}{$1.0\times$\xspace}
\newcommand{\mainDataYCSBLogNumberofWorkerThreadsvsThroughputCompareTaurusDataNOWAITOverPloverDataNOWAITMax}{$1.0\times$\xspace}
\newcommand{\mainDataYCSBLogNumberofWorkerThreadsvsThroughputCompareTaurusDataNOWAITOverSerialRAIDDataNOWAIT}{$1.3\times$\xspace}
\newcommand{\mainDataYCSBLogNumberofWorkerThreadsvsThroughputCompareTaurusDataNOWAITOverSerialRAIDDataNOWAITMax}{$1.3\times$\xspace}
\newcommand{\mainDataYCSBLogNumberofWorkerThreadsvsThroughputCompareTaurusCommandNOWAITOverTaurusDataNOWAIT}{21.6\%\xspace}
\newcommand{\mainDataYCSBLogNumberofWorkerThreadsvsThroughputCompareTaurusCommandNOWAITOverTaurusDataNOWAITMax}{22.5\%\xspace}
\newcommand{\mainDataYCSBLogNumberofWorkerThreadsvsThroughputCompareTaurusCommandNOWAITOverSiloRDataSilo}{23.1\%\xspace}
\newcommand{\mainDataYCSBLogNumberofWorkerThreadsvsThroughputCompareTaurusCommandNOWAITOverSiloRDataSiloMax}{23.1\%\xspace}
\newcommand{\mainDataYCSBLogNumberofWorkerThreadsvsThroughputCompareTaurusDataNOWAITOverSiloRDataSilo}{1.8\%\xspace}
\newcommand{\mainDataYCSBLogNumberofWorkerThreadsvsThroughputCompareTaurusDataNOWAITOverSiloRDataSiloMax}{0.8\%\xspace}

%% file: revision/figs/Data-YCSB-Lr1-Number-of-Worker-Threads-vs-MaxThr.tex
\newcommand{\mainDataYCSBRecNumberofWorkerThreadsvsMaxThrScalabilityNoLoggingDataNOWAIT}{0.0\%\xspace}
\newcommand{\mainDataYCSBRecNumberofWorkerThreadsvsMaxThrScalabilityTaurusCommandNOWAIT}{0.0\%\xspace}
\newcommand{\mainDataYCSBRecNumberofWorkerThreadsvsMaxThrScalabilityTaurusDataNOWAIT}{$6.4\times$\xspace}
\newcommand{\mainDataYCSBRecNumberofWorkerThreadsvsMaxThrCompareTaurusCommandNOWAITOverNoLoggingDataNOWAIT}{0.0\%\xspace}
\newcommand{\mainDataYCSBRecNumberofWorkerThreadsvsMaxThrCompareTaurusCommandNOWAITOverNoLoggingDataNOWAITMax}{0.0\%\xspace}
\newcommand{\mainDataYCSBRecNumberofWorkerThreadsvsMaxThrCompareTaurusCommandNOWAITOverSerialCommandNOWAIT}{0.0\%\xspace}
\newcommand{\mainDataYCSBRecNumberofWorkerThreadsvsMaxThrCompareTaurusCommandNOWAITOverSerialCommandNOWAITMax}{0.0\%\xspace}
\newcommand{\mainDataYCSBRecNumberofWorkerThreadsvsMaxThrCompareTaurusCommandNOWAITOverSerialDataNOWAIT}{$6.8\times$\xspace}
\newcommand{\mainDataYCSBRecNumberofWorkerThreadsvsMaxThrCompareTaurusCommandNOWAITOverSerialDataNOWAITMax}{$6.8\times$\xspace}
\newcommand{\mainDataYCSBRecNumberofWorkerThreadsvsMaxThrCompareTaurusCommandNOWAITOverPloverDataNOWAIT}{$1.2\times$\xspace}
\newcommand{\mainDataYCSBRecNumberofWorkerThreadsvsMaxThrCompareTaurusCommandNOWAITOverPloverDataNOWAITMax}{$1.2\times$\xspace}
\newcommand{\mainDataYCSBRecNumberofWorkerThreadsvsMaxThrCompareTaurusCommandNOWAITOverSerialRAIDCommandNOWAIT}{0.0\%\xspace}
\newcommand{\mainDataYCSBRecNumberofWorkerThreadsvsMaxThrCompareTaurusCommandNOWAITOverSerialRAIDCommandNOWAITMax}{0.0\%\xspace}
\newcommand{\mainDataYCSBRecNumberofWorkerThreadsvsMaxThrCompareTaurusDataNOWAITOverSerialDataNOWAIT}{$5.5\times$\xspace}
\newcommand{\mainDataYCSBRecNumberofWorkerThreadsvsMaxThrCompareTaurusDataNOWAITOverSerialDataNOWAITMax}{$5.5\times$\xspace}
\newcommand{\mainDataYCSBRecNumberofWorkerThreadsvsMaxThrCompareTaurusDataNOWAITOverPloverDataNOWAIT}{2.8\%\xspace}
\newcommand{\mainDataYCSBRecNumberofWorkerThreadsvsMaxThrCompareTaurusDataNOWAITOverPloverDataNOWAITMax}{6.3\%\xspace}
\newcommand{\mainDataYCSBRecNumberofWorkerThreadsvsMaxThrCompareTaurusDataNOWAITOverSerialRAIDDataNOWAIT}{$1.7\times$\xspace}
\newcommand{\mainDataYCSBRecNumberofWorkerThreadsvsMaxThrCompareTaurusDataNOWAITOverSerialRAIDDataNOWAITMax}{$1.7\times$\xspace}
\newcommand{\mainDataYCSBRecNumberofWorkerThreadsvsMaxThrCompareTaurusCommandNOWAITOverTaurusDataNOWAIT}{$1.2\times$\xspace}
\newcommand{\mainDataYCSBRecNumberofWorkerThreadsvsMaxThrCompareTaurusCommandNOWAITOverTaurusDataNOWAITMax}{$1.2\times$\xspace}
\newcommand{\mainDataYCSBRecNumberofWorkerThreadsvsMaxThrCompareTaurusCommandNOWAITOverSiloRDataSilo}{$1.2\times$\xspace}
\newcommand{\mainDataYCSBRecNumberofWorkerThreadsvsMaxThrCompareTaurusCommandNOWAITOverSiloRDataSiloMax}{$1.2\times$\xspace}
\newcommand{\mainDataYCSBRecNumberofWorkerThreadsvsMaxThrCompareTaurusDataNOWAITOverSiloRDataSilo}{2.7\%\xspace}
\newcommand{\mainDataYCSBRecNumberofWorkerThreadsvsMaxThrCompareTaurusDataNOWAITOverSiloRDataSiloMax}{3.1\%\xspace}

%% file: revision/figs/Ln16-silo-Tm0-TPCC-Lr0-Number-of-Worker-Threads-vs-MaxThr.tex
\newcommand{\AdvmainLnXVIsiloNewOrderTPCCLogNumberofWorkerThreadsvsMaxThrScalabilityNoLoggingDataSilo}{$45.2\times$\xspace}
\newcommand{\AdvmainLnXVIsiloNewOrderTPCCLogNumberofWorkerThreadsvsMaxThrScalabilityTaurusCommandSilo}{$39.9\times$\xspace}
\newcommand{\AdvmainLnXVIsiloNewOrderTPCCLogNumberofWorkerThreadsvsMaxThrScalabilityTaurusDataSilo}{$22.7\times$\xspace}
\newcommand{\AdvmainLnXVIsiloNewOrderTPCCLogNumberofWorkerThreadsvsMaxThrCompareTaurusCommandSiloOverNoLoggingDataSilo}{32.9\%\xspace}
\newcommand{\AdvmainLnXVIsiloNewOrderTPCCLogNumberofWorkerThreadsvsMaxThrCompareTaurusCommandSiloOverNoLoggingDataSiloMax}{32.9\%\xspace}
\newcommand{\AdvmainLnXVIsiloNewOrderTPCCLogNumberofWorkerThreadsvsMaxThrCompareTaurusCommandSiloOverTaurusDataSilo}{$1.9\times$\xspace}
\newcommand{\AdvmainLnXVIsiloNewOrderTPCCLogNumberofWorkerThreadsvsMaxThrCompareTaurusCommandSiloOverTaurusDataSiloMax}{$1.9\times$\xspace}
\newcommand{\AdvmainLnXVIsiloNewOrderTPCCLogNumberofWorkerThreadsvsMaxThrCompareTaurusCommandSiloOverSiloRDataSilo}{$1.8\times$\xspace}
\newcommand{\AdvmainLnXVIsiloNewOrderTPCCLogNumberofWorkerThreadsvsMaxThrCompareTaurusCommandSiloOverSiloRDataSiloMax}{$1.8\times$\xspace}
\newcommand{\AdvmainLnXVIsiloNewOrderTPCCLogNumberofWorkerThreadsvsMaxThrCompareTaurusDataSiloOverSiloRDataSilo}{4.2\%\xspace}
\newcommand{\AdvmainLnXVIsiloNewOrderTPCCLogNumberofWorkerThreadsvsMaxThrCompareTaurusDataSiloOverSiloRDataSiloMax}{5.7\%\xspace}

%% file: revision/figs/Ln16-silo-Tm0-TPCC-Lr1-Number-of-Worker-Threads-vs-MaxThr.tex
\newcommand{\AdvmainLnXVIsiloNewOrderTPCCRecNumberofWorkerThreadsvsMaxThrScalabilityNoLoggingDataSilo}{0.0\%\xspace}
\newcommand{\AdvmainLnXVIsiloNewOrderTPCCRecNumberofWorkerThreadsvsMaxThrScalabilityTaurusCommandSilo}{$36.1\times$\xspace}
\newcommand{\AdvmainLnXVIsiloNewOrderTPCCRecNumberofWorkerThreadsvsMaxThrScalabilityTaurusDataSilo}{$6.4\times$\xspace}
\newcommand{\AdvmainLnXVIsiloNewOrderTPCCRecNumberofWorkerThreadsvsMaxThrCompareTaurusCommandSiloOverNoLoggingDataSilo}{$9.5\times$\xspace}
\newcommand{\AdvmainLnXVIsiloNewOrderTPCCRecNumberofWorkerThreadsvsMaxThrCompareTaurusCommandSiloOverNoLoggingDataSiloMax}{$9.5\times$\xspace}
\newcommand{\AdvmainLnXVIsiloNewOrderTPCCRecNumberofWorkerThreadsvsMaxThrCompareTaurusCommandSiloOverTaurusDataSilo}{$3.7\times$\xspace}
\newcommand{\AdvmainLnXVIsiloNewOrderTPCCRecNumberofWorkerThreadsvsMaxThrCompareTaurusCommandSiloOverTaurusDataSiloMax}{$3.7\times$\xspace}
\newcommand{\AdvmainLnXVIsiloNewOrderTPCCRecNumberofWorkerThreadsvsMaxThrCompareTaurusCommandSiloOverSiloRDataSilo}{$6.4\times$\xspace}
\newcommand{\AdvmainLnXVIsiloNewOrderTPCCRecNumberofWorkerThreadsvsMaxThrCompareTaurusCommandSiloOverSiloRDataSiloMax}{$6.2\times$\xspace}
\newcommand{\AdvmainLnXVIsiloNewOrderTPCCRecNumberofWorkerThreadsvsMaxThrCompareTaurusDataSiloOverSiloRDataSilo}{$1.7\times$\xspace}
\newcommand{\AdvmainLnXVIsiloNewOrderTPCCRecNumberofWorkerThreadsvsMaxThrCompareTaurusDataSiloOverSiloRDataSiloMax}{$1.7\times$\xspace}

%% file: revision/figs/Ln16-silo-Tm1-TPCC-Lr0-Number-of-Worker-Threads-vs-MaxThr.tex
\newcommand{\AdvmainLnXVIsiloPayTPCCLogNumberofWorkerThreadsvsMaxThrScalabilityNoLoggingDataSilo}{$43.5\times$\xspace}
\newcommand{\AdvmainLnXVIsiloPayTPCCLogNumberofWorkerThreadsvsMaxThrScalabilityTaurusCommandSilo}{$36.0\times$\xspace}
\newcommand{\AdvmainLnXVIsiloPayTPCCLogNumberofWorkerThreadsvsMaxThrScalabilityTaurusDataSilo}{$12.8\times$\xspace}
\newcommand{\AdvmainLnXVIsiloPayTPCCLogNumberofWorkerThreadsvsMaxThrCompareTaurusCommandSiloOverNoLoggingDataSilo}{42.1\%\xspace}
\newcommand{\AdvmainLnXVIsiloPayTPCCLogNumberofWorkerThreadsvsMaxThrCompareTaurusCommandSiloOverNoLoggingDataSiloMax}{42.1\%\xspace}
\newcommand{\AdvmainLnXVIsiloPayTPCCLogNumberofWorkerThreadsvsMaxThrCompareTaurusCommandSiloOverTaurusDataSilo}{$3.1\times$\xspace}
\newcommand{\AdvmainLnXVIsiloPayTPCCLogNumberofWorkerThreadsvsMaxThrCompareTaurusCommandSiloOverTaurusDataSiloMax}{$3.0\times$\xspace}
\newcommand{\AdvmainLnXVIsiloPayTPCCLogNumberofWorkerThreadsvsMaxThrCompareTaurusCommandSiloOverSiloRDataSilo}{$2.8\times$\xspace}
\newcommand{\AdvmainLnXVIsiloPayTPCCLogNumberofWorkerThreadsvsMaxThrCompareTaurusCommandSiloOverSiloRDataSiloMax}{$2.6\times$\xspace}
\newcommand{\AdvmainLnXVIsiloPayTPCCLogNumberofWorkerThreadsvsMaxThrCompareTaurusDataSiloOverSiloRDataSilo}{11.6\%\xspace}
\newcommand{\AdvmainLnXVIsiloPayTPCCLogNumberofWorkerThreadsvsMaxThrCompareTaurusDataSiloOverSiloRDataSiloMax}{12.4\%\xspace}

%% file: revision/figs/Ln16-silo-Tm1-TPCC-Lr1-Number-of-Worker-Threads-vs-MaxThr.tex
\newcommand{\AdvmainLnXVIsiloPayTPCCRecNumberofWorkerThreadsvsMaxThrScalabilityNoLoggingDataSilo}{0.0\%\xspace}
\newcommand{\AdvmainLnXVIsiloPayTPCCRecNumberofWorkerThreadsvsMaxThrScalabilityTaurusCommandSilo}{$29.7\times$\xspace}
\newcommand{\AdvmainLnXVIsiloPayTPCCRecNumberofWorkerThreadsvsMaxThrScalabilityTaurusDataSilo}{$7.6\times$\xspace}
\newcommand{\AdvmainLnXVIsiloPayTPCCRecNumberofWorkerThreadsvsMaxThrCompareTaurusCommandSiloOverNoLoggingDataSilo}{$33.7\times$\xspace}
\newcommand{\AdvmainLnXVIsiloPayTPCCRecNumberofWorkerThreadsvsMaxThrCompareTaurusCommandSiloOverNoLoggingDataSiloMax}{$33.7\times$\xspace}
\newcommand{\AdvmainLnXVIsiloPayTPCCRecNumberofWorkerThreadsvsMaxThrCompareTaurusCommandSiloOverTaurusDataSilo}{$3.7\times$\xspace}
\newcommand{\AdvmainLnXVIsiloPayTPCCRecNumberofWorkerThreadsvsMaxThrCompareTaurusCommandSiloOverTaurusDataSiloMax}{$3.6\times$\xspace}
\newcommand{\AdvmainLnXVIsiloPayTPCCRecNumberofWorkerThreadsvsMaxThrCompareTaurusCommandSiloOverSiloRDataSilo}{$4.1\times$\xspace}
\newcommand{\AdvmainLnXVIsiloPayTPCCRecNumberofWorkerThreadsvsMaxThrCompareTaurusCommandSiloOverSiloRDataSiloMax}{$4.0\times$\xspace}
\newcommand{\AdvmainLnXVIsiloPayTPCCRecNumberofWorkerThreadsvsMaxThrCompareTaurusDataSiloOverSiloRDataSilo}{$1.1\times$\xspace}
\newcommand{\AdvmainLnXVIsiloPayTPCCRecNumberofWorkerThreadsvsMaxThrCompareTaurusDataSiloOverSiloRDataSiloMax}{$1.1\times$\xspace}

%% file: revision/figs/Ln16-silo-YCSB-Lr0-Number-of-Worker-Threads-vs-MaxThr.tex
\newcommand{\AdvmainLnXVIsiloYCSBLogNumberofWorkerThreadsvsMaxThrScalabilityNoLoggingDataSilo}{$65.1\times$\xspace}
\newcommand{\AdvmainLnXVIsiloYCSBLogNumberofWorkerThreadsvsMaxThrScalabilityTaurusCommandSilo}{$62.2\times$\xspace}
\newcommand{\AdvmainLnXVIsiloYCSBLogNumberofWorkerThreadsvsMaxThrScalabilityTaurusDataSilo}{$40.7\times$\xspace}
\newcommand{\AdvmainLnXVIsiloYCSBLogNumberofWorkerThreadsvsMaxThrCompareTaurusCommandSiloOverNoLoggingDataSilo}{21.7\%\xspace}
\newcommand{\AdvmainLnXVIsiloYCSBLogNumberofWorkerThreadsvsMaxThrCompareTaurusCommandSiloOverNoLoggingDataSiloMax}{21.7\%\xspace}
\newcommand{\AdvmainLnXVIsiloYCSBLogNumberofWorkerThreadsvsMaxThrCompareTaurusCommandSiloOverTaurusDataSilo}{$1.6\times$\xspace}
\newcommand{\AdvmainLnXVIsiloYCSBLogNumberofWorkerThreadsvsMaxThrCompareTaurusCommandSiloOverTaurusDataSiloMax}{$1.4\times$\xspace}
\newcommand{\AdvmainLnXVIsiloYCSBLogNumberofWorkerThreadsvsMaxThrCompareTaurusCommandSiloOverSiloRDataSilo}{$1.3\times$\xspace}
\newcommand{\AdvmainLnXVIsiloYCSBLogNumberofWorkerThreadsvsMaxThrCompareTaurusCommandSiloOverSiloRDataSiloMax}{$1.3\times$\xspace}
\newcommand{\AdvmainLnXVIsiloYCSBLogNumberofWorkerThreadsvsMaxThrCompareTaurusDataSiloOverSiloRDataSilo}{16.3\%\xspace}
\newcommand{\AdvmainLnXVIsiloYCSBLogNumberofWorkerThreadsvsMaxThrCompareTaurusDataSiloOverSiloRDataSiloMax}{7.7\%\xspace}

%% file: revision/figs/Ln16-silo-YCSB-Lr1-Number-of-Worker-Threads-vs-MaxThr.tex
\newcommand{\AdvmainLnXVIsiloYCSBRecNumberofWorkerThreadsvsMaxThrScalabilityNoLoggingDataSilo}{0.0\%\xspace}
\newcommand{\AdvmainLnXVIsiloYCSBRecNumberofWorkerThreadsvsMaxThrScalabilityTaurusCommandSilo}{$53.0\times$\xspace}
\newcommand{\AdvmainLnXVIsiloYCSBRecNumberofWorkerThreadsvsMaxThrScalabilityTaurusDataSilo}{$31.9\times$\xspace}
\newcommand{\AdvmainLnXVIsiloYCSBRecNumberofWorkerThreadsvsMaxThrCompareTaurusCommandSiloOverNoLoggingDataSilo}{$16.9\times$\xspace}
\newcommand{\AdvmainLnXVIsiloYCSBRecNumberofWorkerThreadsvsMaxThrCompareTaurusCommandSiloOverNoLoggingDataSiloMax}{$16.9\times$\xspace}
\newcommand{\AdvmainLnXVIsiloYCSBRecNumberofWorkerThreadsvsMaxThrCompareTaurusCommandSiloOverTaurusDataSilo}{$1.5\times$\xspace}
\newcommand{\AdvmainLnXVIsiloYCSBRecNumberofWorkerThreadsvsMaxThrCompareTaurusCommandSiloOverTaurusDataSiloMax}{$1.5\times$\xspace}
\newcommand{\AdvmainLnXVIsiloYCSBRecNumberofWorkerThreadsvsMaxThrCompareTaurusCommandSiloOverSiloRDataSilo}{$1.4\times$\xspace}
\newcommand{\AdvmainLnXVIsiloYCSBRecNumberofWorkerThreadsvsMaxThrCompareTaurusCommandSiloOverSiloRDataSiloMax}{$1.4\times$\xspace}
\newcommand{\AdvmainLnXVIsiloYCSBRecNumberofWorkerThreadsvsMaxThrCompareTaurusDataSiloOverSiloRDataSilo}{11.3\%\xspace}
\newcommand{\AdvmainLnXVIsiloYCSBRecNumberofWorkerThreadsvsMaxThrCompareTaurusDataSiloOverSiloRDataSiloMax}{11.3\%\xspace}

%% file: revision/figs/Ln16-Tm0-TPCC-Lr0-Number-of-Worker-Threads-vs-MaxThr.tex
\newcommand{\AdvmainLnXVINewOrderTPCCLogNumberofWorkerThreadsvsMaxThrScalabilityNoLoggingDataNOWAIT}{$35.1\times$\xspace}
\newcommand{\AdvmainLnXVINewOrderTPCCLogNumberofWorkerThreadsvsMaxThrScalabilityTaurusCommandNOWAIT}{$35.3\times$\xspace}
\newcommand{\AdvmainLnXVINewOrderTPCCLogNumberofWorkerThreadsvsMaxThrScalabilityTaurusDataNOWAIT}{$19.1\times$\xspace}
\newcommand{\AdvmainLnXVINewOrderTPCCLogNumberofWorkerThreadsvsMaxThrCompareTaurusCommandNOWAITOverNoLoggingDataNOWAIT}{29.7\%\xspace}
\newcommand{\AdvmainLnXVINewOrderTPCCLogNumberofWorkerThreadsvsMaxThrCompareTaurusCommandNOWAITOverNoLoggingDataNOWAITMax}{29.7\%\xspace}
\newcommand{\AdvmainLnXVINewOrderTPCCLogNumberofWorkerThreadsvsMaxThrCompareTaurusCommandNOWAITOverSerialCommandNOWAIT}{$1.3\times$\xspace}
\newcommand{\AdvmainLnXVINewOrderTPCCLogNumberofWorkerThreadsvsMaxThrCompareTaurusCommandNOWAITOverSerialCommandNOWAITMax}{$1.3\times$\xspace}
\newcommand{\AdvmainLnXVINewOrderTPCCLogNumberofWorkerThreadsvsMaxThrCompareTaurusCommandNOWAITOverSerialDataNOWAIT}{$15.7\times$\xspace}
\newcommand{\AdvmainLnXVINewOrderTPCCLogNumberofWorkerThreadsvsMaxThrCompareTaurusCommandNOWAITOverSerialDataNOWAITMax}{$15.5\times$\xspace}
\newcommand{\AdvmainLnXVINewOrderTPCCLogNumberofWorkerThreadsvsMaxThrCompareTaurusCommandNOWAITOverPloverDataNOWAIT}{$22.0\times$\xspace}
\newcommand{\AdvmainLnXVINewOrderTPCCLogNumberofWorkerThreadsvsMaxThrCompareTaurusCommandNOWAITOverPloverDataNOWAITMax}{$20.6\times$\xspace}
\newcommand{\AdvmainLnXVINewOrderTPCCLogNumberofWorkerThreadsvsMaxThrCompareTaurusCommandNOWAITOverSerialRAIDCommandNOWAIT}{12.8\%\xspace}
\newcommand{\AdvmainLnXVINewOrderTPCCLogNumberofWorkerThreadsvsMaxThrCompareTaurusCommandNOWAITOverSerialRAIDCommandNOWAITMax}{12.8\%\xspace}
\newcommand{\AdvmainLnXVINewOrderTPCCLogNumberofWorkerThreadsvsMaxThrCompareTaurusDataNOWAITOverSerialDataNOWAIT}{$7.8\times$\xspace}
\newcommand{\AdvmainLnXVINewOrderTPCCLogNumberofWorkerThreadsvsMaxThrCompareTaurusDataNOWAITOverSerialDataNOWAITMax}{$7.8\times$\xspace}
\newcommand{\AdvmainLnXVINewOrderTPCCLogNumberofWorkerThreadsvsMaxThrCompareTaurusDataNOWAITOverPloverDataNOWAIT}{$10.9\times$\xspace}
\newcommand{\AdvmainLnXVINewOrderTPCCLogNumberofWorkerThreadsvsMaxThrCompareTaurusDataNOWAITOverPloverDataNOWAITMax}{$10.4\times$\xspace}
\newcommand{\AdvmainLnXVINewOrderTPCCLogNumberofWorkerThreadsvsMaxThrCompareTaurusDataNOWAITOverSerialRAIDDataNOWAIT}{$1.2\times$\xspace}
\newcommand{\AdvmainLnXVINewOrderTPCCLogNumberofWorkerThreadsvsMaxThrCompareTaurusDataNOWAITOverSerialRAIDDataNOWAITMax}{$1.2\times$\xspace}
\newcommand{\AdvmainLnXVINewOrderTPCCLogNumberofWorkerThreadsvsMaxThrCompareTaurusCommandNOWAITOverTaurusDataNOWAIT}{$2.0\times$\xspace}
\newcommand{\AdvmainLnXVINewOrderTPCCLogNumberofWorkerThreadsvsMaxThrCompareTaurusCommandNOWAITOverTaurusDataNOWAITMax}{$2.0\times$\xspace}
\newcommand{\AdvmainLnXVINewOrderTPCCLogNumberofWorkerThreadsvsMaxThrCompareTaurusCommandNOWAITOverSiloRDataSilo}{$4.1\times$\xspace}
\newcommand{\AdvmainLnXVINewOrderTPCCLogNumberofWorkerThreadsvsMaxThrCompareTaurusCommandNOWAITOverSiloRDataSiloMax}{$4.1\times$\xspace}
\newcommand{\AdvmainLnXVINewOrderTPCCLogNumberofWorkerThreadsvsMaxThrCompareTaurusDataNOWAITOverSiloRDataSilo}{$2.0\times$\xspace}
\newcommand{\AdvmainLnXVINewOrderTPCCLogNumberofWorkerThreadsvsMaxThrCompareTaurusDataNOWAITOverSiloRDataSiloMax}{$2.1\times$\xspace}

%% file: revision/figs/Ln16-Tm0-TPCC-Lr1-Number-of-Worker-Threads-vs-MaxThr.tex
\newcommand{\AdvmainLnXVINewOrderTPCCRecNumberofWorkerThreadsvsMaxThrScalabilityNoLoggingDataNOWAIT}{0.0\%\xspace}
\newcommand{\AdvmainLnXVINewOrderTPCCRecNumberofWorkerThreadsvsMaxThrScalabilityTaurusCommandNOWAIT}{$35.4\times$\xspace}
\newcommand{\AdvmainLnXVINewOrderTPCCRecNumberofWorkerThreadsvsMaxThrScalabilityTaurusDataNOWAIT}{$6.2\times$\xspace}
\newcommand{\AdvmainLnXVINewOrderTPCCRecNumberofWorkerThreadsvsMaxThrCompareTaurusCommandNOWAITOverNoLoggingDataNOWAIT}{$9.7\times$\xspace}
\newcommand{\AdvmainLnXVINewOrderTPCCRecNumberofWorkerThreadsvsMaxThrCompareTaurusCommandNOWAITOverNoLoggingDataNOWAITMax}{$9.7\times$\xspace}
\newcommand{\AdvmainLnXVINewOrderTPCCRecNumberofWorkerThreadsvsMaxThrCompareTaurusCommandNOWAITOverSerialCommandNOWAIT}{$75.6\times$\xspace}
\newcommand{\AdvmainLnXVINewOrderTPCCRecNumberofWorkerThreadsvsMaxThrCompareTaurusCommandNOWAITOverSerialCommandNOWAITMax}{$75.6\times$\xspace}
\newcommand{\AdvmainLnXVINewOrderTPCCRecNumberofWorkerThreadsvsMaxThrCompareTaurusCommandNOWAITOverSerialDataNOWAIT}{$61.1\times$\xspace}
\newcommand{\AdvmainLnXVINewOrderTPCCRecNumberofWorkerThreadsvsMaxThrCompareTaurusCommandNOWAITOverSerialDataNOWAITMax}{$60.6\times$\xspace}
\newcommand{\AdvmainLnXVINewOrderTPCCRecNumberofWorkerThreadsvsMaxThrCompareTaurusCommandNOWAITOverPloverDataNOWAIT}{$2.4\times$\xspace}
\newcommand{\AdvmainLnXVINewOrderTPCCRecNumberofWorkerThreadsvsMaxThrCompareTaurusCommandNOWAITOverPloverDataNOWAITMax}{$2.4\times$\xspace}
\newcommand{\AdvmainLnXVINewOrderTPCCRecNumberofWorkerThreadsvsMaxThrCompareTaurusCommandNOWAITOverSerialRAIDCommandNOWAIT}{$75.9\times$\xspace}
\newcommand{\AdvmainLnXVINewOrderTPCCRecNumberofWorkerThreadsvsMaxThrCompareTaurusCommandNOWAITOverSerialRAIDCommandNOWAITMax}{$75.8\times$\xspace}
\newcommand{\AdvmainLnXVINewOrderTPCCRecNumberofWorkerThreadsvsMaxThrCompareTaurusDataNOWAITOverSerialDataNOWAIT}{$16.2\times$\xspace}
\newcommand{\AdvmainLnXVINewOrderTPCCRecNumberofWorkerThreadsvsMaxThrCompareTaurusDataNOWAITOverSerialDataNOWAITMax}{$16.1\times$\xspace}
\newcommand{\AdvmainLnXVINewOrderTPCCRecNumberofWorkerThreadsvsMaxThrCompareTaurusDataNOWAITOverPloverDataNOWAIT}{36.0\%\xspace}
\newcommand{\AdvmainLnXVINewOrderTPCCRecNumberofWorkerThreadsvsMaxThrCompareTaurusDataNOWAITOverPloverDataNOWAITMax}{35.8\%\xspace}
\newcommand{\AdvmainLnXVINewOrderTPCCRecNumberofWorkerThreadsvsMaxThrCompareTaurusDataNOWAITOverSerialRAIDDataNOWAIT}{$15.9\times$\xspace}
\newcommand{\AdvmainLnXVINewOrderTPCCRecNumberofWorkerThreadsvsMaxThrCompareTaurusDataNOWAITOverSerialRAIDDataNOWAITMax}{$15.9\times$\xspace}
\newcommand{\AdvmainLnXVINewOrderTPCCRecNumberofWorkerThreadsvsMaxThrCompareTaurusCommandNOWAITOverTaurusDataNOWAIT}{$3.8\times$\xspace}
\newcommand{\AdvmainLnXVINewOrderTPCCRecNumberofWorkerThreadsvsMaxThrCompareTaurusCommandNOWAITOverTaurusDataNOWAITMax}{$3.8\times$\xspace}
\newcommand{\AdvmainLnXVINewOrderTPCCRecNumberofWorkerThreadsvsMaxThrCompareTaurusCommandNOWAITOverSiloRDataSilo}{$9.7\times$\xspace}
\newcommand{\AdvmainLnXVINewOrderTPCCRecNumberofWorkerThreadsvsMaxThrCompareTaurusCommandNOWAITOverSiloRDataSiloMax}{$9.7\times$\xspace}
\newcommand{\AdvmainLnXVINewOrderTPCCRecNumberofWorkerThreadsvsMaxThrCompareTaurusDataNOWAITOverSiloRDataSilo}{$2.6\times$\xspace}
\newcommand{\AdvmainLnXVINewOrderTPCCRecNumberofWorkerThreadsvsMaxThrCompareTaurusDataNOWAITOverSiloRDataSiloMax}{$2.6\times$\xspace}

%% file: revision/figs/Ln16-Tm1-TPCC-Lr0-Number-of-Worker-Threads-vs-MaxThr.tex
\newcommand{\AdvmainLnXVIPayTPCCLogNumberofWorkerThreadsvsMaxThrScalabilityNoLoggingDataNOWAIT}{$43.4\times$\xspace}
\newcommand{\AdvmainLnXVIPayTPCCLogNumberofWorkerThreadsvsMaxThrScalabilityTaurusCommandNOWAIT}{$35.4\times$\xspace}
\newcommand{\AdvmainLnXVIPayTPCCLogNumberofWorkerThreadsvsMaxThrScalabilityTaurusDataNOWAIT}{$10.8\times$\xspace}
\newcommand{\AdvmainLnXVIPayTPCCLogNumberofWorkerThreadsvsMaxThrCompareTaurusCommandNOWAITOverNoLoggingDataNOWAIT}{43.4\%\xspace}
\newcommand{\AdvmainLnXVIPayTPCCLogNumberofWorkerThreadsvsMaxThrCompareTaurusCommandNOWAITOverNoLoggingDataNOWAITMax}{43.4\%\xspace}
\newcommand{\AdvmainLnXVIPayTPCCLogNumberofWorkerThreadsvsMaxThrCompareTaurusCommandNOWAITOverSerialCommandNOWAIT}{$2.9\times$\xspace}
\newcommand{\AdvmainLnXVIPayTPCCLogNumberofWorkerThreadsvsMaxThrCompareTaurusCommandNOWAITOverSerialCommandNOWAITMax}{$2.9\times$\xspace}
\newcommand{\AdvmainLnXVIPayTPCCLogNumberofWorkerThreadsvsMaxThrCompareTaurusCommandNOWAITOverSerialDataNOWAIT}{$24.1\times$\xspace}
\newcommand{\AdvmainLnXVIPayTPCCLogNumberofWorkerThreadsvsMaxThrCompareTaurusCommandNOWAITOverSerialDataNOWAITMax}{$24.0\times$\xspace}
\newcommand{\AdvmainLnXVIPayTPCCLogNumberofWorkerThreadsvsMaxThrCompareTaurusCommandNOWAITOverPloverDataNOWAIT}{$52.6\times$\xspace}
\newcommand{\AdvmainLnXVIPayTPCCLogNumberofWorkerThreadsvsMaxThrCompareTaurusCommandNOWAITOverPloverDataNOWAITMax}{$39.1\times$\xspace}
\newcommand{\AdvmainLnXVIPayTPCCLogNumberofWorkerThreadsvsMaxThrCompareTaurusCommandNOWAITOverSerialRAIDCommandNOWAIT}{$3.2\times$\xspace}
\newcommand{\AdvmainLnXVIPayTPCCLogNumberofWorkerThreadsvsMaxThrCompareTaurusCommandNOWAITOverSerialRAIDCommandNOWAITMax}{$3.2\times$\xspace}
\newcommand{\AdvmainLnXVIPayTPCCLogNumberofWorkerThreadsvsMaxThrCompareTaurusDataNOWAITOverSerialDataNOWAIT}{$6.7\times$\xspace}
\newcommand{\AdvmainLnXVIPayTPCCLogNumberofWorkerThreadsvsMaxThrCompareTaurusDataNOWAITOverSerialDataNOWAITMax}{$6.9\times$\xspace}
\newcommand{\AdvmainLnXVIPayTPCCLogNumberofWorkerThreadsvsMaxThrCompareTaurusDataNOWAITOverPloverDataNOWAIT}{$14.7\times$\xspace}
\newcommand{\AdvmainLnXVIPayTPCCLogNumberofWorkerThreadsvsMaxThrCompareTaurusDataNOWAITOverPloverDataNOWAITMax}{$11.2\times$\xspace}
\newcommand{\AdvmainLnXVIPayTPCCLogNumberofWorkerThreadsvsMaxThrCompareTaurusDataNOWAITOverSerialRAIDDataNOWAIT}{4.6\%\xspace}
\newcommand{\AdvmainLnXVIPayTPCCLogNumberofWorkerThreadsvsMaxThrCompareTaurusDataNOWAITOverSerialRAIDDataNOWAITMax}{9.5\%\xspace}
\newcommand{\AdvmainLnXVIPayTPCCLogNumberofWorkerThreadsvsMaxThrCompareTaurusCommandNOWAITOverTaurusDataNOWAIT}{$3.6\times$\xspace}
\newcommand{\AdvmainLnXVIPayTPCCLogNumberofWorkerThreadsvsMaxThrCompareTaurusCommandNOWAITOverTaurusDataNOWAITMax}{$3.5\times$\xspace}
\newcommand{\AdvmainLnXVIPayTPCCLogNumberofWorkerThreadsvsMaxThrCompareTaurusCommandNOWAITOverSiloRDataSilo}{$23.4\times$\xspace}
\newcommand{\AdvmainLnXVIPayTPCCLogNumberofWorkerThreadsvsMaxThrCompareTaurusCommandNOWAITOverSiloRDataSiloMax}{$23.4\times$\xspace}
\newcommand{\AdvmainLnXVIPayTPCCLogNumberofWorkerThreadsvsMaxThrCompareTaurusDataNOWAITOverSiloRDataSilo}{$6.5\times$\xspace}
\newcommand{\AdvmainLnXVIPayTPCCLogNumberofWorkerThreadsvsMaxThrCompareTaurusDataNOWAITOverSiloRDataSiloMax}{$6.7\times$\xspace}

%% file: revision/figs/Ln16-Tm1-TPCC-Lr1-Number-of-Worker-Threads-vs-MaxThr.tex
\newcommand{\AdvmainLnXVIPayTPCCRecNumberofWorkerThreadsvsMaxThrScalabilityNoLoggingDataNOWAIT}{0.0\%\xspace}
\newcommand{\AdvmainLnXVIPayTPCCRecNumberofWorkerThreadsvsMaxThrScalabilityTaurusCommandNOWAIT}{$30.8\times$\xspace}
\newcommand{\AdvmainLnXVIPayTPCCRecNumberofWorkerThreadsvsMaxThrScalabilityTaurusDataNOWAIT}{$7.5\times$\xspace}
\newcommand{\AdvmainLnXVIPayTPCCRecNumberofWorkerThreadsvsMaxThrCompareTaurusCommandNOWAITOverNoLoggingDataNOWAIT}{$36.5\times$\xspace}
\newcommand{\AdvmainLnXVIPayTPCCRecNumberofWorkerThreadsvsMaxThrCompareTaurusCommandNOWAITOverNoLoggingDataNOWAITMax}{$36.5\times$\xspace}
\newcommand{\AdvmainLnXVIPayTPCCRecNumberofWorkerThreadsvsMaxThrCompareTaurusCommandNOWAITOverSerialCommandNOWAIT}{$57.3\times$\xspace}
\newcommand{\AdvmainLnXVIPayTPCCRecNumberofWorkerThreadsvsMaxThrCompareTaurusCommandNOWAITOverSerialCommandNOWAITMax}{$57.3\times$\xspace}
\newcommand{\AdvmainLnXVIPayTPCCRecNumberofWorkerThreadsvsMaxThrCompareTaurusCommandNOWAITOverSerialDataNOWAIT}{$35.0\times$\xspace}
\newcommand{\AdvmainLnXVIPayTPCCRecNumberofWorkerThreadsvsMaxThrCompareTaurusCommandNOWAITOverSerialDataNOWAITMax}{$34.9\times$\xspace}
\newcommand{\AdvmainLnXVIPayTPCCRecNumberofWorkerThreadsvsMaxThrCompareTaurusCommandNOWAITOverPloverDataNOWAIT}{$6.3\times$\xspace}
\newcommand{\AdvmainLnXVIPayTPCCRecNumberofWorkerThreadsvsMaxThrCompareTaurusCommandNOWAITOverPloverDataNOWAITMax}{$3.8\times$\xspace}
\newcommand{\AdvmainLnXVIPayTPCCRecNumberofWorkerThreadsvsMaxThrCompareTaurusCommandNOWAITOverSerialRAIDCommandNOWAIT}{$57.1\times$\xspace}
\newcommand{\AdvmainLnXVIPayTPCCRecNumberofWorkerThreadsvsMaxThrCompareTaurusCommandNOWAITOverSerialRAIDCommandNOWAITMax}{$57.1\times$\xspace}
\newcommand{\AdvmainLnXVIPayTPCCRecNumberofWorkerThreadsvsMaxThrCompareTaurusDataNOWAITOverSerialDataNOWAIT}{$8.8\times$\xspace}
\newcommand{\AdvmainLnXVIPayTPCCRecNumberofWorkerThreadsvsMaxThrCompareTaurusDataNOWAITOverSerialDataNOWAITMax}{$8.8\times$\xspace}
\newcommand{\AdvmainLnXVIPayTPCCRecNumberofWorkerThreadsvsMaxThrCompareTaurusDataNOWAITOverPloverDataNOWAIT}{$1.6\times$\xspace}
\newcommand{\AdvmainLnXVIPayTPCCRecNumberofWorkerThreadsvsMaxThrCompareTaurusDataNOWAITOverPloverDataNOWAITMax}{5.0\%\xspace}
\newcommand{\AdvmainLnXVIPayTPCCRecNumberofWorkerThreadsvsMaxThrCompareTaurusDataNOWAITOverSerialRAIDDataNOWAIT}{$8.9\times$\xspace}
\newcommand{\AdvmainLnXVIPayTPCCRecNumberofWorkerThreadsvsMaxThrCompareTaurusDataNOWAITOverSerialRAIDDataNOWAITMax}{$8.8\times$\xspace}
\newcommand{\AdvmainLnXVIPayTPCCRecNumberofWorkerThreadsvsMaxThrCompareTaurusCommandNOWAITOverTaurusDataNOWAIT}{$4.0\times$\xspace}
\newcommand{\AdvmainLnXVIPayTPCCRecNumberofWorkerThreadsvsMaxThrCompareTaurusCommandNOWAITOverTaurusDataNOWAITMax}{$4.0\times$\xspace}
\newcommand{\AdvmainLnXVIPayTPCCRecNumberofWorkerThreadsvsMaxThrCompareTaurusCommandNOWAITOverSiloRDataSilo}{$36.5\times$\xspace}
\newcommand{\AdvmainLnXVIPayTPCCRecNumberofWorkerThreadsvsMaxThrCompareTaurusCommandNOWAITOverSiloRDataSiloMax}{$36.5\times$\xspace}
\newcommand{\AdvmainLnXVIPayTPCCRecNumberofWorkerThreadsvsMaxThrCompareTaurusDataNOWAITOverSiloRDataSilo}{$9.2\times$\xspace}
\newcommand{\AdvmainLnXVIPayTPCCRecNumberofWorkerThreadsvsMaxThrCompareTaurusDataNOWAITOverSiloRDataSiloMax}{$9.2\times$\xspace}

%% file: revision/figs/Ln16-YCSB-Lr0-Number-of-Worker-Threads-vs-MaxThr.tex
\newcommand{\AdvmainLnXVIYCSBLogNumberofWorkerThreadsvsMaxThrScalabilityNoLoggingDataNOWAIT}{$70.8\times$\xspace}
\newcommand{\AdvmainLnXVIYCSBLogNumberofWorkerThreadsvsMaxThrScalabilityTaurusCommandNOWAIT}{$61.6\times$\xspace}
\newcommand{\AdvmainLnXVIYCSBLogNumberofWorkerThreadsvsMaxThrScalabilityTaurusDataNOWAIT}{$41.0\times$\xspace}
\newcommand{\AdvmainLnXVIYCSBLogNumberofWorkerThreadsvsMaxThrCompareTaurusCommandNOWAITOverNoLoggingDataNOWAIT}{29.5\%\xspace}
\newcommand{\AdvmainLnXVIYCSBLogNumberofWorkerThreadsvsMaxThrCompareTaurusCommandNOWAITOverNoLoggingDataNOWAITMax}{29.5\%\xspace}
\newcommand{\AdvmainLnXVIYCSBLogNumberofWorkerThreadsvsMaxThrCompareTaurusCommandNOWAITOverSerialCommandNOWAIT}{$1.3\times$\xspace}
\newcommand{\AdvmainLnXVIYCSBLogNumberofWorkerThreadsvsMaxThrCompareTaurusCommandNOWAITOverSerialCommandNOWAITMax}{$1.3\times$\xspace}
\newcommand{\AdvmainLnXVIYCSBLogNumberofWorkerThreadsvsMaxThrCompareTaurusCommandNOWAITOverSerialDataNOWAIT}{$15.5\times$\xspace}
\newcommand{\AdvmainLnXVIYCSBLogNumberofWorkerThreadsvsMaxThrCompareTaurusCommandNOWAITOverSerialDataNOWAITMax}{$14.6\times$\xspace}
\newcommand{\AdvmainLnXVIYCSBLogNumberofWorkerThreadsvsMaxThrCompareTaurusCommandNOWAITOverPloverDataNOWAIT}{$2.4\times$\xspace}
\newcommand{\AdvmainLnXVIYCSBLogNumberofWorkerThreadsvsMaxThrCompareTaurusCommandNOWAITOverPloverDataNOWAITMax}{$2.4\times$\xspace}
\newcommand{\AdvmainLnXVIYCSBLogNumberofWorkerThreadsvsMaxThrCompareTaurusCommandNOWAITOverSerialRAIDCommandNOWAIT}{$1.7\times$\xspace}
\newcommand{\AdvmainLnXVIYCSBLogNumberofWorkerThreadsvsMaxThrCompareTaurusCommandNOWAITOverSerialRAIDCommandNOWAITMax}{$1.7\times$\xspace}
\newcommand{\AdvmainLnXVIYCSBLogNumberofWorkerThreadsvsMaxThrCompareTaurusDataNOWAITOverSerialDataNOWAIT}{$9.9\times$\xspace}
\newcommand{\AdvmainLnXVIYCSBLogNumberofWorkerThreadsvsMaxThrCompareTaurusDataNOWAITOverSerialDataNOWAITMax}{$10.0\times$\xspace}
\newcommand{\AdvmainLnXVIYCSBLogNumberofWorkerThreadsvsMaxThrCompareTaurusDataNOWAITOverPloverDataNOWAIT}{$1.6\times$\xspace}
\newcommand{\AdvmainLnXVIYCSBLogNumberofWorkerThreadsvsMaxThrCompareTaurusDataNOWAITOverPloverDataNOWAITMax}{$1.7\times$\xspace}
\newcommand{\AdvmainLnXVIYCSBLogNumberofWorkerThreadsvsMaxThrCompareTaurusDataNOWAITOverSerialRAIDDataNOWAIT}{$1.5\times$\xspace}
\newcommand{\AdvmainLnXVIYCSBLogNumberofWorkerThreadsvsMaxThrCompareTaurusDataNOWAITOverSerialRAIDDataNOWAITMax}{$1.6\times$\xspace}
\newcommand{\AdvmainLnXVIYCSBLogNumberofWorkerThreadsvsMaxThrCompareTaurusCommandNOWAITOverTaurusDataNOWAIT}{$1.6\times$\xspace}
\newcommand{\AdvmainLnXVIYCSBLogNumberofWorkerThreadsvsMaxThrCompareTaurusCommandNOWAITOverTaurusDataNOWAITMax}{$1.5\times$\xspace}
\newcommand{\AdvmainLnXVIYCSBLogNumberofWorkerThreadsvsMaxThrCompareTaurusCommandNOWAITOverSiloRDataSilo}{$14.6\times$\xspace}
\newcommand{\AdvmainLnXVIYCSBLogNumberofWorkerThreadsvsMaxThrCompareTaurusCommandNOWAITOverSiloRDataSiloMax}{$14.6\times$\xspace}
\newcommand{\AdvmainLnXVIYCSBLogNumberofWorkerThreadsvsMaxThrCompareTaurusDataNOWAITOverSiloRDataSilo}{$9.3\times$\xspace}
\newcommand{\AdvmainLnXVIYCSBLogNumberofWorkerThreadsvsMaxThrCompareTaurusDataNOWAITOverSiloRDataSiloMax}{$9.9\times$\xspace}

%% file: revision/figs/Ln16-YCSB-Lr1-Number-of-Worker-Threads-vs-MaxThr.tex
\newcommand{\AdvmainLnXVIYCSBRecNumberofWorkerThreadsvsMaxThrScalabilityNoLoggingDataNOWAIT}{0.0\%\xspace}
\newcommand{\AdvmainLnXVIYCSBRecNumberofWorkerThreadsvsMaxThrScalabilityTaurusCommandNOWAIT}{$49.4\times$\xspace}
\newcommand{\AdvmainLnXVIYCSBRecNumberofWorkerThreadsvsMaxThrScalabilityTaurusDataNOWAIT}{$30.9\times$\xspace}
\newcommand{\AdvmainLnXVIYCSBRecNumberofWorkerThreadsvsMaxThrCompareTaurusCommandNOWAITOverNoLoggingDataNOWAIT}{$16.0\times$\xspace}
\newcommand{\AdvmainLnXVIYCSBRecNumberofWorkerThreadsvsMaxThrCompareTaurusCommandNOWAITOverNoLoggingDataNOWAITMax}{$16.0\times$\xspace}
\newcommand{\AdvmainLnXVIYCSBRecNumberofWorkerThreadsvsMaxThrCompareTaurusCommandNOWAITOverSerialCommandNOWAIT}{$42.6\times$\xspace}
\newcommand{\AdvmainLnXVIYCSBRecNumberofWorkerThreadsvsMaxThrCompareTaurusCommandNOWAITOverSerialCommandNOWAITMax}{$42.6\times$\xspace}
\newcommand{\AdvmainLnXVIYCSBRecNumberofWorkerThreadsvsMaxThrCompareTaurusCommandNOWAITOverSerialDataNOWAIT}{$33.9\times$\xspace}
\newcommand{\AdvmainLnXVIYCSBRecNumberofWorkerThreadsvsMaxThrCompareTaurusCommandNOWAITOverSerialDataNOWAITMax}{$33.8\times$\xspace}
\newcommand{\AdvmainLnXVIYCSBRecNumberofWorkerThreadsvsMaxThrCompareTaurusCommandNOWAITOverPloverDataNOWAIT}{$1.1\times$\xspace}
\newcommand{\AdvmainLnXVIYCSBRecNumberofWorkerThreadsvsMaxThrCompareTaurusCommandNOWAITOverPloverDataNOWAITMax}{$1.0\times$\xspace}
\newcommand{\AdvmainLnXVIYCSBRecNumberofWorkerThreadsvsMaxThrCompareTaurusCommandNOWAITOverSerialRAIDCommandNOWAIT}{$42.6\times$\xspace}
\newcommand{\AdvmainLnXVIYCSBRecNumberofWorkerThreadsvsMaxThrCompareTaurusCommandNOWAITOverSerialRAIDCommandNOWAITMax}{$42.6\times$\xspace}
\newcommand{\AdvmainLnXVIYCSBRecNumberofWorkerThreadsvsMaxThrCompareTaurusDataNOWAITOverSerialDataNOWAIT}{$22.9\times$\xspace}
\newcommand{\AdvmainLnXVIYCSBRecNumberofWorkerThreadsvsMaxThrCompareTaurusDataNOWAITOverSerialDataNOWAITMax}{$23.1\times$\xspace}
\newcommand{\AdvmainLnXVIYCSBRecNumberofWorkerThreadsvsMaxThrCompareTaurusDataNOWAITOverPloverDataNOWAIT}{28.5\%\xspace}
\newcommand{\AdvmainLnXVIYCSBRecNumberofWorkerThreadsvsMaxThrCompareTaurusDataNOWAITOverPloverDataNOWAITMax}{29.7\%\xspace}
\newcommand{\AdvmainLnXVIYCSBRecNumberofWorkerThreadsvsMaxThrCompareTaurusDataNOWAITOverSerialRAIDDataNOWAIT}{$22.8\times$\xspace}
\newcommand{\AdvmainLnXVIYCSBRecNumberofWorkerThreadsvsMaxThrCompareTaurusDataNOWAITOverSerialRAIDDataNOWAITMax}{$23.1\times$\xspace}
\newcommand{\AdvmainLnXVIYCSBRecNumberofWorkerThreadsvsMaxThrCompareTaurusCommandNOWAITOverTaurusDataNOWAIT}{$1.5\times$\xspace}
\newcommand{\AdvmainLnXVIYCSBRecNumberofWorkerThreadsvsMaxThrCompareTaurusCommandNOWAITOverTaurusDataNOWAITMax}{$1.5\times$\xspace}
\newcommand{\AdvmainLnXVIYCSBRecNumberofWorkerThreadsvsMaxThrCompareTaurusCommandNOWAITOverSiloRDataSilo}{$16.0\times$\xspace}
\newcommand{\AdvmainLnXVIYCSBRecNumberofWorkerThreadsvsMaxThrCompareTaurusCommandNOWAITOverSiloRDataSiloMax}{$16.0\times$\xspace}
\newcommand{\AdvmainLnXVIYCSBRecNumberofWorkerThreadsvsMaxThrCompareTaurusDataNOWAITOverSiloRDataSilo}{$10.8\times$\xspace}
\newcommand{\AdvmainLnXVIYCSBRecNumberofWorkerThreadsvsMaxThrCompareTaurusDataNOWAITOverSiloRDataSiloMax}{$11.0\times$\xspace}

%% file: revision/figs/RAM-YCSB-Lr0-Number-of-Worker-Threads-vs-Throughput.tex
\newcommand{\mainRAMYCSBLogNumberofWorkerThreadsvsThroughputScalabilityNoLoggingDataNOWAIT}{$70.6\times$\xspace}
\newcommand{\mainRAMYCSBLogNumberofWorkerThreadsvsThroughputScalabilityTaurusCommandNOWAIT}{$64.7\times$\xspace}
\newcommand{\mainRAMYCSBLogNumberofWorkerThreadsvsThroughputScalabilityTaurusDataNOWAIT}{$56.8\times$\xspace}
\newcommand{\mainRAMYCSBLogNumberofWorkerThreadsvsThroughputCompareTaurusCommandNOWAITOverNoLoggingDataNOWAIT}{31.6\%\xspace}
\newcommand{\mainRAMYCSBLogNumberofWorkerThreadsvsThroughputCompareTaurusCommandNOWAITOverNoLoggingDataNOWAITMax}{31.6\%\xspace}
\newcommand{\mainRAMYCSBLogNumberofWorkerThreadsvsThroughputCompareTaurusCommandNOWAITOverSerialCommandNOWAIT}{$1.5\times$\xspace}
\newcommand{\mainRAMYCSBLogNumberofWorkerThreadsvsThroughputCompareTaurusCommandNOWAITOverSerialCommandNOWAITMax}{$1.5\times$\xspace}
\newcommand{\mainRAMYCSBLogNumberofWorkerThreadsvsThroughputCompareTaurusCommandNOWAITOverSerialDataNOWAIT}{$6.9\times$\xspace}
\newcommand{\mainRAMYCSBLogNumberofWorkerThreadsvsThroughputCompareTaurusCommandNOWAITOverSerialDataNOWAITMax}{$6.4\times$\xspace}
\newcommand{\mainRAMYCSBLogNumberofWorkerThreadsvsThroughputCompareTaurusCommandNOWAITOverPloverDataNOWAIT}{$2.6\times$\xspace}
\newcommand{\mainRAMYCSBLogNumberofWorkerThreadsvsThroughputCompareTaurusCommandNOWAITOverPloverDataNOWAITMax}{$2.6\times$\xspace}
\newcommand{\mainRAMYCSBLogNumberofWorkerThreadsvsThroughputCompareTaurusCommandNOWAITOverSerialRAIDCommandNOWAIT}{$14.2\times$\xspace}
\newcommand{\mainRAMYCSBLogNumberofWorkerThreadsvsThroughputCompareTaurusCommandNOWAITOverSerialRAIDCommandNOWAITMax}{$14.2\times$\xspace}
\newcommand{\mainRAMYCSBLogNumberofWorkerThreadsvsThroughputCompareTaurusDataNOWAITOverSerialDataNOWAIT}{$5.7\times$\xspace}
\newcommand{\mainRAMYCSBLogNumberofWorkerThreadsvsThroughputCompareTaurusDataNOWAITOverSerialDataNOWAITMax}{$5.3\times$\xspace}
\newcommand{\mainRAMYCSBLogNumberofWorkerThreadsvsThroughputCompareTaurusDataNOWAITOverPloverDataNOWAIT}{$2.1\times$\xspace}
\newcommand{\mainRAMYCSBLogNumberofWorkerThreadsvsThroughputCompareTaurusDataNOWAITOverPloverDataNOWAITMax}{$2.1\times$\xspace}
\newcommand{\mainRAMYCSBLogNumberofWorkerThreadsvsThroughputCompareTaurusDataNOWAITOverSerialRAIDDataNOWAIT}{$11.7\times$\xspace}
\newcommand{\mainRAMYCSBLogNumberofWorkerThreadsvsThroughputCompareTaurusDataNOWAITOverSerialRAIDDataNOWAITMax}{$11.7\times$\xspace}
\newcommand{\mainRAMYCSBLogNumberofWorkerThreadsvsThroughputCompareTaurusCommandNOWAITOverTaurusDataNOWAIT}{$1.2\times$\xspace}
\newcommand{\mainRAMYCSBLogNumberofWorkerThreadsvsThroughputCompareTaurusCommandNOWAITOverTaurusDataNOWAITMax}{$1.2\times$\xspace}
\newcommand{\mainRAMYCSBLogNumberofWorkerThreadsvsThroughputCompareTaurusCommandNOWAITOverSiloRDataSilo}{$1.2\times$\xspace}
\newcommand{\mainRAMYCSBLogNumberofWorkerThreadsvsThroughputCompareTaurusCommandNOWAITOverSiloRDataSiloMax}{$1.2\times$\xspace}
\newcommand{\mainRAMYCSBLogNumberofWorkerThreadsvsThroughputCompareTaurusDataNOWAITOverSiloRDataSilo}{$1.0\times$\xspace}
\newcommand{\mainRAMYCSBLogNumberofWorkerThreadsvsThroughputCompareTaurusDataNOWAITOverSiloRDataSiloMax}{$1.0\times$\xspace}

%% file: revision/figs/RAM-YCSB-Lr1-Number-of-Worker-Threads-vs-Throughput.tex
\newcommand{\mainRAMYCSBRecNumberofWorkerThreadsvsThroughputScalabilityNoLoggingDataNOWAIT}{0.0\%\xspace}
\newcommand{\mainRAMYCSBRecNumberofWorkerThreadsvsThroughputScalabilityTaurusCommandNOWAIT}{$51.8\times$\xspace}
\newcommand{\mainRAMYCSBRecNumberofWorkerThreadsvsThroughputScalabilityTaurusDataNOWAIT}{$36.3\times$\xspace}
\newcommand{\mainRAMYCSBRecNumberofWorkerThreadsvsThroughputCompareTaurusCommandNOWAITOverNoLoggingDataNOWAIT}{$16.6\times$\xspace}
\newcommand{\mainRAMYCSBRecNumberofWorkerThreadsvsThroughputCompareTaurusCommandNOWAITOverNoLoggingDataNOWAITMax}{$16.6\times$\xspace}
\newcommand{\mainRAMYCSBRecNumberofWorkerThreadsvsThroughputCompareTaurusCommandNOWAITOverSerialCommandNOWAIT}{$43.9\times$\xspace}
\newcommand{\mainRAMYCSBRecNumberofWorkerThreadsvsThroughputCompareTaurusCommandNOWAITOverSerialCommandNOWAITMax}{$43.9\times$\xspace}
\newcommand{\mainRAMYCSBRecNumberofWorkerThreadsvsThroughputCompareTaurusCommandNOWAITOverSerialDataNOWAIT}{$32.6\times$\xspace}
\newcommand{\mainRAMYCSBRecNumberofWorkerThreadsvsThroughputCompareTaurusCommandNOWAITOverSerialDataNOWAITMax}{$32.6\times$\xspace}
\newcommand{\mainRAMYCSBRecNumberofWorkerThreadsvsThroughputCompareTaurusCommandNOWAITOverPloverDataNOWAIT}{$1.2\times$\xspace}
\newcommand{\mainRAMYCSBRecNumberofWorkerThreadsvsThroughputCompareTaurusCommandNOWAITOverPloverDataNOWAITMax}{$1.2\times$\xspace}
\newcommand{\mainRAMYCSBRecNumberofWorkerThreadsvsThroughputCompareTaurusCommandNOWAITOverSerialRAIDCommandNOWAIT}{$16.6\times$\xspace}
\newcommand{\mainRAMYCSBRecNumberofWorkerThreadsvsThroughputCompareTaurusCommandNOWAITOverSerialRAIDCommandNOWAITMax}{$16.6\times$\xspace}
\newcommand{\mainRAMYCSBRecNumberofWorkerThreadsvsThroughputCompareTaurusDataNOWAITOverSerialDataNOWAIT}{$29.7\times$\xspace}
\newcommand{\mainRAMYCSBRecNumberofWorkerThreadsvsThroughputCompareTaurusDataNOWAITOverSerialDataNOWAITMax}{$29.7\times$\xspace}
\newcommand{\mainRAMYCSBRecNumberofWorkerThreadsvsThroughputCompareTaurusDataNOWAITOverPloverDataNOWAIT}{$1.1\times$\xspace}
\newcommand{\mainRAMYCSBRecNumberofWorkerThreadsvsThroughputCompareTaurusDataNOWAITOverPloverDataNOWAITMax}{$1.1\times$\xspace}
\newcommand{\mainRAMYCSBRecNumberofWorkerThreadsvsThroughputCompareTaurusDataNOWAITOverSerialRAIDDataNOWAIT}{$15.1\times$\xspace}
\newcommand{\mainRAMYCSBRecNumberofWorkerThreadsvsThroughputCompareTaurusDataNOWAITOverSerialRAIDDataNOWAITMax}{$15.1\times$\xspace}
\newcommand{\mainRAMYCSBRecNumberofWorkerThreadsvsThroughputCompareTaurusCommandNOWAITOverTaurusDataNOWAIT}{$1.1\times$\xspace}
\newcommand{\mainRAMYCSBRecNumberofWorkerThreadsvsThroughputCompareTaurusCommandNOWAITOverTaurusDataNOWAITMax}{$1.1\times$\xspace}
\newcommand{\mainRAMYCSBRecNumberofWorkerThreadsvsThroughputCompareTaurusCommandNOWAITOverSiloRDataSilo}{$1.0\times$\xspace}
\newcommand{\mainRAMYCSBRecNumberofWorkerThreadsvsThroughputCompareTaurusCommandNOWAITOverSiloRDataSiloMax}{$1.0\times$\xspace}
\newcommand{\mainRAMYCSBRecNumberofWorkerThreadsvsThroughputCompareTaurusDataNOWAITOverSiloRDataSilo}{6.1\%\xspace}
\newcommand{\mainRAMYCSBRecNumberofWorkerThreadsvsThroughputCompareTaurusDataNOWAITOverSiloRDataSiloMax}{6.1\%\xspace}

%% file: revision/figs/Tm0-TPCC-Lr0-Number-of-Worker-Threads-vs-Throughput.tex
\newcommand{\mainNewOrderTPCCLogNumberofWorkerThreadsvsThroughputScalabilityNoLoggingDataNOWAIT}{$27.6\times$\xspace}
\newcommand{\mainNewOrderTPCCLogNumberofWorkerThreadsvsThroughputScalabilityTaurusCommandNOWAIT}{$27.9\times$\xspace}
\newcommand{\mainNewOrderTPCCLogNumberofWorkerThreadsvsThroughputScalabilityTaurusDataNOWAIT}{$8.1\times$\xspace}
\newcommand{\mainNewOrderTPCCLogNumberofWorkerThreadsvsThroughputCompareTaurusCommandNOWAITOverNoLoggingDataNOWAIT}{22.4\%\xspace}
\newcommand{\mainNewOrderTPCCLogNumberofWorkerThreadsvsThroughputCompareTaurusCommandNOWAITOverNoLoggingDataNOWAITMax}{22.4\%\xspace}
\newcommand{\mainNewOrderTPCCLogNumberofWorkerThreadsvsThroughputCompareTaurusCommandNOWAITOverSerialCommandNOWAIT}{$6.2\times$\xspace}
\newcommand{\mainNewOrderTPCCLogNumberofWorkerThreadsvsThroughputCompareTaurusCommandNOWAITOverSerialCommandNOWAITMax}{$6.1\times$\xspace}
\newcommand{\mainNewOrderTPCCLogNumberofWorkerThreadsvsThroughputCompareTaurusCommandNOWAITOverSerialDataNOWAIT}{$64.3\times$\xspace}
\newcommand{\mainNewOrderTPCCLogNumberofWorkerThreadsvsThroughputCompareTaurusCommandNOWAITOverSerialDataNOWAITMax}{$56.6\times$\xspace}
\newcommand{\mainNewOrderTPCCLogNumberofWorkerThreadsvsThroughputCompareTaurusCommandNOWAITOverPloverDataNOWAIT}{$11.2\times$\xspace}
\newcommand{\mainNewOrderTPCCLogNumberofWorkerThreadsvsThroughputCompareTaurusCommandNOWAITOverPloverDataNOWAITMax}{$10.0\times$\xspace}
\newcommand{\mainNewOrderTPCCLogNumberofWorkerThreadsvsThroughputCompareTaurusCommandNOWAITOverSerialRAIDCommandNOWAIT}{5.1\%\xspace}
\newcommand{\mainNewOrderTPCCLogNumberofWorkerThreadsvsThroughputCompareTaurusCommandNOWAITOverSerialRAIDCommandNOWAITMax}{5.1\%\xspace}
\newcommand{\mainNewOrderTPCCLogNumberofWorkerThreadsvsThroughputCompareTaurusDataNOWAITOverSerialDataNOWAIT}{$8.3\times$\xspace}
\newcommand{\mainNewOrderTPCCLogNumberofWorkerThreadsvsThroughputCompareTaurusDataNOWAITOverSerialDataNOWAITMax}{$7.5\times$\xspace}
\newcommand{\mainNewOrderTPCCLogNumberofWorkerThreadsvsThroughputCompareTaurusDataNOWAITOverPloverDataNOWAIT}{$1.4\times$\xspace}
\newcommand{\mainNewOrderTPCCLogNumberofWorkerThreadsvsThroughputCompareTaurusDataNOWAITOverPloverDataNOWAITMax}{$1.3\times$\xspace}
\newcommand{\mainNewOrderTPCCLogNumberofWorkerThreadsvsThroughputCompareTaurusDataNOWAITOverSerialRAIDDataNOWAIT}{$1.3\times$\xspace}
\newcommand{\mainNewOrderTPCCLogNumberofWorkerThreadsvsThroughputCompareTaurusDataNOWAITOverSerialRAIDDataNOWAITMax}{$1.3\times$\xspace}
\newcommand{\mainNewOrderTPCCLogNumberofWorkerThreadsvsThroughputCompareTaurusCommandNOWAITOverTaurusDataNOWAIT}{$7.7\times$\xspace}
\newcommand{\mainNewOrderTPCCLogNumberofWorkerThreadsvsThroughputCompareTaurusCommandNOWAITOverTaurusDataNOWAITMax}{$7.6\times$\xspace}
\newcommand{\mainNewOrderTPCCLogNumberofWorkerThreadsvsThroughputCompareTaurusCommandNOWAITOverSiloRDataSilo}{$7.7\times$\xspace}
\newcommand{\mainNewOrderTPCCLogNumberofWorkerThreadsvsThroughputCompareTaurusCommandNOWAITOverSiloRDataSiloMax}{$7.4\times$\xspace}
\newcommand{\mainNewOrderTPCCLogNumberofWorkerThreadsvsThroughputCompareTaurusDataNOWAITOverSiloRDataSilo}{0.0\%\xspace}
\newcommand{\mainNewOrderTPCCLogNumberofWorkerThreadsvsThroughputCompareTaurusDataNOWAITOverSiloRDataSiloMax}{1.5\%\xspace}

%% file: revision/figs/Tm1-TPCC-Lr0-Number-of-Worker-Threads-vs-Throughput.tex
\newcommand{\mainPayTPCCLogNumberofWorkerThreadsvsThroughputScalabilityNoLoggingDataNOWAIT}{$31.5\times$\xspace}
\newcommand{\mainPayTPCCLogNumberofWorkerThreadsvsThroughputScalabilityTaurusCommandNOWAIT}{$12.0\times$\xspace}
\newcommand{\mainPayTPCCLogNumberofWorkerThreadsvsThroughputScalabilityTaurusDataNOWAIT}{$10.2\times$\xspace}
\newcommand{\mainPayTPCCLogNumberofWorkerThreadsvsThroughputCompareTaurusCommandNOWAITOverNoLoggingDataNOWAIT}{72.6\%\xspace}
\newcommand{\mainPayTPCCLogNumberofWorkerThreadsvsThroughputCompareTaurusCommandNOWAITOverNoLoggingDataNOWAITMax}{72.1\%\xspace}
\newcommand{\mainPayTPCCLogNumberofWorkerThreadsvsThroughputCompareTaurusCommandNOWAITOverSerialCommandNOWAIT}{$6.8\times$\xspace}
\newcommand{\mainPayTPCCLogNumberofWorkerThreadsvsThroughputCompareTaurusCommandNOWAITOverSerialCommandNOWAITMax}{$6.8\times$\xspace}
\newcommand{\mainPayTPCCLogNumberofWorkerThreadsvsThroughputCompareTaurusCommandNOWAITOverSerialDataNOWAIT}{$51.2\times$\xspace}
\newcommand{\mainPayTPCCLogNumberofWorkerThreadsvsThroughputCompareTaurusCommandNOWAITOverSerialDataNOWAITMax}{$47.5\times$\xspace}
\newcommand{\mainPayTPCCLogNumberofWorkerThreadsvsThroughputCompareTaurusCommandNOWAITOverPloverDataNOWAIT}{$13.8\times$\xspace}
\newcommand{\mainPayTPCCLogNumberofWorkerThreadsvsThroughputCompareTaurusCommandNOWAITOverPloverDataNOWAITMax}{$14.1\times$\xspace}
\newcommand{\mainPayTPCCLogNumberofWorkerThreadsvsThroughputCompareTaurusCommandNOWAITOverSerialRAIDCommandNOWAIT}{$1.2\times$\xspace}
\newcommand{\mainPayTPCCLogNumberofWorkerThreadsvsThroughputCompareTaurusCommandNOWAITOverSerialRAIDCommandNOWAITMax}{$1.1\times$\xspace}
\newcommand{\mainPayTPCCLogNumberofWorkerThreadsvsThroughputCompareTaurusDataNOWAITOverSerialDataNOWAIT}{$9.5\times$\xspace}
\newcommand{\mainPayTPCCLogNumberofWorkerThreadsvsThroughputCompareTaurusDataNOWAITOverSerialDataNOWAITMax}{$8.7\times$\xspace}
\newcommand{\mainPayTPCCLogNumberofWorkerThreadsvsThroughputCompareTaurusDataNOWAITOverPloverDataNOWAIT}{$2.6\times$\xspace}
\newcommand{\mainPayTPCCLogNumberofWorkerThreadsvsThroughputCompareTaurusDataNOWAITOverPloverDataNOWAITMax}{$2.6\times$\xspace}
\newcommand{\mainPayTPCCLogNumberofWorkerThreadsvsThroughputCompareTaurusDataNOWAITOverSerialRAIDDataNOWAIT}{$1.3\times$\xspace}
\newcommand{\mainPayTPCCLogNumberofWorkerThreadsvsThroughputCompareTaurusDataNOWAITOverSerialRAIDDataNOWAITMax}{$1.3\times$\xspace}
\newcommand{\mainPayTPCCLogNumberofWorkerThreadsvsThroughputCompareTaurusCommandNOWAITOverTaurusDataNOWAIT}{$5.4\times$\xspace}
\newcommand{\mainPayTPCCLogNumberofWorkerThreadsvsThroughputCompareTaurusCommandNOWAITOverTaurusDataNOWAITMax}{$5.5\times$\xspace}
\newcommand{\mainPayTPCCLogNumberofWorkerThreadsvsThroughputCompareTaurusCommandNOWAITOverSiloRDataSilo}{$5.1\times$\xspace}
\newcommand{\mainPayTPCCLogNumberofWorkerThreadsvsThroughputCompareTaurusCommandNOWAITOverSiloRDataSiloMax}{$5.2\times$\xspace}
\newcommand{\mainPayTPCCLogNumberofWorkerThreadsvsThroughputCompareTaurusDataNOWAITOverSiloRDataSilo}{5.2\%\xspace}
\newcommand{\mainPayTPCCLogNumberofWorkerThreadsvsThroughputCompareTaurusDataNOWAITOverSiloRDataSiloMax}{5.1\%\xspace}

%% file: revision/figs/TPCF-Lr0-Number-of-Worker-Threads-vs-Throughput.tex
\newcommand{\mainTPCFLogNumberofWorkerThreadsvsThroughputScalabilityNoLoggingDataNOWAIT}{$14.4\times$\xspace}
\newcommand{\mainTPCFLogNumberofWorkerThreadsvsThroughputScalabilityTaurusCommandNOWAIT}{$14.4\times$\xspace}
\newcommand{\mainTPCFLogNumberofWorkerThreadsvsThroughputScalabilityTaurusDataNOWAIT}{$14.7\times$\xspace}
\newcommand{\mainTPCFLogNumberofWorkerThreadsvsThroughputCompareTaurusCommandNOWAITOverNoLoggingDataNOWAIT}{13.0\%\xspace}
\newcommand{\mainTPCFLogNumberofWorkerThreadsvsThroughputCompareTaurusCommandNOWAITOverNoLoggingDataNOWAITMax}{11.7\%\xspace}
\newcommand{\mainTPCFLogNumberofWorkerThreadsvsThroughputCompareTaurusCommandNOWAITOverSerialCommandNOWAIT}{9.6\%\xspace}
\newcommand{\mainTPCFLogNumberofWorkerThreadsvsThroughputCompareTaurusCommandNOWAITOverSerialCommandNOWAITMax}{8.3\%\xspace}
\newcommand{\mainTPCFLogNumberofWorkerThreadsvsThroughputCompareTaurusCommandNOWAITOverSerialDataNOWAIT}{$3.2\times$\xspace}
\newcommand{\mainTPCFLogNumberofWorkerThreadsvsThroughputCompareTaurusCommandNOWAITOverSerialDataNOWAITMax}{$3.0\times$\xspace}
\newcommand{\mainTPCFLogNumberofWorkerThreadsvsThroughputCompareTaurusCommandNOWAITOverPloverDataNOWAIT}{34.7\%\xspace}
\newcommand{\mainTPCFLogNumberofWorkerThreadsvsThroughputCompareTaurusCommandNOWAITOverPloverDataNOWAITMax}{33.4\%\xspace}
\newcommand{\mainTPCFLogNumberofWorkerThreadsvsThroughputCompareTaurusCommandNOWAITOverSerialRAIDCommandNOWAIT}{7.4\%\xspace}
\newcommand{\mainTPCFLogNumberofWorkerThreadsvsThroughputCompareTaurusCommandNOWAITOverSerialRAIDCommandNOWAITMax}{5.5\%\xspace}
\newcommand{\mainTPCFLogNumberofWorkerThreadsvsThroughputCompareTaurusDataNOWAITOverSerialDataNOWAIT}{$3.1\times$\xspace}
\newcommand{\mainTPCFLogNumberofWorkerThreadsvsThroughputCompareTaurusDataNOWAITOverSerialDataNOWAITMax}{$2.9\times$\xspace}
\newcommand{\mainTPCFLogNumberofWorkerThreadsvsThroughputCompareTaurusDataNOWAITOverPloverDataNOWAIT}{36.8\%\xspace}
\newcommand{\mainTPCFLogNumberofWorkerThreadsvsThroughputCompareTaurusDataNOWAITOverPloverDataNOWAITMax}{36.8\%\xspace}
\newcommand{\mainTPCFLogNumberofWorkerThreadsvsThroughputCompareTaurusDataNOWAITOverSerialRAIDDataNOWAIT}{5.0\%\xspace}
\newcommand{\mainTPCFLogNumberofWorkerThreadsvsThroughputCompareTaurusDataNOWAITOverSerialRAIDDataNOWAITMax}{5.0\%\xspace}
\newcommand{\mainTPCFLogNumberofWorkerThreadsvsThroughputCompareTaurusCommandNOWAITOverTaurusDataNOWAIT}{$1.0\times$\xspace}
\newcommand{\mainTPCFLogNumberofWorkerThreadsvsThroughputCompareTaurusCommandNOWAITOverTaurusDataNOWAITMax}{$1.1\times$\xspace}
\newcommand{\mainTPCFLogNumberofWorkerThreadsvsThroughputCompareTaurusCommandNOWAITOverSiloRDataSilo}{34.7\%\xspace}
\newcommand{\mainTPCFLogNumberofWorkerThreadsvsThroughputCompareTaurusCommandNOWAITOverSiloRDataSiloMax}{33.4\%\xspace}
\newcommand{\mainTPCFLogNumberofWorkerThreadsvsThroughputCompareTaurusDataNOWAITOverSiloRDataSilo}{36.8\%\xspace}
\newcommand{\mainTPCFLogNumberofWorkerThreadsvsThroughputCompareTaurusDataNOWAITOverSiloRDataSiloMax}{36.8\%\xspace}

%% file: revision/figs/TPCF-Lr1-Number-of-Worker-Threads-vs-Throughput.tex
\newcommand{\mainTPCFRecNumberofWorkerThreadsvsThroughputScalabilityNoLoggingDataNOWAIT}{0.0\%\xspace}
\newcommand{\mainTPCFRecNumberofWorkerThreadsvsThroughputScalabilityTaurusCommandNOWAIT}{$8.4\times$\xspace}
\newcommand{\mainTPCFRecNumberofWorkerThreadsvsThroughputScalabilityTaurusDataNOWAIT}{$5.9\times$\xspace}
\newcommand{\mainTPCFRecNumberofWorkerThreadsvsThroughputCompareTaurusCommandNOWAITOverNoLoggingDataNOWAIT}{18.4\%\xspace}
\newcommand{\mainTPCFRecNumberofWorkerThreadsvsThroughputCompareTaurusCommandNOWAITOverNoLoggingDataNOWAITMax}{15.7\%\xspace}
\newcommand{\mainTPCFRecNumberofWorkerThreadsvsThroughputCompareTaurusCommandNOWAITOverSerialCommandNOWAIT}{$12.8\times$\xspace}
\newcommand{\mainTPCFRecNumberofWorkerThreadsvsThroughputCompareTaurusCommandNOWAITOverSerialCommandNOWAITMax}{$13.2\times$\xspace}
\newcommand{\mainTPCFRecNumberofWorkerThreadsvsThroughputCompareTaurusCommandNOWAITOverSerialDataNOWAIT}{$9.4\times$\xspace}
\newcommand{\mainTPCFRecNumberofWorkerThreadsvsThroughputCompareTaurusCommandNOWAITOverSerialDataNOWAITMax}{$9.7\times$\xspace}
\newcommand{\mainTPCFRecNumberofWorkerThreadsvsThroughputCompareTaurusCommandNOWAITOverPloverDataNOWAIT}{18.4\%\xspace}
\newcommand{\mainTPCFRecNumberofWorkerThreadsvsThroughputCompareTaurusCommandNOWAITOverPloverDataNOWAITMax}{15.7\%\xspace}
\newcommand{\mainTPCFRecNumberofWorkerThreadsvsThroughputCompareTaurusCommandNOWAITOverSerialRAIDCommandNOWAIT}{$13.6\times$\xspace}
\newcommand{\mainTPCFRecNumberofWorkerThreadsvsThroughputCompareTaurusCommandNOWAITOverSerialRAIDCommandNOWAITMax}{$13.9\times$\xspace}
\newcommand{\mainTPCFRecNumberofWorkerThreadsvsThroughputCompareTaurusDataNOWAITOverSerialDataNOWAIT}{$7.5\times$\xspace}
\newcommand{\mainTPCFRecNumberofWorkerThreadsvsThroughputCompareTaurusDataNOWAITOverSerialDataNOWAITMax}{$7.5\times$\xspace}
\newcommand{\mainTPCFRecNumberofWorkerThreadsvsThroughputCompareTaurusDataNOWAITOverPloverDataNOWAIT}{34.3\%\xspace}
\newcommand{\mainTPCFRecNumberofWorkerThreadsvsThroughputCompareTaurusDataNOWAITOverPloverDataNOWAITMax}{34.3\%\xspace}
\newcommand{\mainTPCFRecNumberofWorkerThreadsvsThroughputCompareTaurusDataNOWAITOverSerialRAIDDataNOWAIT}{$8.0\times$\xspace}
\newcommand{\mainTPCFRecNumberofWorkerThreadsvsThroughputCompareTaurusDataNOWAITOverSerialRAIDDataNOWAITMax}{$7.9\times$\xspace}
\newcommand{\mainTPCFRecNumberofWorkerThreadsvsThroughputCompareTaurusCommandNOWAITOverTaurusDataNOWAIT}{$1.2\times$\xspace}
\newcommand{\mainTPCFRecNumberofWorkerThreadsvsThroughputCompareTaurusCommandNOWAITOverTaurusDataNOWAITMax}{$1.3\times$\xspace}
\newcommand{\mainTPCFRecNumberofWorkerThreadsvsThroughputCompareTaurusCommandNOWAITOverSiloRDataSilo}{18.4\%\xspace}
\newcommand{\mainTPCFRecNumberofWorkerThreadsvsThroughputCompareTaurusCommandNOWAITOverSiloRDataSiloMax}{15.7\%\xspace}
\newcommand{\mainTPCFRecNumberofWorkerThreadsvsThroughputCompareTaurusDataNOWAITOverSiloRDataSilo}{34.3\%\xspace}
\newcommand{\mainTPCFRecNumberofWorkerThreadsvsThroughputCompareTaurusDataNOWAITOverSiloRDataSiloMax}{34.3\%\xspace}

%% file: revision/figs/YCSB-Lr0-Number-of-Worker-Threads-vs-Throughput.tex
\newcommand{\mainYCSBLogNumberofWorkerThreadsvsThroughputScalabilityNoLoggingDataNOWAIT}{$42.3\times$\xspace}
\newcommand{\mainYCSBLogNumberofWorkerThreadsvsThroughputScalabilityTaurusCommandNOWAIT}{$38.0\times$\xspace}
\newcommand{\mainYCSBLogNumberofWorkerThreadsvsThroughputScalabilityTaurusDataNOWAIT}{$7.6\times$\xspace}
\newcommand{\mainYCSBLogNumberofWorkerThreadsvsThroughputCompareTaurusCommandNOWAITOverNoLoggingDataNOWAIT}{26.9\%\xspace}
\newcommand{\mainYCSBLogNumberofWorkerThreadsvsThroughputCompareTaurusCommandNOWAITOverNoLoggingDataNOWAITMax}{26.9\%\xspace}
\newcommand{\mainYCSBLogNumberofWorkerThreadsvsThroughputCompareTaurusCommandNOWAITOverSerialCommandNOWAIT}{$3.8\times$\xspace}
\newcommand{\mainYCSBLogNumberofWorkerThreadsvsThroughputCompareTaurusCommandNOWAITOverSerialCommandNOWAITMax}{$3.0\times$\xspace}
\newcommand{\mainYCSBLogNumberofWorkerThreadsvsThroughputCompareTaurusCommandNOWAITOverSerialDataNOWAIT}{$66.4\times$\xspace}
\newcommand{\mainYCSBLogNumberofWorkerThreadsvsThroughputCompareTaurusCommandNOWAITOverSerialDataNOWAITMax}{$65.8\times$\xspace}
\newcommand{\mainYCSBLogNumberofWorkerThreadsvsThroughputCompareTaurusCommandNOWAITOverPloverDataNOWAIT}{$9.4\times$\xspace}
\newcommand{\mainYCSBLogNumberofWorkerThreadsvsThroughputCompareTaurusCommandNOWAITOverPloverDataNOWAITMax}{$9.4\times$\xspace}
\newcommand{\mainYCSBLogNumberofWorkerThreadsvsThroughputCompareTaurusCommandNOWAITOverSerialRAIDCommandNOWAIT}{$1.6\times$\xspace}
\newcommand{\mainYCSBLogNumberofWorkerThreadsvsThroughputCompareTaurusCommandNOWAITOverSerialRAIDCommandNOWAITMax}{$1.6\times$\xspace}
\newcommand{\mainYCSBLogNumberofWorkerThreadsvsThroughputCompareTaurusDataNOWAITOverSerialDataNOWAIT}{$7.1\times$\xspace}
\newcommand{\mainYCSBLogNumberofWorkerThreadsvsThroughputCompareTaurusDataNOWAITOverSerialDataNOWAITMax}{$7.1\times$\xspace}
\newcommand{\mainYCSBLogNumberofWorkerThreadsvsThroughputCompareTaurusDataNOWAITOverPloverDataNOWAIT}{$1.0\times$\xspace}
\newcommand{\mainYCSBLogNumberofWorkerThreadsvsThroughputCompareTaurusDataNOWAITOverPloverDataNOWAITMax}{$1.0\times$\xspace}
\newcommand{\mainYCSBLogNumberofWorkerThreadsvsThroughputCompareTaurusDataNOWAITOverSerialRAIDDataNOWAIT}{$1.3\times$\xspace}
\newcommand{\mainYCSBLogNumberofWorkerThreadsvsThroughputCompareTaurusDataNOWAITOverSerialRAIDDataNOWAITMax}{$1.3\times$\xspace}
\newcommand{\mainYCSBLogNumberofWorkerThreadsvsThroughputCompareTaurusCommandNOWAITOverTaurusDataNOWAIT}{$9.4\times$\xspace}
\newcommand{\mainYCSBLogNumberofWorkerThreadsvsThroughputCompareTaurusCommandNOWAITOverTaurusDataNOWAITMax}{$9.3\times$\xspace}
\newcommand{\mainYCSBLogNumberofWorkerThreadsvsThroughputCompareTaurusCommandNOWAITOverSiloRDataSilo}{$9.2\times$\xspace}
\newcommand{\mainYCSBLogNumberofWorkerThreadsvsThroughputCompareTaurusCommandNOWAITOverSiloRDataSiloMax}{$9.2\times$\xspace}
\newcommand{\mainYCSBLogNumberofWorkerThreadsvsThroughputCompareTaurusDataNOWAITOverSiloRDataSilo}{1.8\%\xspace}
\newcommand{\mainYCSBLogNumberofWorkerThreadsvsThroughputCompareTaurusDataNOWAITOverSiloRDataSiloMax}{0.8\%\xspace}

%% file: revision/figs/YCSB-Lr1-Number-of-Worker-Threads-vs-MaxThr.tex
\newcommand{\mainYCSBRecNumberofWorkerThreadsvsMaxThrScalabilityNoLoggingDataNOWAIT}{0.0\%\xspace}
\newcommand{\mainYCSBRecNumberofWorkerThreadsvsMaxThrScalabilityTaurusCommandNOWAIT}{$15.5\times$\xspace}
\newcommand{\mainYCSBRecNumberofWorkerThreadsvsMaxThrScalabilityTaurusDataNOWAIT}{$6.4\times$\xspace}
\newcommand{\mainYCSBRecNumberofWorkerThreadsvsMaxThrCompareTaurusCommandNOWAITOverNoLoggingDataNOWAIT}{$5.3\times$\xspace}
\newcommand{\mainYCSBRecNumberofWorkerThreadsvsMaxThrCompareTaurusCommandNOWAITOverNoLoggingDataNOWAITMax}{$5.3\times$\xspace}
\newcommand{\mainYCSBRecNumberofWorkerThreadsvsMaxThrCompareTaurusCommandNOWAITOverSerialCommandNOWAIT}{$11.3\times$\xspace}
\newcommand{\mainYCSBRecNumberofWorkerThreadsvsMaxThrCompareTaurusCommandNOWAITOverSerialCommandNOWAITMax}{$11.3\times$\xspace}
\newcommand{\mainYCSBRecNumberofWorkerThreadsvsMaxThrCompareTaurusCommandNOWAITOverSerialDataNOWAIT}{$36.0\times$\xspace}
\newcommand{\mainYCSBRecNumberofWorkerThreadsvsMaxThrCompareTaurusCommandNOWAITOverSerialDataNOWAITMax}{$36.1\times$\xspace}
\newcommand{\mainYCSBRecNumberofWorkerThreadsvsMaxThrCompareTaurusCommandNOWAITOverPloverDataNOWAIT}{$6.4\times$\xspace}
\newcommand{\mainYCSBRecNumberofWorkerThreadsvsMaxThrCompareTaurusCommandNOWAITOverPloverDataNOWAITMax}{$6.1\times$\xspace}
\newcommand{\mainYCSBRecNumberofWorkerThreadsvsMaxThrCompareTaurusCommandNOWAITOverSerialRAIDCommandNOWAIT}{$11.0\times$\xspace}
\newcommand{\mainYCSBRecNumberofWorkerThreadsvsMaxThrCompareTaurusCommandNOWAITOverSerialRAIDCommandNOWAITMax}{$10.9\times$\xspace}
\newcommand{\mainYCSBRecNumberofWorkerThreadsvsMaxThrCompareTaurusDataNOWAITOverSerialDataNOWAIT}{$5.5\times$\xspace}
\newcommand{\mainYCSBRecNumberofWorkerThreadsvsMaxThrCompareTaurusDataNOWAITOverSerialDataNOWAITMax}{$5.5\times$\xspace}
\newcommand{\mainYCSBRecNumberofWorkerThreadsvsMaxThrCompareTaurusDataNOWAITOverPloverDataNOWAIT}{2.8\%\xspace}
\newcommand{\mainYCSBRecNumberofWorkerThreadsvsMaxThrCompareTaurusDataNOWAITOverPloverDataNOWAITMax}{6.3\%\xspace}
\newcommand{\mainYCSBRecNumberofWorkerThreadsvsMaxThrCompareTaurusDataNOWAITOverSerialRAIDDataNOWAIT}{$1.7\times$\xspace}
\newcommand{\mainYCSBRecNumberofWorkerThreadsvsMaxThrCompareTaurusDataNOWAITOverSerialRAIDDataNOWAITMax}{$1.7\times$\xspace}
\newcommand{\mainYCSBRecNumberofWorkerThreadsvsMaxThrCompareTaurusCommandNOWAITOverTaurusDataNOWAIT}{$6.6\times$\xspace}
\newcommand{\mainYCSBRecNumberofWorkerThreadsvsMaxThrCompareTaurusCommandNOWAITOverTaurusDataNOWAITMax}{$6.6\times$\xspace}
\newcommand{\mainYCSBRecNumberofWorkerThreadsvsMaxThrCompareTaurusCommandNOWAITOverSiloRDataSilo}{$6.4\times$\xspace}
\newcommand{\mainYCSBRecNumberofWorkerThreadsvsMaxThrCompareTaurusCommandNOWAITOverSiloRDataSiloMax}{$6.3\times$\xspace}
\newcommand{\mainYCSBRecNumberofWorkerThreadsvsMaxThrCompareTaurusDataNOWAITOverSiloRDataSilo}{2.7\%\xspace}
\newcommand{\mainYCSBRecNumberofWorkerThreadsvsMaxThrCompareTaurusDataNOWAITOverSiloRDataSiloMax}{3.1\%\xspace}

%% file: abstract.tex
\begin{abstract}

Existing single-stream logging schemes are unsuitable for in-memory database 
management systems (DBMSs) as the single log is often a performance bottleneck. 
To overcome this problem, we present \name, an efficient parallel logging scheme that 
uses multiple log streams, and is compatible with both data and command logging.
\name tracks and encodes transaction dependencies using a vector of log sequence numbers 
(LSNs). These vectors ensure that the dependencies are 
fully captured in logging and correctly enforced in recovery.
\myhighlightpar{Our experimental evaluation with an in-memory DBMS shows that \name's parallel 
logging achieves up 
to
\AdvmainLnXVIYCSBLogNumberofWorkerThreadsvsMaxThrCompareTaurusDataNOWAITOverSerialDataNOWAIT
 and 
 \AdvmainLnXVIPayTPCCLogNumberofWorkerThreadsvsMaxThrCompareTaurusCommandNOWAITOverSerialCommandNOWAIT
 speedups over single-streamed data logging and command logging, respectively. It 
also enables the DBMS to recover up to 
\AdvmainLnXVIYCSBRecNumberofWorkerThreadsvsMaxThrCompareTaurusDataNOWAITOverSerialDataNOWAIT and
\AdvmainLnXVINewOrderTPCCRecNumberofWorkerThreadsvsMaxThrCompareTaurusCommandNOWAITOverSerialCommandNOWAIT
faster than 
these baselines for data and command logging, respectively. We also compare \name 
with two state-of-the-art parallel logging schemes and show that the DBMS achieves up to  
\AdvmainLnXVIsiloPayTPCCLogNumberofWorkerThreadsvsMaxThrCompareTaurusCommandSiloOverSiloRDataSilo
better performance on NVMe drives and \mainYCSBLogNumberofWorkerThreadsvsThroughputCompareTaurusCommandNOWAITOverSiloRDataSilo on HDDs.}

\end{abstract}

%% file: intro.tex
\section{Introduction} \label{sec:intro}

A database management system (DBMS) guarantees that a transaction's modifications to the database 
persist even if the system crashes. The most common method to enforce durability is 
\textit{write-ahead-logging}, where each transaction sequentially writes its changes to a persistent 
storage device (e.g., HDD, SSD, NVM) before it commits~\cite{mohan1992aries}. With increasing 
parallelism in modern multicore hardware and the rising trend of high-throughput in-memory DBMSs, 
the scalability bottleneck caused by sequential logging~\cite{johnson2010aether, tu13, 
wang2014using, zheng2014fast} is onerous, motivating the need for a parallel 
solution. 

It is non-trivial, however, to perform parallel logging because the system must 
ensure the correct recovery order of transactions. Although this 
is straightforward in sequential logging because the LSNs (the positions of transaction 
records in the log file) explicitly define the order of 
transactions, it is not easy to efficiently recover transactions that are distributed across 
multiple logs without central LSNs. A parallel logging scheme must maintain transactions' order 
information across multiple logs to recover correctly. 

There are several parallel logging and recovery proposals in the literature~\cite{johnson2010aether, 
tu13, wang2014using, zheng2014fast}. These previous designs, however, are limited in their scope and 
applicability. Some algorithms support only parallel data logging but not parallel command 
logging~\cite{tu13, zheng2014fast, hong2013kuafu}; some can only parallelize the recovery process 
but not the logging process~\cite{dewitt1984implementation, nakamura2019integration}; a few 
protocols assume NVM hardware but do not work for conventional storage devices~\cite{arulraj2015let, 
arulraj2016write, chatzistergiou2015rewind, fang2011high, huang2014nvram, kim2016nvwal, 
kimura2015foedus, wang2014scalable}. As such, previously proposed methods are insufficient for 
modern DBMSs in diverse operating environments.

To overcome these limitations, we present \textbf{\name}, a lightweight protocol that performs 
both logging and recovery in parallel, supports both data and command logging, and is compatible 
with multiple concurrency control schemes. \name achieves this by tracking the 
inter-transaction dependencies. The recovery algorithm uses this information to determine the order 
of transactions. \codename encodes dependencies into a vector of LSNs, which we define as the 
\textit{LSN Vector} (\dbLV). LSN Vectors are inspired by vector clocks to enforce partial orderings 
in message-passing systems \cite{fidge1987timestamps, mattern1988virtual}. To reduce the overhead of 
maintaining \dbLV{s}, \codename compresses the vector based on the observation that a DBMS 
can recover transactions with no dependencies in any order. Thus, \codename does not need to store 
many \dbLV{s}, thereby reducing the space overhead.

\myhighlightpar{We compare the performance of \codename to a sequential logging scheme (with and without RAID-0 
setups) and state-of-the-art 
parallel logging schemes (i.e., Silo-R~\cite{tu13, zheng2014fast}
\myhighlightpar{and Plover~\cite{zhou2020plover}}) on YCSB and TPC-C benchmarks. 
Our evaluation on eight NVMe SSDs shows that \codename with data logging 
outperforms 
sequential data logging by 
\AdvmainLnXVIYCSBLogNumberofWorkerThreadsvsMaxThrCompareTaurusDataNOWAITOverSerialDataNOWAIT
at runtime, 
and \codename with command 
logging outperforms the sequential command logging by 
\AdvmainLnXVIPayTPCCLogNumberofWorkerThreadsvsMaxThrCompareTaurusCommandNOWAITOverSerialCommandNOWAIT.
During recovery, \codename with data logging and command logging are 
\AdvmainLnXVIYCSBRecNumberofWorkerThreadsvsMaxThrCompareTaurusDataNOWAITOverSerialDataNOWAIT and
\AdvmainLnXVINewOrderTPCCRecNumberofWorkerThreadsvsMaxThrCompareTaurusCommandNOWAITOverSerialCommandNOWAIT
faster than the serial baselines, respectively. \name with data logging matches the performance of 
the other parallel schemes, and \name with command logging is 
\AdvmainLnXVIsiloPayTPCCLogNumberofWorkerThreadsvsMaxThrCompareTaurusCommandSiloOverSiloRDataSilo faster 
at both runtime and recovery. Another evaluation on eight HDDs shows that 
\name{} with command logging achieves 
\mainYCSBLogNumberofWorkerThreadsvsThroughputCompareTaurusCommandNOWAITOverSiloRDataSilo and
\mainYCSBRecNumberofWorkerThreadsvsMaxThrCompareTaurusCommandNOWAITOverSiloRDataSilo 
faster than these parallel algorithms in logging and recovery, respectively.}

The main contributions of this paper include:
\squishitemize
    \item
    We propose the \name parallel scheme that supports both command 
    logging and data logging. We formally prove the correctness and liveness in \cref{sec:proof}. 

    \item
    We propose optimizations to reduce the memory footprint of the dependency information that 
    \codename maintains and extensions for supporting multiple concurrency control algorithms.

    \item
    We evaluate \name against sequential and the parallel logging 
    schemes, and demonstrate its advantages and generality. 

    \item We open source \name and evaluation scripts at \url{https://github.com/yuxiamit/DBx1000\_logging}.
\squishend

%% file: background.tex
\section{Background}
\label{sec:background}

We first provide an overview of conventional serial logging protocols and 
then discuss the challenges of extending a logging algorithm to support a parallel environment.

\subsection{Serial Logging}
\label{sec:back-serial}

In a serial logging protocol, the DBMS constructs a single log stream for all transactions. The 
protocol maintains the ordering invariant that, if \txn{2} depends on \txn{1}, then the DBMS writes 
\txn{2} to disk after \txn{1}. The DBMS ensures a transaction commits only after it 
successfully writes the transaction's log records to disk.
During recovery, the DBMS reads the log sequentially, starting from the last checkpoint. The 
DBMS replays each transaction sequentially until it encounters an incomplete log record or the end 
of the file.

In general, there are two categories of logging schemes. The first is \textit{data logging}, where 
log records contain the physical modifications that transactions made to the database. The recovery 
process with this scheme is to re-apply these changes back to the database. The other category, 
called \textit{command logging}~\cite{malviya13}, reduces the amount of log data by only recording 
transactions' high-level commands (i.e., invocations of stored procedures). The log records for these 
commands are typically smaller in size than the physical changes made to the database. The recovery 
process involves more computation, as all transactions are re-executed sequentially. When the disk 
bandwidth is the bottleneck, command logging can substantially outperform data logging. 

Although serial logging is inherently sequential, one can improve its performance by using RAID 
disks that act as a single storage device to increase disk 
bandwidth~\cite{patterson1988case}. Serial logging can also support parallel recovery if the DBMS 
uses data logging~\cite{speer2007c, tu13, zheng2014fast}. But the fundamental property that 
distinguishes serial logging from parallel logging is that it relies on a single log stream that 
respects all the data dependencies among transactions. On a modern in-memory DBMS with many CPU 
cores, such a single log stream is a contention point that becomes a scalability 
bottleneck~\cite{tu13}. \myhighlightpar{Competing for the single atomic
LSN counter inhibits performance due to cache coherence traffic~\cite{yu2014}.} 

\subsection{Parallel Logging Challenges}
\label{sec:back-parallel}

\begin{figure}
    \centering
    \adjincludegraphics[scale=0.25, trim={0 {0.1\height} 0 0}, clip]{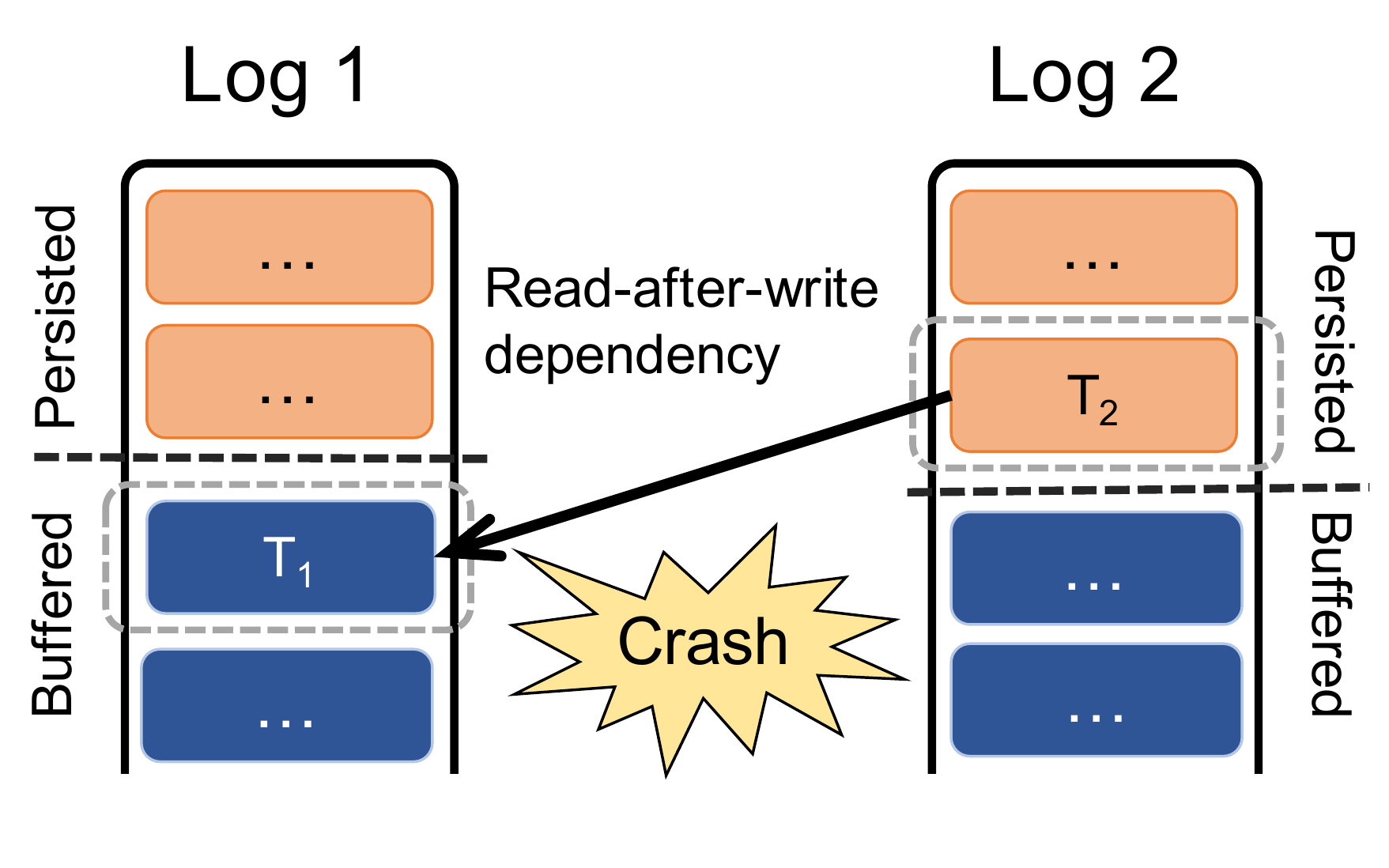}
    \caption{
        \textbf{Data Dependency in Parallel Logging} ---
        Transaction \txn{2} depends on \txn{1}.
        The two transactions write to different logs. %
    }
    \label{fig:back-challenges}
\end{figure}

Parallel logging allows transactions to write to multiple log streams (e.g., one stream per disk), 
thereby avoiding serial logging's scalability bottlenecks to satisfy the high throughput demands of 
in-memory DBMSs. Multiple streams inhibit an inherent natural ordering of transactions. 
Therefore, other mechanisms are required to track and enforce the ordering among these transactions. 
\cref{fig:back-challenges} shows an example with transactions \txn{1} and \txn{2}, where \txn{2} 
depends on \txn{1} with a read-after-write (RAW) data dependency. 

In this example, we assume that \txn{1} writes to \logger{1} and \txn{2} writes to \logger{2} and 
they may be flushed in any order. If \txn{2} is already persistent in \logger{2} while \txn{1} is 
still in the log buffer (shown in \cref{fig:back-challenges}), the DBMS must \textit{not} commit 
\txn{2} since \txn{1} has not committed. Furthermore, if the DBMS crashes, then the 
recovery process must be aware of such data dependency and therefore should not recover \txn{2}. 
Specifically, parallel logging faces the following three challenges.
\vspace*{-0.1in} \\

\textbf{Challenge \#1 -- When to Commit a Transaction:} %
The DBMS can only commit a transaction if it is persistent and all the 
transactions that it depends on can commit. In serial logging, this requirement
is satisfied if the transaction itself is persistent, indicating all the preceding transactions are also persistent. In parallel logging, however, a transaction must 
identify when other transactions that it depends on can commit, especially those that are 
storing their log records on other log streams. For the example shown in 
\cref{fig:back-challenges}, the DBMS can commit \txn{2} only after \txn{1} is already persistent. 
\vspace*{-0.1in} \\

\textbf{Challenge \#2 -- Whether to Recover a Transaction:} %
\myhighlightpar{A technique like \textit{Early-Lock-Release} (ELR) prevents transactions from waiting for 
log persistency during execution by allowing a transaction to release locks early before the log records hit disks~\cite{dewitt1984implementation}. But this means that during recovery, the DBMS has to 
determine whether transactions successfully committed before a crash.}
The DBMS ignores any transaction that fails to complete properly.
For the example in \cref{fig:back-challenges}, if \txn{2} is in the 
log but \txn{1} is not, then the DBMS should not process \txn{2} during recovery.
\vspace*{-0.1in} \\

\textbf{Challenge \#3 -- Determine the Recovery Order:} %
The DBMS must recover transactions in the order that respects data dependencies. 
If both \txn{1} and \txn{2} are persistent and have committed before 
the crash, the DBMS must recover \txn{1} before \txn{2}, since \txn{2} reads the value that is 
written by \txn{1}.
\vspace*{-0.1in} \\

One can resolve some of the above issues if the DBMS satisfies certain assumptions. 
For example, if the concurrency control algorithm enforces dependent transactions to write to disks 
in the corresponding order, then this solves the first and second challenges: the 
persistence of one transaction implies that any transactions that it depends on are also persistent. 
If the DBMS uses data logging, then it needs to handle write-after-write (\conflictWAW) 
dependencies, but not read-after-write (\conflictRAW) or \myhighlight{write-after-read 
(\conflictWAR)} dependencies. 
\myhighlightpar{For example, consider a transaction $\txn{1}$ that writes A=1, and a transaction $\txn{2}$ that
reads A and then writes B=A+1. Suppose the initial value of A is 0, and 
the DBMS schedules $\txn{2}$ before $\txn{1}$, resulting in A=1 and B=1. With this schedule, 
$\txn{1}$ has a \conflictWAR dependency on $\txn{2}$. If the DBMS does not 
track \conflictWAR dependencies and perform command logging, running $\txn{1}$ before $\txn{2}$ will result in 
A=1 and B=2, which violates correctness. But if the DBMS performs data logging, then $\txn{1}$ 
will have a record of A=1 and $\txn{2}$ will have a record of B=1. Regardless of the 
recovery order between $\txn{1}$ and $\txn{2}$, the resulting state is always correct. Supporting only data 
logging simplifies the protocol~\cite{tu13, zheng2014fast}. }
These assumptions, 
however, would hurt either performance or generality of the DBMS. 
\myhighlightpar{Our experiments in \cref{sec:evaluation} show that \name command logging outperforms 
all the data logging baselines by up to 6.4$\times$ in both logging and recovery.}

%% file: protocol.tex
\section{\name Parallel Logging}
\label{sec:protocol}

We now present the \name protocol in detail. The core idea of \name is to use a lightweight 
dependency tracking mechanism called \textit{LSN Vector}. After first describing LSN Vectors, we 
then explain how \name uses them in \cref{sec:taurus-logging} and 
\cref{sec:taurus-recovery} during runtime and recovery operations, respectively.
\myhighlightpar{We then discuss how \name{} supports index operations like range scan, 
insertions, and deletions. Lastly, we describe limitations of \name and 
potential solutions.}

Although \name supports multiple concurrency control schemes (see 
\cref{sec:protocol-occ-variant}), for the sake of simplicity, we assume strict two-phase locking 
(S2PL) in this section unless otherwise stated. We also assume that the DBMS uses multiple disks 
with each log file residing 
on one disk. Each transaction writes only a single log entry to one log file at commit time. The 
design of a single log entry per transaction simplifies the protocol and is used by other in-memory 
DBMSs, including Hekaton~\cite{diaconu13}, Silo~\cite{tu13, zheng2014fast}, and 
H-Store~\cite{kallman08}.

\subsection{LSN Vector}
\label{sec:protocol-lv}

\begin{figure}[t!]
    \centering
    \adjincludegraphics[width=0.55\columnwidth, trim={0 0 0 {0.04\height}}, clip]{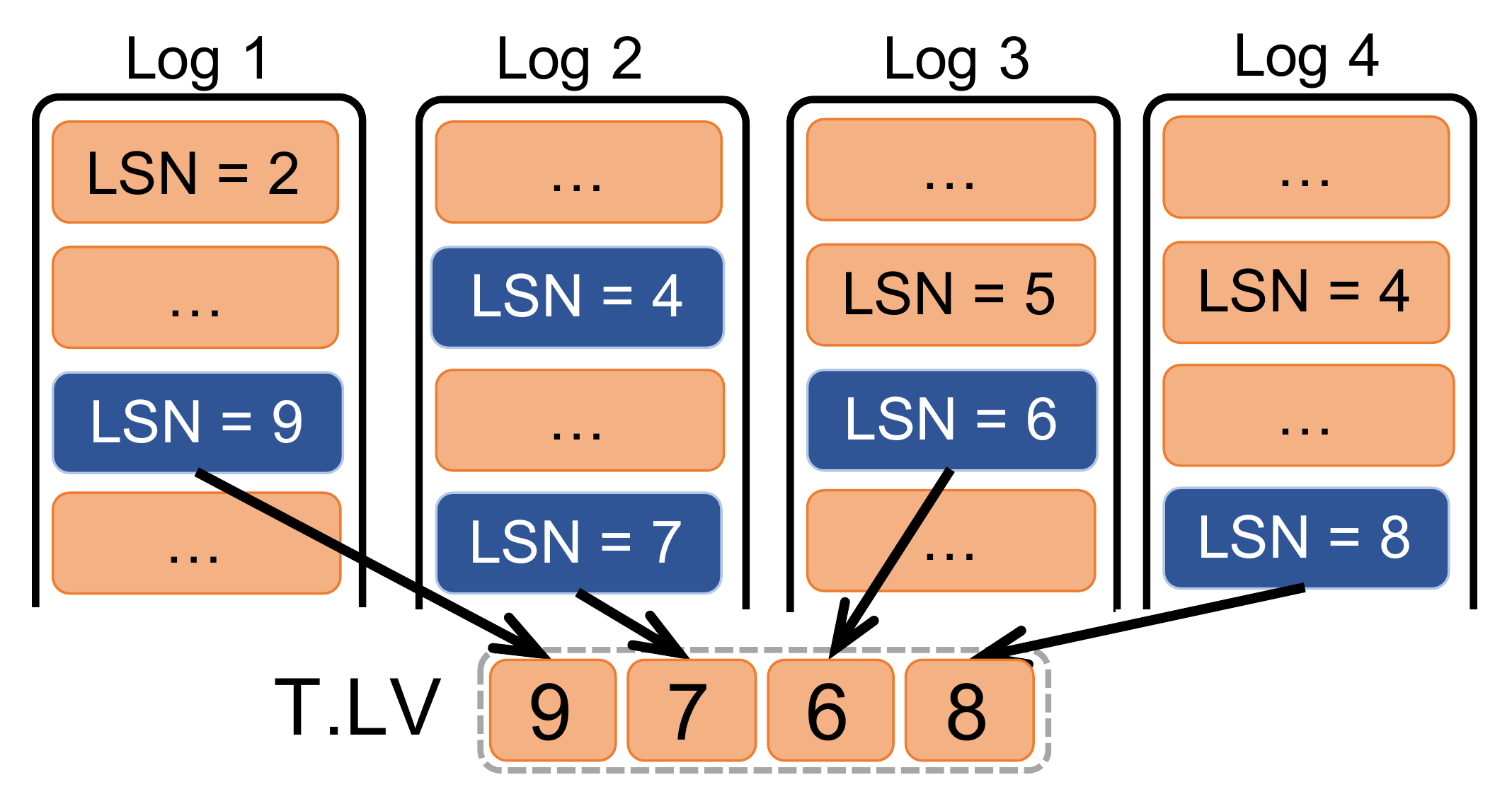}
    \caption{
        \textbf{LSN Vector (LV) example} --- 
        The $i^{th}$ element of transaction \txn{}'s LV is an LSN of log$_i$, indicating that \txn{} 
        depends on one or more transactions (rendered in dark blue) in log$_i$ before that LSN.
    }
    \label{fig:protocol-lv}
\end{figure}

An \textit{LSN Vector} (LV) is a vector of LSNs that encodes the dependencies between transactions. 
The DBMS assigns it to either (1) a transaction to record its dependency information or (2) a data 
item to capture the dependencies between transactions accessing it. The dimension of 
an \LV is the same as the number of logs. Each element of \LV indicates that a transaction \txn{} may depend 
on transactions before a certain position in the corresponding log. Specifically, given a 
transaction \txn{} and its assigned LV: {$\txn{}.\LV = (LV[1], LV[2], \dots, LV[n])$}, for any $1 \leq i 
\leq n$, the following property holds: 

{
\begin{property}\label{prop:lsn}
    Transaction $\txn{}$ does not depend on any transaction $\txn{}'$ that maps to the $i$-th log 
    with LSN $>LV[i]$. 
\end{property}
}

\cref{fig:protocol-lv} shows the LV of an example transaction \txn{}. The second element in $\txn{}.\LV$ is 
$7$, meaning that \txn{} may depend on any transaction that maps to \logger{2} with an LSN $\leq 7$ 
but no transaction with an LSN $> 7$. In this example, \txn{} depends on two transactions in 
\logger{2} with both LSNs no greater than $7$. 
The semantics of LV is similar to vector clocks~\cite{fidge1987timestamps, mattern1988virtual}. In 
particular, the following two operations will be frequently used on LVs: \elementWiseMax and 
\textit{comparison}. The \elementWiseMax is the element-wise maximum function:
{\small
\begin{align*}
    & LV = \textit{\elementWiseMax}(LV', LV'') 
    \Rightarrow \forall{i}, LV[i] = \textit{max}(LV'[i], LV''[i])
\end{align*}
}%
\myhighlightpar{For \emph{comparison}, the relationships are defined as follows:
{\small 
\begin{align*}
    LV &\leq LV' \iff \forall i, \LV[i] \leq \LV'[i].
\end{align*}
}
}
Following the semantics of vector clocks, \LV captures an approximation of the partial order among 
transactions --- LVs of dependent transactions are always ordered and LVs of independent 
transactions may or may not be ordered. 

An \LV of a transaction is written to the log together with the rest of the log entry. The 
dependency information captured by the partial order is sufficient for resolving the three 
challenges from \cref{sec:back-parallel}. \name's \LVs address these challenges in the 
following way: (1) A transaction \txn{} can commit if 
it is persistent and each log has flushed to the point specified by \txn{}.\LV, indicating that 
all transactions that \txn{} depends on are persistent. (2) During recovery, the DBMS determines 
that a transaction \txn{} has committed before the crash if each log has flushed to the point 
specified by \txn{}.\LV. (3) The recovery order follows the partial order specified by LVs, and 
the DBMS can recover unordered transactions in parallel.

\subsection{Logging Operations}
\label{sec:taurus-logging}

\myhighlightpar{The \name protocol is based on \textit{workers} and \textit{log 
managers} (denoted as $L_1, L_2, \dots, L_n$).}
Each log manager writes to a unique log file. Each worker is 
assigned to a log manager and we assume every log manager has exactly $p$ workers. The log 
managers and workers run on separate threads. We first describe the protocol's internal data 
structures and then explain its algorithms. 
\vspace*{-0.1in} \\

\textbf{Data Structures:} On top of a conventional 2PL protocol, \name adds the following data 
structures to the system.

\squishitemize
    \item \myitem{\txn{}.\LV} -- 
    Each transaction \txn{} contains a \txn{}.\LV encoding its dependency as 
    discussed in \cref{sec:protocol-lv}. When \txn{} initially starts, \txn{}.\LV is a vector of 
    zeroes.

    \item \myitem{\textit{Tuple.}\readLV/\writeLV} -- 
    Each tuple contains two LVs that serve as a medium for transaction \LVs to propagate between 
    transactions. Intuitively, these vectors are the maximum \LV of transactions that have 
    read/written the tuple. Initially, all elements are zeroes. This does not necessarily incur
    extra linear storage
    because the DBMS maintains this metadata in its lock table (cf. \cref{sec:lv-compression}).

    \item \myitem{\textit{L.logLSN}} -- 
    The highest position in the log file that has not been allocated for log manager \textit{L}. 
    It is initialized as zero. \myhighlightpar{Workers reserve space for log records by 
    incrementing \textit{L.logLSN}.}

    \item \myitem{\textit{L.allocatedLSN}} -- 
    A vector of length $p$ that stores the last LSN allocated by each worker of log manager 
    \textit{L}. Initially, all elements of \textit{allocatedLSN} are $\infty$. 

    \item \myitem{\textit{L.filledLSN}} -- 
    A vector of length $p$, storing the last LSN filled by each worker of log manager \textit{L}. 
    Initially, all elements are zeroes. The purpose of \textit{L.allocatedLSN} and 
    \textit{L.filledLSN} is to determine the point to which the log manager \textit{L} can safely 
    flush its log. They are irrelevant to the idea of LV but are important for the 
    correctness of \name. 

    \item \myitem{\textit{Global.PLV}} --
    \textit{PLV} stands for \textit{Persistent LSN Vector} that is a global vector of length $n$. 
    The element \PLV{}$_i$ denotes the LSN that log manager $L_i$ has successfully flushed up to. 
\squishend

\textbf{Worker Threads:}
\myhighlightpar{Worker threads track dependencies by enforcing partial orders on the LSN Vectors.} The logic of a worker thread is contained in the \textit{Lock} and \textit{Commit} functions shown 
in \cref{alg:logging}. The 2PL locking logic is in the \textit{FetchLock} function 
(\cref{line:fetchlock}); \name supports 
any variant of 2PL (e.g., deadlock-detection, no-wait, wait-and-die). 
After a worker thread acquires a lock, it 
executes Lines~\ref{line:element-wise-max-1}--\ref{line:element-wise-max-2} to update the LV of 
the transaction. 
\myhighlightpar{
It first updates $\txn{}.\LV$ to be the element-wise maximum of the current $\txn{}.\LV$ and the tuple's 
\writeLV (\cref{line:element-wise-max-1}). This enforces \txn{}.\LV to be no less than the LV of 
previous writing transactions. 
If the access is a write, it 
also updates $\txn{}.\LV$ using the tuple's \readLV.}

The DBMS calls the \textit{Commit} function shown in 
Lines~\ref{line:commit-start}--\ref{line:commit-end} when the transaction finishes. At this moment, 
\txn{} has locked tuples it accessed. Since the DBMS updates $\txn{}.\LV$ for 
each access, it already captures $\txn{}$'s dependency information. 
\myhighlight{The DBMS first checks if \textit{\txn{}} is read-only, and skip generating log records if so.}
Otherwise, it creates the log record 
for transaction $\txn{}$ (\cref{line:create-log-record}). The log record contains two parts: the 
\textit{redo log} \myhighlight{and a copy of \textit{\txn{}'s LV} at this moment}. The contents 
of the redo log depends on the logging scheme: the keys and values that $\txn{}$ 
modified (for data logging), or the information sufficient to reconstruct $\txn{}$ 
(for command logging). The DBMS writes the record into the corresponding log 
manager's buffer by \textit{WriteLogBuffer} (\cref{line:write-buffer}).
The algorithm then updates the $i$-th dimension of \txn{}.\LV to the returned LSN 
(\cref{line:lsn-update}), thereby allowing future transactions to capture their 
dependencies on \txn. \myhighlight{This update only changes \txn{}.\LV,
while the copy of \txn{}.\LV in the buffer does not contain this update.} %
Lines \ref{line:release-start}--\ref{line:release-end} update the \readLV and/or \writeLV of each 
tuple 
that \txn{} accessed before releasing the locks on those tuples. If \txn{} reads a tuple, 
it updates the tuple's \readLV using \txn{}.\LV, indicating that the tuple was read by 
\txn{} and future transactions must respect this dependency. Similarly, if \txn{} has written a 
tuple, the tuple's \writeLV is updated accordingly. Updating the LVs and releasing the 
lock must be executed in an atomic section, otherwise multiple transactions simultaneously updating 
the \readLV can cause race conditions leading to incorrect dependencies. 
As most 2PL protocols use 
latches to protect the release function, updating LVs can be piggybacked within those latches. %
For simplicity, we present a long atomic section covering Lines~\ref{line:release-start}--\ref{line:release-end} (shaded in gray).

After the DBMS releases transaction $\txn{}$'s locks, it has to wait for \PLV to catch up such 
that $\PLV \geq \txn{}.\LV$ (indicating $\txn{}$ is durable). All 
transactions within the same log manager commit sequentially. Since each log manager flushes 
records sequentially, this does not introduce a scalability bottleneck. %
\myhighlight{We employ the ELR optimization~\cite{dewitt1984implementation} to reduce lock
contention by allowing transactions to release locks before they are durable.}

The \textit{Commit} function calls \textit{WriteLogBuffer} 
(Lines~\ref{lines:write-buffer-begin}--\ref{lines:write-buffer-end}) to write a log entry into the 
log buffer.
It first allocates space in the log manager's ($L_i$) buffer by atomically incrementing its LSN by 
the size of the log 
record (\cref{line:atom-fetch-and-add}).
It then copies the log record into the log buffer (\cref{line:memcpy}). 
\cref{lines:allocate-greater-than-filled-begin,lines:allocate-greater-than-filled-end} are 
indicators for the log manager to decide up to which point it can flush the log buffer to disk. 
Specifically, before a transaction 
increments the LSN, it notifies the log manager ($L_i$) that its allocated space is no earlier than 
its current LSN (\cref{lines:allocate-greater-than-filled-begin}). This leads to 
$\textit{allocatedLSN}{[j]} \geq \textit{filledLSN}{[j]}$, which instructs $L_i$ that 
the log buffer contents after $\textit{allocatedLSN}{[j]}$ are unstable and should not be 
flushed to the disk. 
After the log buffer is filled, the transaction updates \textit{$L_i.\textit{filledLSN}[j]$} so 
that $\textit{allocatedLSN}{[j]} < \textit{filledLSN}{[j]}$, indicating that the worker thread has 
no ongoing operations on the log buffer.

\begin{algorithm}[t!]
    \caption{\small
        \textbf{Worker Thread} --- 
        We assume the worker is the $j$-th worker for log manager $L_i$.
    }\algoSize
    \label{alg:logging}
    {\input{algos/worker-thread.tex}}
\end{algorithm}

\myhighlightpar{To demonstrate how \name{} tracks dependencies, we use the 
example 
in \cref{fig:working-thread-example} with three transactions ($\txn{1}$, $\txn{2}$, $\txn{3}$) 
and two database objects A,B.
WLOG, we assume $\txn{1}$ and $\txn{2}$ are assigned to Log 1 and $\txn{3}$ is assigned to Log 2. 
In the beginning, A has a \writeLV [4,2] and a \readLV [3,7] while object B has [8,6] and 
[5,11]. 
\textbf{\circled{1}} The DBMS initializes the transactions' \LV{s} as [0,0].
\textbf{\circled{2}} $\txn{1}$ acquires an exclusive lock on A and writes to it. Then, $\txn{1}$ updates 
$\txn{1}.\LV$ to be the element-wise maximum among A.\writeLV, A.\readLV, and \txn{1}.\LV.
In this example, $\txn{1}.\LV$=[max(4,3,0), max(2,7,0)] = [4,7]. Enforcing these partial orders
enforces \conflictWAR and \conflictWAW dependencies. Namely, any previous transactions that ever 
read or wrote A will have an \LV no greater than $\txn{1}.\LV$.
\textbf{\circled{3}} $\txn{1}$ acquires a shared lock on B and then reads it. Then, $\txn{1}$ updates $\txn{1}.\LV$ 
to be the element-wise maximum among $B.\writeLV$ and $\txn{1}.\LV$. This is to track \conflictRAW 
dependencies. Now $\txn{1}.\LV$=[max(8, 4), max(6, 7)] = [8,7].
\textbf{\circled{4}} $\txn{2}$ wants to read A but has to wait for $\txn{1}$ to release the lock.
\textbf{\circled{5}} Similarly, $\txn{3}$ wants to write B but has to wait as well.
\textbf{\circled{6}} After $\txn{1}$ finishes, $\txn{1}$ writes its redo record and a copy of $\txn{1}.\LV$ into 
the log buffer. After successfully writing to the buffer, $\txn{1}$ learns 
its LSN in Log 1 is 16. Then, $\txn{1}$ updates the first dimension of $\txn{1}.\LV$ to be 16, resulting in 
$\txn{1}.\LV$=[16,7].
\textbf{\circled{7}} For each item $\txn{1}$ accessed, $\txn{1}$ updates the \readLV (or \writeLV) 
accordingly. $\txn{1}$ updates
$A.\writeLV=\elementWiseMax \allowbreak (A.\writeLV, \txn{1}.\LV) = \txn{1}.LV$ = [16,7], and 
$B.\readLV =\elementWiseMax \allowbreak (B.\readLV, \txn{1}.\LV)$ = [16,11]. 
Then, $\txn{1}$ releases the locks. After this, $\txn{1}$ waits for itself and all the transactions it depends 
on to become persistent, equivalently, $\PLV \geq \txn{1}.\LV$. The workers can process other transactions, 
and periodically check if $\txn{1}$ should be marked as committed.
\textbf{\circled{8}} $\txn{2}$ acquires the shared lock on A. $\txn{2}$ then updates 
$\txn{2}.\LV$=\elementWiseMax(\txn{2}.\LV, A.\writeLV) = [16,7]. This update 
enforces the partial order that $\txn{1}.\LV \leq \txn{2}.\LV$ because $\txn{2}$ depends on $\txn{1}$. Since $\txn{2}$ is 
read-only, it does not create a log record. It also enters the asynchronous commit by 
waiting for $\PLV \geq \txn{2}.\LV$.
\textbf{\circled{9}} $\txn{3}$ acquires an exclusive lock on B and 
updates 
$\txn{3}.\LV=$ $\elementWiseMax \allowbreak ($ $\txn{3}.\LV,$ $B.\readLV,$ $B.\writeLV)$ = [16,11]. The fact that $\txn{3}$ depends on 
$\txn{1}$ reflects on $\txn{3}.\LV \geq \txn{1}.\LV$.
\textbf{\circled{10}} The logging threads have flushed all transactions before $\txn{1}.\LV=\txn{2}.\LV$ = [16,7] 
and updated \PLV. Observing $\PLV\geq$ [16,7], \name{} marks $\txn{1}$ and $\txn{2}$ 
as committed.
\textbf{\circled{11}} $\txn{3}$ writes its redo record and a copy of $\txn{3}.\LV$ to the buffer of Log 2, 
and gets its LSN as 21. $\txn{3}.\LV$ increases to [16,21].
\textbf{\circled{12}} $\txn{3}$ sets B.\writeLV to [16, 21] and releases the lock.
\textbf{\circled{13}} When \PLV achieves $\txn{3}.\LV=[16, 21]$, \name{} commits $\txn{3}$.}
\vspace*{-0.1in} \\

\begin{figure}
    \includegraphics[scale=0.4]{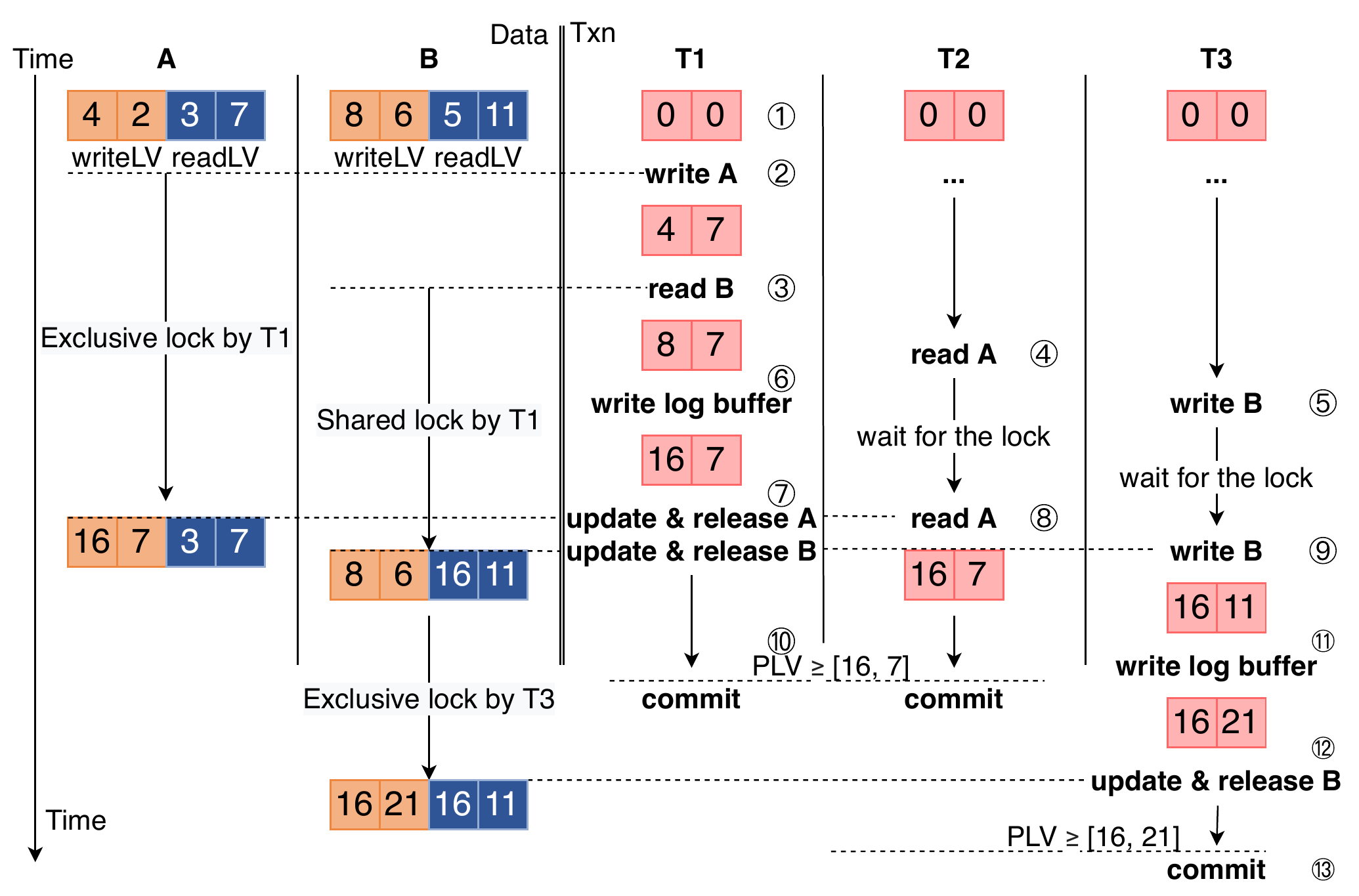}
    \caption{\myhighlightpar{\textbf{Worker Thread Example} -- Three transactions ($\txn{1}$, $\txn{2}$, and 
    \txn{3}) are accessing two objects A and B. Transactions are logged to two files. The diagram is 
    drawn in the time order with the axis on the left.}}
    \label{fig:working-thread-example}
\end{figure}

\begin{algorithm}[t!]
    \caption{\small
        \textbf{Log Manager Thread $L_i$}
    }\algoSize
    \label{alg:log-manager-logging}
    {\input{algos/manager-thread.tex} }
\end{algorithm}

\textbf{Log Manager Threads:}
We use a dedicated thread serving as the log manager for each log file. The main job of the log 
manager is to flush the contents in the log buffer into the file on disk. It periodically invokes 
\cref{alg:log-manager-logging} when a timeout period has 
passed or when the buffer is half full, whichever happens first. The algorithm
identifies up to which point the DBMS can flush to the disk so that it does not flush data 
that active transactions are still processing.

\name{} uses two arrays, \textit{allocatedLSN} and \textit{filledLSN}, to achieve this goal. 
\textit{readyLSN} is the log buffer position up to which the DBMS can safely flush; 
its initial value is \loglsn{}$[i]$ (\cref{line:init-readylsn}). For each worker thread $j$ that 
belongs to $L_i$, if \textit{allocatedLSN$[j]$} $\geq$\\ \textit{filledLSN$[j]$}, then the transaction 
in thread $j$ is filling the log buffer at a position after \textit{allocatedLSN$[j]$} 
(\cref{alg:logging}, \cref{lines:allocate-greater-than-filled-begin} and 
\cref{lines:allocate-greater-than-filled-end}), so \textit{readyLSN} should not be greater than 
\textit{allocatedLSN$[j]$}. Otherwise, no transaction in worker $j$ is filling the log buffer, so 
\textit{readyLSN} is not changed (Lines~\ref{line:state-check-begin}--\ref{line:state-check-end}). 
Lastly, the log manager flushes the buffer to the disk up to \textit{readyLSN} and updates 
$\PLV[i]$ (Lines~\ref{lines:try-flush-end}--\ref{lines:update-readylsn}). 

\myhighlightpar{The frequency that the DBMS flushes log records to disk is based on the performance profile of the storage devices. Although each flush might enable a number of transactions to commit, transactions in the same log file still commit in a sequential order. This removes ambiguity of transaction dependency during recovery. Sequential committing will not affect scalability because ELR prevents transactions waiting for log record duration or sequential committing from being on the critical path.
}

\subsection{Recovery Operations}
\label{sec:taurus-recovery}
\myhighlightpar{
\name{}' recovery algorithm replays transactions following the partial orders between 
their \LVs, which is sufficient to respect all the data dependencies. This is equivalent to 
performing topological sorting in parallel on a dependency graph.
Each log manager thread reads log records from a file, and the worker threads recover 
transactions by re-applying the log records.}
\vspace*{-0.1in} \\

\textbf{Data Structures:}
The recovery process contains the following:

\squishitemize
    \item \myitem{\textit{L.pool}} --
    For each log manager, \textit{pool} is a queue containing transactions 
    that are read from the log but not recovered.

    \item \myitem{\textit{L.maxLSN}} -- 
    For each log manager, \textit{maxLSN} is the LSN of the latest 
    transaction that has been read from the log file.

    \item \myitem{\textit{Global.\RLV}} --
    \RLV is a vector of length $n$ ($n$ is the number of log managers). 
    An element \RLV{}$_i$ means that all transactions mapping to $L_i$ with $\textit{LSN} \leq 
    \RLV{}_i$ have been successfully recovered. Therefore, a transaction $\txn{}$ can start its recovery 
    if $\txn{}.\LV \leq \RLV$, at which point all transactions that $\txn{}$ depends on have been recovered. 
    Initially, \RLV is a vector of zeroes.

    \item \myitem{\textit{Global.\ELV}} --
    \ELV is a vector of length $n$. An element \ELV{}$_i$ is the number 
    of bytes in \logger{i}. The DBMS uses this vector to determine if a transaction committed 
    before the crash.
     \myhighlightpar{Before the recovery starts, \name{} fetches the sizes of the 
     log files to initialize \ELV, namely, $\ELV[i]$ is the size of \logger{i}.}
\squishend

\begin{algorithm}[t!]
    \caption{\small\textbf{Log Manager Recovery for Thread $L_i$.}}
    \label{alg:log-manager-recovery}
    {\algoSize \input{algos/manager-recovery.tex}}
\end{algorithm}

\textbf{Log Manager Threads:}
In \cref{alg:log-manager-recovery}, the thread reads the log file and decodes 
records into transactions (\cref{line:decode-next}). For a transaction $\txn{}$, if $\txn{}.LV \leq \ELV$, 
then $\txn{}$ committed before the crash and is therefore considered for recovery; 
otherwise, $\txn{}$ and transactions after it are ignored for recovery. A transaction is enqueued 
into the tail of $pool$ and the value of \textit{maxLSN} is updated to be the LSN of \txn{} 
(Lines~\ref{line:enqueue}--\ref{line:update-maxlsn}). It is important that the thread updates 
\textit{maxLSN} \underline{after} it executes \textit{Enqueue}, otherwise the DBMS may 
recover transactions in an incorrect order.
\myhighlightpar{If the pool is empty after the DBMS updates \textit{maxLSN} but before it enqueues 
$\txn{}$, then it sets \RLV[i]=$\txn{}.LSN$ to indicate that $\txn{}$ is recovered; this prevents another 
worker from recovering a transaction that depends on $\txn{}$ before $\txn{}$ is recovered.}
\vspace*{-0.1in} \\

\begin{algorithm}[t!]
    \caption{\small\textbf{Worker Recovery Thread}}
    \label{alg:recovery-worker}
    {\algoSize \input{algos/worker-recovery.tex}}
\end{algorithm}

\textbf{Worker Threads:}
In \cref{alg:recovery-worker}, the worker threads keep executing until the log manager finishes 
decoding all the transactions and the pool is empty. A worker thread tries to get a transaction $\txn{}$ 
from $pool$ such that $\txn{}.\LV \! \leq \! \RLV$ (\cref{line:pool-fetch-next}). Then, the worker 
thread 
recovers $\txn{}$ (\cref{line:recover-txn}). For data logging, the data elements in the log record are 
copied to the database; for command logging, the transaction is re-executed. During the 
re-execution, no concurrency control algorithm is needed, since \name guarantees no conflicts during 
recovery. Then, $\RLV[i]$ is updated (Lines~\ref{line:rlv-update-begin}-\ref{line:rlv-update-end}).
\myhighlightpar{If \textit{pool} is empty, then the thread sets $\RLV{}[i]$ to 
$\textit{pool.maxLSN}$, the largest LSN of any transaction added to \textit{pool}, if it is larger;} 
otherwise, 
$\RLV{}[i]$ is set to 
one less than the first transaction's LSN, indicating that the previous transaction has been 
recovered but not the one blocking the head of \textit{pool}. In the pseudo-code, the 
code for \RLV update is protected with an atomic section for correctness. We 
use a lock-free design to avoid this critical section in our implementation.

The \textit{pool} data structure described above can become a potential scalability bottleneck if a 
large number of workers are mapped to a single log manager. There are additional optimizations that 
address this issue. For example, we partition each 
$pool$ into multiple queues. We also split \RLV into local 
copies and add delegations to reduce false sharing in CPU caches. 

\subsection{\myhighlight{Supporting Index Operations}}
\label{sec:index-ops}
\myhighlightpar{
Although our discussion has focused on \emph{read} and \emph{update} operations, \name can also 
support \emph{scan}, \emph{insert}, and \emph{delete} operations with an additional index locking 
protocol. %
}

\myhighlightpar{For a range scan, the transaction (atomically) fetches a shared lock on each of the 
result rows using the \textit{Lock} function in \cref{alg:logging}. When the transaction commits, it 
goes through the \textit{Commit} function and update the \readLV's of the rows. To avoid phantoms, 
the transaction performs the same scan again before releasing the locks in \textit{Commit} function. 
If the result rows are different, some other transactions have inserted or deleted rows within the 
scan range, we abort the transaction. This scan-twice trick is from Silo \cite{tu13}. We notice 
that, assuming 2PL, the transaction only needs to record the number of rows returned. During the 
second scan, it's guaranteed that the rows in the previous scan still exist because shared locks 
are held by the transaction. Therefore, if the row count is still the same, the result rows are not 
changed. %
}

\myhighlightpar{If a transaction $\txn{}$ inserts a row with primary key \textit{key}, it initializes 
\textit{DB[key].\readLV} and \textit{DB[key].\writeLV} to be {0}. Because the index for 
\textit{DB[key]} is not updated yet, other transactions will not see the new row. 
In \textit{Commit} function after $\txn{}$ releases the locks, it updates 
\textit{DB[key].\writeLV = \txn{}.LV}.
Finally, $\txn{}$ inserts \textit{key} into the index.}

\myhighlightpar{When a transaction $\txn{}$ deletes a row with primary key \textit{key}, it first grabs 
an exclusive lock of the row. And updates $\txn{}.\LV=ElemWiseMax \allowbreak (\txn{}.\LV, DB[key].readLV, DB[key].writeLV)$. Any 
other transaction trying to access this row will abort due to lock conflicts. In 
the \textit{Commit} function before $\txn{}$ releases the locks, it removes \textit{key} from the index.
}

\subsection{\myhighlight{Limitations of \codename}}
\label{sec:limitation}
We now discuss the limitations of \name's design and potential ways to mitigate them.

One potential problem is that the size of \LV is proportional to the number of log 
managers. For a large number of log managers, the computation and storage overhead of \LV will 
increase. In contrast, serial logging maintains a single LSN and therefore avoids this problem. 
Although we believe most DBMSs will use a relatively small number of log files and thus this 
overhead is acceptable, \name can also leverage \LV compression (\cref{sec:lv-compression}) 
and SIMD instructions (\cref{sec:eval-log-num}) to partially resolve this issue. If necessary, a 
dependency-aware transaction-to-log mapping mechanism can also potentially reduce inter-log 
dependencies.

Another limitation of \name is the amount of parallelism during recovery for workloads with high 
contention. For these workloads, the inherent recovery parallelism can be lower than the number of 
log managers. During recovery, a large number of inter-log dependencies will exist. In \name, the 
dependencies propagate through \RLV (\cref{alg:recovery-worker}), which leads to inter-thread 
communication, incurring relatively long latency between the recovery of dependent transactions. In 
contrast, a serial recovery scheme has no delay between consecutive transactions and, therefore, may 
deliver better performance. To address this, when the contention is high, \name will degrade to 
using serial recovery. Specifically, a single worker recovers all the transactions sequentially. The 
worker checks every \textit{pool} of log managers and recovers the transaction that satisfies $\txn{}.\LV 
\leq \RLV$; this approach incurs no delay between two consecutive transactions. We evaluate 
this aspect in \cref{sec:eval-sens}.

During recovery, if the pool size is large and contention is high, workers might need to scan 
the whole pool to find the next transaction that is ready to be recovered. Heuristic optimizations 
like zig-zag scans could help. We defer the problem of developing a data 
structure specialized for \name{} recovery to future work.

%% file: algos/worker-thread.tex
\SetKwProg{myfun}{Function}{}{}
\myfun{Lock(key, type, \txn{})}{ \label{line:lock-start}	
    \nonl\codeComment{Lock the tuple following the 2PL protocol.}\\
    \textit{FetchLock(key, type, \txn{})}\; \label{line:fetchlock}
    \textit{\txn{}.\LV = \elementWiseMax{}(\txn{}.\LV, DB[key].\writeLV)}\; \label{line:element-wise-max-1}
    \If{type is write}{ \label{line:element-wise-max-if}
        \textit{\txn{}.\LV = \elementWiseMax{}(\txn{}.\LV, DB[key].\readLV)}\; \label{line:element-wise-max-2}
    }
} \label{line:lock-end}	
\myfun{Commit(\txn{})}{ \label{line:commit-start}
    \myhighlightpar{
    \If{$\txn{}$ is not read-only}{\label{line:not-read-only-skip-1}
        \nonl\codeComment{Include \txn{}'s LV into the log record.}\\
        \textit{logRecord = \{CreateLogRecord(\txn{}), copy(\txn{}.LV)\}}\; \label{line:create-log-record}
        \textit{recordSize = GetSize(logRecord)}\;
        \textit{LSN = WriteLogBuffer(logRecord, recordSize)}\; \label{line:write-buffer}
        \textit{\txn{}.$\LV[i]$ = LSN}\codeComment{Update \txn{}.$\LV[i]$ in the memory.}\; \label{line:lsn-update} 
    }}
    \For{key $\in$ \txn{}'s access set}{ \label{line:release-start}
        \DrawBox{a}{b}
        \tikzmark{a}
        \If(\hfill \codeComment{Atomic Section}){\txn{} reads DB[key]}{ 
            \textit{DB[key].\readLV = \elementWiseMax{}(\txn{}.LV, DB[key].\readLV)}\;
        }
        \If{\txn{} writes DB[key]}{
            \nonl\codeComment{\txn{}.\LV is always no less than DB[key].\writeLV}\\
            \myhighlightpar{\textit{DB[key].\writeLV = \txn{}.\LV}}\;
            \label{line:update-writelv}
        }
        \textit{Release(key)}\hfill\tikzmark{b}%
    }
    \label{line:release-end}		
    \textit{Asynchronously commit \txn{} if $\PLV \geq \txn{}.\LV$ and all transactions in $L_i$ with smaller 
LSNs have committed}\;\label{line:commit-end}
}
\myfun{WriteLogBuffer(logRecord, recordSize)}{ \label{lines:write-buffer-begin}
    \textit{$L_i$.allocatedLSN$[j]$ = $L_i$.\loglsn}\; 
    \label{lines:allocate-greater-than-filled-begin}
    \textit{lsn = AtomicFetchAndAdd($L_i$.\loglsn, recordSize)}\; \label{line:atom-fetch-and-add}
    \textit{memcpy($L_i$.logBuffer + lsn, logRecord, recordSize)}\; \label{line:memcpy}
    \textit{L$_i$.filledlSN$[j]$ = lsn + recordSize}\; \label{line:assign-filledlsn} 
\label{lines:allocate-greater-than-filled-end}
    \Return{lsn + recordSize} \label{lines:write-buffer-end}
}

%% file: algos/manager-thread.tex
    	\textit{readyLSN = $L_i$.\loglsn{}}\; \label{line:init-readylsn}
    	\ForEach{worker thread j that maps to $L_i$}{ \label{line:state-check-begin}
    		\nonl\codeComment{We assume $allocatedLSN[j]$ and $filledLSN[j]$ are fetched together atomically}\;
    		\If{$allocatedLSN[j]$ $\geq$ $filledLSN[j]$}{ 
    			\textit{readyLSN = min(readyLSN, $allocatedLSN[j]$)} \label{line:state-check-end}
    		}
    	}
    	\textit{flush the buffer up to readyLSN}\; \label{lines:try-flush-end}
    	\textit{$\PLV[i]$ = readyLSN}\; \label{lines:update-readylsn}

%% file: algos/manager-recovery.tex
\SetKwProg{myfun}{Function}{}{}
    \While{\txn{} = $L_i$.DecodeNext() \textbf{and} \txn{}.\LV $\leq$ \ELV}{ \label{line:decode-next}
        \textit{pool.Enqueue(\txn{})}\; \label{line:enqueue}
        \textit{pool.maxLSN = \txn{}.LSN}\; \label{line:update-maxlsn}
    } \label{line:run-end}

%% file: algos/worker-recovery.tex
\While{\textbf{not} IsRecoveryDone()}{ \label{line:isrecoverydone}
    \nonl\codeComment{FetchNext atomically dequeues a transaction $\txn{}$ such that $\txn{}.\LV \leq \RLV$}\;
    \textit{\txn{} = pool.FetchNext(\RLV)}\; \label{line:pool-fetch-next}
    \textit{Recover(\txn{})}\; \label{line:recover-txn}
    \DrawBox{c}{d}
    \tikzmark{c}
    \eIf(\hfill \codeComment{Atomic Section}){pool is empty}{ \label{line:rlv-update-begin}
        \textit{RLV$[i]$ = Max(RLV$[i]$, pool.maxLSN)}\; 
        \label{line:update-rlv-by-maxlsn}
    }{
        \textit{RLV$[i]$ = Max(RLV$[i]$, pool.head.LSN - 1)}%
        \label{line:update-rlv-by-head}\hfill \tikzmark{d} 
    }
    \label{line:rlv-update-end}
}

%% file: discussion.tex
\section{Optimizations and Extensions}
\label{sec:discussion}

We now discuss optimizations to reduce \name' \LV storage overhead and computational overhead, and two 
extensions to support Optimistic Concurrency Control (OCC) and MVCC.

\subsection{Optimization: LV Compression}
\label{sec:lv-compression}

The design of \name as described in \cref{sec:protocol} has two issues: (1) the DBMS 
stores \readLV{} and \writeLV{} for every tuple, which changes the data 
layout and incurs extra storage overhead; (2) the transaction's \LV is stored for each log record, 
which can significantly increase the log size especially for command logging where each log record 
is relatively small. We describe optimizations that address these problems. 
\vspace*{-0.1in} \\

\textbf{Tuple LV Compression:}
To reduce a tuple's \LV storage, we observe that keeping \LVs for tuples that were accessed a long 
time ago is unnecessary. The \LVs in these tuples are too small to affect active transactions. 
This optimization thus stores \LVs only for active tuples in the lock table.
For these tuples, transactions operate on their \LVs following 
\cref{alg:logging}. If the DBMS inserts a  tuple into the lock table, then the algorithm assigns its 
\readLV and \writeLV to be the current \PLV. The system can evict a tuple from the 
lock table if no transactions hold locks on it and both its \readLV and \writeLV are not 
greater than the current \PLV. 

For the tuples previously evicted from the lock table and later inserted back, the optimization 
increases the \readLV and \writeLV of these tuples and also the \LVs of transactions accessing them. 
This modification makes transactions depend on more transactions than before. %
To make the trade-off between higher compression ratio and fewer 
artificial dependencies, we introduce a new parameter $\delta$ and evict a tuple from the lock table 
only if $\forall i, \PLV{}[i] - \LV{}[i] \geq \delta$ is true for both \readLV and \writeLV. 
Accordingly, a newly inserted tuple will have $\readLV{}[i] = \writeLV{}[i] = \PLV{}[i] - \delta$. 
Larger $\delta$ means fewer artificial dependencies, but more tuples will stay in the lock table 
waiting for eviction, and vice versa.
\vspace*{-0.1in} \\

\begin{algorithm}[t!]
    \caption{\small
        \textbf{\LV Compression for Log Records} --
        Each log manager $L_i$ periodically calls \textit{FlushPLV}. \
        The worker threads mapped to $L_i$ call \textit{Compress} and \textit{Decompress}. 
    }
    \label{alg:log-lsn-compression}
    {\algoSize \input{algos/compression.tex}}
\end{algorithm}

\textbf{Log Record LV Compression:}
We next address the second issue that each log record must store the \LV of the transaction. We 
propose an optimization where each log record stores only a few but not all dimensions of a 
transaction's \LV. The motivating insight is that for workloads with low to medium contention, most 
dimensions of a log record's \LV are too small to be interesting. For example, suppose that a transaction $\txn{}$ 
depends on a committed transaction $\txn{}'$. It is not critical 
to remember precisely which $\txn{}'$ that $\txn{}$ depends on, but only that $\txn{}$ depends on some transaction 
that happened before a specific point in time. Therefore, we can set anchor points (in the form of 
\LVs) into each log such that, if $\txn{}$ depends on only transactions before an anchor point, it 
stores the anchor point instead of the detailed \LV.

In \cref{alg:log-lsn-compression}, we introduce a variable \LPLV as the anchor point. 
\textbf{\textit{L.\LPLV}} is an LSN Vector that is maintained by each log manager $L$. It keeps a 
copy of the most recent \PLV written into $L$'s log buffer. Periodically, the log manager appends 
\PLV into the log buffer and updates $L.\LPLV$ 
(Lines~\ref{line:flush-plv-start}--\ref{line:flush-plv-end}). 

To compress a transaction $\txn{}$'s \LV, we check for every dimension if $\txn{}.\LV$ is no greater than $L_i.\LPLV$. 
If this is true for dimension $j$, namely, $\txn{}.\LV{}[j] \leq L_i.\LPLV{}[j]$, we 
can artificially increase $\txn{}.\LV{}[j]$ to $L_i.\LPLV{}[j]$. Since $L_i.\LPLV$ is already in the buffer, 
the system no longer needs to store $\txn{}.\LV{}[j]$ 
(Lines~\ref{line:compress-start}--\ref{line:compress-omit-end}). During recovery, the DBMS performs the 
opposite operation; if the $j$-th dimension of an \LV was compressed, it replaces it with 
the value of $\LPLV{}[j]$ (Lines~\ref{line:rebuild-start}--\ref{line:rebuild-lv-end}). If it reads an 
anchor from the log, it updates \LPLV.

\begin{figure}
    \centering
    \subfloat[Logging]{
        \adjincludegraphics[width=0.52\columnwidth, trim={0 {0.05\height} {0.0\width} {0.05\height}}, clip]{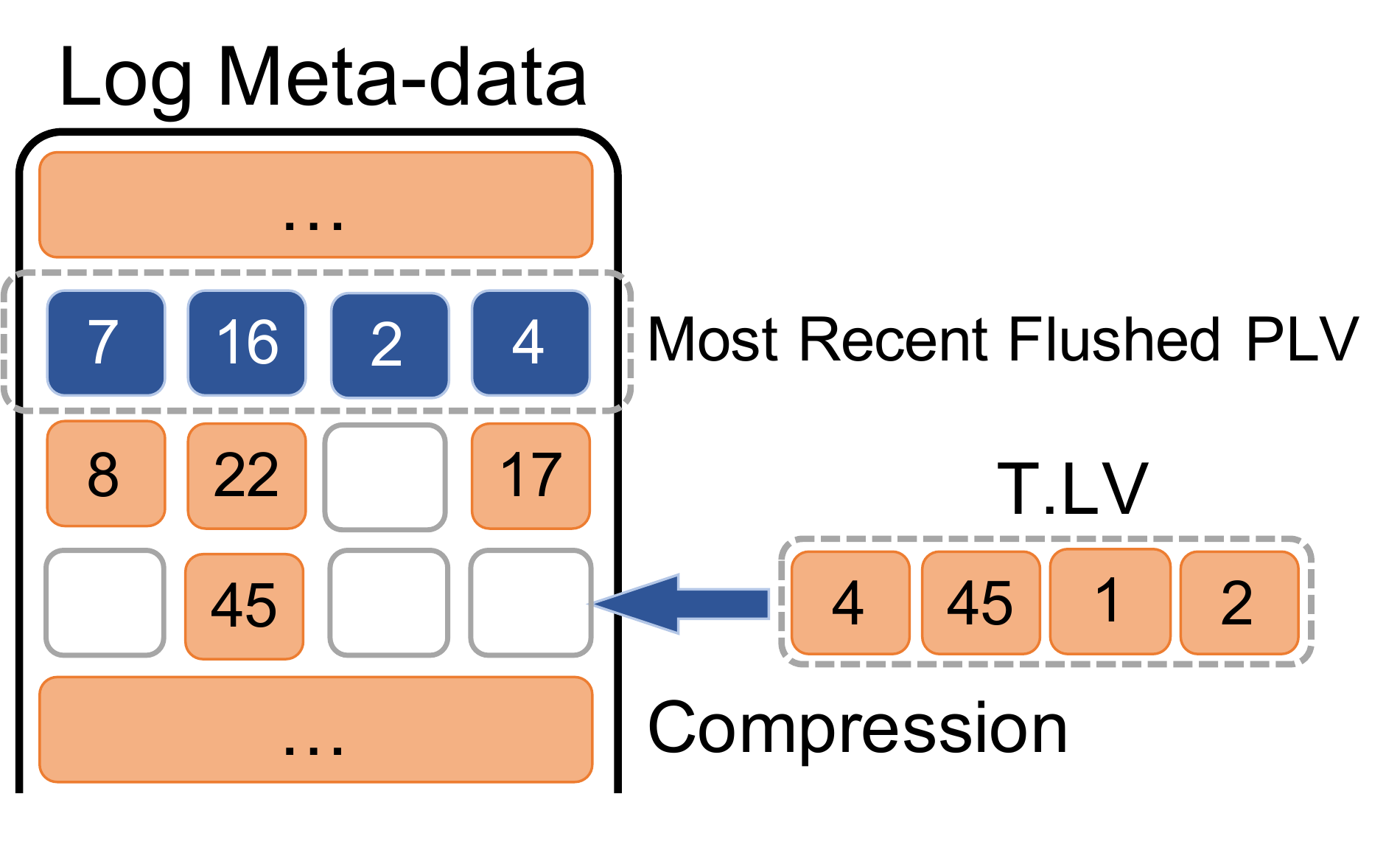}
        \label{fig:lv-compression}
    }\hspace*{-0.1in} 
    \subfloat[Recovery]{
        \adjincludegraphics[width=0.52\columnwidth, trim={{0.0\width} {0.05\height} 0 {0.05\height}}, clip]{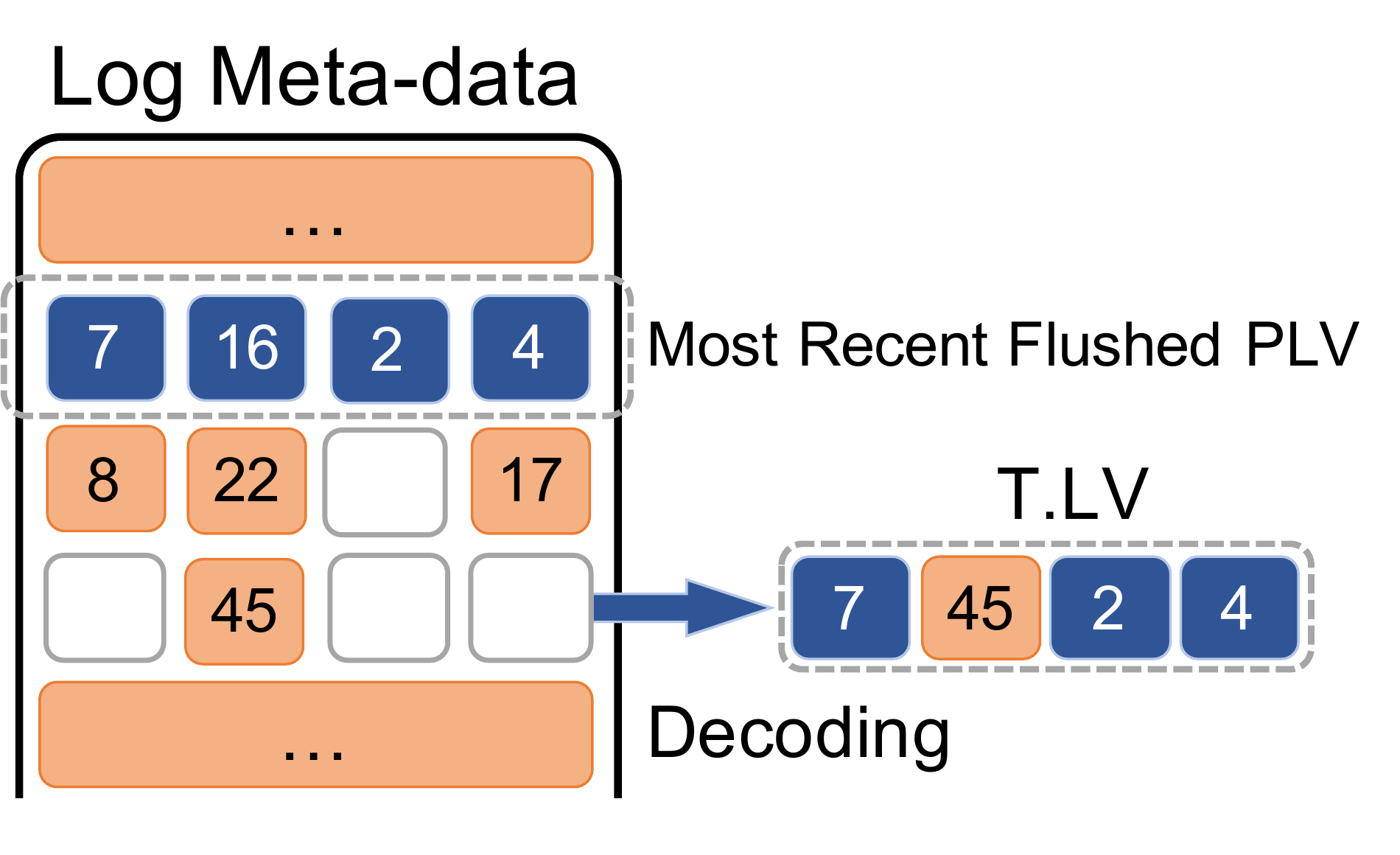}
        \label{fig:lv-decompression}
    }
    \caption{
        \textbf{LV Compression} -- 
        Example of \name's compression method.
    }
    \label{fig:lsn-vector-compression-example}
\end{figure}

\cref{fig:lsn-vector-compression-example} shows an example of \LV compression. In 
\cref{fig:lv-compression}, transaction $\txn{}$'s \LV = [4, 45, 1, 2] is written to the 
log. The system 
compares it against \LPLV and finds that $\txn{}.$\LV has only one dimension (the 2nd dimension with 
value 45) greater than \LPLV. Only the 2nd dimension is written into the log. During recovery, 
\cref{fig:lv-decompression} shows that \name fills in the blanks with the most recently seen anchor, 
\LPLV = [7, 16, 2, 4]. The compressed \LV is decoded into [7, 45, 2, 4]. Note that the 1$^{st}$, 
3$^{rd}$, and 4$^{th}$ dimension of the decompressed \LV are greater than the original $\txn{}.\LV$.

The frequency of \LPLV flushing makes a trade-off between parallelism in recovery and \LV 
compression 
ratio. When the frequency is high, a dimension of \LV is smaller than \LPLV 
and thus it enables better compression, but some amount of recovery parallelism is sacrificed since 
the decompressed \LVs have larger values.

\subsection{Optimization: Vectorization}
\label{sec:vectorization}
\myhighlightpar{
The logging overhead mainly consists of four parts: (1) the overhead introduced 
by \name where we calculate \LVs and move them around; (2) the overhead of creating the 
log records and writing them to the in-memory log buffer; (3) for lock-based concurrency control 
algorithms, the extra latency caused by (1) and (2) will result in extra lock contention; (4) the time cost in persisting the log records to the disk.
Among 
them, (4) is moved off the critical path by ELR; (2) and (3) are shared by essentially all the write-ahead logging algorithms. }
\myhighlightpar{
These overheads will not block the DBMS from scaling up. Overheads (1) and 
(2) are linear in the number of total transactions executed.
Overhead (1) is also related to the number of log files. If the system is 
writing to many log files and the transactions have a short execution time, it is 
up to {13.8\%} of the total execution time if implemented naively.  We can exploit the data parallelism in the LSN Vector as the values 
within a single vector are processed independently. Modern CPUs provide SIMD extensions 
that allow the DBMS to process multiple vector elements items in a single instruction. For example, the instruction \texttt{\_mm512\_max\_epu32} can compute the element-wise maximum of two vectors of 16 32-bit integers. In \cref{sec:eval-log-num}, we show that switching to vectorized operations reduces \name' overhead by 89.5\%.
}

\subsection{Extension: Support for OCC}
\label{sec:protocol-occ-variant}
Our overview of \name thus far assumes that the DBMS uses 2PL. 
\name is also compatible with other schemes. We next discuss how \name can support
Optimistic Concurrency Control (OCC)~\cite{kung1981optimistic}.

\cref{alg:life-cycle-occ} shows the protocol for a worker thread. Different from a 2PL protocol 
(\cref{alg:logging}), an OCC transaction calls \textit{Access} when accessing a tuple and 
\textit{Commit} after finishing execution. The \readset and \writeset are maintained by the 
\textit{read/write} functions in the conventional OCC algorithm, from which \textit{Access} is 
called. In the \textit{Access} function, the transaction atomically reads the value, 
\readLV, \writeLV, and potentially other auxiliary data. Commonly seen in OCC algorithms, 
the \textit{ValidateSuccess} function returns true
if the values in the \readset are not modified by other transactions.
The atomicity is guaranteed through a 
latch, or by reading a version number twice before and after reading the value~\cite{tu13}.

\begin{algorithm}[t!]
    \caption{\small
        \textbf{OCC Logging for Worker Threads}
    }
    \label{alg:life-cycle-occ}
    {\algoSize \input{algos/occ.tex}}
\end{algorithm}

For high concurrency, we choose a reader-lock-free design of the \textit{Commit} function. 
The transaction first locks all the tuples in the \writeset 
(Lines~\ref{line:lock-ws-start}--\ref{line:lock-ws-end}). Before validating the \readset 
(\cref{line:validation}), it updates the \readLV of tuples in the \readset one dimension at 
a time (Lines~\ref{line:commit-update-readlv-begin}--\ref{line:commit-update-readlv-end}). Each update 
happens atomically using compare-and-swap instructions. This is 
necessary because the data item might appear in the \readset of multiple transactions, and 
concurrent updates of \readLV might cause loss of data. The reason that the \readLV extension must 
occur before the validation is to enforce write-after-read dependencies. To see a failure example, 
consider a transaction $\txn{1}$ modifying the data after $\txn{2}$'s validation but before $\txn{2}$s updates on 
\readLV. Then, it is possible that $\txn{1}$ does not observe the latest \readLV and therefore fails to 
capture the write-after-read dependency to $\txn{2}$. Note that updating \readLV before the validation 
might result in extra non-existing dependencies (i.e., \LVs larger than necessary) where the 
transaction aborts later in the validation but has already updated the \readLV of some data tuples. Such 
aborts only affect performance but not correctness. The design of log managers stays the same as in 
\cref{alg:log-manager-logging}.

\subsection{Extension: Multi-Versioning}
\label{sec:work-mvcc}

\myhighlightpar{We next discuss how \name{} works with MVCC. 
We assume the recovery process also uses multi-versions. Otherwise, the DBMS has to reorder 
the transactions either by changing already persistent data or appending extra information. 
Concurrency control algorithms based 
on logical timestamps allow physically late transactions to access versions early and commit 
transactions logically early. The DBMS, however, creates log records and flushes them in the 
physical time order. 
Solving the decoupled order requires extra design. Allowing multi-versions in the recovery process 
relaxes the decoupling by allowing physically late transactions to commit logically early in the 
recovery. This assumption frees \name{} from tracking the \conflictWAR dependencies because the read 
operation can still fetch the correct historic version even after the tuple has been modified. 
Therefore, \name{} only needs to track \conflictWAW and \conflictRAW dependencies. Different from 
\cref{sec:protocol-lv}, \name{} for MVCC only adds a single metadata field for the data versions, 
the LSN Vector \LV. Our 
discussion is based on the MVCC scheme \cite{larson11} used in Hekaton~\cite{diaconu13}. The 
algorithm adds three extra fields to the data version tuples, namely, Begin Timestamp, End 
Timestamp, and a hash pointer. 
}

\myhighlightpar{Whenever a transaction reads a data version $v$, the transaction updates $\txn{}.\LV$ to 
be $\elementWiseMax(\txn{}.\LV, v.\LV)$ to catch \conflictRAW dependencies. When a transaction updates 
the data by adding a new data version $v$ after the old version $u$ during normal processing phase, 
it first updates the timestamps as in MVCC, then it updates $\txn{}.\LV$ to be 
$\elementWiseMax(\txn{}.\LV,$ $u.\LV)$, and $v.\LV$ to be empty.}

\myhighlightpar{In the postprocessing phase, if the transaction $\txn{}$ commits, before it replaces its 
transaction ID with its end timestamp, it iterates data versions in the {\writeset}. For a 
data version $v$ in the {\writeset}, it replaces $v.\LV$ to be $\txn{}.\LV$. The log records of $\txn{}$ contains $\txn{}.\LV$ and the commit timestamp of $\txn{}$. The former identifies whether $\txn{}$ should recover and the recovery order, and the latter determines the visible version when reading the data as well as the logical timestamp of the new versions when writing the data.}

\myhighlightpar{During recovery, \cref{alg:log-manager-recovery} and \cref{alg:recovery-worker} are 
executed. Only the visible version is returned for read operations. Whenever a write happens, 
the transaction writes a new version with the commit timestamp. Different from MVCC, 
transactions no longer acquire locks during recovery because \name guarantees no conflicts will 
occur. Without \name{}, the log records described in \cite{larson11} contain 
only the payload and the logical timestamps, enforcing a total order among transactions. \name{} 
exploits the parallelism to recover non-conflicting transactions in parallel.}

%% file: algos/compression.tex
    \SetKwProg{myfun}{Function}{}{}	
    \myfun{FlushPLV()}{ \label{line:flush-plv-start}
        \textit{currentPLV = Global.\PLV}\;
        \textit{logBuffer.append(currentPLV)}\;
        \textit{\LPLV = currentPLV}\; \label{line:assign-lplv}
    } \label{line:flush-plv-end}

    \myfun{Compress(\LV)}{
        \textit{compressedLV = \LV}\;  \label{line:compress-start}
        \ForEach{\LV{}$[j]$ $\in$ \LV}{
            \If{\LV{}$[j] \leq L_i.\LPLV{}[j]$}{ \label{line:compress-omit-start}
                \textit{compressedLV$[j]$ = NaN}\;
            } \label{line:compress-omit-end}
        }
        \Return{compressedLV}\;
    }

    \myfun{Decompress(compressedLV)}{
        \textit{\LV = compressedLV}\; \label{line:rebuild-start}
        \ForEach{LV$[j]$ $\in$ \LV}{
            \If{LV$[j]$ = NaN}{ \label{line:rebuild-lv-start}
                \textit{LV$[j]$ = $L_i$.LPLV$[j]$}\;
            } \label{line:rebuild-lv-end}
        }
        \Return{LV}\;
    }

%% file: algos/occ.tex
\SetKwProg{myfun}{Function}{}{}
	\myfun{Access(key, \txn{})}{
		\textit{value, \readLV, \writeLV = load(key)} \codeComment{load atomically}\;
        \textit{\txn{}.\LV = \elementWiseMax(\txn{}.\LV, \writeLV)}\;
		\Return{value}
	}
	\myfun{Commit(\txn{})}{
		\For{key $\in$ sorted(\txn{}.\writeset)}{	\label{line:lock-ws-start}		
			\textit{DB[key].lock()}\;\label{line:lock-ws-end}
		}

		\For{key $\in$ \txn{}.\readset}{
			\ForEach{dimension i of \LV}{ \label{line:commit-update-readlv-begin}
				\DrawBox{e}{f}
				\tikzmark{e}
				\If(\hfill\codeComment{Atomic}){$DB[key].\readLV{}[i]$ < $\txn{}.LV[i]$}{
					$DB[key].\readLV{}[i]$ = $\txn{}.LV[i]$ \hfill \tikzmark{f}
				}
				\label{line:commit-update-readlv-end}
			}
		}

		\If{\textbf{not} ValidateSuccess())}{	\label{line:validation}
			\textit{Abort(\txn{})}\;
		}
		\label{line:occ-right-after-validation}
		\textit{Create log record and write to log buffer similar to Lines~\ref{line:create-log-record}--\ref{line:lsn-update} in \cref{alg:logging}}\; \label{line:occ-log-buffer}

		\For{key $\in$ \txn{}.\writeset}{
			\textit{DB[key].\writeLV = \elementWiseMax(DB[key].\writeLV, \txn{}.\LV)}\;		
			\textit{DB[key].release()}\;
		}
		\textit{Asynchronously commit \txn{} if $\PLV \geq \txn{}.\LV$ and all transactions in $L_i$ with smaller LSNs have committed}\;
	}

%% file: evaluation.tex
\section{Evaluation}
\label{sec:evaluation}
\myhighlightpar{We implemented \name in the DBx1000 in-memory DBMS~\cite{dbx1000} to evaluate its 
performance.
We include both the 2PL and OCC variants of \name{} in the evaluation.}
\myhighlightpar{We evaluate the DBMS on three storage devices: (2) NVMe SSDs, (2) hard 
drives (HDDs), and (3) Persistent Memory (PM) simulated by a RAM disk. The performance profiles of 
these devices highlight different properties of the logging algorithms. As the mainstream 
fast storage, NVMe SSDs provide a high bandwidth, affording insights of the performance in production. HDDs have limited bandwidth, which is better for command logging. The cutting-edge PM 
largely eliminates disk bandwidth restrictions and exposes CPU and memory overheads.}

We compare \name to the following protocols all in DBx1000:
\vspace*{-0.1in} \\

    \textbf{No Logging:}
    The DBMS has all logging functionalities disabled.
    It does not incur any logging-related overhead and therefore serves as a performance 
    upper bound. 
    \vspace*{-0.1in} \\
    
    \textbf{Serial Logging:}
    This is our baseline implementation that uses a single disk and supports both data logging 
    and command logging. For data logging, the DBMS saturates the 160~MB/s bandwidth.
    With command logging, the DBMS generates a smaller log and its performance is 
    limited by the atomic increment of the central LSN.
    \vspace*{-0.1in} \\

    \textbf{Serial Logging with RAID-0 Setup:}
    This is the same configuration as Serial Logging, except that it uses a RAID-0 
    array across the eight disks using Linux's software RAID driver.
    \vspace*{-0.1in} \\

    \myhighlightpar{
    \textbf{Plover:}
    This a parallel data logging scheme that partitions log records based on data 
    accesses~\cite{zhou2020plover}.
    It uses per-log sequence numbers to enforce a total order among transactions.
    Each transaction generates multiple log entities.}
    \vspace*{-0.1in} \\

    \textbf{Silo-R:}
    Lastly, we also implemented the parallel logging scheme from 
    Silo~\cite{zheng2014fast,tu13}. %
    Silo uses a variant OCC that commits transactions in epochs. 
    The DBMS logs transactions in batches that are processed by multiple threads in parallel. 
    Silo-R only supports data logging because the system does not track write-after-read 
    dependencies.

\subsection{Workloads}
\label{sec:workloads}
We first describe the benchmarks used in the evaluation:
\vspace*{-0.1in} \\

\textbf{Yahoo! Cloud Serving Benchmark (YCSB):}
\myhighlightpar{This benchmark simulates the workload pattern of cloud-based OLTP 
systems~\cite{cooper10}. In our experiments, we simulate a DBMS with a single table. Each data row 
has 10 fields and each field contains 100 bytes. We evaluate two databases with 10~GB 
and 500~GB of data.
We build a single index for the table. The access 
pattern of transactions visiting the rows follows a Zipfian distribution; unless otherwise stated, 
we set the distribution parameter to $0.6$ to simulate moderate contention. By default, each 
transaction accesses two tuples and each access has a 50\% chance to be a read operation and a 50\% 
chance to be a write operation. We will perform sensitivity studies regarding these workload 
parameters in \cref{sec:eval-sens}.}
The size of each transaction's command log record is smaller than that of its data log records.
\vspace*{-0.1in} \\

\textbf{TPC-C:}
This is the standard OLTP benchmark that simulates a wholesale company operating on 
warehouses~\cite{tpc-c}. The database has nine tables covering a 
variety of necessary information and transactions are performing daily order-processing business. We 
simulate two (\tpccPayment and \tpccNewOrder) out of the five transaction types in TPC-C as around 
90\% of 
the default TPC-C mix consists of these two types of transactions. When \name is running in command 
logging mode, each transaction log record consists of the input parameters to the stored 
procedure. 
\myhighlight{The workload is logically partitioned by warehouses. We use 80 warehouses in the evaluation. We evaluate the full TPC-C workload in \cref{sec:eval_tpcc_full_mix}.} 

\myhighlightpar{The choices of the benchmarks provide a comprehensive evaluation of \name and baselines. YCSB, 
TPC-C \tpccPayment, and TPC-C \tpccNewOrder represent short transactions with moderate 
contention, short transactions with low contention, and long transactions.}

\begin{figure*}[t]
    \centering
    \hspace{0.4in}{\fbox{\adjincludegraphics[scale=0.3, trim={{0.13\width} {0.6\height} 0 {0.15\height}}, clip]
            {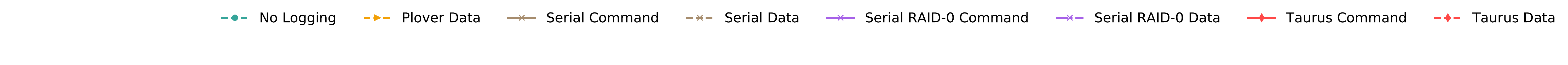}}
    } %
    \vspace{-0.0in}\newline
    \subfloat[YCSB-500G]{
        \adjincludegraphics[scale=\globalGraphScale, trim={0 0 0 {0.2\height}}, clip]
            {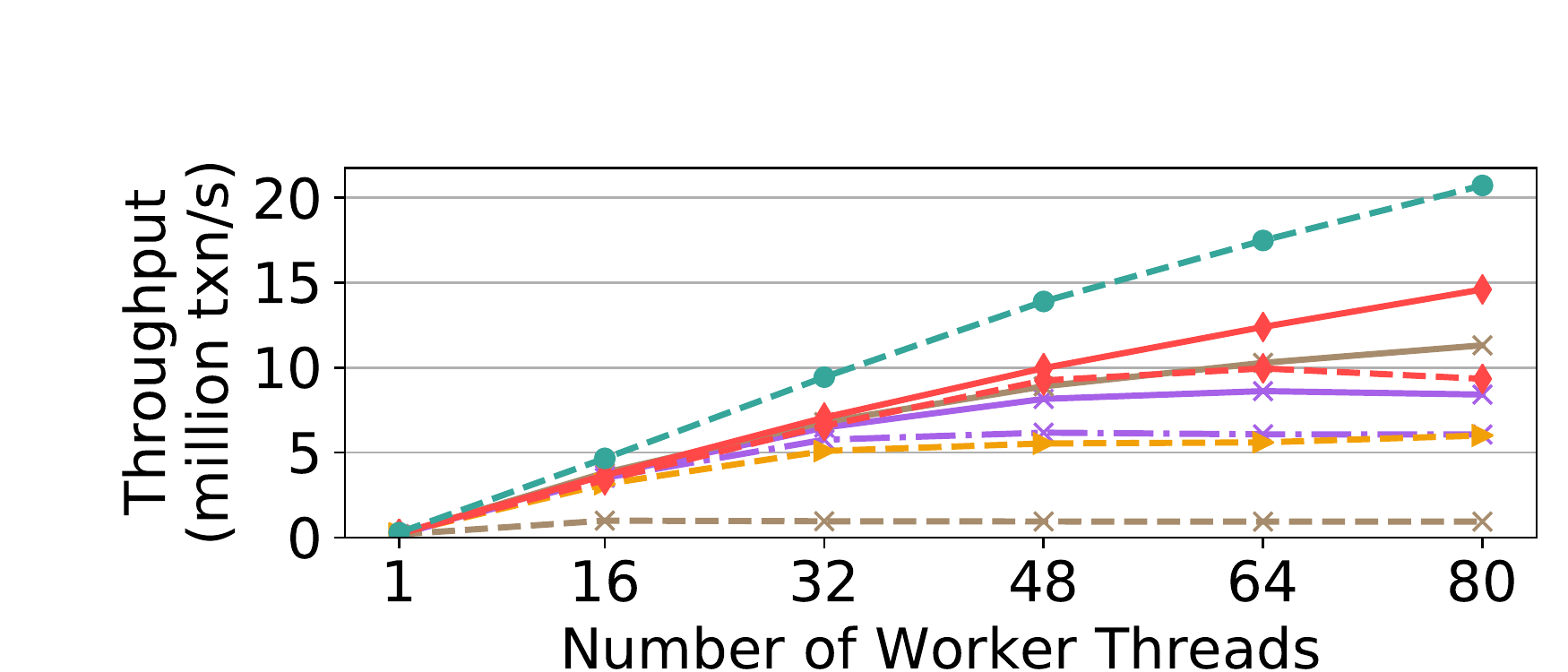}
            \label{fig:logging-performance-adv-2pl-a}
    }
    \hfill
    \subfloat[TPC-C \tpccPayment]{
        \adjincludegraphics[scale=\globalGraphScale, trim={0 0 0 {0.2\height}}, clip]
            {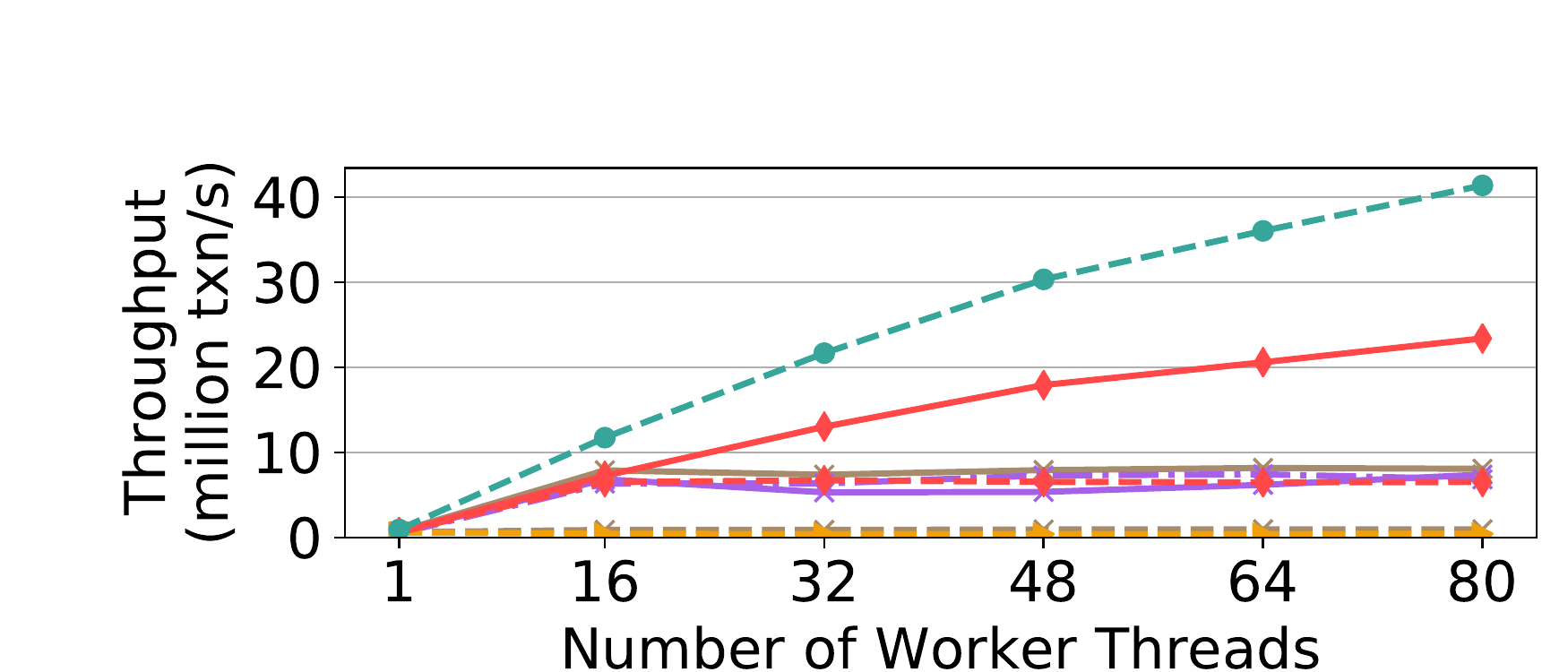}
            \label{fig:logging-performance-adv-2pl-b}
    }
    \hfill
    \subfloat[TPC-C \tpccNewOrder]{
        \adjincludegraphics[scale=\globalGraphScale, trim={0 0 0 {0.2\height}}, clip]
            {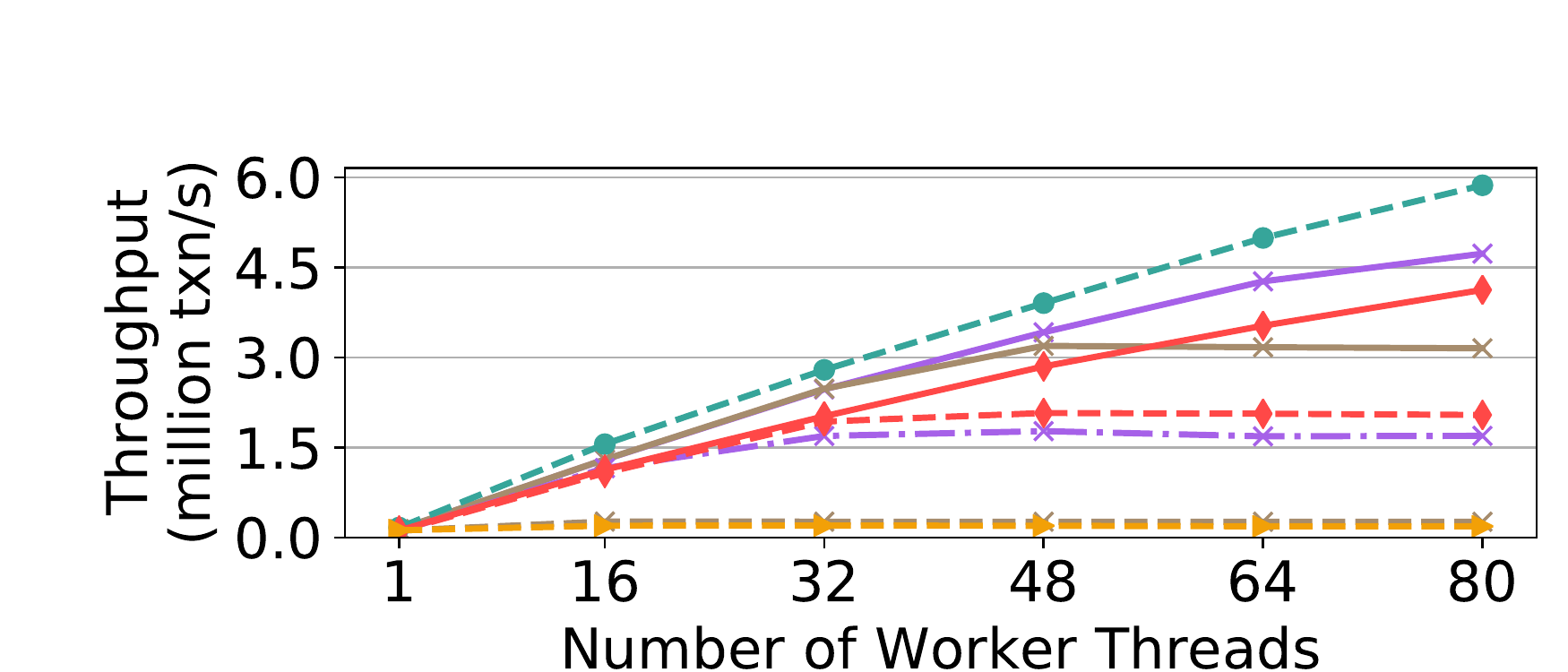}
            \label{fig:logging-performance-adv-2pl-c}
    }
    \caption{
        \textbf{Logging Performance (2PL)} --
            Performance comparison on YCSB-500G, TPC-C \tpccPayment, and TPC-C \tpccNewOrder on NVMe drives.
    }
    \label{fig:logging-performance-adv-2pl}
\end{figure*}

\begin{figure*}[t]
    \centering
    \hspace{1.1in}{
        \fbox{\hspace{-0.17in}\adjincludegraphics[scale=0.3, trim={{0.17\width} {0.6\height} 0 {0.15\height}}, clip]
            {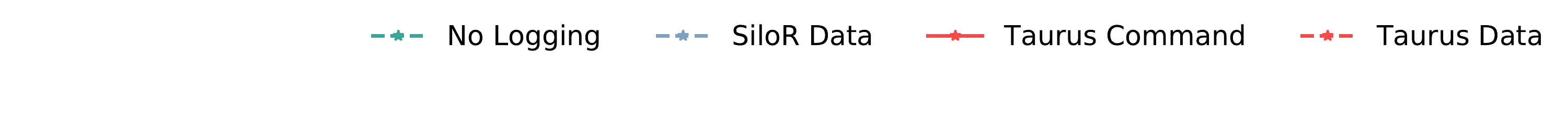}}
    } %
    \vspace{-0.0in}\newline
    \subfloat[YCSB-500G]{
        \adjincludegraphics[scale=\globalGraphScale, trim={0 0 0 {0.2\height}}, clip]
            {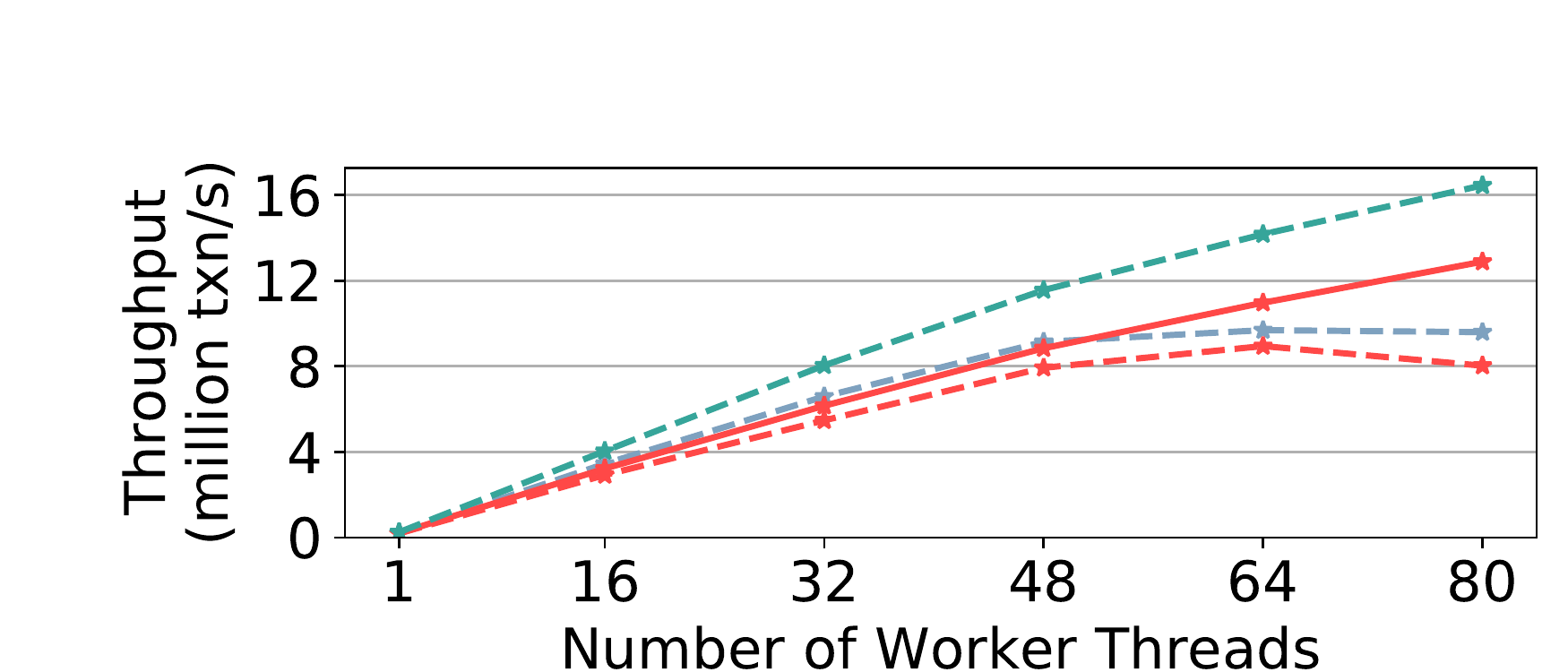}
            \label{fig:logging-performance-adv-occ-a}
    }
    \hfill
    \subfloat[TPC-C \tpccPayment]{
        \adjincludegraphics[scale=\globalGraphScale, trim={0 0 0 {0.2\height}}, clip]
            {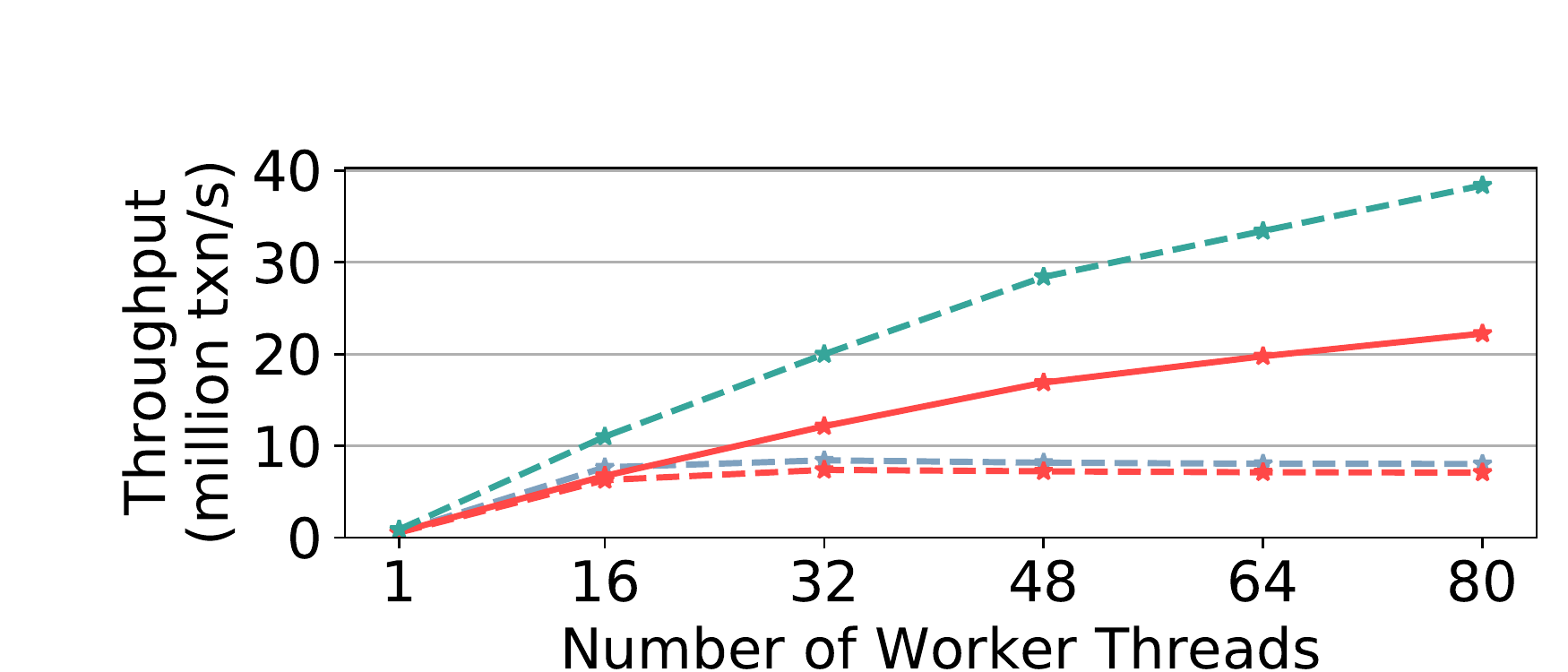}
            \label{fig:logging-performance-adv-occ-b}
    }
    \hfill
    \subfloat[TPC-C \tpccNewOrder]{
        \adjincludegraphics[scale=\globalGraphScale, trim={0 0 0 {0.2\height}}, clip]
            {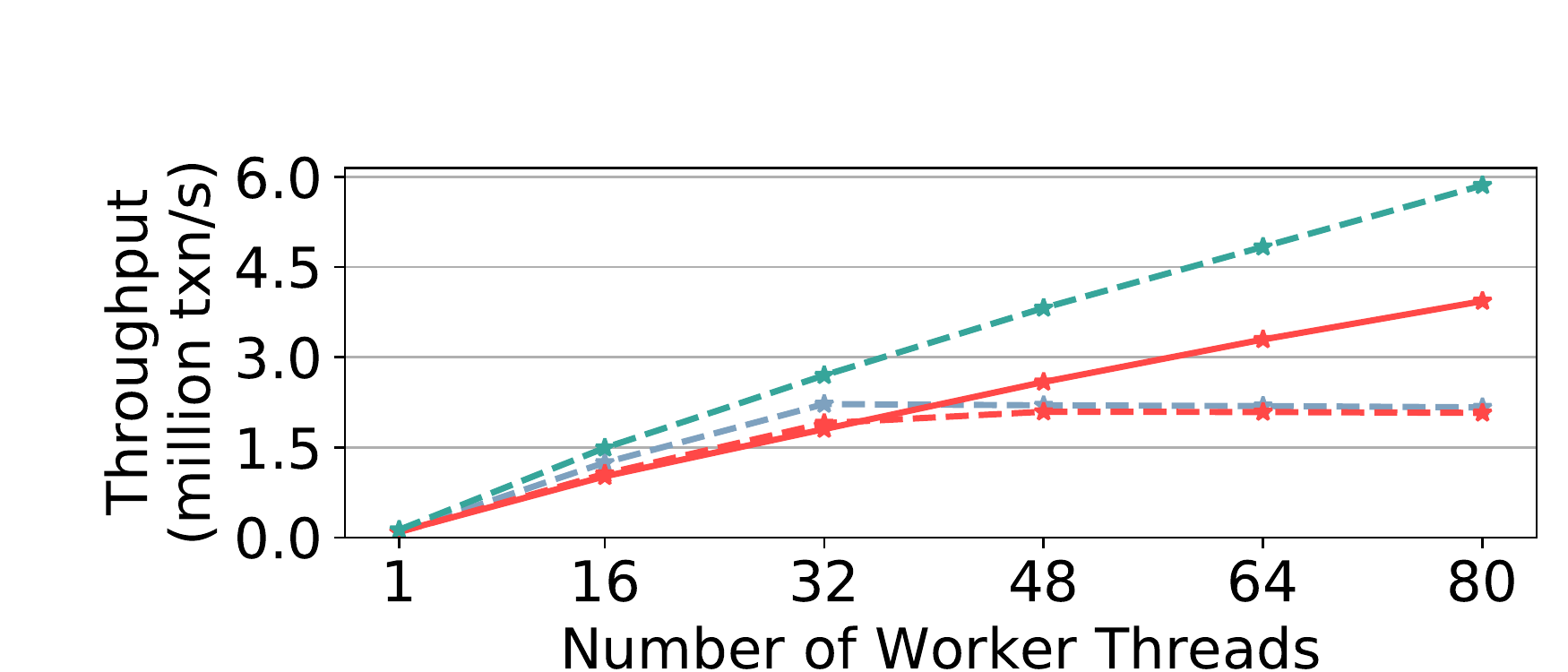}
            \label{fig:logging-performance-adv-occ-c}
    }
    \caption{
        \textbf{Logging Performance (OCC)} --
        Performance comparison on YCSB-500G, TPC-C \tpccPayment, and TPC-C \tpccNewOrder on NVMe drives.
    }
    \label{fig:logging-performance-adv-occ}
\end{figure*}

\begin{figure*}[t]
    \centering
    \subfloat[YCSB-500G]{
        \adjincludegraphics[scale=\globalGraphScale, trim={0 0 0 {0.2\height}}, clip]
            {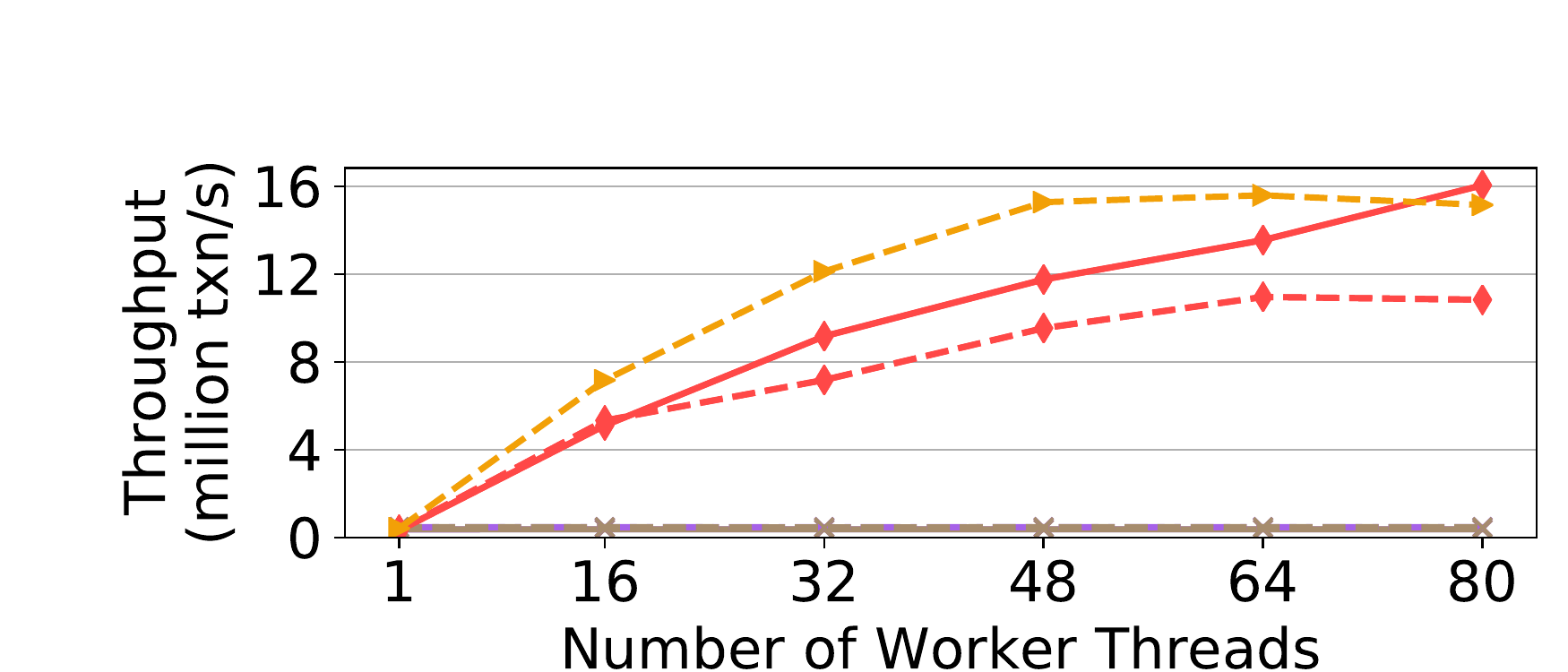}
            \label{fig:recovery-performance-adv-2pl-a}
    }
    \hfill
    \subfloat[TPC-C \tpccPayment]{
        \adjincludegraphics[scale=\globalGraphScale, trim={0 0 0 {0.2\height}}, clip]
            {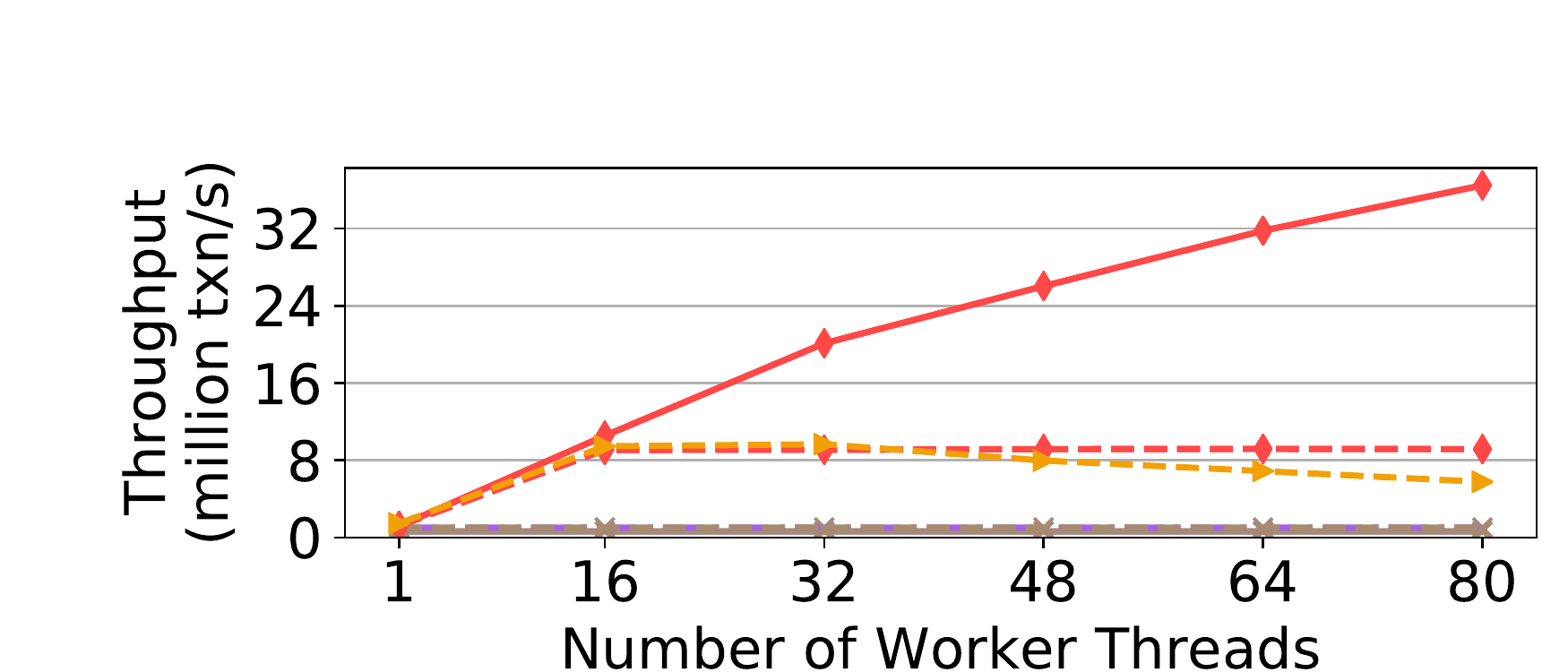}
            \label{fig:recovery-performance-adv-2pl-b}
    }
    \hfill
    \subfloat[TPC-C \tpccNewOrder]{
        \adjincludegraphics[scale=\globalGraphScale, trim={0 0 0 {0.2\height}}, clip]
            {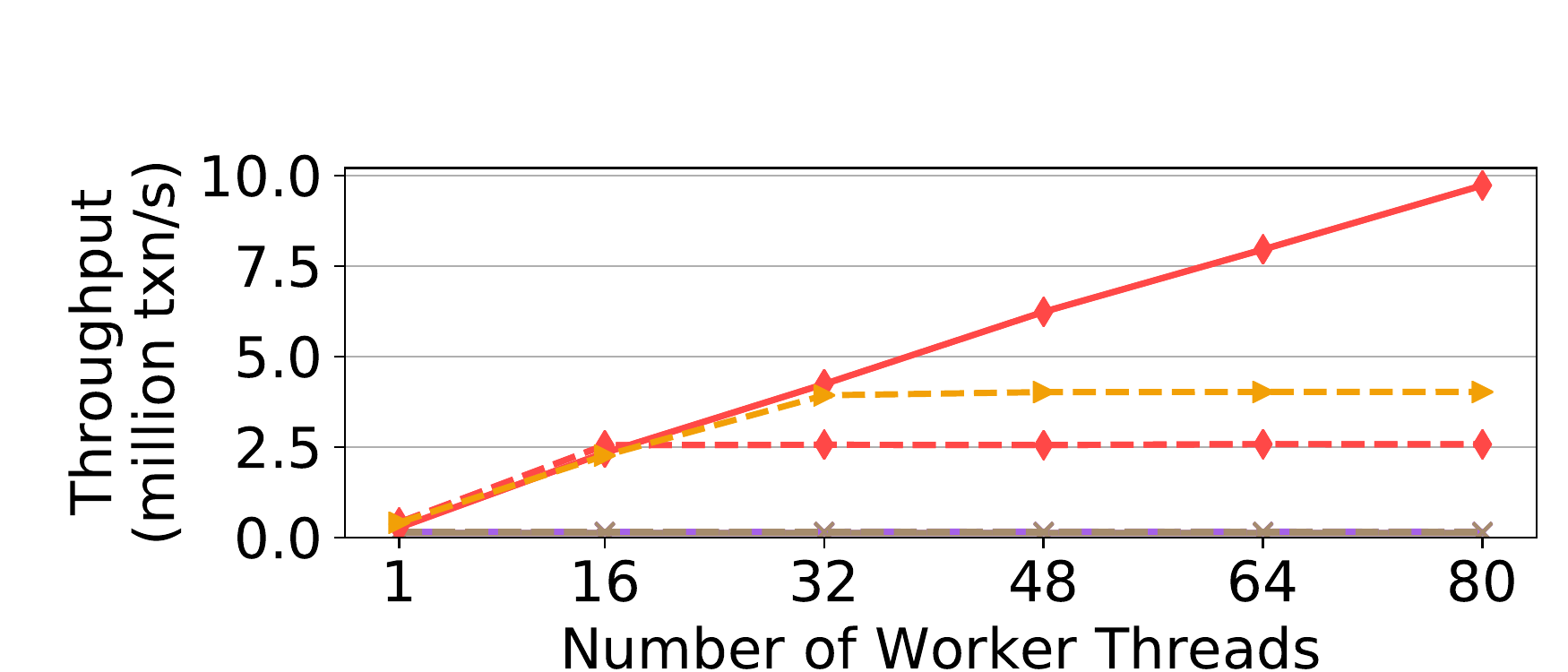}
            \label{fig:recovery-performance-adv-2pl-c}
    }
    \caption{
        \textbf{Recovery Performance (2PL)} --
        Performance comparison on YCSB-500G, TPC-C \tpccPayment, and TPC-C \tpccNewOrder on NVMe drives.
    }
    \label{fig:recovery-performance-adv-2pl}
\end{figure*}

\begin{figure*}[t]
    \centering
    \subfloat[YCSB-500G]{
        \adjincludegraphics[scale=\globalGraphScale, trim={0 0 0 {0.2\height}}, clip]
            {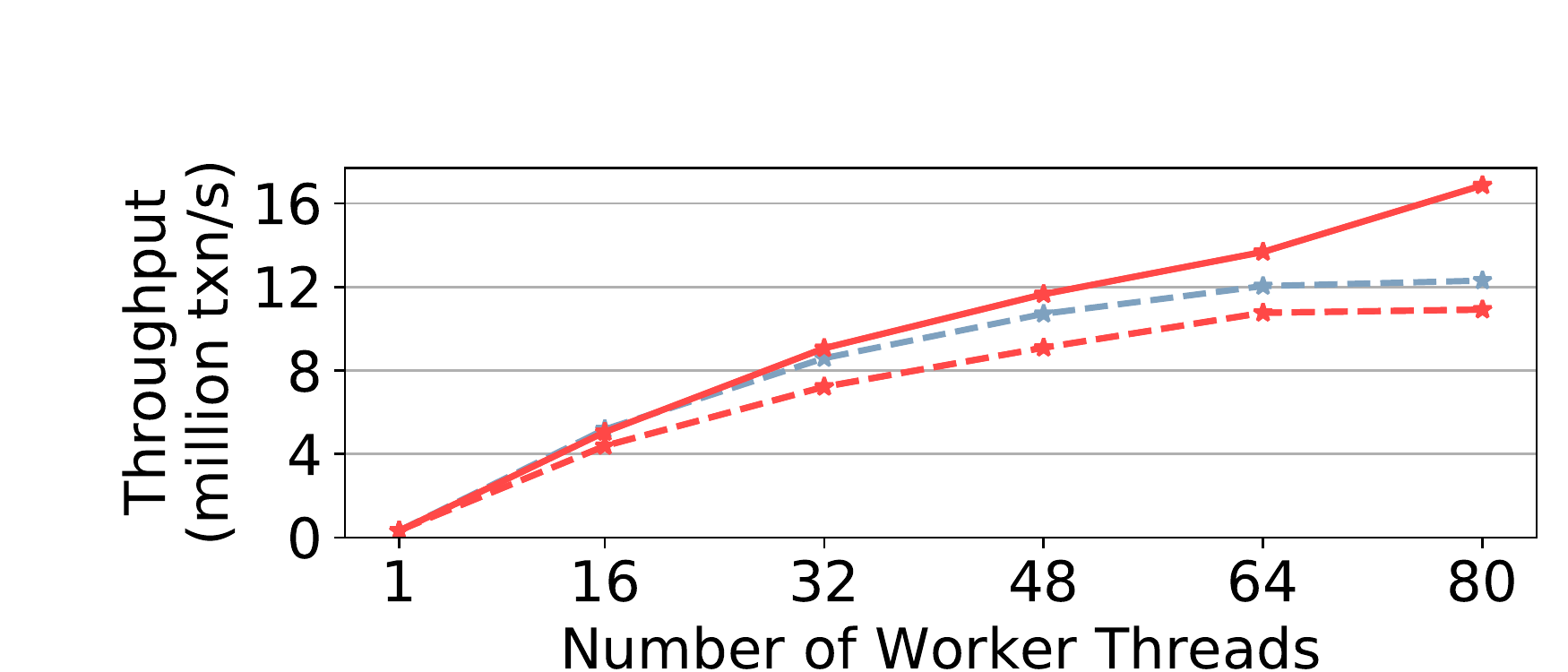}
            \label{fig:recovery-performance-adv-occ-a}
    }
    \hfill
    \subfloat[TPC-C \tpccPayment]{
        \adjincludegraphics[scale=\globalGraphScale, trim={0 0 0 {0.2\height}}, clip]
            {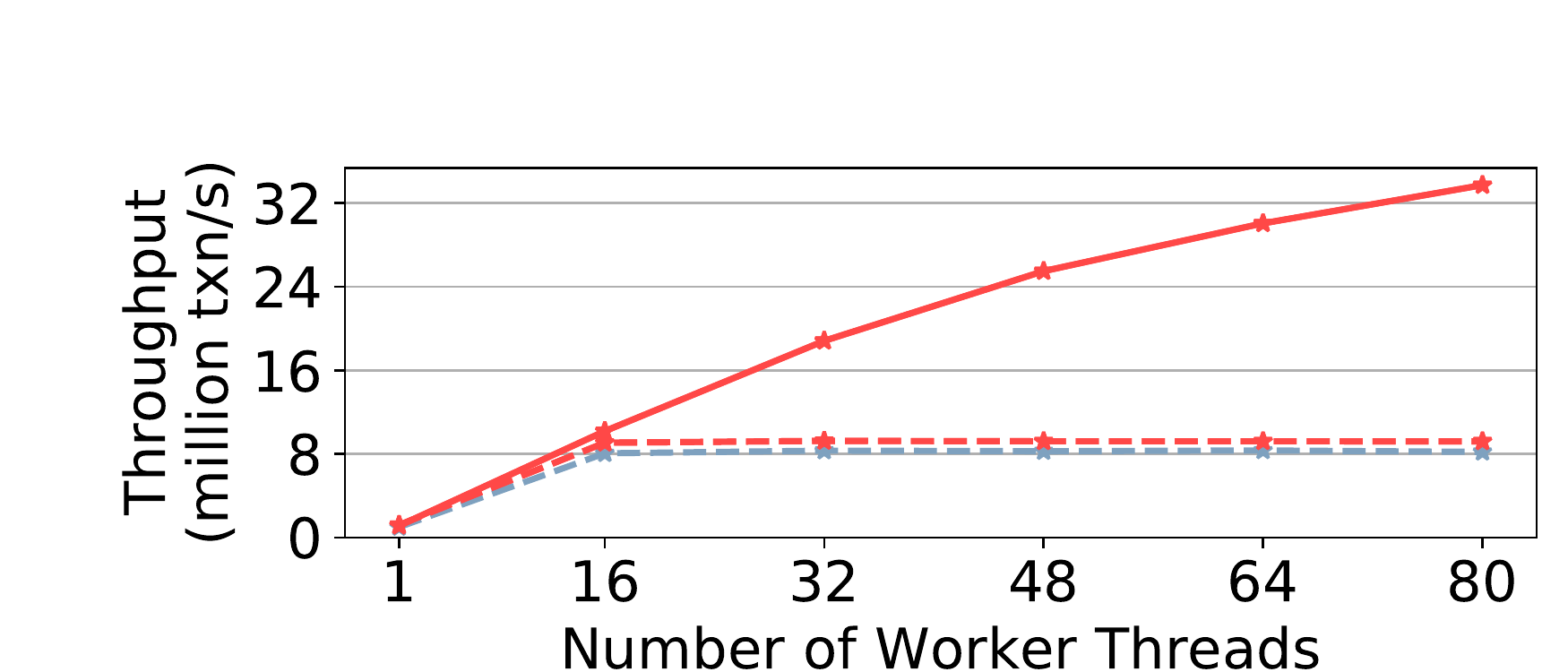}
            \label{fig:recovery-performance-adv-occ-b}
    }
    \hfill
    \subfloat[TPC-C \tpccNewOrder]{
        \adjincludegraphics[scale=\globalGraphScale, trim={0 0 0 {0.2\height}}, clip]
            {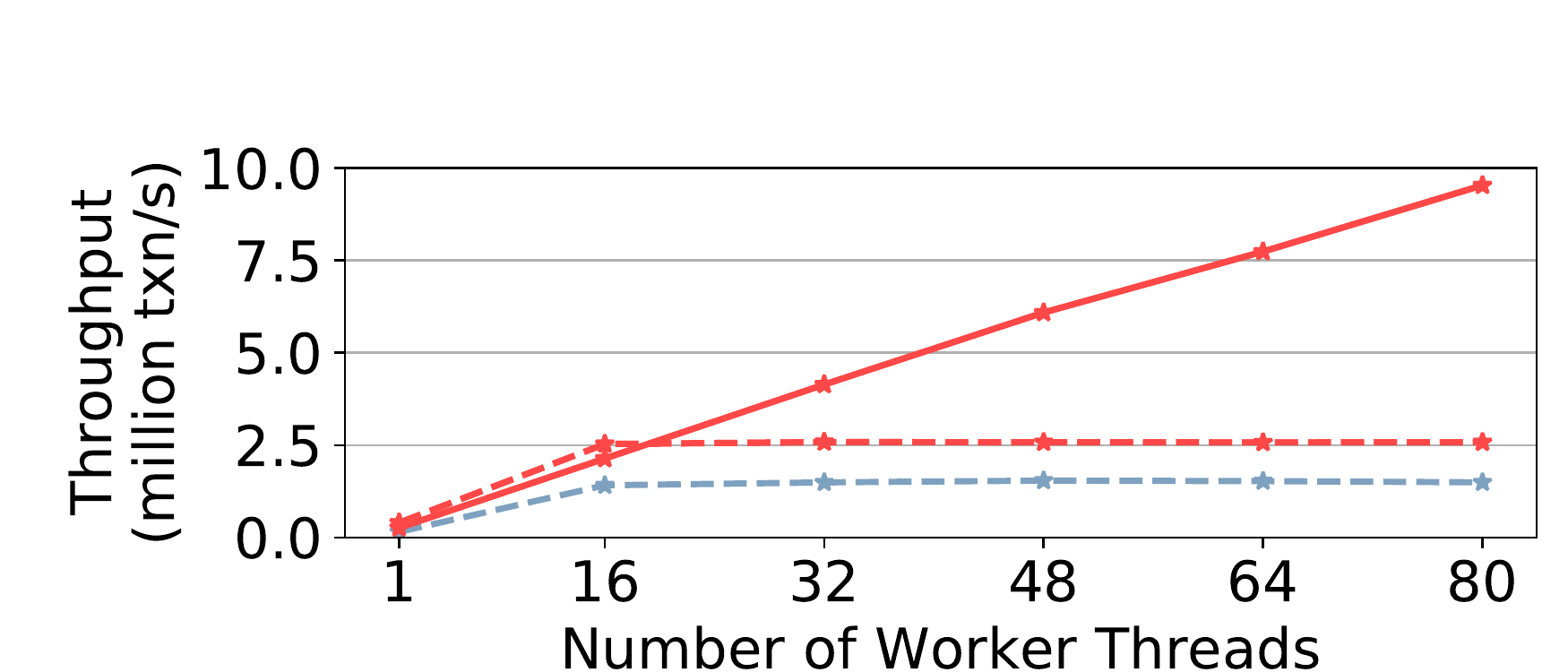}
            \label{fig:recovery-performance-adv-occ-c}
    }
    \caption{
        \textbf{Recovery Performance (OCC)} --
        Performance comparison on YCSB-500G, TPC-C \tpccPayment, and TPC-C \tpccNewOrder on NVMe drives.
    }
    \label{fig:recovery-performance-adv-occ}
\end{figure*}

\subsection{\myhighlightpar{Performances with NVMe SSDs}}
\label{sec:eval-nvme-drives}
{We run the DBMS on an Amazon EC2 \texttt{i3en.metal} instance with two Intel Xeon 
8175M CPUs (24 cores per CPU) with hyperthreading (96 virtual cores in total). The server has eight 
NVMe SSDs. Each device provides around 2~GB/s bandwidth and in total the server has 
theoretically 16~GB/s I/O bandwidth. We run the DBMS with at most 80 worker threads and 16 
log manager threads. Every disk contains two log files to better exploit the bandwidth.}
\vspace*{-0.1in} \\

\textbf{{Logging Performance:}}
Our first experiment evaluates the runtime performance of \codename by measuring 
the throughput when the number of worker threads changes. We test the logging protocols with 
YCSB-500G and TPC-C benchmarks.
We measure the throughput by the number of transactions committed by the worker threads per 
second. We keep the 2PL and OCC results separate to avoid comparisons based on the concurrency 
control algorithm performance.
We show the 2PL results in \cref{fig:logging-performance-adv-2pl} 
and the OCC results in \cref{fig:logging-performance-adv-occ}.
The x-axes are the number of worker threads (excluding the log managers), and the y-axes are the 
execution throughput. 

\cref{fig:logging-performance-adv-2pl-a} presents the logging performance for the YCSB-500G 
benchmark. \name{} with command logging scales linearly, while \name{} with data logging plateaus 
after 48 threads because it is bounded by the I/O of 16 dedicated writers. The serial command baseline also reaches a high 
throughput due to the succinctness of the command logging. It grows slower after 48 threads. 
This is not because of the disk bandwidth because it achieves similar 
performance with the RAID-0 disk array. It is instead because every transaction that spans multiple 
threads increments the shared 
LSN; this leads to excessive cache coherence traffic that inhibits scalability~\cite{tu13}. \name 
command logging is more scalable as the number of worker threads increases because 
each log manager maintains a separate LSN.
Serial data saturates the single disk's bandwidth. 
Similar to \name{}, Plover writes records across multiple files. For every 
transaction, it generates a log record for each accessed partition, and 
accesses the per-log LSN to generate a global LSN for the 
transaction. The DBMS then uses this global LSN to update the per-log sequence 
numbers. These updates are atomic to prevent data races. Plover is limited by the 
contention of the local counters. \name{} with command logging is up to 
\AdvmainLnXVIYCSBLogNumberofWorkerThreadsvsMaxThrCompareTaurusCommandNOWAITOverPloverDataNOWAIT 
faster than Plover.

\cref{fig:logging-performance-adv-2pl-b} shows the performance for the short and low-contended
\tpccPayment transactions. These results are similar to YCSB. %
All the logging baselines incur a significant overhead compared to 
No Logging. The gap between No Logging and \name{} reflects the overheads discussed in 
\cref{sec:vectorization}.
The LV maintenance in \name only costs 1.6\% of the running time. \name with command logging has the 
best performance. %
Plover suffers from the increased data accesses, causing the 
worker threads to compete for one of the few latches on the local sequence numbers, essentially 
downgrading to a single stream logging. \cref{fig:logging-performance-adv-2pl-c} shows 
the comparison for the TPC-C \tpccNewOrder transactions. These transactions access a larger number 
of tuples ($\sim$30 tuples per invocation) than the previous workloads. The overall 
throughput is lower, making it difficult for the DBMS to hit the LSN allocation bottleneck. 
Therefore, serial command logging scales well. The
gap between serial command with RAID-0 and \name command corresponds to LV-related overheads.
\name with command logging shows advantages when the number of workers is adequate. We project 
that the serial command logging will plateau
reaching the cache traffic limit when there are more than 120 workers whereas \name{} should still scale.
Similar to \tpccPayment transactions, Plover is bounded by the contention. 

\cref{fig:logging-performance-adv-occ} shows the comparison between the OCC variant of \name and 
Silo-R. The No Logging baseline also uses the OCC algorithm to keep comparison fair. For all the 
benchmarks we observe that both Silo-R and \name data logging plateaus at a similar level, saturating the disk bandwidth. Before they saturate the bandwidth,
Silo-R performs slightly better than \name because it does not need to track LSN Vectors.
However, Silo-R cannot track RAW 
dependencies, so it is incompatible with command logging. \name command logging, 
benefiting 
from the conciseness of the log records, outperforms Silo-R in every benchmark, by up 
to 
\AdvmainLnXVIsiloPayTPCCLogNumberofWorkerThreadsvsMaxThrCompareTaurusCommandSiloOverSiloRDataSilo.
\vspace*{-0.1in} \\

\textbf{{Recovery Performance:}} We next evaluate the DBMS's recovery time for all the 
protocols. We use the log files generated by 80 worker threads for 
better recovery parallelism. These files are large enough for steady performance measurements and 
are stored in uncompressed bytes across the disks with I/O caches cleaned.

\cref{fig:recovery-performance-adv-2pl-a} shows the recovery peformance on YCSB-500G. Plover 
outperforms \name below 80 threads because it does not need to resolve dependencies. Every 
Plover log file corresponds to a partition that contains totally ordered records, which is 
sufficient to recover transactions independently. Plover saturates the 
disk's 16~GB/s bandwidth after 48 threads and plateaus. \name{} command scales 
linearly and exceeds Plover at 80 threads. The serial baselines, regardless of data or command 
logging, with a RAID-0 setup or not, are limited by the total sequence order of transactions. 
\name{} recovery is up to 
\AdvmainLnXVIYCSBRecNumberofWorkerThreadsvsMaxThrCompareTaurusCommandNOWAITOverSerialCommandNOWAIT 
faster than the serial baselines.

The recovery performance of TPC-C \tpccPayment is in \cref{fig:recovery-performance-adv-2pl-b}. 
Both Plover and \name data logging hit the I/O bottleneck quickly, while \name{} command logging scales linearly. %
\cref{fig:recovery-performance-2pl-command-c} shows the 
comparison for TPC-C \tpccNewOrder. \name command scales well and outperforms Plover by 
up to 
\AdvmainLnXVINewOrderTPCCRecNumberofWorkerThreadsvsMaxThrCompareTaurusCommandNOWAITOverPloverDataNOWAIT. 
The gap between Plover and \name data logging corresponds to dependency resolution and the 
resulting memory overhead. \name{} command is slower than \name{} data 
at 16 threads due to the cost of re-running the transactions.

\cref{fig:recovery-performance-adv-occ} presents the comparisons for the OCC baselines. Similar to 
Plover, Silo-R requires data logging and therefore falls behind \name command logging in all three benchmarks. But Silo-R does not require dependency resolution so it outperforms \name with 
data 
logging when the number of transactions is large. Silo-R uses latches when transactions update 
tuples to ensure that they only perform updates with a higher version number. This 
overhead is more 
significant when the transactions are long.
\name{} command logging outperforms Silo-R by up to 
\AdvmainLnXVINewOrderTPCCRecNumberofWorkerThreadsvsMaxThrCompareTaurusCommandNOWAITOverSiloRDataSilo
. 

\subsection{\myhighlightpar{Performance with Hard Disks}}
\label{sec:evaluation-hdd-drives}

\begin{figure*}[t]
    \centering
    \hspace{0.4in}{
        \fbox{\adjincludegraphics[scale=0.3, trim={{0.31\width} {0.55\height} 0 {0.1\height}}, clip]
            {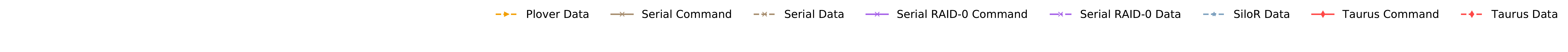}}
    }
    \vspace{-0.0in}\newline
    \subfloat[YCSB-10G Data]{
        \adjincludegraphics[scale=\globalGraphScale, trim={0 0 0 {0.2\height}}, clip]
            {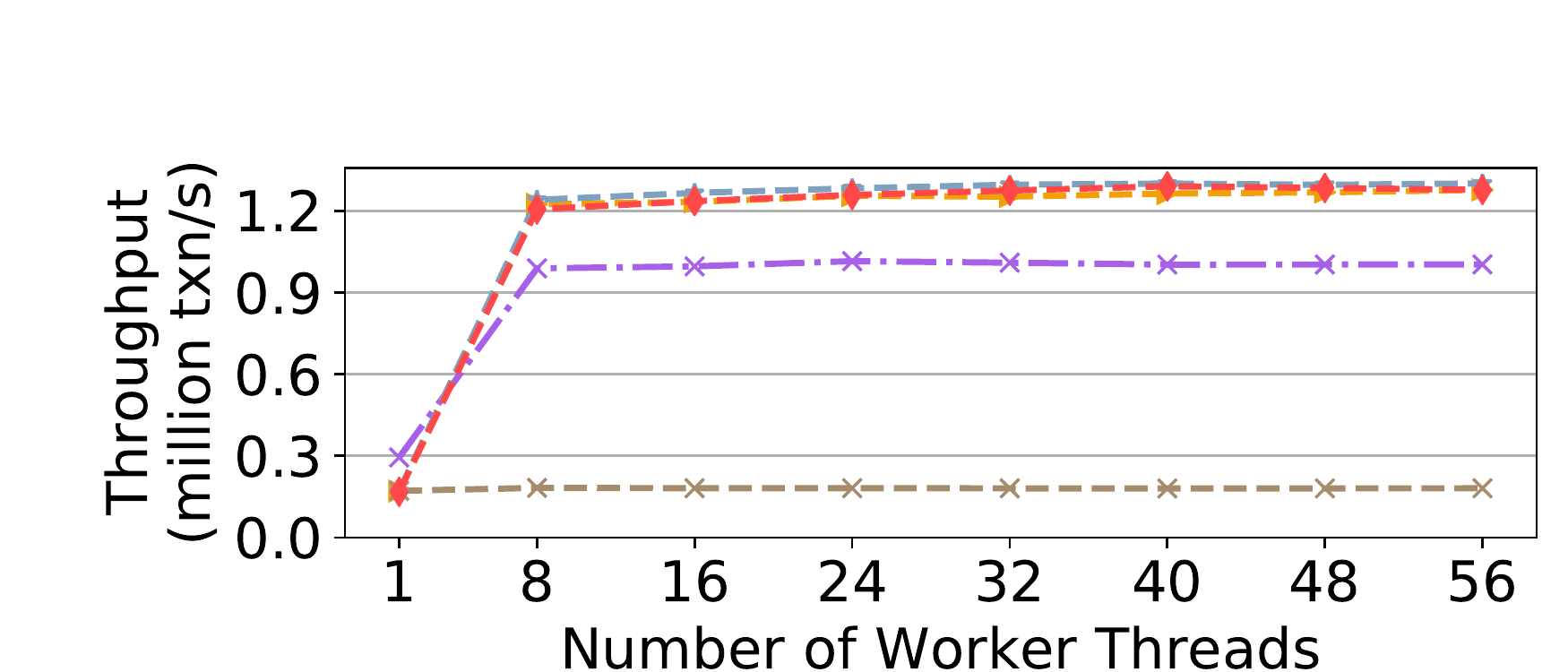}
            \label{fig:logging-performance-2pl-data-a}
    }
    \hfill
    \subfloat[TPC-C \tpccPayment Data]{
        \adjincludegraphics[scale=\globalGraphScale, trim={0 0 0 {0.2\height}}, clip]
            {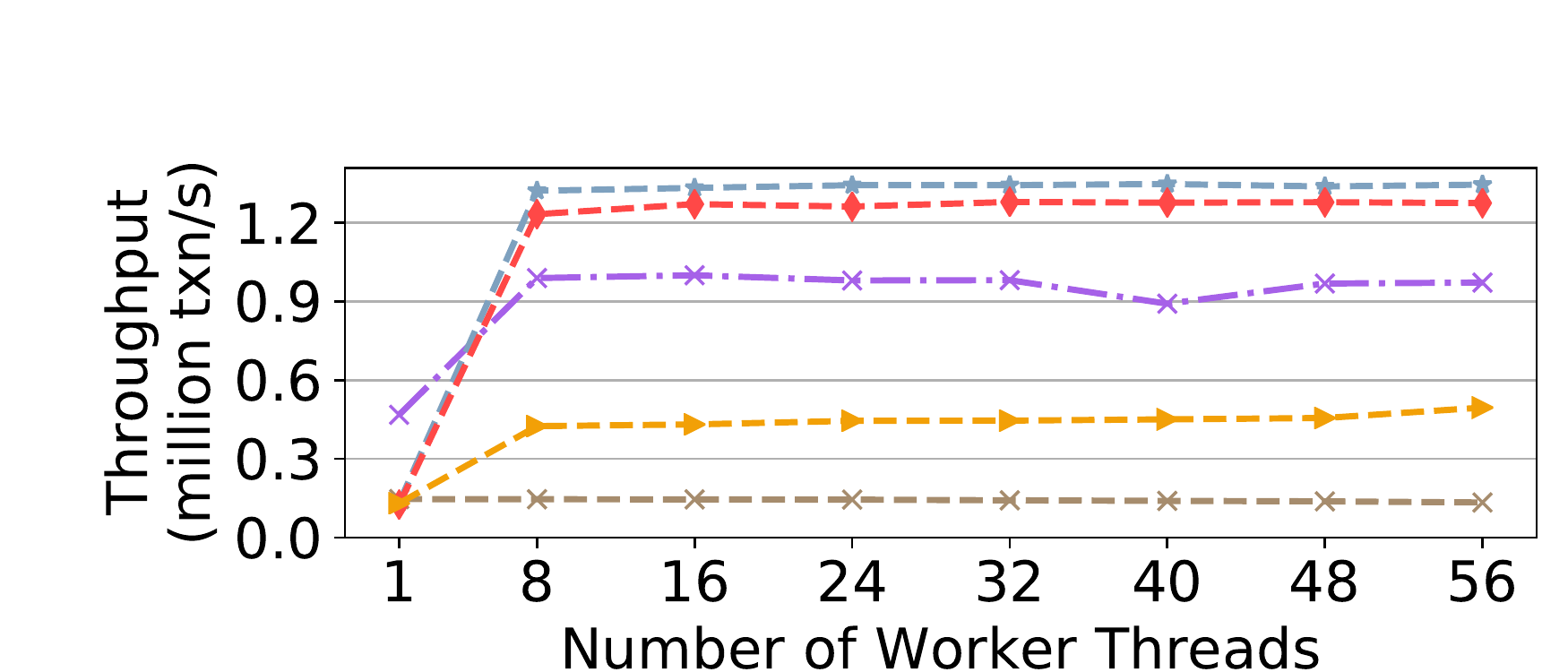}
            \label{fig:logging-performance-2pl-data-b}
    }
    \hfill
    \subfloat[TPC-C \tpccNewOrder Data]{
        \adjincludegraphics[scale=\globalGraphScale, trim={0 0 0 {0.2\height}}, clip]
            {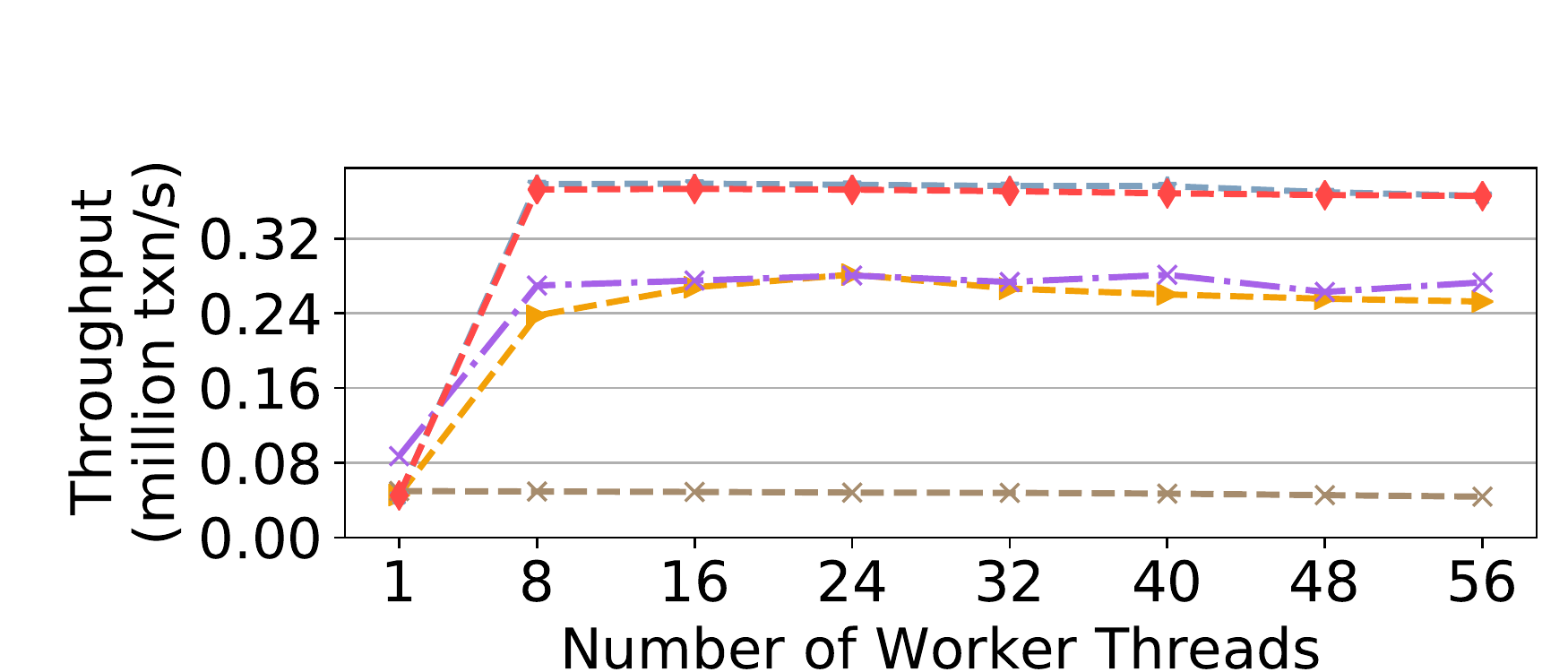}
            \label{fig:logging-performance-2pl-data-c}
    }

    \subfloat[YCSB-10G Command]{
        \adjincludegraphics[scale=\globalGraphScale, trim={0 0 0 {0.2\height}}, clip]
            {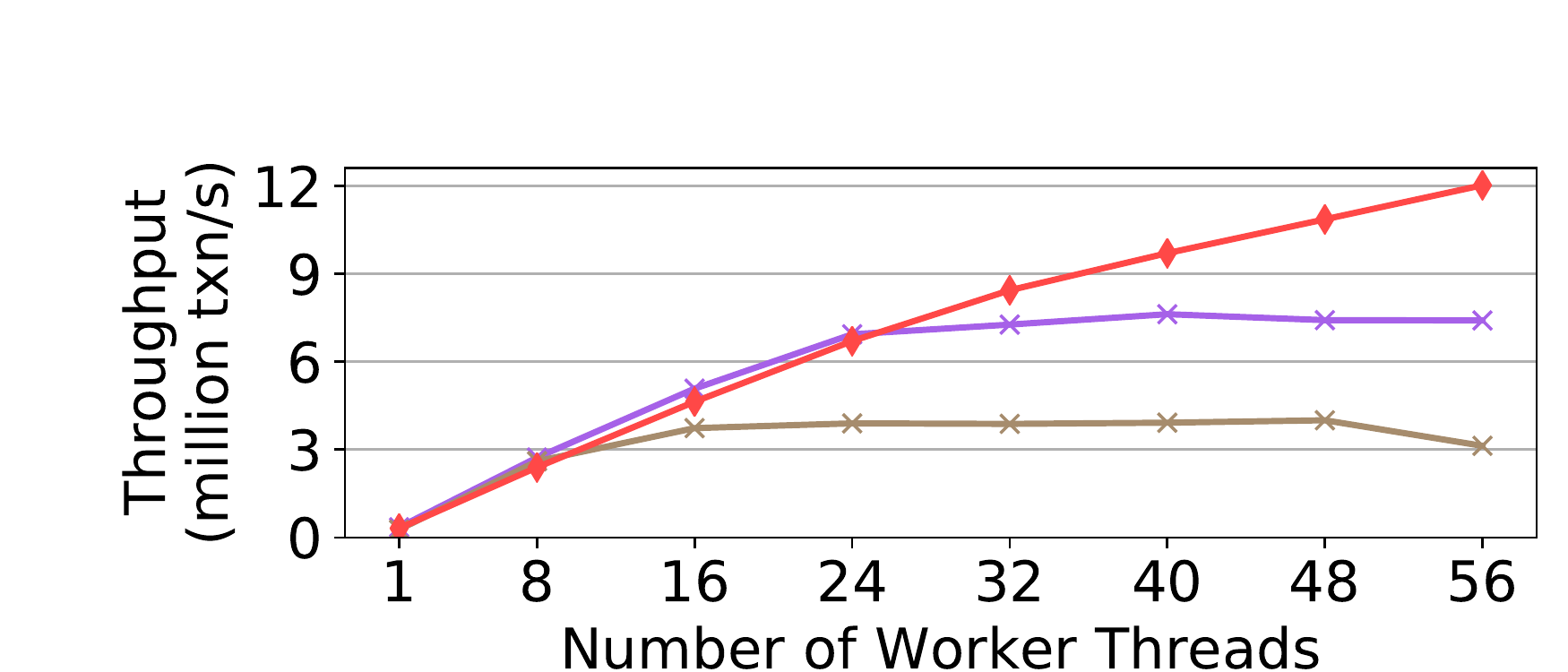}
            \label{fig:logging-performance-2pl-command-a}
    }
    \hfill
    \subfloat[TPC-C \tpccPayment Command]{
        \adjincludegraphics[scale=\globalGraphScale, trim={0 0 0 {0.2\height}}, clip]
            {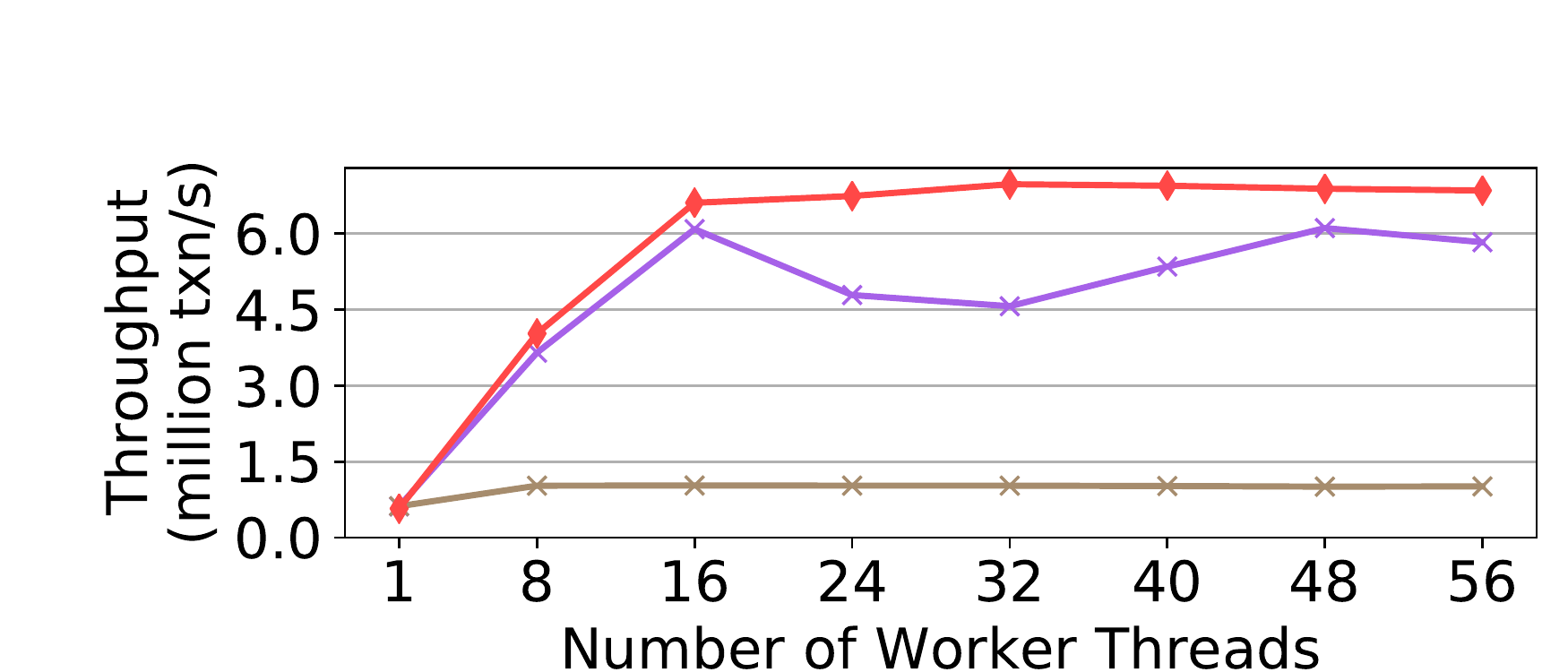}
            \label{fig:logging-performance-2pl-command-b}
    }
    \hfill
    \subfloat[TPC-C \tpccNewOrder Command]{
        \adjincludegraphics[scale=\globalGraphScale, trim={0 0 0 {0.2\height}}, clip]
            {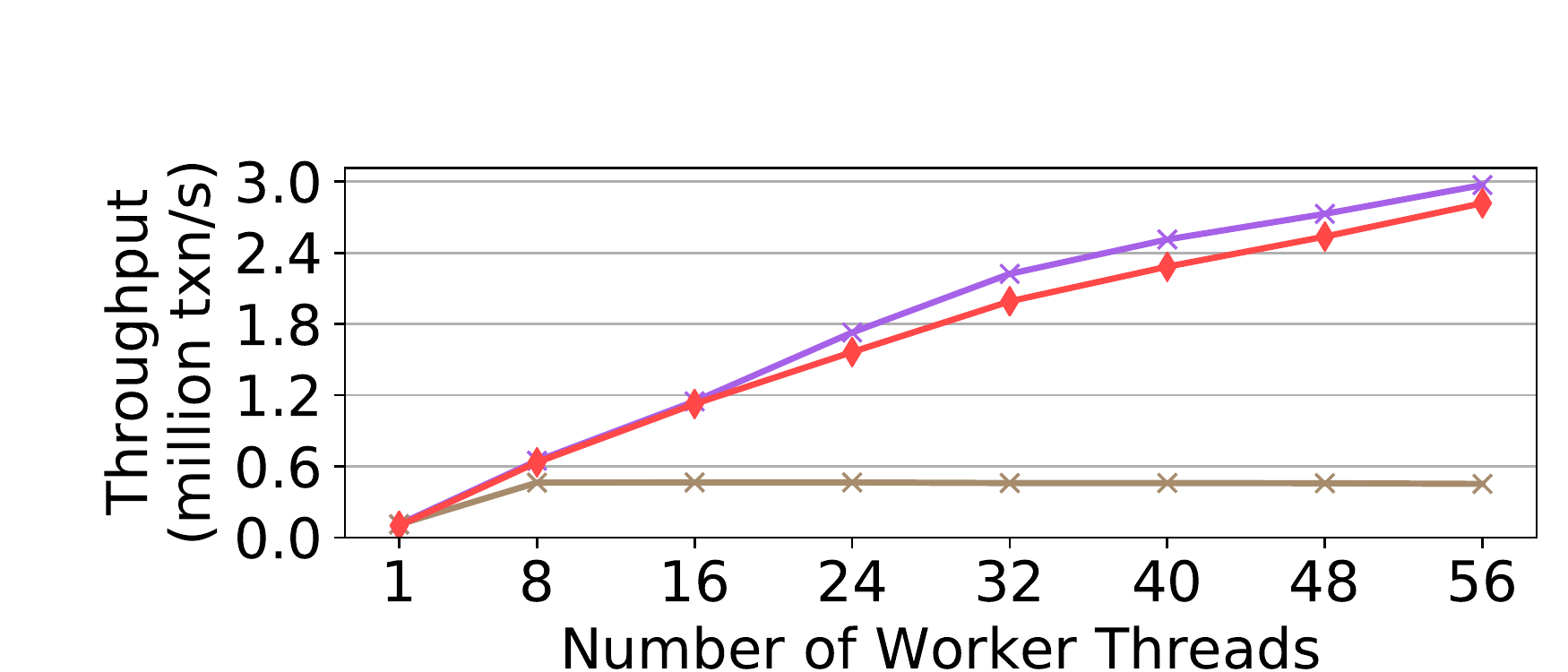}
            \label{fig:logging-performance-2pl-command-c}
    }
    \caption{
        \textbf{Data and Command Logging Performance} --
        Performance comparison on YCSB-10G, TPC-C \tpccPayment, and TPC-C \tpccNewOrder on HDDs.
    }
    \label{fig:logging-performance-2pl}
    \vspace{-0.05in}
\end{figure*}

\begin{figure*}[t]
    \centering
    \subfloat[YCSB-10G Data]{
        \adjincludegraphics[scale=\globalGraphScale, trim={0 0 0 {0.2\height}}, clip]
            {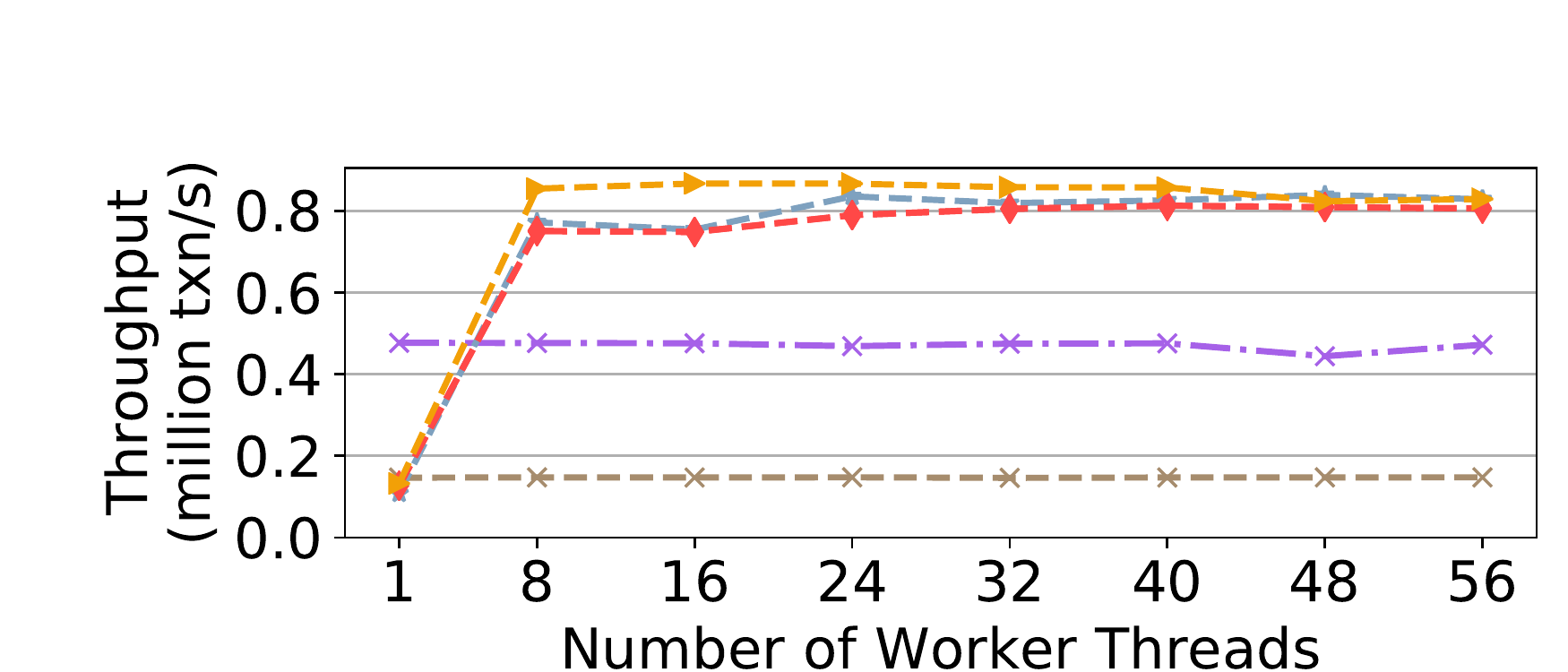}\label{fig:recovery-performance-2pl-data-a}
    }\hfill
    \subfloat[TPC-C \tpccPayment Data]{
        \adjincludegraphics[scale=\globalGraphScale, trim={0 0 0 {0.2\height}}, clip]
            {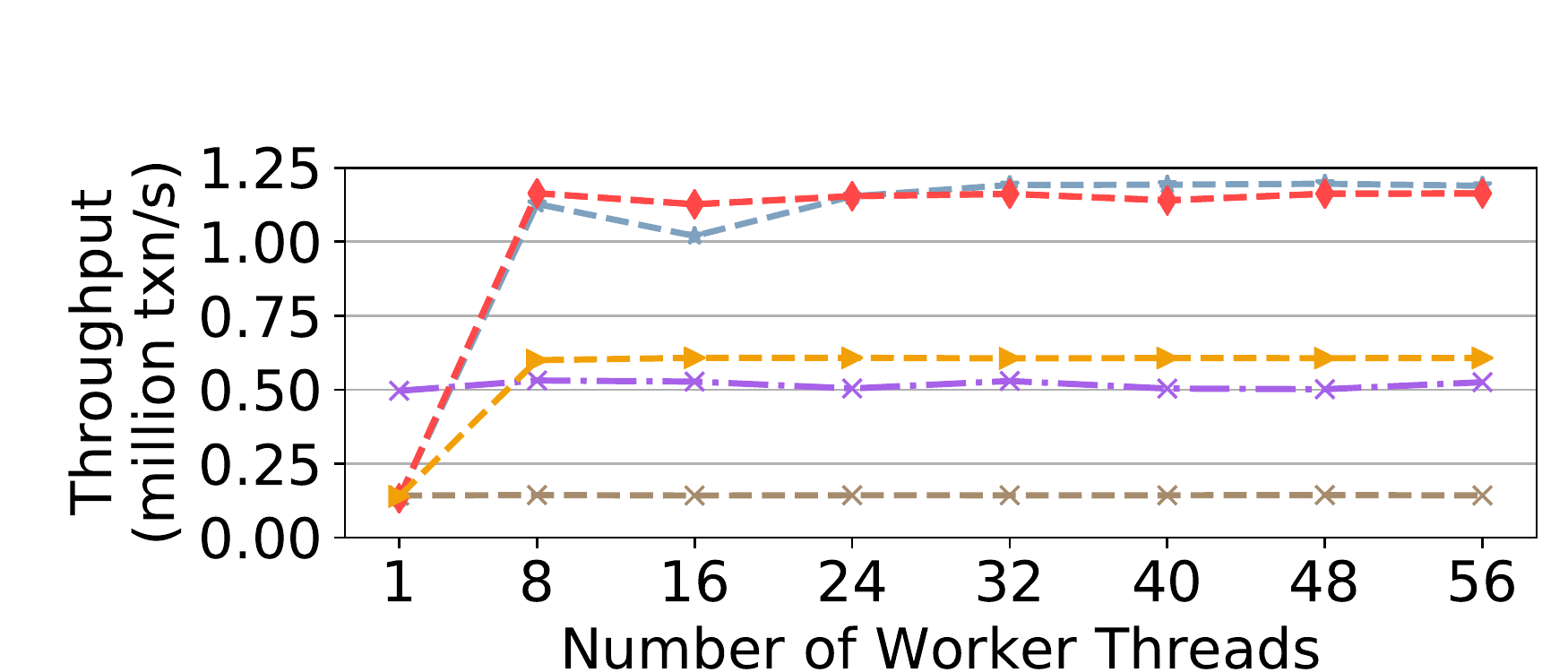}\label{fig:recovery-performance-2pl-data-b}
    }
    \hfill
    \subfloat[TPC-C \tpccNewOrder Data]{
        \adjincludegraphics[scale=\globalGraphScale, trim={0 0 0 {0.2\height}}, clip]
            {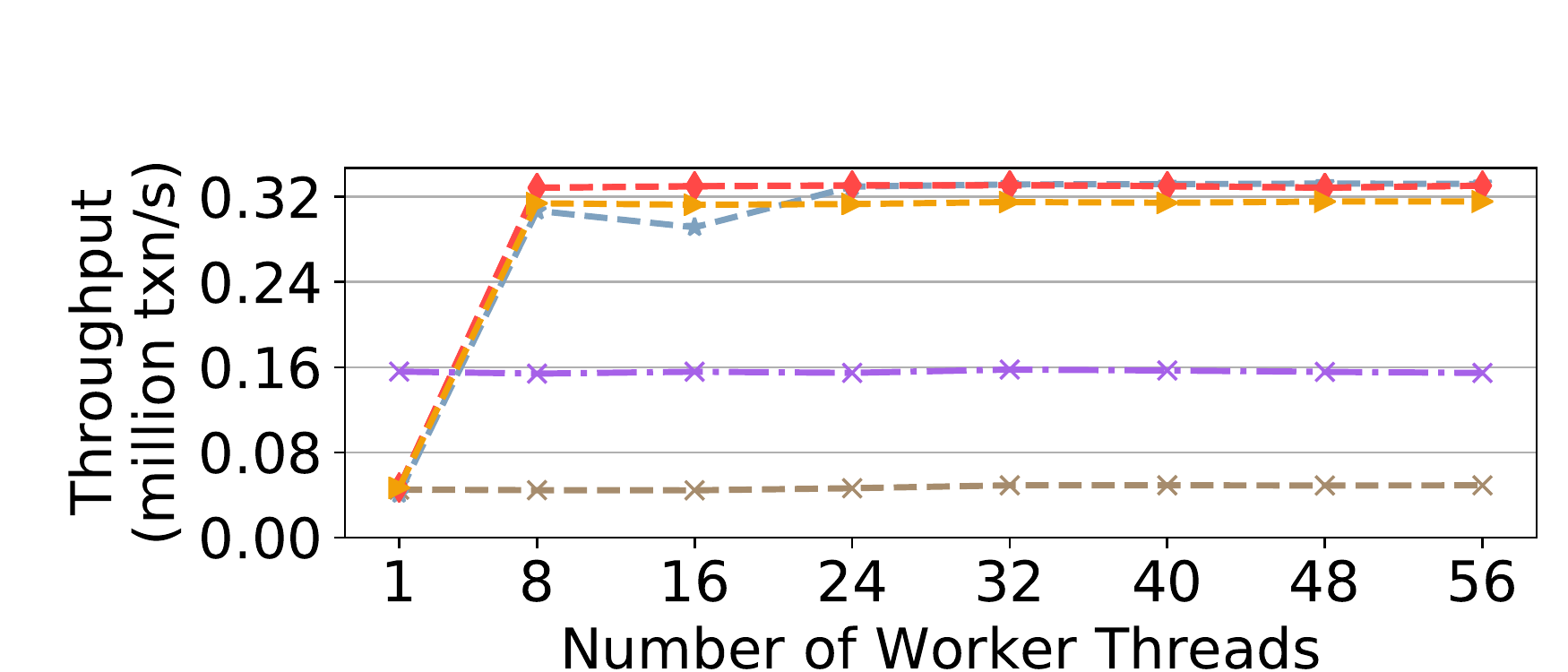}\label{fig:recovery-performance-2pl-data-c}
    }
    \newline
    \subfloat[YCSB-10G Command]{
        \adjincludegraphics[scale=\globalGraphScale, trim={0 0 0 {0.2\height}}, clip]
            {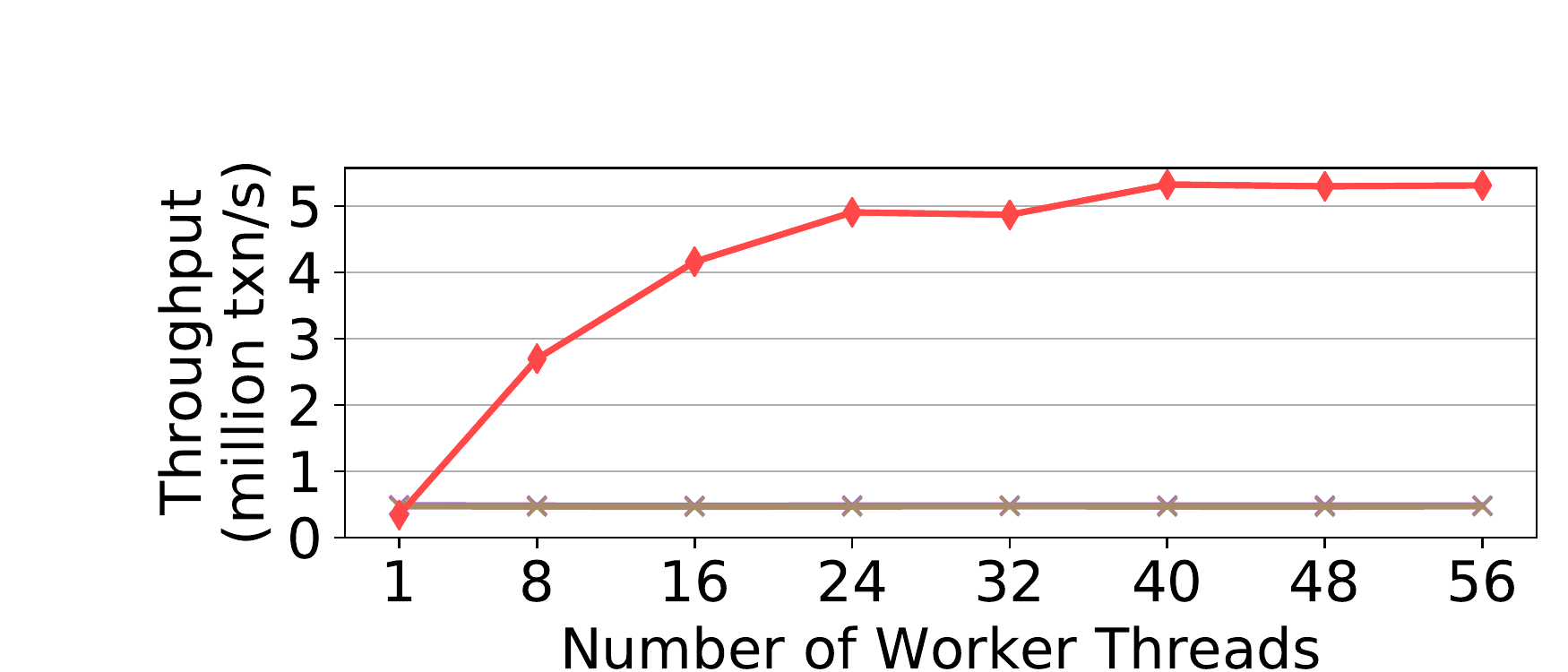}\label{fig:recovery-performance-2pl-command-a}
    }\hfill
    \subfloat[TPC-C \tpccPayment Command]{
        \adjincludegraphics[scale=\globalGraphScale, trim={0 0 0 {0.2\height}}, clip]
            {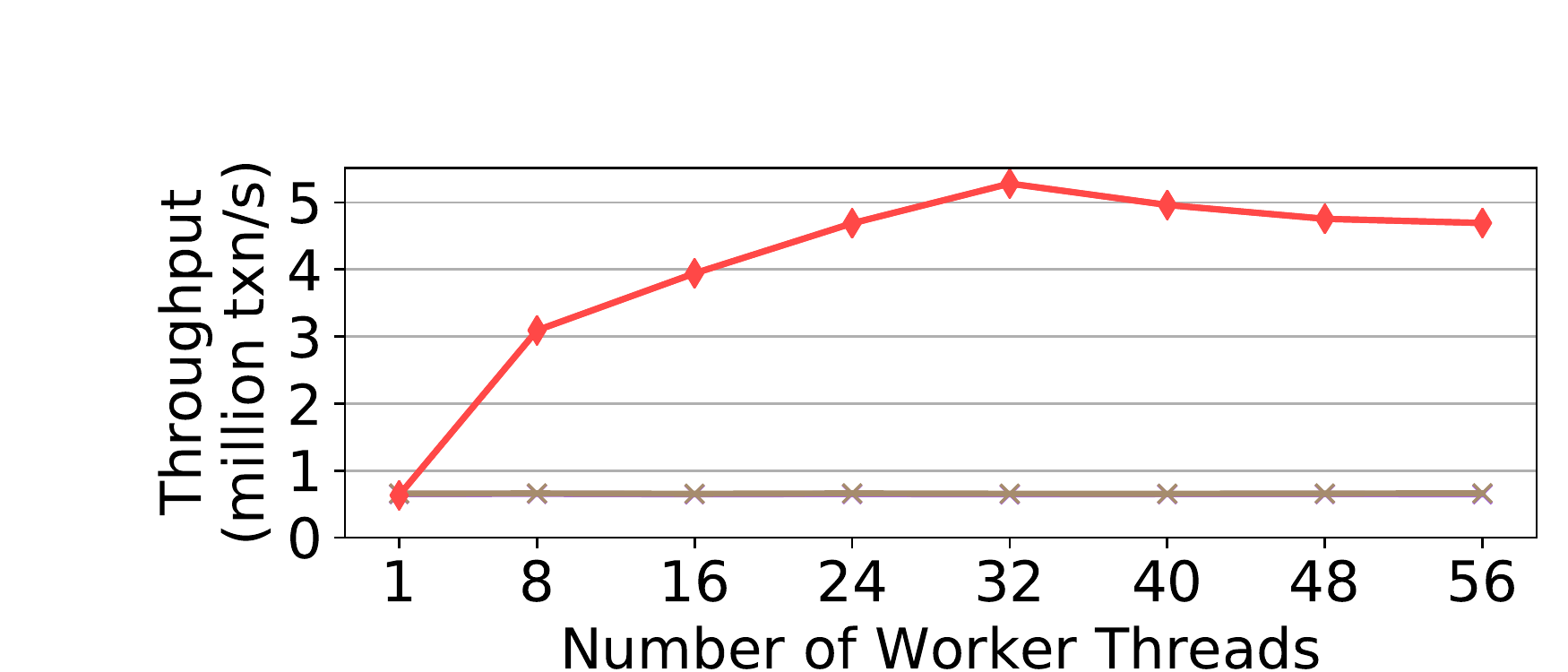}\label{fig:recovery-performance-2pl-command-b}
    }
    \hfill
    \subfloat[TPC-C \tpccNewOrder Command]{
        \adjincludegraphics[scale=\globalGraphScale, trim={0 0 0 {0.2\height}}, clip]
            {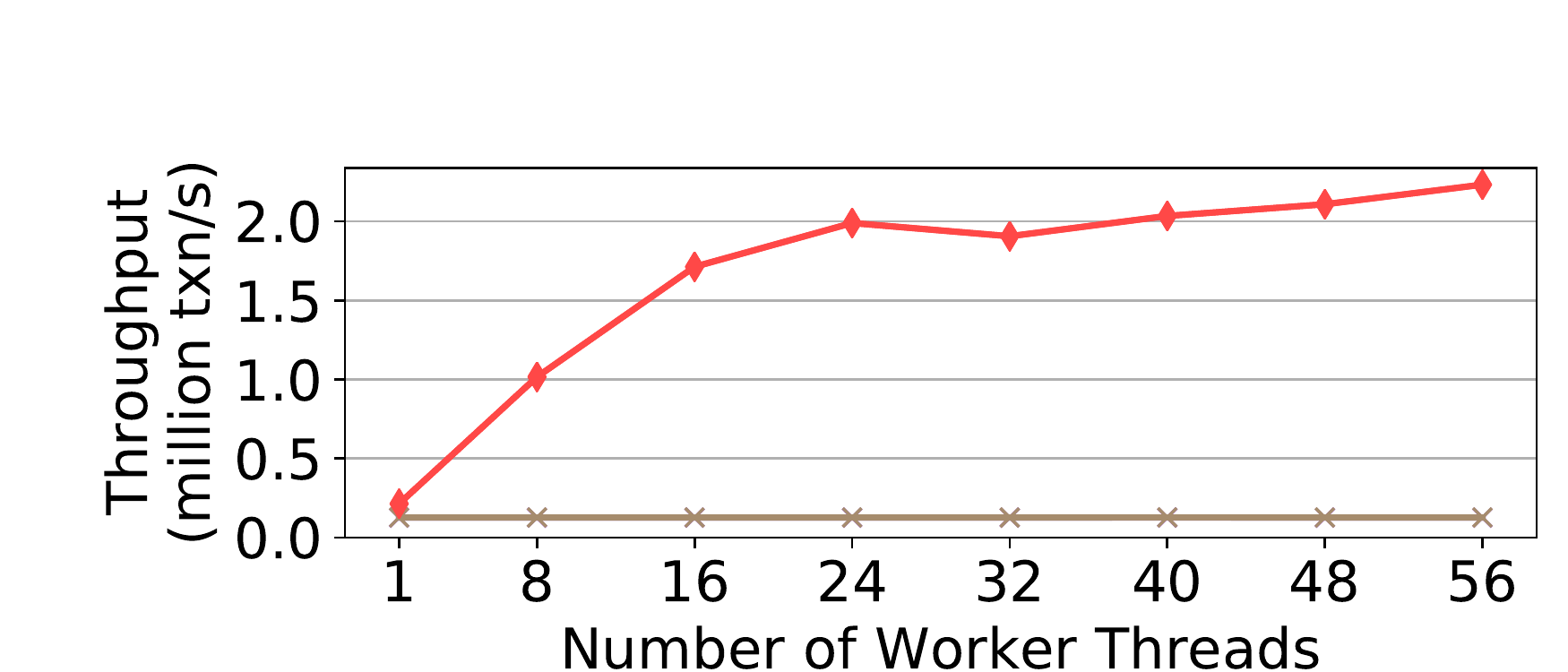}\label{fig:recovery-performance-2pl-command-c}
    }
    \caption{
        \textbf{Data and Command Recovery Performance} --
        Performance comparison on YCSB-10G, TPC-C \tpccPayment, and TPC-C \tpccNewOrder on HDDs.
    }
    \label{fig:recovery-performance-2pl}
    \vspace{-0.05in}
\end{figure*}

To better understand the performance of baselines with limited bandwidth,
we performed the evaluation on an Amazon EC2 \texttt{h1.16xlarge} machine with eight HDD drives. 
Each disk provides around 160~MB/s bandwidth and in total the server has 1.3~GB/s I/O bandwidth. 
Because the server only has 256~GB memory, we use YCSB-10G.
The data logging and command logging baselines differ in absolute throughput on 
HDDs, so we present them separately. Silo-R is bound by the disk bandwidth often, and the 
difference in concurrency control does not contribute to the relative order. Therefore, 
we display the results for Silo-R and 2PL baselines together.
\vspace*{-0.1in} \\

\textbf{Logging Performance:} \cref{fig:logging-performance-2pl-data-a} 
shows the logging performance of data logging baselines for YCSB-10G. We observe that serial 
data saturates the bandwidth of a single disk quickly.
\name Data achieves 
\mainDataYCSBLogNumberofWorkerThreadsvsThroughputCompareTaurusDataNOWAITOverSerialDataNOWAIT higher 
throughput than serial data as it writes to eight 
disks in parallel. Serial data logging on RAID-0 delivers similar performance since the 
bandwidth of the disk array is $8\times$ greater. Silo-R and Plover also write the log across 
eight disks uniformly, thereby achieving similar 
performance to \name. 
In \cref{fig:logging-performance-2pl-data-b,fig:logging-performance-2pl-data-c}, we 
also observe this pattern for the TPC-C transactions except that 
Plover plateaus because of the high contention.

\cref{fig:logging-performance-2pl-command-a} shows the command logging baselines for the YCSB benchmark.
Serial command logging outperforms serial data logging, benefiting from the smaller log record sizes.
Starting from 16 threads, its performance is limited by 
the bandwidth of a single disk.
The serial command baseline on a RAID array plateaus after 24 threads, limited by the cache coherence traffic.  
\name with command logging is \mainYCSBLogNumberofWorkerThreadsvsThroughputCompareTaurusCommandNOWAITOverSiloRDataSiloMax faster than Silo-R and Plover.
\cref{fig:logging-performance-2pl-command-b} shows the DBMS's throughput for the TPC-C \tpccPayment
transaction. \name plateaus after 16 threads, limited by the disk bandwidth, achieving \mainPayTPCCLogNumberofWorkerThreadsvsThroughputCompareTaurusCommandNOWAITOverSiloRDataSiloMax speedup over Silo-R. Serial command logging suffers from NUMA issues between 16 threads and 48 threads as the log buffer resides on a single socket.
For the TPC-C \tpccNewOrder workload in \cref{fig:logging-performance-2pl-command-c}, both serial command logging on the RAID-0 array and \name command logging have good scalability. 
\vspace*{-0.1in} \\

\textbf{Recovery Performance:} \cref{fig:recovery-performance-2pl} presents the recovery performance on HDDs. 
The serial baselines are again limited by the transaction total order.
For \name, we can see that the recovery
performance of data logging plateaus after the number of worker threads exceeds 8. It is up to \mainDataYCSBRecNumberofWorkerThreadsvsMaxThrCompareTaurusDataNOWAITOverSerialRAIDDataNOWAITMax faster than the serial data
logging on a disk array.
\name data logging achieves similar throughput as Silo-R, while \name command logging
is up to 
\mainYCSBRecNumberofWorkerThreadsvsMaxThrCompareTaurusCommandNOWAITOverSiloRDataSiloMax faster
than Silo-R for recovery. Plover parallels Silo-R except for \tpccPayment, where the contention devolves Plover to single stream logging.

The peak performance of \name command logging and \name data logging are \mainCommandYCSBRecNumberofWorkerThreadsvsMaxThrCompareTaurusCommandNOWAITOverSerialCommandNOWAITMax and \mainDataYCSBRecNumberofWorkerThreadsvsMaxThrCompareTaurusDataNOWAITOverSerialDataNOWAITMax faster than serial baselines for YCSB recovery. 
For TPC-C \tpccPayment in \cref{fig:recovery-performance-2pl-command-b}, the DBMS achieves its peak
recovery performance using \name command logging where it is \mainCommandPayTPCCRecNumberofWorkerThreadsvsMaxThrCompareTaurusCommandNOWAITOverSerialCommandNOWAIT
faster than the serial command
logging baseline. The performance of \name command logging decreases when the number of workers increases beyond 24 because the parallelism is fully exploited and more threads will only incur more contention.

For TPC-C \tpccNewOrder, the performance ratios between \name and the serial baselines are
\mainCommandNewOrderTPCCRecNumberofWorkerThreadsvsMaxThrCompareTaurusCommandNOWAITOverSerialCommandNOWAITMax and
\mainDataNewOrderTPCCRecNumberofWorkerThreadsvsMaxThrCompareTaurusDataNOWAITOverSerialDataNOWAITMax for command logging (or data logging lifted by disk arrays)
and data logging (without disk arrays), respectively. If the DBMS
uses \name command logging protocol instead of its data logging protocol, then it
improves the performance by \mainNewOrderTPCCLogNumberofWorkerThreadsvsThroughputCompareTaurusCommandNOWAITOverTaurusDataNOWAIT. This is up to \mainNewOrderTPCCLogNumberofWorkerThreadsvsThroughputCompareTaurusCommandNOWAITOverSerialDataNOWAITMax better than serial data logging. 
Databases with limited bandwidth can benefit from \name supporting command logging.

\begin{figure}
    \centering
    \hspace{0.0in}{
        \fbox{\adjincludegraphics[scale=0.3, trim={{0.17\width} {0.3\height} 0 {0.1\height}}, clip]
            {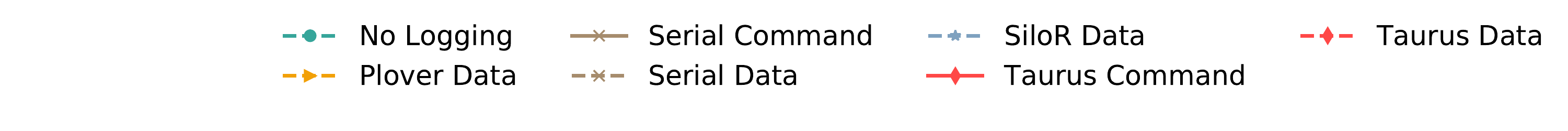}}
    }
    \vspace{0in}\newline
    \hspace{-0.3in}
    \subfloat[Logging YCSB-500G]{
        \adjincludegraphics[scale=0.27, trim={{0.03\width} 0 {0.1in} {0.2\height}}, clip]
            {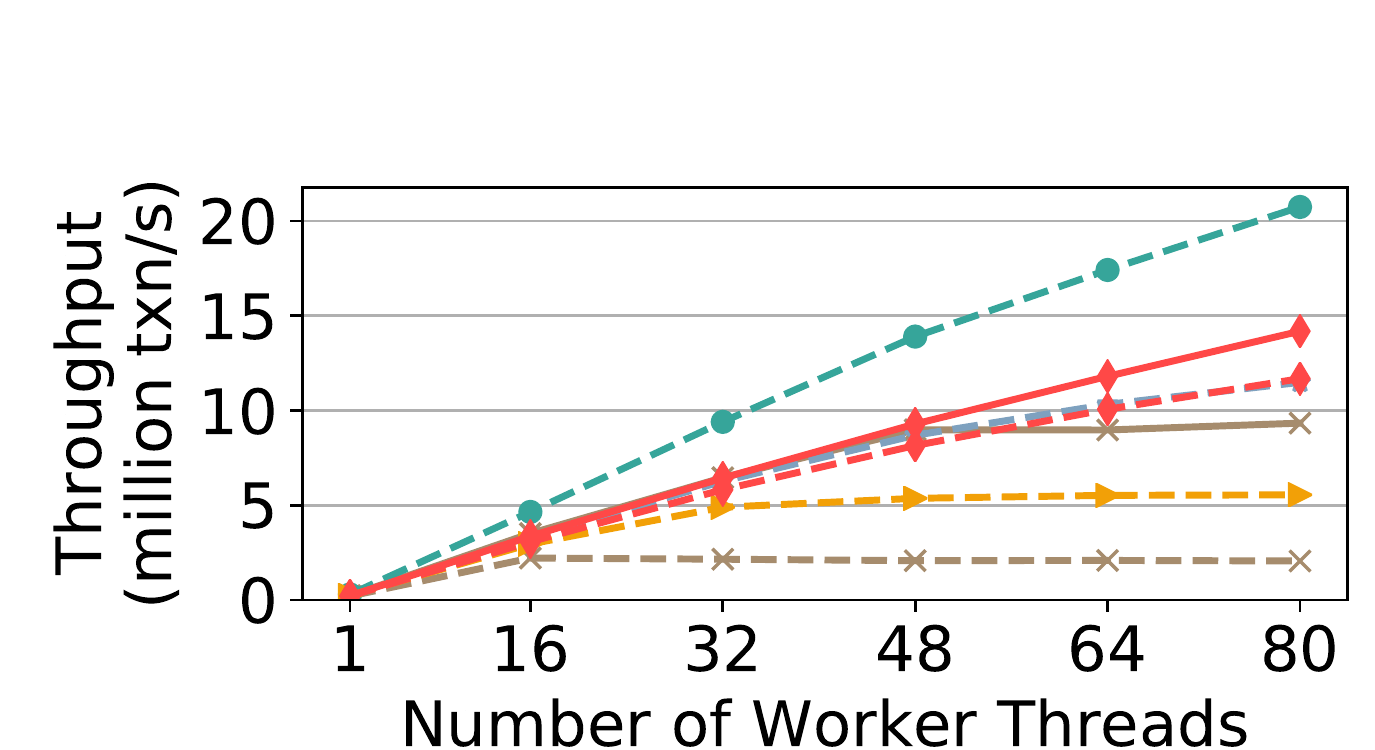}
            \label{fig:dram-2pl-logging}
    }
    \hspace{-0.1in}
    \subfloat[Recovery YCSB-500G]{
        \adjincludegraphics[scale=0.27, trim={{0.15\width} 0 0 {0.2\height}}, clip]
            {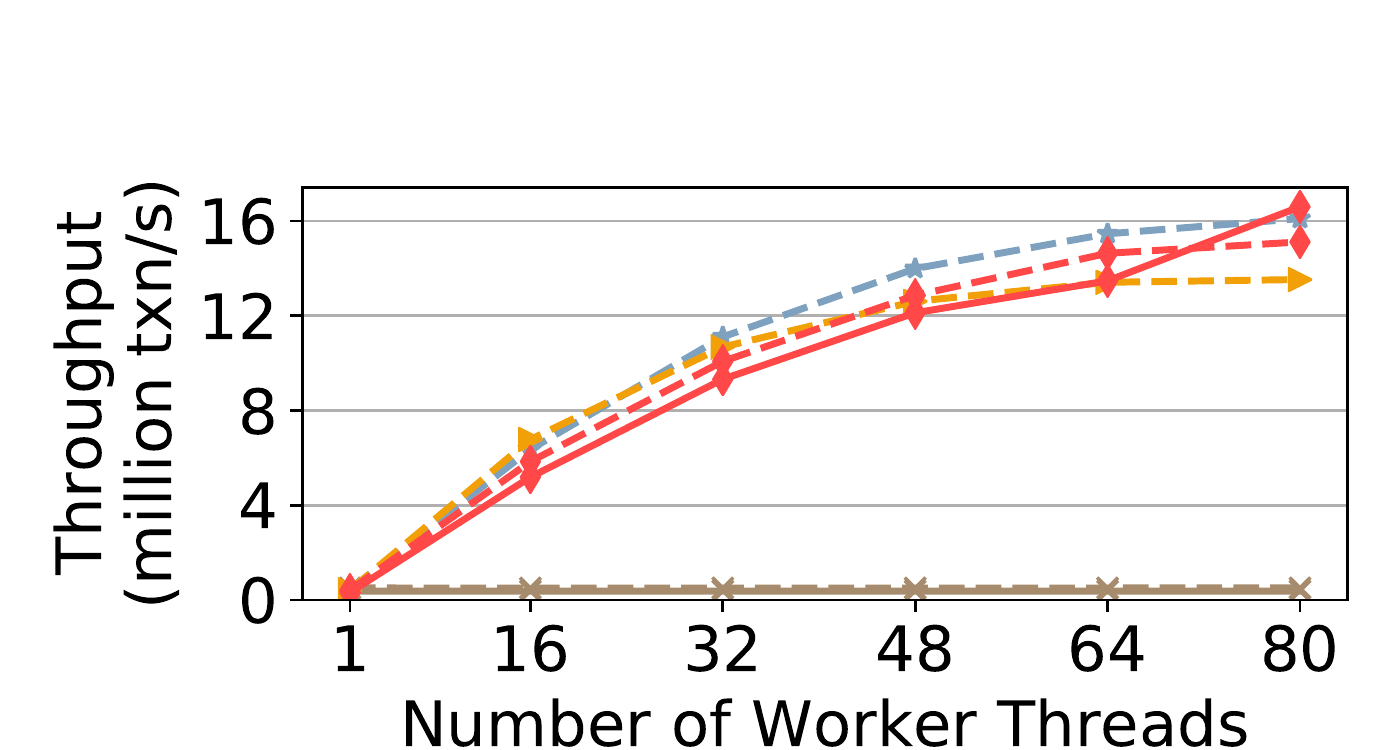}
            \label{fig:dram-2pl-recovery}
    }
    \caption{\textbf{DRAM Performance} -- Performance comparison on DRAM filesystems.} %
    \label{fig:dram-2pl}
\end{figure}

\subsection{\myhighlightpar{Performance with PM (RAM Disk)}}
\label{sec:eval-dram}
Since emerging Persistent Memory (PM) has higher I/O bandwidth than of SSDs and HDDs,
we evaluated the performance of \name{} on DRAM filesystems to simulate a PM environment. Every 
operation to this filesystem goes through the OS. This overhead is also shared in 
the architecture with a real PM. 
The PM incurs a higher latency 
($<$1~us for 99.99\%) and has a bandwidth $3-13\times$ lower than DRAM~\cite{yang2020empirical}. We conjecture that \name{} command logging would perform relatively better than other baselines on a real PM because the bandwidth might become the bottleneck.

The results on the DRAM filesystem are shown in \cref{fig:dram-2pl}. 
The advantage of command logging over data logging is greatly reduced when the bandwidth is sufficient.
\name{} command logging scales linearly, while the serial command logging is again restricted by the 
contention of the single counter.
All the parallel algorithms scale well in recovery. Silo-R outperforms \name slightly because it does not have to resolve dependencies during recovery.
We can infer that \name{} does not incur observable overhead that would preclude it from a PM-based 
DBMS.

\subsection{TPC-C Full Mix}
\label{sec:eval_tpcc_full_mix}
\myhighlightpar{
To demonstrate the generality of \name{} and to evaluate \name{} in a more realistic DBMS OLTP workload, we added the support for range scans, row insertions, and row deletions. We implement all the types of transactions from the TPC-C benchmark with the 2PL concurrency control algorithm. %
The full TPC-C mix consists of 45\% \tpccNewOrder, 43\% \tpccPayment, 4\% \tpccOrderStatus, 4\% \tpccDelivery, and 4\% \tpccStockLevel. Among them, \tpccOrderStatus and \tpccStockLevel are read-only transactions, and therefore \name{} does not create log records for them. Figure~\ref{fig:tpcc-full-2pl} shows the logging performance and recovery performance. We observe that, starting from 32 threads, the logging algorithms are limited by the workload parallelism since the no logging baseline plateaus at a similar level. Compared to the no logging baseline, the overhead caused by \name{} is around \mainTPCFLogNumberofWorkerThreadsvsThroughputCompareTaurusCommandNOWAITOverNoLoggingDataNOWAITMax{}. In recovery, the serial algorithms are again limited by the loss of parallelism. \name{} command logging outperforms the serial baselines by \mainTPCFRecNumberofWorkerThreadsvsThroughputCompareTaurusCommandNOWAITOverSerialCommandNOWAIT{}. \cref{fig:tpcc-full-breakdown} shows the time breakdown among all the five TPC-C transactions. Although \tpccStockLevel transactions are read-only, they take a significant proportion of the total running time. This proportion increases with the number of threads because \tpccStockLevel transactions perform massive read operations. %
}
\begin{figure}[t]
    \centering
    \hspace{0.in}{
        \fbox{\adjincludegraphics[scale=0.3, trim={{0.29\width} {0.25\height} {0.01\width} {0.1\height}}, clip]
            {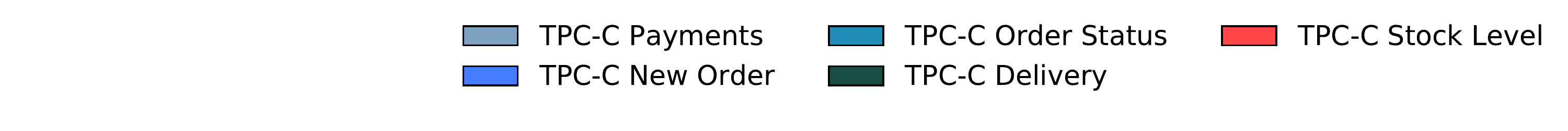}}
    }
    \vspace{-0.0in}\newline
    \adjincludegraphics[scale=0.27, trim={0 0 0 {0.05\height}}, clip]
            {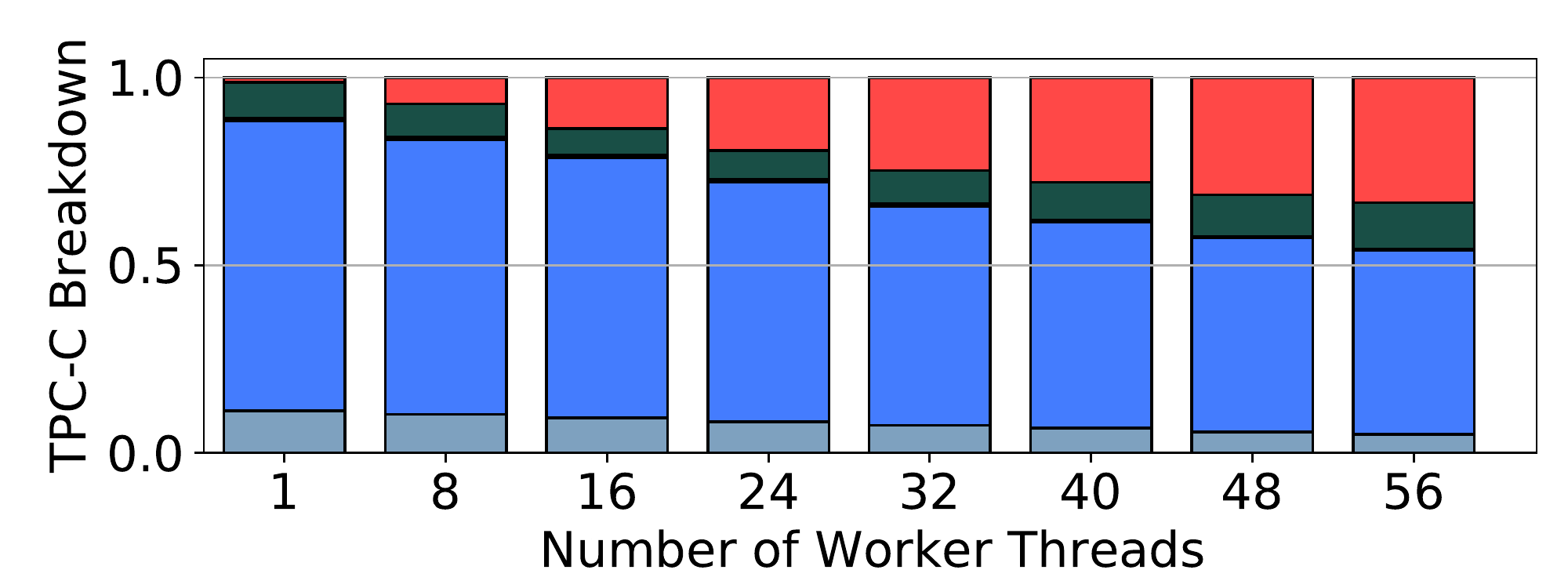}
            \label{fig:tpcc-full-breakdown-stats}
    \caption{\textbf{TPC-C Full Mix Time Breakdown} -- Time cost breakdown by different TPC-C transaction types.}
    \label{fig:tpcc-full-breakdown}
\end{figure}

\subsection{\myhighlightpar{Sensitivity Study}}
\label{sec:eval-sens}
In this section we show that \codename{} robustly provides relatively good
performance even when various factors change.
\vspace*{-0.1in} \\

\begin{figure}[t]
    \centering
    \hspace{-0.1in}{
        \fbox{\adjincludegraphics[scale=0.3, trim={{0.08\width} {0.0\height} 0 {0.1\height}}, clip]
             {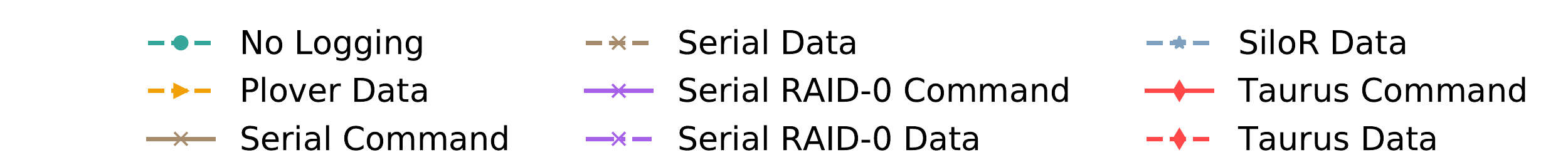}}
    }
    \hspace{-0.3in}\subfloat[YCSB-10G Logging]{
        \adjincludegraphics[scale=0.3, trim={{0.03\width} 0 0 {0.2\height}}, clip]
            {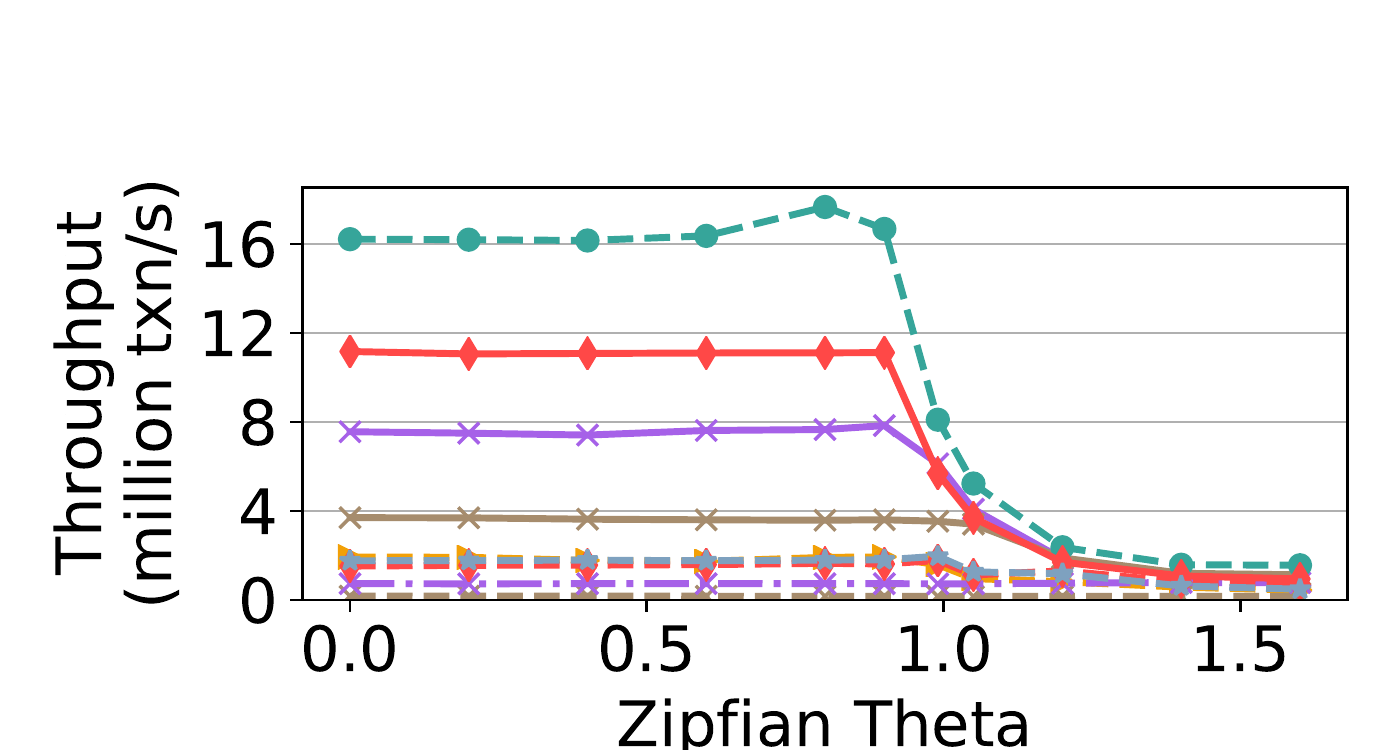}\label{fig:contention-ycsb-2pl-a}
    }
\hspace{-0.05in}\subfloat[YCSB-10G Recovery]{
        \adjincludegraphics[scale=0.3, trim={{0.17\width} 0 0 {0.2\height}}, clip]
            {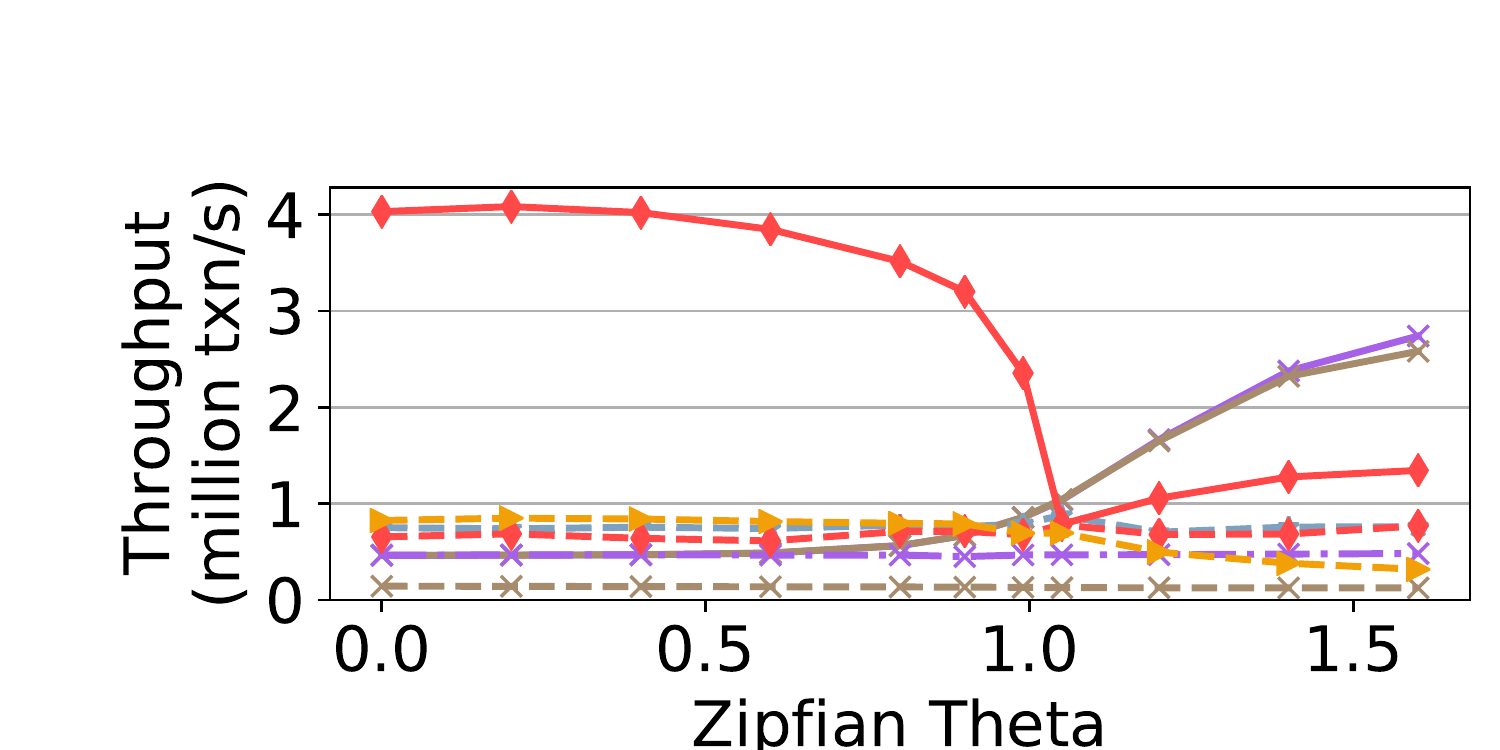}\label{fig:contention-ycsb-2pl-b}
    }
    \caption{
        \textbf{Contention} --
        Zipfian Theta in YCSB. %
    }
    \label{fig:contention-ycsb-2pl}
    \vspace{-0.05in}
\end{figure}

\textbf{Contention:}
We use the YCSB-10G benchmark on the \texttt{h1.16xlarge} server to study how the contention level 
impacts performance. We control the contention level by adjusting the $\theta$ parameter of the 
Zipfian distribution. A higher $\theta$ value corresponds to more contention. We include a 
baseline (No Logging) to provide an upper bound at different contention levels. 
Every baseline uses 56 worker threads.

\cref{fig:contention-ycsb-2pl} shows the DBMS's throughput measurements when varying the Zipfian
$\theta$ parameter for the logging and recovery procedures. 
The gap between \name and No Logging is insensitive to data contention.
When $\theta$ is greater than 1.0, the performance of all logging
schemes decreases due to the reduced parallelism in the workload.
\cref{fig:contention-ycsb-2pl-b} shows the recovery measurements when the
contention level increases. These results indicate different trends for serial algorithms
and \name. For \name, the performance drops when $\theta$ goes beyond 0.8 because
of the inter-log dependency problem (\cref{sec:limitation}): dependencies
between transactions spanning different logs incur extra latency that hurts performance at high
contention. In contrast, serial algorithms have low throughput at low contention, but
their throughput increases with higher $\theta$. This is because higher data skew
makes the working set fit in on-chip caches, resulting in a higher cache hit rate and thus
better performance. Since the recovery proceeds sequentially, contention does not introduce
data races, so it does not harm the performance of the serial baselines.
When the contention level is high (i.e., $\theta > 1$), we run \name with serial recovery to
avoid the high latency between dependent transactions. As shown in
\cref{fig:contention-ycsb-2pl-b}, this configuration enables \name to achieve good performance 
under high contention.
\vspace*{-0.1in} \\

\begin{figure}[t]
    \centering
    \hspace{0.3in}{
        \fbox{\adjincludegraphics[scale=0.3, trim={{0.2\width} {0.55\height} 0 {0.07\height}}, clip]
            {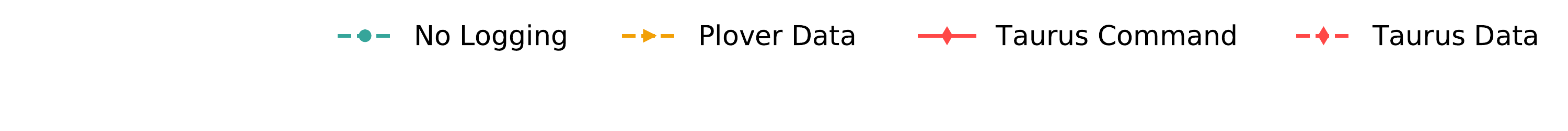}}
    }
    \vspace{-0.in}
    \adjincludegraphics[scale=\globalGraphScale, trim={0 0 0 {0.2\height}}, clip]
            {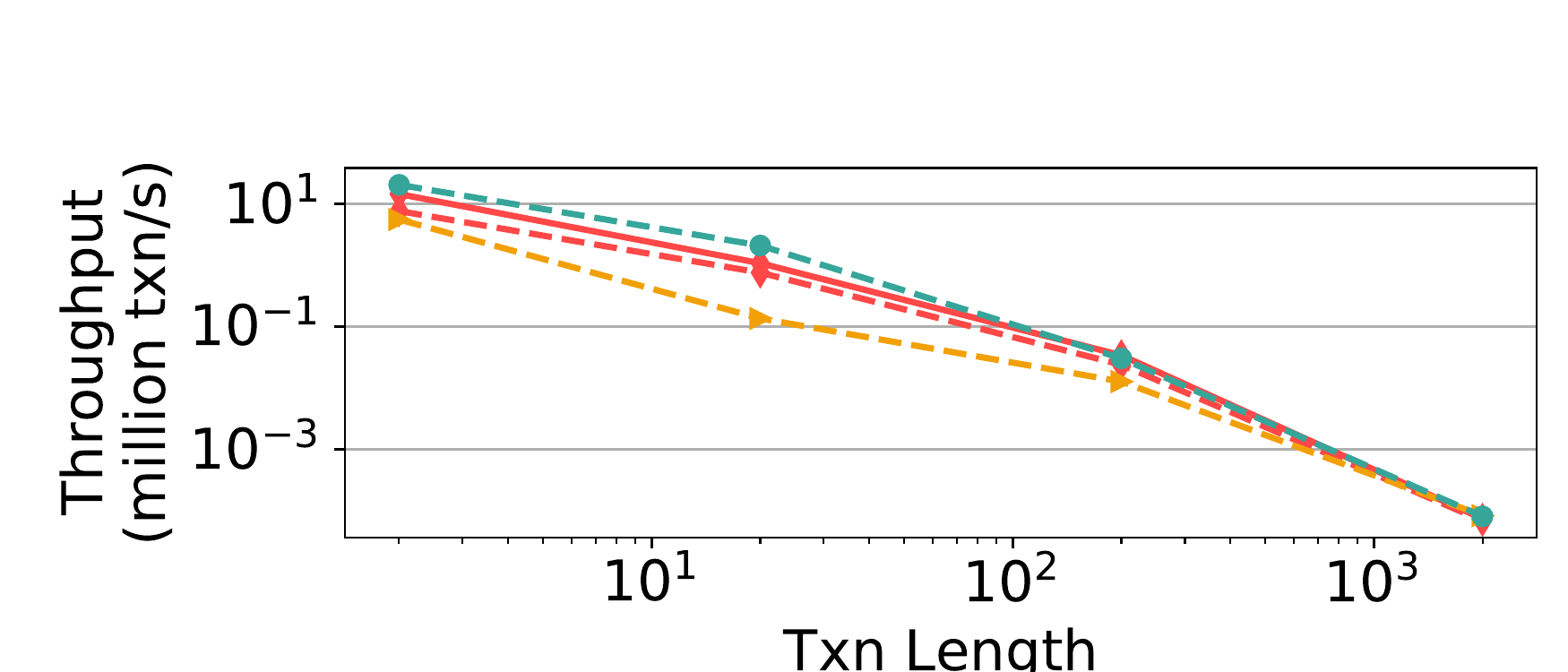}
            
    \caption{\textbf{Transaction Impact} -- We vary the number of tuples per transaction touches from 2 to 2000. %
    }
    \label{fig:sweepr-2pl}
\end{figure}

\begin{figure}[t]
    \centering
    \hspace{-0.0in}{
        \fbox{\adjincludegraphics[scale=0.3, trim={{0.313\width} {-0.05\height} {0.01\width} {0.1\height}}, clip]
            {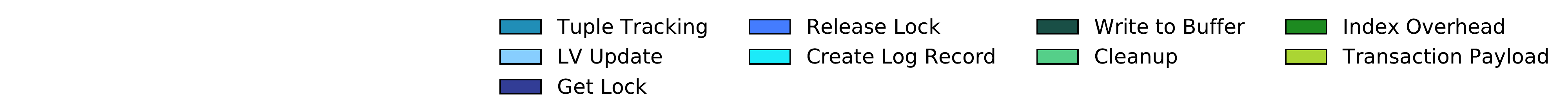}}
    }\vspace{-0.0in}
    \newline
    \adjincludegraphics[scale=0.3, trim={0 {0.05\height} 0 {0.05\height}}, clip]
            {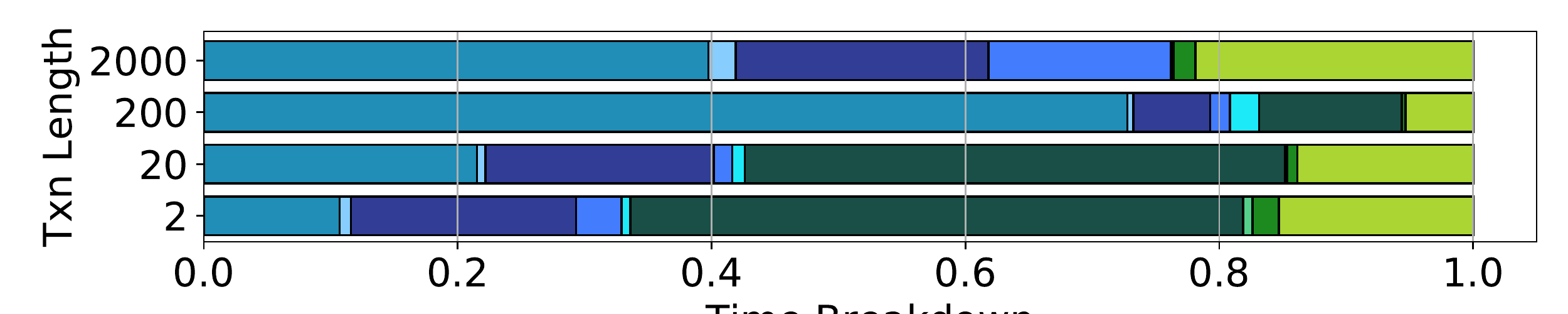}
            \label{fig:sens-sweepR-breakdown}
    \caption{\textbf{Transaction Impact} -- The percentages of the tuple tracking overhead for \name{} Data Logging.}
    \label{fig:sens-sweepR-breakdown}
\end{figure}

\textbf{Transaction Impact}
\label{sec:eval-transaction-impact}
We evaluate how \name{} performs when every transaction touches a large number of 
tuples. We ran YCSB-500G on an Amazon EC2 \texttt{i3en.metal} instance and vary the number 
of tuples every transaction accesses from 2 to 2,000. \cref{fig:sweepr-2pl} shows the throughput is 
inversely proportional with the transaction length.
\cref{fig:sens-sweepR-breakdown} shows the time breakdown of \name data logging for YCSB-10G. We can 
observe that when the number of tuples accessed per transaction increases from 2 to 200, the LV 
update overhead stays fixed at 0.6\%, while the tuple tracking overhead of the 2PL 
implementation increases from 10.7\% to 72.8\%. With the \texttt{NO\_WAIT} policy~\cite{yu2014} to 
prevent deadlocks, the abort rate grows quickly with the number of tuples accessed. At 2000 tuples per 
transaction, the abort rate is high, causing the overhead distribution to change greatly because 
overheads grow differently. Some overheads like writing the log buffer only occur after a 
transaction finishes, some overheads like tuple tracking occurs linearly in the number of tuples 
accessed, and some overheads like getting the lock are more sensitive to the growing contention. At 
2,000 tuples per transaction, the LV updating overhead is around 2.1\%.
\vspace*{-0.1in} \\

\begin{figure}[t]
    \centering
    \hspace{0.3in}{
        \fbox{\adjincludegraphics[scale=0.3, trim={{0.05\width} {0.2\height} 0 {0.1\height}}, clip]
            {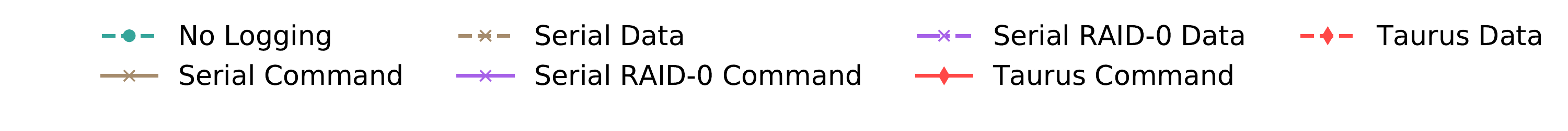}}
    }
    \hspace{-0.1in}\subfloat[Logging]{
        \adjincludegraphics[scale=0.32, trim={0 0 {0.015\width} {0.2\height}}, clip]
            {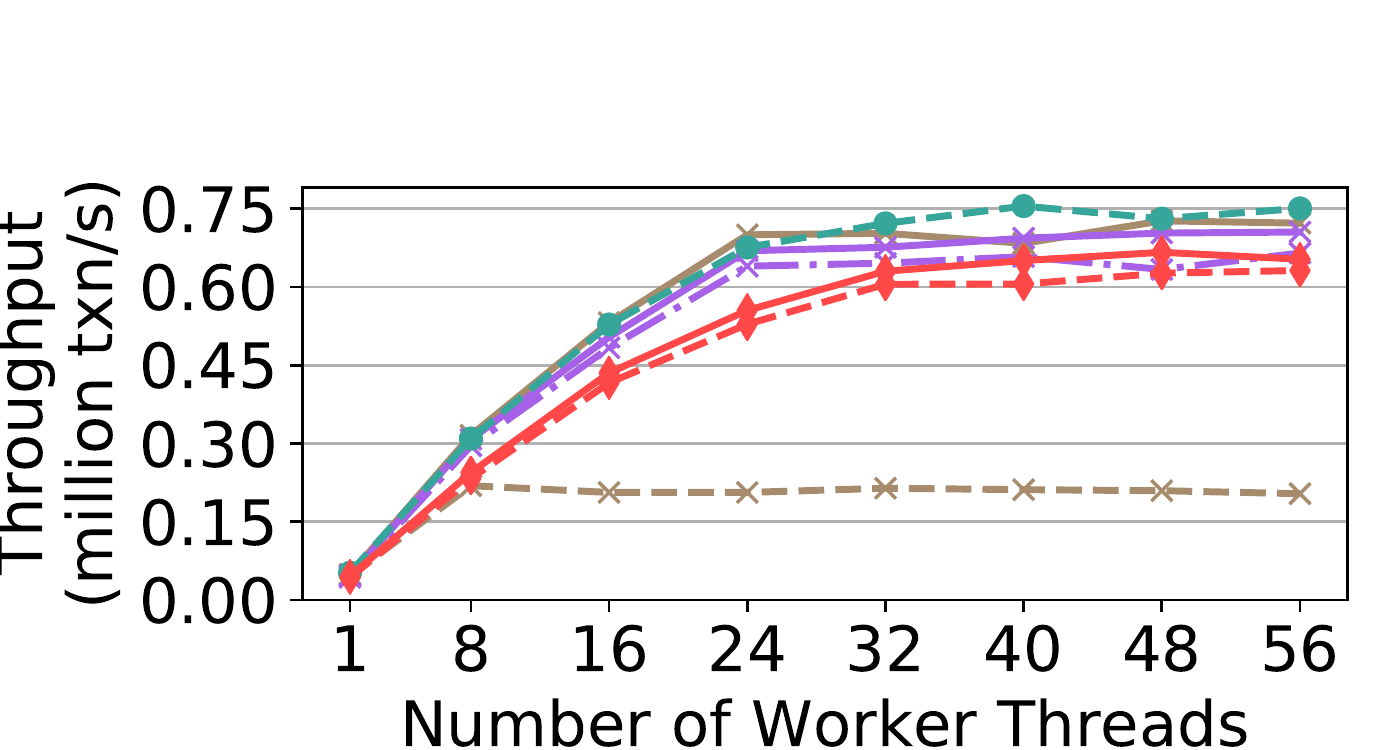}
            \label{fig:tpcc-full-2pl-logging}
    }
    \subfloat[Recovery]{
        \adjincludegraphics[scale=0.32, trim={{0.13\width} 0 0 {0.2\height}}, clip]
            {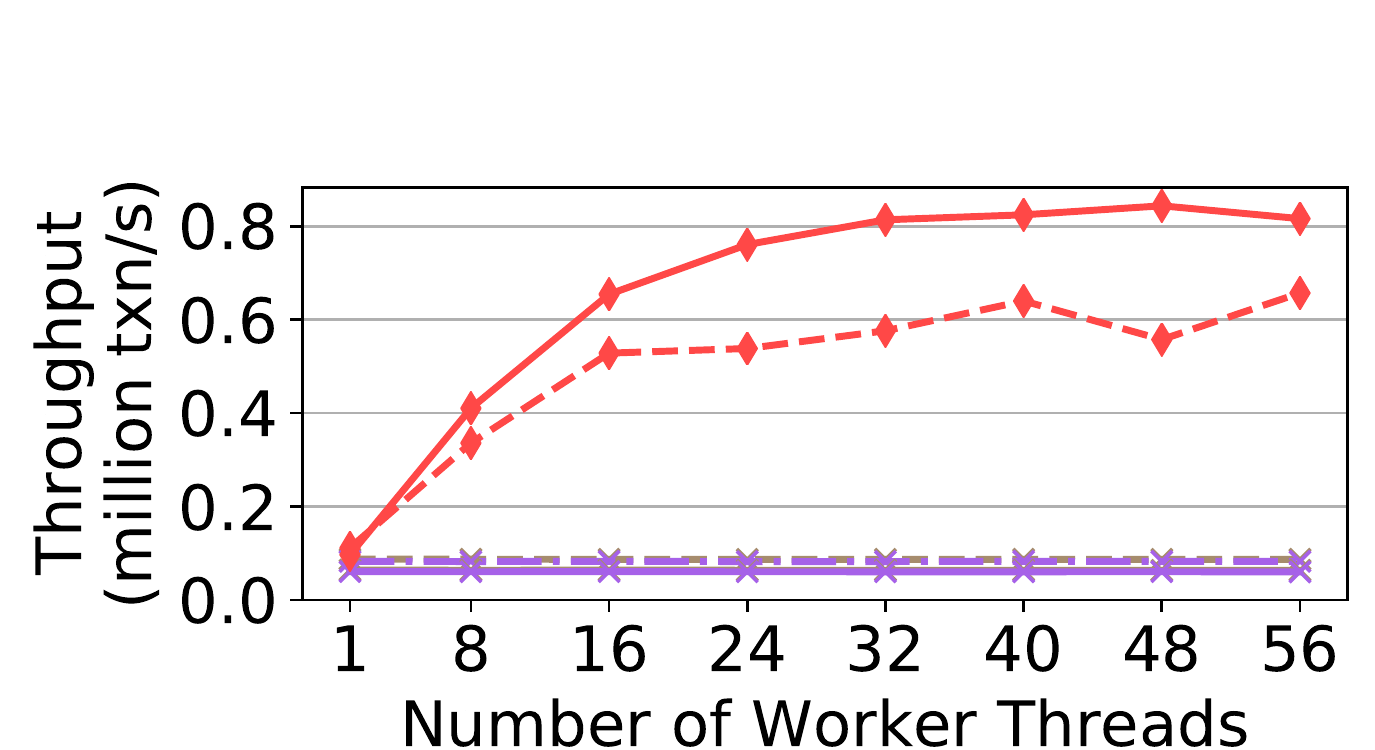}
            \label{fig:tpcc-full-2pl-recovery}
    }
    \caption{\textbf{TPC-C Full Mix Performance (2PL)} -- Performance comparison and scalability on TPC-C Full Mix.} %
    \label{fig:tpcc-full-2pl}
\end{figure}

\begin{figure}[t]
    \centering
    \hspace{0.25in}{
        \fbox{\adjincludegraphics[scale=0.3, trim={{-0.025\width} {0.55\height} {0.02\width} {0.1\height}}, clip]
            {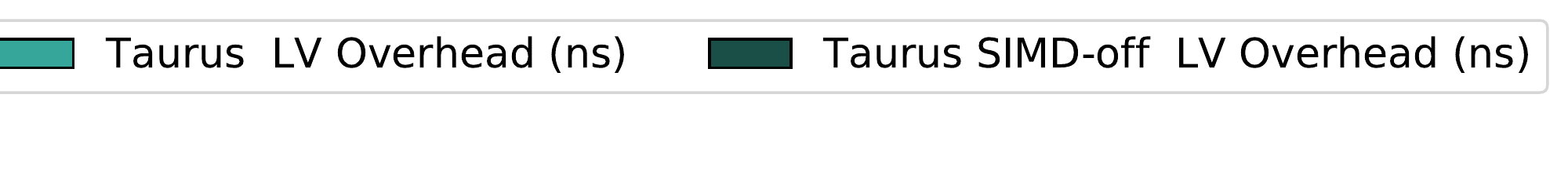}}
    }\vspace{-0.0in}
    \newline
    \adjincludegraphics[scale=0.27, trim={0 0 0 {0.05\height}}, clip]
            {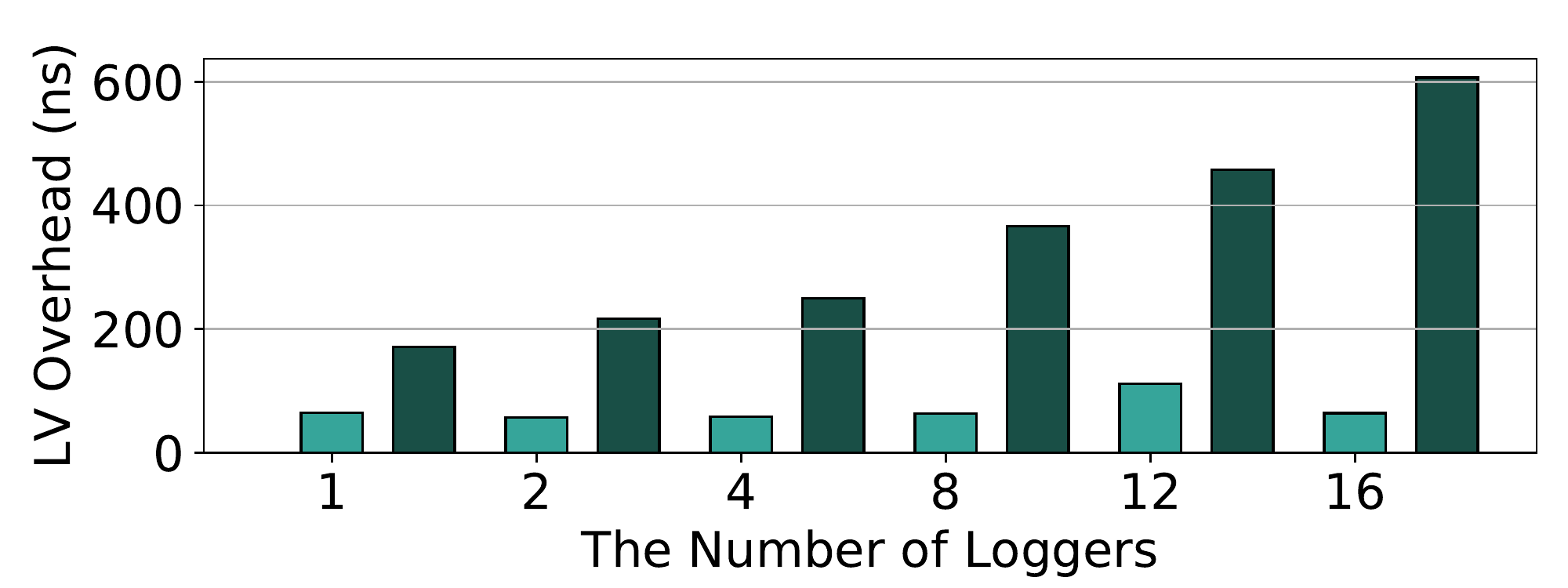}
            \label{fig:sens-log-num}
    
    \caption{\textbf{Vectorization} -- The LV Update overhead (in nanosec).}
    \label{fig:sens-log-num}
\end{figure}

\textbf{Number of Log Files}
\label{sec:eval-log-num}
We also evaluate the effectiveness of the SIMD optimizations. %
We run \name{} command logging with SIMD on and off against the YCSB-10G workload with 64 threads. The results are shown in \cref{fig:sens-log-num}. The x-axis is the number of log files \name{} used, and the y-axis represents the time (in nanoseconds) consumed in the LV operations and updates per transaction. The gap between the two baselines increases with the number of log files. %
Turning on SIMD reduces the overhead by up to 89.5\%.

%% file: related.tex
\section{Related Work}
\label{sec:related}

\textbf{\myhighlight{Early-lock-release (ELR):}} ELR \cite{dewitt1984implementation, graefe2013controlled, kimura2012efficient, 
soisalon1995partial} allows a transaction to release locks before flushing to log files. Controlled 
Lock Violation \cite{graefe2016controlled} is similar. It permits the acquisition of locks if the lock
holder has appended its log record into the buffer. \name{} includes ELR in its design for high 
performance. 
\vspace*{-0.1in} \\

\textbf{\myhighlight{Single-Storage Logging Algorithms:}}
ARIES \cite{mohan1992aries} has been the gold standard in database logging and has been widely
implemented. However, ARIES does not scale well on multicore processors, as many recent works
have observed \cite{johnson2010aether, tu13, wang2014using, zheng2014fast}. C-ARIES 
\cite{speer2007c} was proposed to support parallel recovery, and CTR~\cite{antonopoulos2019constant} improves the
recovery time by using multi-versioned concurrency control and aggressive checkpointing, but the contention caused by the
original ARIES logging remains. %

Aether \cite{johnson2010aether}, ELEDA \cite{jung2017scalable}, and \textsc{Border-Collie} 
\cite{kim2019border} have optimized ARIES by reducing the length of critical sections during
logging. But these protocols still use a single storage device and suffer from the centralized LSN 
bottleneck. TwinBuf~\cite{meng2018twin} uses two log buffers to support parallel buffer filling. 
Besides the single storage bottleneck, TwinBuf relies on global timestamps to order the log records. 
Aether, ELEDA, \textsc{Border-Collie}, and TwinBuf are similar to the serial data baseline we 
evaluated in \cref{sec:evaluation}.
\vspace*{-0.1in} \\

\textbf{\myhighlight{Single-Stream Parallel Logging Algorithms:}}
P-WAL~\cite{nakamura2019integration} realizes parallel logging but relies on a single LSN counter to
order transactions, which will incur scalability issues. Besides, the enforced order causes serial
recovery. Adaptive logging \cite{yao2016adaptive} achieves parallel recovery for command
logging in a distributed partitioned database. Different from \name{},
it infers dependency information from the transactions' read/write set. This approach
requires that the DBMS maintain each transaction's start and end times to detect dependencies. 
PACMAN \cite{wu2017fast} enables parallel command logging recovery by using program analysis
techniques to determine what computation can be performed in parallel, while \name{} does not
require any program analysis and is simpler to implement. Also, \name{} supports both parallel 
logging and recovery, while \cite{wu2017fast} only supports parallel recovery.
\vspace*{-0.1in} \\

\textbf{\myhighlight{Logging Algorithms for Modern Hardware:}} Fast recovery based on non-volatile memory (NVM) is an active research area~\cite{arulraj2015let, 
arulraj2016write, chatzistergiou2015rewind, fang2011high, huang2014nvram, wang2014scalable, 
kimura2015foedus, kim2016nvwal}. This line of work leverages the high read/write bandwidth and
byte-addressable nature of NVM to improve the logging and recovery performance. \name{}, in
contrast, can be applied to both traditional HDD/SSD devices and the new NVM devices. Since NVM 
devices are randomly accessible, \name{} can work with multiple log files per NVM device.
\vspace*{-0.1in} \\

\textbf{\myhighlight{Dependency-Tracking Algorithms:}}
Similar to \name{}, \cite{dewitt1984implementation} also uses dependency tracking to log to multiple 
log files, but does not log dependency information to the log records. This leads to two
shortcomings: (1) transactions with dependencies have to be logged in order, which leads to
significant performance overhead when there are many inter-log dependencies; (2) they do not support
parallel recovery. DistDGCC~\cite{yao2018scaling} is also coupled with a dependency tracking logging
scheme, but it logs fine-grained dependency graphs 
. A transaction
visiting lots of items results in a huge log record, which may potentially incur scalability issues. 
In contrast, \name{} only logs LSN Vectors to enforce dependencies. 
\myhighlightpar{In~\cite{johnson2012scalability}, Johnson et al. proposed a parallel logging scheme that relies on a single-dimension Lamport clock to achieve a global total order of transactions. \name{} uses multi-dimension vector clocks and only preserves partial orders between dependent transactions, enabling moderate parallelism in recovery. Enforcing a total order can accelerate the recovery if the inherent parallelism is significantly low, where a serial recovery is expected to outperform parallel recovery because of the extra inter-thread communications. \name{} provides a serial recovery fallback to fit low-parallelism workloads.
}

Kuafu \cite{hong2013kuafu} presents an algorithm for replaying transactions in parallel on a 
secondary database. Similar to \name{}, Kuafu also encodes dependency information in the log to
replay transactions in parallel. Kuafu supports data logging but not command logging, and it
maintains the whole dependency graph while \name{} maintains only minimal dependency information. 
Therefore, Kuafu will suffer from a bandwidth bottleneck if applied to our setting.
\vspace*{-0.1in} \\

\textbf{\myhighlight{Multi-Partition Logging:}}
Bernstein et al. present a logging algorithm \cite{bernstein2015scaling} for multi-partition databases. They 
distinguish transactions that only visit a single partition from those that visit multiple partitions.
These single-partition transactions are sent to the log file corresponding to the partition (called a single-partition log) as no cross-partition dependencies will occur.
All the transactions that visit multiple partitions are sent to a single log file (called the multi-partition log).
Their design also uses vector clocks%
---two-dimensional
vector clocks are maintained by each log to keep a partial order between the multi-partition
log and the corresponding single-partition log. 
Given a partitioning scheme, picking out transactions without cross-log dependencies is orthogonal
to the problem we solve here; \name{} can be plugged in to better deal with multi-partition transactions by enabling
multiple-stream logging.

%% file: conclusion.tex
\section{Conclusion}
\label{sec:conclusion}
We presented \name, a lightweight parallel logging scheme for high-throughput main memory 
DBMSs. It is designed to support not only data logging but also command logging, and is compatible
with multiple concurrency control algorithms. It is both efficient and
scalable compared to state-of-the-art logging algorithms. %

%% file: appendix.tex
\input{proof}

\section{Correctness Proof of LV Compression}
\label{sec:correctness-compression}
We next prove that the theorems in \cref{sec:proof} still hold with the two optimizations discussed 
above. Both optimizations share the same basic idea: if a \LV is too small, it suffices to store an 
upper bound of it, which can be shared by multiple \LVs. Therefore, for a transaction $\txn{}$ with \LV 
before the optimizations and $\LV{}'$ after the optimizations, we must have $\LV{}' \geq \LV$. The 
optimizations will not affect the correctness of \cref{thm:correctness-enforce} since it does not 
violate existing dependencies. \cref{thm:all-commit} is also not affected since $\LV{}'$ will never 
exceed \PLV and thus \ELV as well. \cref{thm:liveness-1} is not affected by the second optimization; 
for the first optimization, it will not block the asynchronous commit of transactions since an 
increased \LV of a transaction is not higher than the current \PLV. Finally, \cref{thm:liveness-2} 
is not as straightforward --- since the optimizations make transactions depend on more transactions, 
it may potentially create cycles in the dependency graph. In the following, we will prove that 
although it increases \LVs, it only increases them to the point that the dependencies still follow 
the real-time order. Therefore, following the proof of \cref{thm:liveness-2}, no cycle may exist.

For the first optimization, the DBMS may copy the tuple's \LV{} from \PLV. For the second 
optimization, the DBMS increases some dimensions of a log record's \LV to the corresponding values 
of the \PLV. In both cases, the log record of a transaction $\txn{}$ will have an \LV no greater than the 
current \PLV when the DBMS writes $\txn{}$'s log records to the buffer. Therefore, for each $\txn{}'$ 
in \logger{i} s.t. $\txn{}'.LSN \leq \txn{}.\LV[i]$, we still have $ct(\txn{})$ > $ct(\txn{}')$.

As an intuition, the increment of a tuple's \LV is equivalent to having a dummy transaction that 
reads a tuple and later writes the same values back to it. Similarly, the increment of a 
transaction's \LV is equivalent to having the transaction visit a dummy tuple with \PLV as its \LV. 
Inserting such dummy transactions or visiting dummy tuples only causes \LVs to artificially increase 
in the same way. These operations do not affect the correctness of the database, because the same tuple 
accesses and transaction interleavings form valid runtime events for the un-optimized 
\codename{} described in \cref{sec:protocol}, and it can handle them by the proofs in 
\cref{sec:proof}.

\section{Evaluation of the LV Compression}
\label{sec:eval-compression}

\begin{figure}[h]
    \centering
    \hspace{1in}{
        \adjincludegraphics[scale=0.4, trim={{0.01\width} {0.5\height} 0 0}, clip]
            {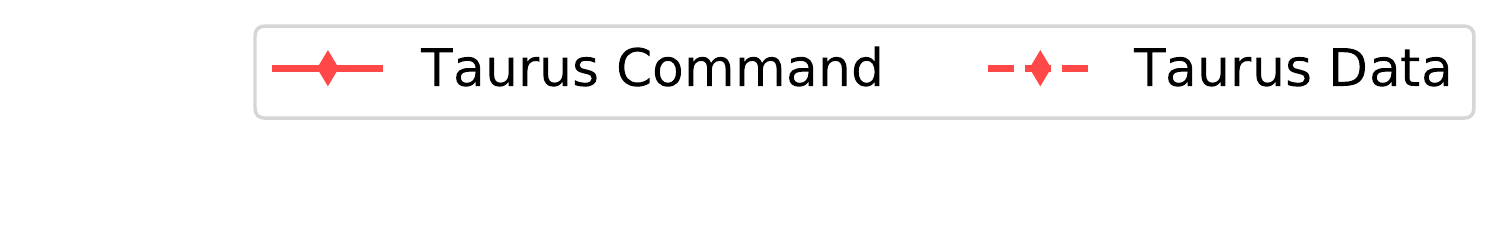}
    }
    \newline
    \subfloat[YCSB Logging]{
        \adjincludegraphics[scale=0.35, trim={0 0 0 {0.23\height}},clip]
            {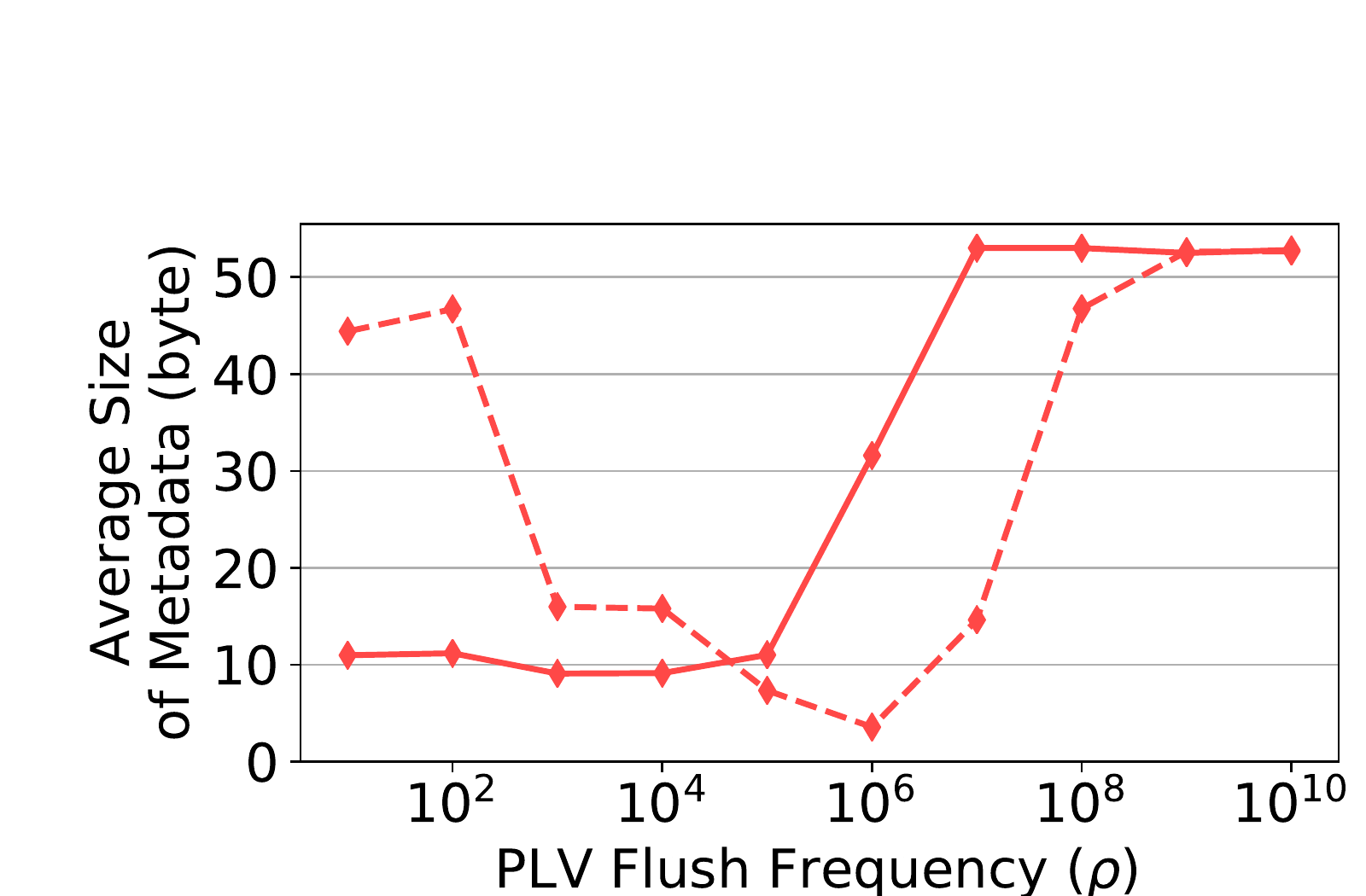}
        \label{fig:compression-ycsb-2pl-a}
    }\\
    \subfloat[YCSB Recovery]{
        \adjincludegraphics[scale=0.35, trim={0 0 0 {0.23\height}}, clip]
            {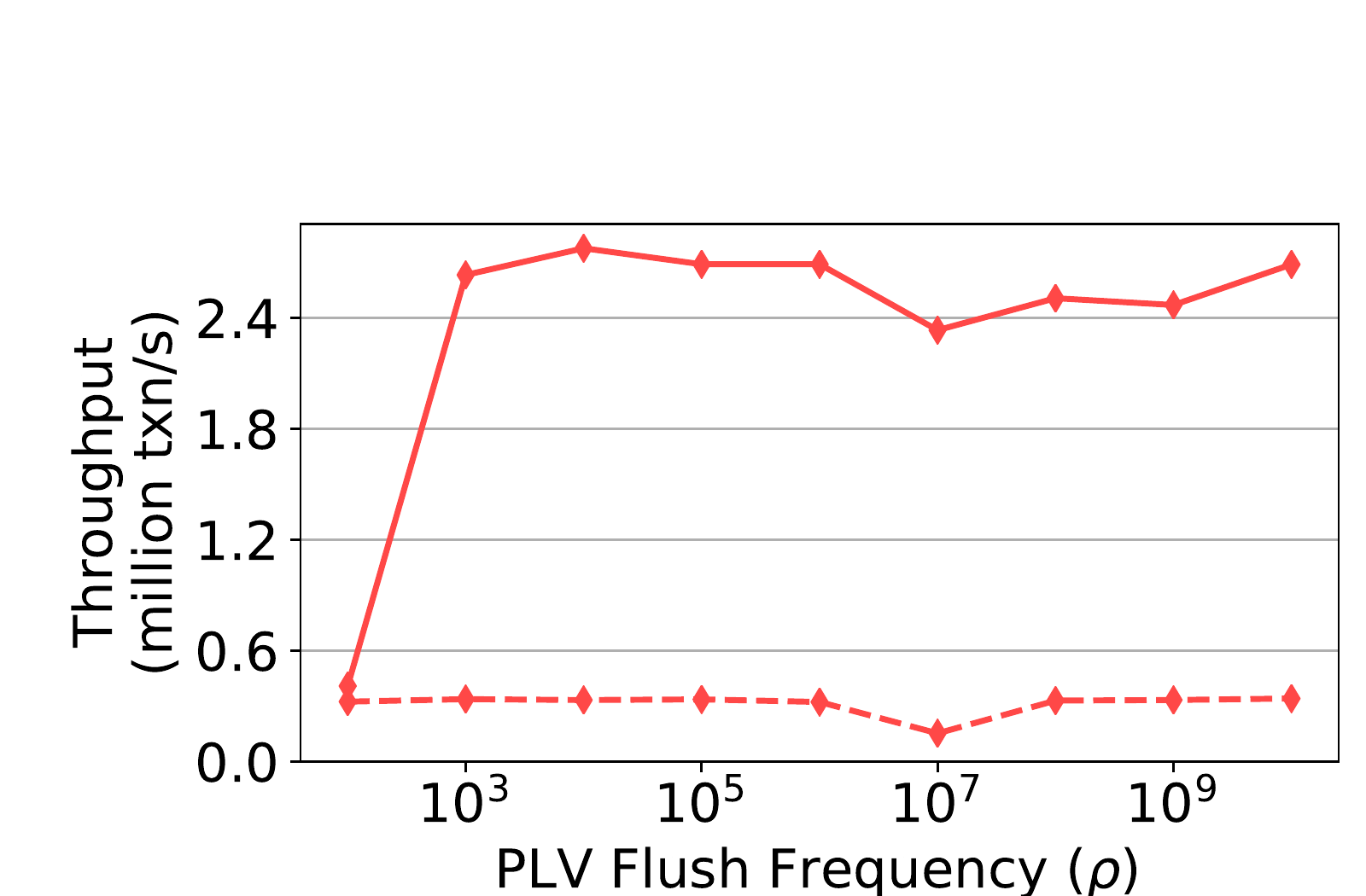}
            \label{fig:compression-ycsb-2pl-b}
    }\\
    \subfloat[YCSB Logging]{
        \adjincludegraphics[scale=0.35, trim={0 0 0 {0.23\height}}, clip]
{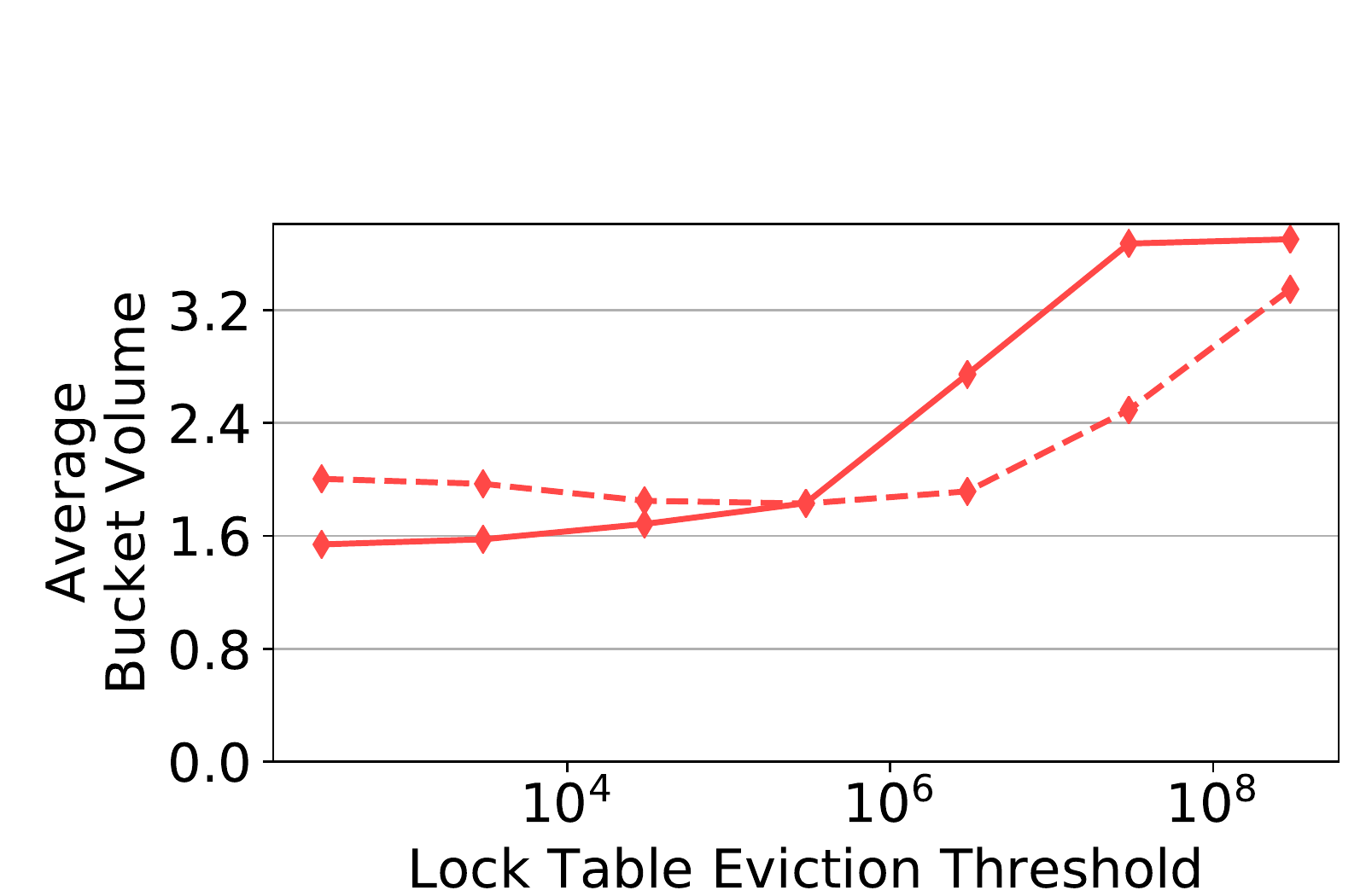}
        \label{fig:compression-ycsb-2pl-c}
    }
    \caption{
        \textbf{LV Compression} --
        (a) PLV Flush Frequency ($\rho$) vs. the average size of metadata in a high-contention workload;
        (b) PLV Flush Frequency ($\rho$) vs. the recovery throughput; and
        (c) lock table eviction threshold vs. the average bucket volume.
    }
    \label{fig:compression-ycsb-2pl}
\end{figure}

We next evaluate the scalability optimizations that we presented in \cref{sec:lv-compression} with a high-contention workload.
We use a single byte to denote the number of elements in the
\textit{compressedLV}, and a 64-bit integer for each element. Without compression, the LV would take 64 bytes to
store the eight 64-bit integers. %
We test the effect of the compression with long YCSB transactions (each visiting 16 rows).
We measure how the amount of metadata %
is affected when adjusting the flush frequency. We vary the \emph{PLV Flush
Frequency} $\rho$ through a large range. %
The DBMS writes a \PLV anchor for
every $\rho$ bytes of the log. %
The DBMS tracks both the LSN Vectors of transactions and the \PLVs when
computing the metadata size. 

From \cref{fig:compression-ycsb-2pl-a}, we observe that even for a high-contention workload, when we set $\rho$
appropriately, on average the DBMS only has to write $3.5$ bytes of metadata per log record for \name{} data logging, 
and $9.1$ bytes per log record for \name{} command logging.
When $\rho$ is small, the average metadata size can be significant because the DBMS flushes too many PLVs, causing extra overhead amortized
on each record.
As \emph{PLV Flush Frequency} grows beyond $10^6$, the size of the metadata
increases and finally reaches a steady value. 
The \name data logging curve is ``right-shifted'' compared to that of \name{} command logging 
because a data logging record is about 26$\times$ larger than a command logging record.
Therefore, the same $\rho$ results in more frequent \PLV flushes for command logging than data
logging. We see in \cref{fig:compression-ycsb-2pl-b} that 
a larger $\rho$ tends to bring better recovery performance\footnote{\small To fully
exploit the recovery parallelism, we use more threads in recovery than in logging: 56 workers in
recovery and 8 workers in logging.} because the \PLVs are flushed less frequently, resulting in
fewer artificial dependencies. This is not observed in \name{} data logging as the recovery is bounded by
the I/O bandwidth.

We also examine how the DBMS's lock table eviction threshold $\delta$ affects the average volume
of the lock table buckets. The results in \cref{fig:compression-ycsb-2pl-c} indicate
that a larger $\delta$ threshold results in a larger lock table.
Qualitatively, when $\delta$ is large, we expect \name to perform better during the recovery because
the DBMS enforces fewer extra dependencies. We contend that such a difference is only possible when
the recovery manager is not the bottleneck, and the underlying workload has a moderate level of
contention.

%% file: proof.tex
\section{Proof of Correctness \& Liveness}
\label{sec:proof}

In this section, we prove both the correctness and liveness of the \name protocol. Specifically, we 
prove that (1) data dependencies are correctly enforced during recovery; (2) all and only committed 
transactions will be recovered; and (3) the protocol will not deadlock or livelock during forward 
execution and recovery if the concurrency control protocol does not deadlock or livelock. 
We use \textit{\txn{}.LSN} to denote $\txn{}$'s position in the corresponding log.%

\begin{theorem}\label{thm:correctness-enforce}
\textbf{[Correctness of Recovery Order]} Data dependencies are correctly enforced during recovery: 
for any transaction $\txn{2}$ that depends on $\txn{1}$, 
$\txn{2}$ is recovered after $\txn{1}$.
\end{theorem}

\begin{proof}
    We prove the theorem in two steps. First, we prove that for any transaction $\txn{2}$ that depends 
    on $\txn{1}$, \name enforces that $\txn{2}.\LV[i] \geq \txn{1}.\textit{LSN}$, where $\txn{1}$ logs to the 
    $i$-th log manager (i.e., $L_i$). Then, we prove that if $\txn{2}.\LV[i] \geq \txn{1}.\textit{LSN}$, 
    $\txn{2}$ will be recovered after $\txn{1}$.

    \textbf{Step 1:} W.l.o.g., we consider a \conflictRAW dependency where $\txn{1}$ writes to tuple 
    $A$ and then $\txn{2}$ reads $A$ (proofs for write-after-read or write-after-write dependencies are 
    similar). According to \cref{line:write-buffer,line:lsn-update,line:update-writelv} in 
    \cref{alg:logging}, $A.\writeLV[i] = \txn{1}.\textit{LSN}$ when $\txn{1}$ releases its write lock 
    on $A$. When $\txn{2}$ reads $A$ at a later time, \cref{line:element-wise-max-1} in 
    \cref{alg:logging} enforces that $\txn{2}.\LV[i] \geq A.\writeLV[i]$. Since \writeLV can only 
    monotonically increase, we have $\txn{2}.\LV[i] \geq \txn{1}.\textit{LSN}$.

    \textbf{Step 2:} During recovery, $\txn{2}$ can be recovered only if $\txn{2}.\LV \leq \RLV$, which 
    means $\txn{2}.\LV[i] \leq \RLV[i]$ (\cref{line:pool-fetch-next} in \cref{alg:recovery-worker}). 
    Given $\txn{1}.\textit{LSN} \leq \txn{2}.\LV[i]$, the recovery of $\txn{2}$ requires $\txn{1}.\textit{LSN} \leq 
    \RLV[i]$. According to Lines~\ref{line:rlv-update-begin}--\ref{line:update-rlv-by-head} in 
    \cref{alg:recovery-worker}, all transactions in $L_i$ with LSNs no greater than \RLV{}$[i]$ 
    have been recovered. Therefore, $\txn{1}.\textit{LSN} \leq \RLV{}[i]$ means $\txn{1}$ is already 
    recovered. The recovery of $\txn{2}$ requires the recovery of $\txn{1}$.%
\end{proof}

\begin{theorem}\label{thm:all-commit}
    \textbf{[Correctness of Completeness]} All and only committed transactions will be recovered.
\end{theorem}

\begin{proof}
    Given that \name \textit{in-order} commits transactions that map to the same log manager 
    (\cref{line:commit-end} in \cref{alg:logging}), we must find the last committed 
    transaction $\txn{}$ and prove that $\txn{}$ and all transactions before $\txn{}$ are recovered and no transaction 
    after $\txn{}$ is recovered. 

    Based on \cref{line:decode-next} in \cref{alg:log-manager-recovery}, the transaction $\txn{}$ that we 
    are looking for is the transaction right before the first transaction $\txn{}'$ that violates 
    $\txn{}'.\LV \leq \ELV$. If no such $\txn{}'$ exists, we choose $\txn{}$ as the last transaction in the log 
    file. Given this, we prove the following three properties:

    \textbf{Property \#1:} $\txn{}$ is the last committed transaction before the crash.
    Since $\txn{}$ is the transaction right before the first transaction $\txn{}'$ violating
    $\txn{}'.\LV \leq \ELV$, all transactions before $\txn{}'$ satisfy this inequality. Therefore, before the  
    crash, for all of them we have $\txn{}.\LV \leq \PLV$. 
    By \cref{line:commit-end} in \cref{alg:logging}, these transactions have committed before the 
    crash.
    We then prove that $\txn{}'$ did not commit before the crash. There are two cases to consider. If 
    $\txn{}'$ does not exist in the log file, then $\txn{}'$ was not persistent and thus never committed 
    before the crash. Otherwise, $\txn{}'$ exists but $\txn{}'.\LV>\ELV$. This means during forward 
    processing, we have $\txn{}'.\LV>\PLV$, indicating that $\txn{}'$ never committed before the crash. 

    \textbf{Property \#2:} All transactions before $\txn{}'$ will be recovered.
    During recovery, $\txn{}$ and transactions before $\txn{}$ form a directed dependency graph 
    where each element in an LV indicates an edge in the graph. Later, in \cref{thm:liveness-2}, we 
    prove that the dependency graph is acyclic, which means that each transaction will 
    be recovered. 

    \textbf{Property \#3:} No transaction after $\txn{}$ will be recovered. 
    By \cref{line:decode-next} in \cref{alg:log-manager-recovery}, all transactions after $\txn{}$ 
    are ignored and will not be recovered. 
\end{proof}    

\begin{theorem}
\textbf{[Liveness in Forward Processing]} The protocol will not deadlock or livelock during forward processing if the concurrency control protocol does not deadlock or livelock. \label{thm:liveness-1}
\end{theorem}
\begin{proof}

Assuming that no thread will be indefinitely suspended, we show that the system will eventually make 
progress (e.g., committing a transaction). We prove this in three parts:

First, if there are transactions in the log buffer, they will eventually be flushed to disks. The 
only possibility that data in the log buffer is not flushed is because \textit{readyLSN} is 
throttled by \textit{allocatedLSN} (\cref{line:state-check-end} in \cref{alg:log-manager-logging}). 
This throttling can happen only if \textit{$allocatedLSN[j]$ $\geq$ $filledLSN[j]$} which means a 
transaction is in the middle of filling a log record. This process, however, completes in a short 
period of time, since all the operations between 
\cref{lines:allocate-greater-than-filled-begin,lines:allocate-greater-than-filled-end} in 
\cref{alg:logging} are non-blocking.

Second, active transactions will eventually be written to the log buffer. 
Assuming the concurrency control algorithm does not incur deadlocks/livelocks, the logic in \cref{alg:logging} is combinational and non-blocking, and the conditional statements are independent from the concurrency control algorithm. Since a log buffer will finally flush, it will contain space for an active transaction to write to. 

Finally, an active transaction will eventually commit assuming no system failure. Since transactions in each log buffer will eventually be flushed, for each transaction $\txn{}$, \PLV will exceed $\txn{}.\LV$. This will also be true for transactions in the same log manager with smaller LSNs. 
Following \cref{line:commit-end} in \cref{alg:logging}, this means $\txn{}$ will eventually commit.%
\end{proof}
\begin{theorem}
\textbf{[Liveness in Recovery]} The protocol will not deadlock or livelock during the \emph{recovery process}. \label{thm:liveness-2}
\end{theorem}
\begin{proof}
    The recovery follows a directed dependency graph where each node is a transaction and each edge 
    corresponds to a value in \LV. Now, we prove that the graph is acyclic.

    \begin{figure}[t]
        \centering
        \adjincludegraphics[width=0.5\columnwidth, trim={0 {0.1\height} 0 {0.05\height}}, 
            clip]{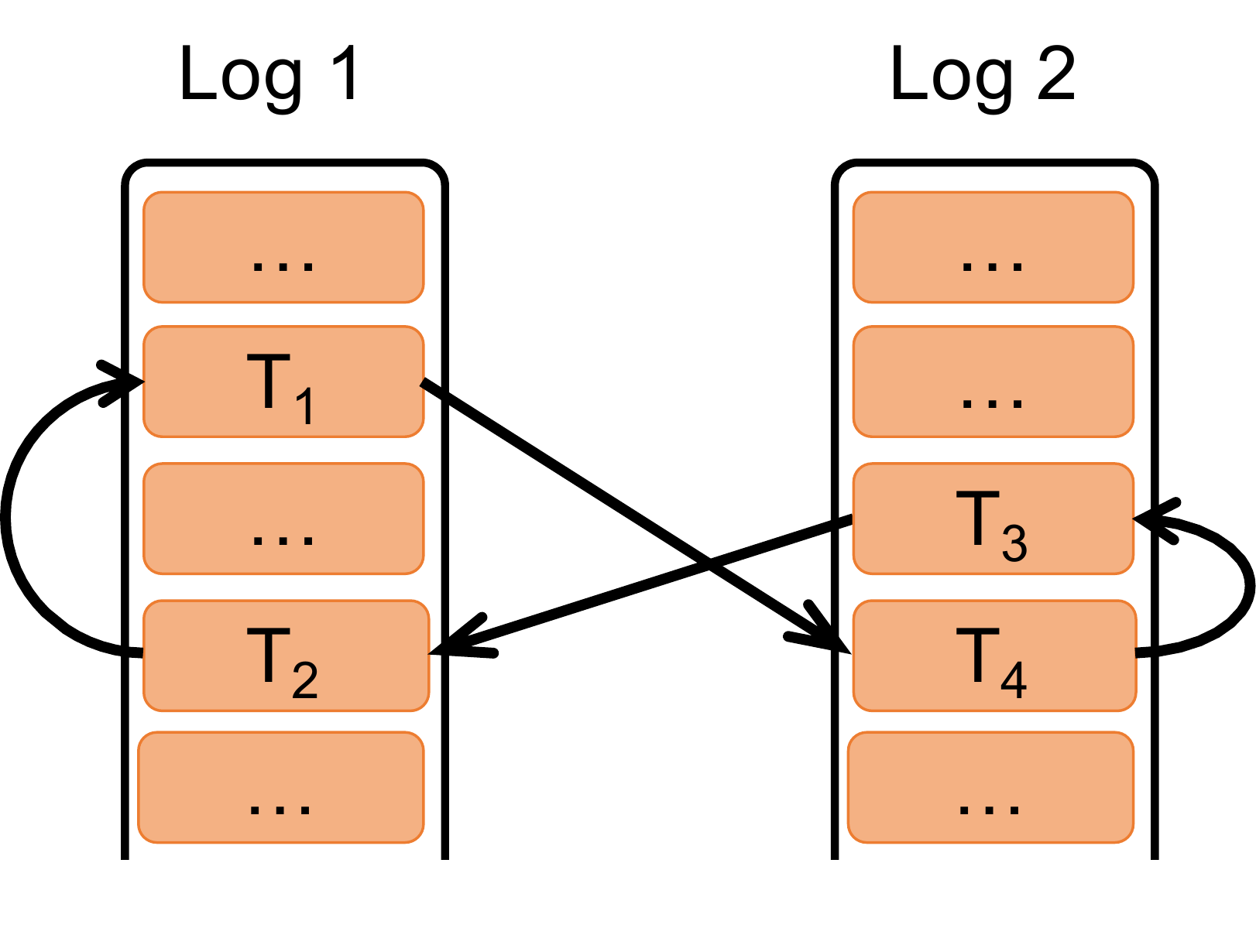}
        \caption{
            \textbf{A Hypothetical Deadlock Situation} --
            The dependencies of transactions form a cycle.
        }
        \label{fig:deadlock}
    \end{figure}

    \cref{fig:deadlock} shows a hypothetical deadlock situation that may occur if \name is 
    incorrectly designed. A cycle is formed between the four transactions: \txn{1} $\leftarrow$ 
    \txn{2} $\leftarrow$ \txn{3} $\leftarrow$ \txn{4} $\leftarrow$ \txn{1}. Although all the 
    transactions are in the \textit{pools}, none of them is able to make any forward progress.
    Note that, the correctness of the concurrency control algorithm is not sufficient to rule out
    this situation because \name{} adds extra dependencies while logging.
 
    To prove that no dependency cycles exist in \name, we define a commit time, \textit{ct(\txn{})}, 
    to each transaction $\txn{}$. We show that every edge in the graph follows the order of commit 
    time, namely, an edge $\txn{2}\!\rightarrow\!\txn{1}$ exists 
    $\Rightarrow$\textit{ct(\txn{2}) $>$ ct(\txn{1})}. Since time specifies a total 
    order, proving this inequality means cycles are impossible. In particular, we choose the time 
    when $\txn{}$'s log record is allocated in the log buffer (i.e., right after atomically incrementing 
    LSN in \cref{alg:logging}, \cref{line:atom-fetch-and-add}) as its commit time \textit{ct(\txn{})}.

    As we proved in \cref{thm:correctness-enforce}, \txn{2} depending on \txn{1} mapped to log manager $i$ $\Rightarrow$ 
    $\txn{2}.\LV{}[i]$ $\geq$ $\txn{1}.\textit{LSN}$. It 
    is clear that for transactions mapping to the same log manager, their \textit{ct} order is the 
    same as their LSN order; therefore$ \textit{ct(\txn{}') > ct(\txn{1})}$, where $\txn{}'$ is the 
    transaction on $L_i$ with $\txn{}'.LSN=\txn{2}.LV{}[i]$. According to \cref{line:lsn-update} in 
    \cref{alg:logging}, a tuple may have its \readLV{}$[i]$ or \writeLV{}$[i]$ equal a particular 
    \textit{LSN} only after the log record has been written to the log buffer at that \textit{LSN}. 
    \txn{2}.\LV{}$[i]$ must be copied from a tuple and thus must occur at a even later time, so 
    \textit{ct(\txn{2}) > ct($\txn{}'$) > ct(\txn{1})}. Therefore, if \txn{2}.\LV{}$[i]$ $\geq$ \txn{1}.\textit{LSN}, 
    we have $\textit{ct(\txn{2}) > ct(\txn{1})}$ proving the theorem.
\end{proof}